\documentclass[11pt]{article}
\usepackage[margin=1in]{geometry}
\usepackage{amsmath,amssymb,amsthm,mathtools,bm}
\usepackage{hyperref}
\usepackage{enumitem}
\usepackage{microtype}
\usepackage{xcolor}
\usepackage{booktabs}
\usepackage{graphicx}
\usepackage{tabularx,longtable,array}
\usepackage{pdflscape}
\usepackage{ragged2e}
\usepackage{siunitx}
\hypersetup{
  colorlinks=true,
  linkcolor=black,
  citecolor=black,
  urlcolor=black,
  pdftitle={Relational Quantum Causal Processes: Exact Models, Continuum Limits, and the Boundary of Emergent Gravity},
  pdfauthor={Yipeng Xu},
  pdfsubject={Relational quantum processes, emergent causal order, Lorentzian geometry, and controlled gravitational limits},
  pdfkeywords={quantum causal processes, open quantum systems, causal order, Lorentzian geometry, quantum gravity}
}
\newcolumntype{P}[1]{>{\RaggedRight\arraybackslash}p{#1}}

\definecolor{statusblue}{HTML}{0072B2}
\definecolor{statusgreen}{HTML}{009E73}
\definecolor{statusorange}{HTML}{D55E00}
\definecolor{statusgray}{HTML}{666666}
\definecolor{statusred}{HTML}{B13B35}
\newcommand{\StatusAbstract}{\textcolor{statusblue}{\textsf{proved abstractly}}}
\newcommand{\StatusModel}{\textcolor{statusgreen}{\textsf{proved in an explicit model}}}
\newcommand{\StatusNumerical}{\textcolor{statusorange}{\textsf{controlled/model-specific}}}
\newcommand{\StatusConditional}{\textcolor{statusgray}{\textsf{conditional interface}}}
\newcommand{\StatusOpen}{\textcolor{statusred}{\textsf{open/phenomenological}}}
\newcommand{\AssumptionArrow}{\ensuremath{\Longrightarrow}}

\newtheorem{theorem}{Theorem}[section]
\newtheorem{lemma}[theorem]{Lemma}
\newtheorem{proposition}[theorem]{Proposition}
\newtheorem{corollary}[theorem]{Corollary}
\newtheorem{definition}[theorem]{Definition}
\newtheorem{assumption}[theorem]{Assumption}
\newtheorem{remark}[theorem]{Remark}
\newtheorem{external}[theorem]{External input}

\newcommand{\A}{\mathcal A}
\newcommand{\B}{\mathcal B}

\newcommand{\D}{\mathcal D}

\newcommand{\F}{\mathcal F}

\newcommand{\Hh}{\mathcal H}
\newcommand{\I}{\mathfrak I}

\newcommand{\K}{\mathcal K}
\newcommand{\M}{\mathcal M}
\newcommand{\N}{\mathcal N}
\newcommand{\R}{\mathcal R}
\newcommand{\Z}{\mathcal Z}
\DeclareMathOperator{\Proj}{Proj}
\newcommand{\Tr}{\operatorname{Tr}}
\newcommand{\Str}{\operatorname{Str}}
\newcommand{\vN}{\operatorname{vN}}
\newcommand{\MD}{\operatorname{MD}}
\newcommand{\Fix}{\operatorname{Fix}}
\newcommand{\Stab}{\operatorname{Stab}}
\newcommand{\Index}{\operatorname{Index}}
\newcommand{\ch}{\operatorname{ch}}

\newcommand{\diff}{\mathrm d}
\newcommand{\Ev}{\mathrm{Ev}}

\title{Relational Quantum Causal Processes:\\
\large Exact Models, Continuum Limits, and the Boundary of Emergent Gravity}
\author{Yipeng Xu\thanks{Email: \texttt{yx488@cam.ac.uk}.}\\
\small University of Cambridge, Cambridge, United Kingdom}
\date{July 2026}

\begin{document}
\maketitle

\begin{abstract}
Relational quantum causal processes describe finite operational contexts by
normal positive functionals on local completely positive maps.  Response
differences define an influence algebra, whose central projections form a
Boolean algebra of jointly readable events.  This article develops that
operator-algebraic starting point through exact models and states explicitly
where its gravitational interpretation remains conditional.  Fresh-cell
unitary collision circuits yield dephasing--exchange kinetics, an exact
charge-center fixed algebra, a system-size-uniform finite-step limit at fixed
response order, and graph-controlled metastable Markov dynamics.  A separate
absorbing-state model produces a sharp transition between a surviving matrix
algebra and Boolean memory records, while a reversal-covariant defect
dynamics generates a locally finite partial order on a restricted graph
family without a prescribed Lyapunov time.

Given a certified order and a positive additive record measure, compactness
and identifiability conditions lead to a Lorentzian metric--measure
subsequential limit, finite error bounds, and uniqueness of admissible smooth
interpretations.  Finite regulators then test distinct parts of the
information-to-gravity problem: modular-to-boost response, null tomography,
same-update variational identities, induced quadratic gravity, and
common-refinement limits.  These statements are exact or controlled within
their declared models.  They do not yet provide one background-independent
microscopic law that simultaneously generates adjacency, time, volume
normalization, dimension and signature, nonlinear Einstein constraints, and
quantum matter.  Operator-growth rigidity and projective-history
constructions further delimit the remaining non-Abelian possibilities.
RQCP-QG is therefore presented here as a theorem-indexed research framework:
established mechanisms, conditional compositions, and open assumptions are
kept separate throughout.
\end{abstract}

\newpage
\tableofcontents

\newpage
\section{Introduction}\label{sec:introduction}

A microscopic theory of spacetime must separate operational input from
dynamical output.  In particular, causal order and metric geometry cannot be
counted as emergent if they have already entered through the laboratory
labels, a calibrated clock, or a prescribed continuum background.  The
relational quantum causal process (RQCP) program starts instead from local
interventions and their observable response relations.  Its central question
is whether classical records, causal order, geometric measure, and
gravitational response can arise as successive long-distance structures of
one quantum process.

For a finite operational context \(C\), the primitive probability object is a
normal process functional
\[
 \Omega_C(\Phi_1,\ldots,\Phi_n),
\]
positive on completely positive local maps and normalized on deterministic
normal unital completely positive maps.  The context specifies the
laboratories, their admissible interventions, boundary data, and the
background strategy class, but not a global state on a preferred time slice
or a continuum metric.  This formulation includes finite process matrices
and quantum combs as special cases \cite{OCB2012,Chiribella2009}.

Interventions outside a target algebra \(Y\) generate predual response
differences at \(Y\).  The unitaries that leave all of these differences
invariant form a stabilizer; its commutant in \(Y\) is the influence algebra
\(\I_Y^{C;\B}\).  The jointly readable exact events are
\[
 \Ev_0(Y|C;\B)=\Proj Z(\I_Y^{C;\B}).
\]
Their Boolean structure follows from centrality.  The full influence algebra
may remain noncommutative, so the construction does not identify quantum
logic with a Boolean lattice or assume that all infrared information is
classical.

This algebraic observation is only a starting point.  Complete positivity
does not imply Markovian kinetics, a commutative infrared algebra, an acyclic
response relation, a Lorentzian continuum, or an Einstein equation.  Earlier
formulations collected these requirements in a conditional synthesis
theorem.  Here they are separated and tested in explicit models.  Each
result is labelled as abstract, exact in a model, controlled or numerical,
conditional, or open.  The labels are part of the statements: conclusions
from different models are not composed unless their shared variables and
limits have been identified.

The first group of results concerns the origin of stable classical records.
A fresh-cell unitary collision circuit on a connected graph reduces exactly
to local dephasing and random exchange after the environment cells are
traced out.  Its finite-step invariant algebra is
\(\operatorname{span}\{P_0,\ldots,P_N\}\), with \(P_q\) the projector onto
total charge \(q\), and its Pauli-response hierarchy is an exact killed
colored-interchange process.  At fixed response order, a palindromic
collision round converges to the associated GKSL evolution with an error
constant independent of the total number of qubits
\cite{Lindblad1976,Gorini1976}.  For the continuous
dephasing--exchange generator, the same charge algebra is the exact fixed
observable algebra.  On a regular graph the complete operator-space gap is
\[
 \lambda_G=\min\!\left\{2\kappa,
  \frac{\gamma}{d}\lambda_2(L_G)\right\}.
\]
Weak bistochastic perturbations then induce a Markov generator on charge
observables, with a finite-time reduction error controlled by
\(\varepsilon\|\mathcal V\|/\lambda_G\).  This establishes the formation and
metastable motion of Boolean charge records without assigning them a
spacetime meaning.

A second model shows that Booleanity is dynamical rather than inevitable.
Contact-process variables gate the dephasing of a separately writable
memory \cite{Harris1974,Liggett1999}.  Every memory matrix unit acquires an
exact occupation-time Feynman--Kac factor.  In the thermodynamic late-write
limit, the surviving algebra is a full matrix algebra in the
absorbing-noise phase and a diagonal algebra in the persistent-noise phase.
For a tagged qubit, the same factor determines a commutator witness,
reference-system negativity, and the memoryless-use quantum capacity.  The
transition therefore changes an operational quantum-information resource,
not merely the notation used for fixed points.

Causal order is addressed by a reversal-covariant dynamics of binary edge
orientations on a block cactus.  Directed triangles are local cycle defects,
and their complete observable algebra has an exact contact-process quotient.
Above the healing threshold, every fixed window becomes a directed acyclic
graph with exponentially increasing probability.  Only after the defects
have disappeared is the longest-path rank
\[
 T_\omega(v)=\max_{p:\,p\to v}|p|
\]
defined; it then increases along every realized directed edge.  Thus the
Lyapunov ranking is an output on this graph family.  The undirected skeleton,
laboratory runtime, and conversion from record number to spacetime volume
remain external.

Two complementary results connect operational data to geometry.  In the
first, Ramsey ratios of a quantum field cancel local self-variances and
isolate spacelike covariances.  For a conformal scalar in the Bunch--Davies
state on \(\mathrm{dS}_4\), six generic anchors determine the curvature
radius from a Lorentzian rank defect and predict target--target invariants
not used in the fit.  This is an inverse theorem in a specified
field--state--geometry family, not an emergence theorem.  In the second, a
margin-certified partial order and a positive additive record measure define
an intrinsic finite volume clock.  Uniform strong-profile covering bounds
give a Lorentzian metric--measure subsequential limit without assuming a
manifold; when an admissible smooth strongly causal limit exists, local
interval volumes recover proper time and the order--measure data determine
the geometry up to measure-preserving isometry.  Dimension and
Alexandrov-nerve topology are inferred with explicit abstention criteria.
This result realizes an ``order and number'' reconstruction under stated
compactness and identifiability conditions
\cite{BraunSaemann2025,BraunOrderNumber2026}; it does not generate the input
order or the absolute record-density unit.

The gravitational part is deliberately divided into finite questions.  A
critical free-fermion interval supplies a fixed physical modular band and a
controlled modular-to-boost response certificate
\cite{BisognanoWichmann1975,Giudici2018}.  Affine null-wire and nonlinear
record-lattice regulators then test null tomography, finite information
balance, focusing, and Ward completion.  A same-update construction compares
Ramsey response with the variation of the same declared phase functional,
and a common-refinement model places order, integer record measure, and that
phase on one register family.  A separate same-link determinant model removes
the bare curvature action in a transverse-traceless sector and produces a
positive two-derivative kernel from matter on the same links, in the spirit
of induced gravity \cite{Sakharov1968,Adler1982,Visser2002}.  Further finite
constructions supply relational clock histories, the quadratic
Fierz--Pauli/ghost/BRST complex, and a regulator-independent Gaussian
coefficient under a derived spectral balance.  These are nonempty-model and
compatibility results.  None selects four dimensions, the Einstein phase,
the record-area conversion, interacting matter, and the nonlinear continuum
from one autonomous dynamics.

The history constructions sharpen a different boundary.  Positive
projective causal-growth laws extend to infinite past-finite histories, and
restricted non-Abelian models yield countably additive operator-valued
measures.  Conversely, central operator-valued growth reduces to mixtures of
scalar laws, while several self-adjoint and triangular square-operator
branches are rigidly commutative.  These results narrow the search for a
covariant non-Abelian growth law with genuine geometric backreaction; they
do not provide such a law on arbitrary causal sets
\cite{RideoutSorkin2000,Sorkin1994QuantumMeasure}.

The status of vacuum, matter, and cosmology is correspondingly limited.
Fixed-sector normalization removes a source-independent pure-volume factor
once an additive charge has already been identified with metric volume.
No theorem here derives that identification, the observed cosmological
constant, the Standard Model representation, three generations, or a dark
equation of state.  Their conditional constructions are retained because
they state precise mathematical targets and failure tests, but they are not
used as evidence for the earlier information-to-geometry results.

The article is organized so that this separation can be checked directly.
Section~\ref{sec:status-map} gives the status conventions and a compact map of
the principal results.  Sections~\ref{sec:v62-axiom-reduction}--%
\ref{sec:closed-loops} introduce the operational assumptions and exact model
results.  Sections~\ref{sec:process-algebra}--%
\ref{sec:generic-order-volume} develop the common algebraic and
order--volume constructions.  The subsequent sections treat finite
gravitational regulators, induced quadratic dynamics, causal histories, and
their no-go boundaries.  The appendices collect proofs and conditional
matter and cosmological extensions.  The synthesis theorem records logical
dependence only; it is not a claim that all of its hypotheses arise in a
single model.

\section{Scope and status of the principal results}
\label{sec:status-map}

The results form a dependency structure rather than a single derivation.
For each link, the manuscript identifies the assumption being used, the
statement proved in its place, and the model or observable that tests the
statement:
\begin{equation}
 \boxed{\begin{gathered}
 \text{stated assumption}
 \ \AssumptionArrow\
 \text{new theorem or controlled mechanism}\\[-1mm]
 \AssumptionArrow\
 \text{explicit model or observable}
 \end{gathered}}
 \label{eq:closed-loop-rule}
\end{equation}
An assumption is considered removed only when it is not replaced by an
equivalent calibration, projection, or geometric input.  This convention is
used throughout to separate operator-algebraic results from conditional
matter and cosmological extensions.

\subsection{Status conventions}

Every major statement below carries one of five statuses.
\begin{description}[leftmargin=3.6cm,style=nextline]
\item[\StatusAbstract.]
The conclusion follows from the stated operator-algebraic hypotheses and is
not tied to a finite numerical example.
\item[\StatusModel.]
The statement is exact for a specified finite-range channel or Lindbladian.
It is evidence that the relevant structure is dynamically realizable, not a
universality theorem for all quantum causal processes.
\item[\StatusNumerical.]
Part of the conclusion is a deterministic finite-size calculation,
finite-shot simulation, imported asymptotic coefficient, or model-specific
inverse theorem.  The data and the analytic components are identified
separately.
\item[\StatusConditional.]
The implication is mathematically valid only after an external continuum,
state, geometric, or field-theoretic identification is supplied.
\item[\StatusOpen.]
The item is a target mechanism or phenomenological branch.  It is not used as
support for any earlier stage of the foundational chain.
\end{description}

\begin{landscape}
\begin{figure}[p]
 \centering
 \includegraphics[width=0.96\linewidth]{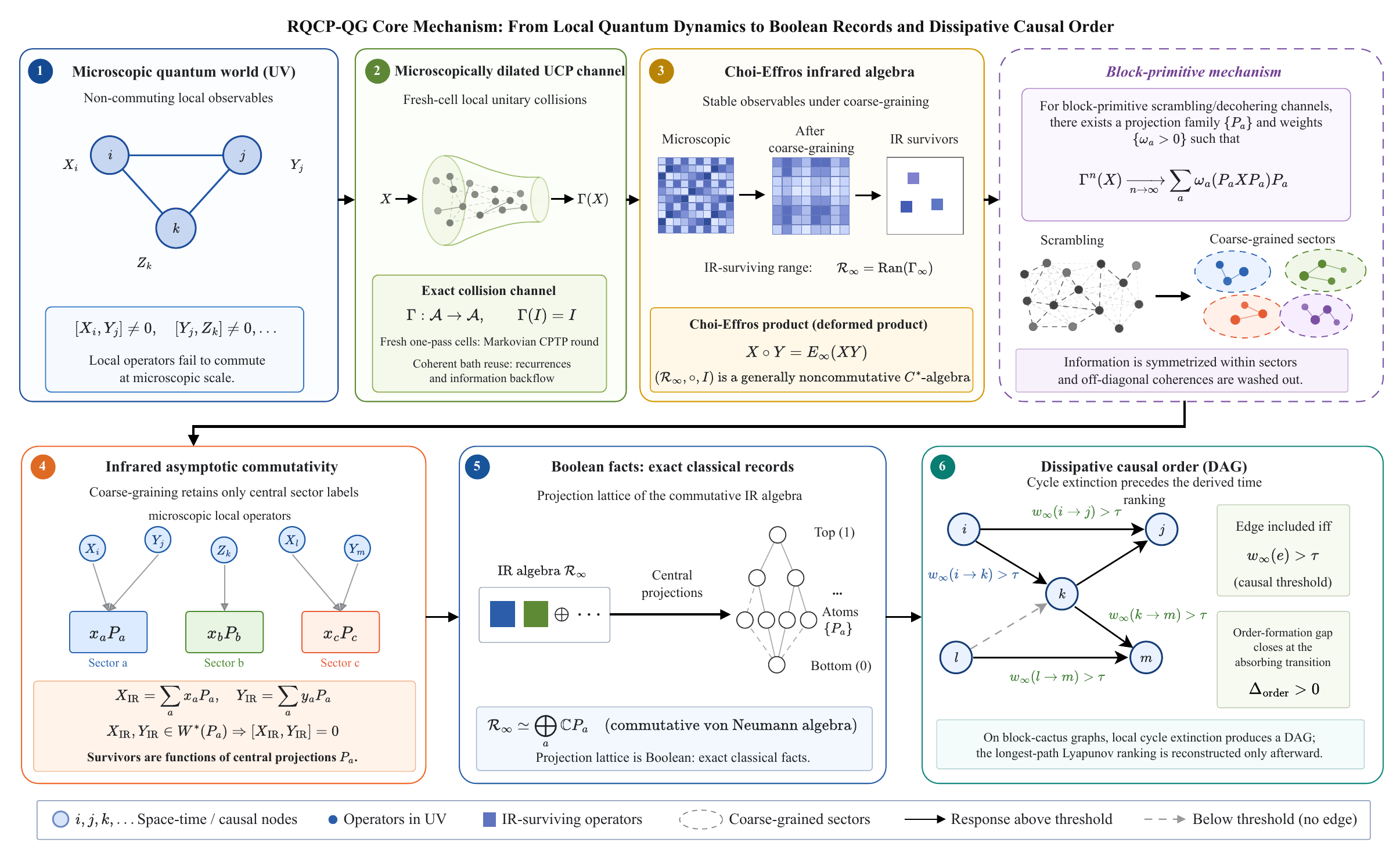}
 \caption{\textbf{The information-to-order core of RQCP-QG.}
 Panels 1--5 retain the mechanism of the original full-process diagram:
 local UCP coarse graining, a generally noncommutative Choi--Effros range,
 and its reduction to Boolean center records in block-primitive dynamics.
 The collision panel displays the fresh-cell unitary dilation realized
 by the solvable microscopic model.  Panel 6 incorporates the later
 causal-order result: local cycle defects on a block cactus become extinct
 before a DAG and longest-path Lyapunov rank are defined.  Metric geometry
 and gravity are not conclusions of this figure.}
 \label{fig:rqcp-core-mechanism}
\end{figure}
\end{landscape}

\begin{landscape}
\begin{figure}[p]
 \centering
 \includegraphics[width=0.819\linewidth]
 {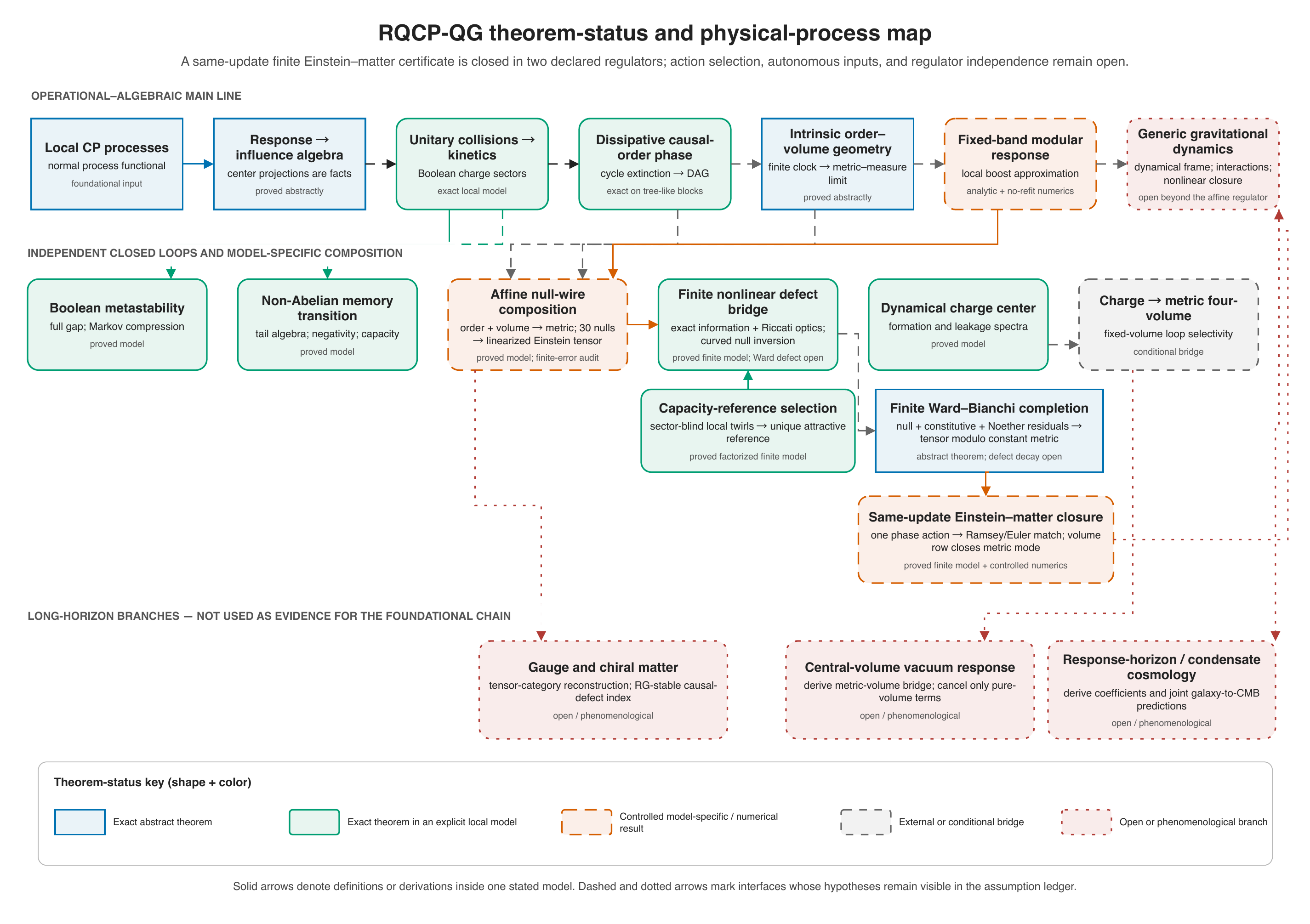}
 \caption{\textbf{RQCP-QG physical processes and theorem status.}
 Blue/green denote exact abstract/model theorems; orange, gray, and red denote
 controlled, conditional, and open statements.  The lower nodes concern
 reference selection, a nonlinear Ward-defect identity, finite tensor
 completion, and a same-update Einstein--matter certificate under their
 displayed inputs.  Common refinement, relational time, quadratic
 constraints, regulator universality, and positive-Hilbert constructions are
 stated in Table~\ref{tab:status-ledger} and in their dedicated sections.  The
 common-refinement mechanism appears separately in
 Fig.~\ref{fig:unified-continuum-mechanism}.  Dashed edges are open; arrows
 between named models record logical dependence and do not identify their
 microscopic generators.}
 \label{fig:rqcp-status-map}
\end{figure}
\end{landscape}

\begin{landscape}
\subsection{Status of the principal results}

Table~\ref{tab:status-ledger} summarizes the strongest statement available
for each module and the assumptions that remain.  Any claim that an input has
been removed is restricted to the scope stated in the third column.

\footnotesize
\begin{longtable}{@{}P{0.16\linewidth}P{0.13\linewidth}
 P{0.36\linewidth}P{0.27\linewidth}@{}}
\caption{Status of the main RQCP-QG modules.}
\label{tab:status-ledger}\\
\toprule
Module & Status & Strongest present result & Remaining boundary\\
\midrule
\endfirsthead
\toprule
Module & Status & Strongest present result & Remaining boundary\\
\midrule
\endhead
Influence algebra and exact facts
& \textcolor{statusblue}{\textsf{abstract}}
& Predual response differences generate an influence algebra; its center
  projections form a complete Boolean event algebra.
& The background strategy class and field-theoretic localization remain part
  of the event context.\\
\addlinespace
Microscopic response kinetics
& \textcolor{statusgreen}{\textsf{exact model}}
& A fresh-cell unitary collision circuit gives the exact random-unitary
  channel, Boolean finite-step fixed algebra, closed Pauli hierarchy, and a
  fixed-order kinetic error independent of total system size.
& One-pass product bath cells, calibrated collision strength, and an external
  schedule remain assumptions.\\
\addlinespace
Boolean formation and metastability
& \textcolor{statusgreen}{\textsf{exact model}}
& Local dephasing and exchange give the exact charge algebra, complete
  graph-controlled gap, metastable Markov compression, topology-dependent
  tolerance, and a noncommuting leakage correction.
& The strongest norm is normalized Hilbert--Schmidt; a general
  size-independent local operator-norm theorem remains open.\\
\addlinespace
Non-Abelian memory transition
& \textcolor{statusgreen}{\textsf{exact model}}
& An occupation-time identity classifies the late-write tail algebra and
  determines reference negativity and memoryless-use quantum capacity.
& The activity quotient is classical; finite systems ultimately absorb, and
  background dephasing rounds the infinite-time phase.\\
\addlinespace
Dissipative causal order
& \textcolor{statusgreen}{\textsf{exact model}}
& Local cycle projectors on a block cactus have an exact contact-process
  quotient.  Their extinction produces a DAG and only then a longest-path
  Lyapunov ranking.
& The undirected skeleton, laboratory runtime, and extension to generic
  graphs are not generated.\\
\addlinespace
Quantum-field inverse geometry
& \textcolor{statusorange}{\textsf{model/numerical}}
& A Ramsey ratio isolates a physical covariance; six anchors fix a de Sitter
  radius, while rank-five completion predicts held-out invariants with
  perturbation and shot-noise bounds.
& Field species, state, canonical normalization, spacelike support, and the
  de Sitter model class are inputs.\\
\addlinespace
Fixed-band modular response
& \textcolor{statusorange}{\textsf{model/numerical}}
& A state-independent physical band gives an optimal finite-size response
  certificate; a sparse range-three correction passes a frozen no-refit
  holdout in the free-fermion model.
& Interactions, higher dimensions, transverse area, multiple null directions,
  and gravity are absent.\\
\addlinespace
Capacity-reference selection
& \textcolor{statusgreen}{\textsf{exact finite model}}
& Sector-blind local product-Weyl collisions select the unique
  capacity-balanced boundary reference with a size-independent
  relative-entropy contraction and an additive-block modular coefficient.
& Boundary/complement token factorization, unbiased exchange,
  \(\widehat A=a_0^2\widehat N\), correlated-record selection, and the
  continuum Newton coefficient remain inputs or open.\\
\addlinespace
Dynamical charge center
& \textcolor{statusgreen}{\textsf{exact model}}
& A local number-conserving channel isolates the complete additive-charge
  center with a graph-dependent formation gap; homogeneous breaking gives an
  exact Ehrenfest spectrum and autocorrelation.
& The microscopic \(U(1)\) law is assumed, and charge has not been proved to
  converge to metric four-volume.\\
\addlinespace
Order--volume geometry
& \textcolor{statusblue}{\textsf{abstract + conditional smooth limit}}
& A robust response DAG and positive additive record measure define an
  intrinsic finite volume clock with an exact reverse triangle.  Strong-
  profile compactness gives a Lorentzian metric--measure limit; local
  subdivision recovers proper time, dimension, measure, and good-cover
  topology for an admissible unknown smooth limit.
& Generic dynamics must still generate adjacency and acyclicity and derive a
  scale-stable capacity-density conversion.  Smoothness is certified, not
  automatic, and one absolute volume unit remains microscopic.\\
\addlinespace
Affine order--volume realization
& \textcolor{statusorange}{\textsf{model-specific}}
& On a contractible affine four-complex, Hodge loss forms an exact response
  order, Boolean record capacity gives an additive volume, and a rank-nine
  finite response frame reconstructs a Lorentzian metric with a stability
  bound.
& The affine adjacency, dimension, and 30-class null frame are inputs; this
  regulator is a composition example, not the generic compactness theorem.\\
\addlinespace
Einstein closure
& \textcolor{statusorange}{\textsf{linearized model}}
& A transverse fixed-band wire bundle, boundary record entropy, capacity
  exchange with an explicit product-reference condition, discrete focusing,
  30 null probes, conservation, and an exact lattice Bianchi identity close to
  a linearized tensor equation.
& The wires are decoupled, transverse area is calibrated,
  \(G_{\rm eff}\) is regulator dependent, and interacting modular
  universality, nonlinear focusing, and graviton self-interaction remain
  open.\\
\addlinespace
Finite nonlinear information--focusing bridge
& \textcolor{statusgreen}{\textsf{exact finite model}}
& A capacity-balanced record reference fixes the modular area coefficient;
  unitary record--bulk exchange has an exact finite relative-entropy
  remainder; a matrix Jacobi recurrence gives exact nonlinear focusing and
  area transport before caustics; and arbitrary invertible coframes preserve
  rank-nine null-tensor identifiability.
& The reference is selected only in the separate factorized token model;
  coupling it to the same nonlinear update, interacting modular response,
  shrinking-diamond locality, the nonlinear information--geometry Ward
  identity, common stress/holonomy evolution, caustics, and regulator
  independence remain open.\\
\addlinespace
Finite Ward--Bianchi tensor completion
& \textcolor{statusblue}{\textsf{abstract + finite certificate}}
& A rank-nine local null frame, compatible discrete divergence, exact
  Noether-defect chain rule, and scalar Poincar\'e gap give a noisy
  full-tensor bound modulo one constant metric mode.  Exact counterexamples
  show that null data and divergence conservation are individually
  insufficient.
& The response tensor must still be derived as the geometric Euler covector
  of the same microscopic update.  Exact holonomy Bianchi identities do not
  imply exact finite diffeomorphism symmetry, and uniform interacting
  continuum control remains open.\\
\addlinespace
Same-update Einstein--matter closure
& \textcolor{statusorange}{\textsf{proved finite model}}
& For one declared local action-phase update, a fresh Ramsey instrument and
  an independent analytic variation have controlled field and divergence
  constitutive defects.  Rank-nine null data, the finite Noether identity,
  a scalar gap, matter residuals, and one volume score then bound the complete
  geometric Euler covector, including the constant metric mode.
& The Einstein--matter action, four-dimensional regulator, signature,
  adjacency, fresh phase access, and refinement family remain inputs.
  Interacting modular dynamics, uniform general-regulator convergence, and
  a universal Newton coefficient are not derived.\\
\addlinespace
Unified operational continuum limit
& \textcolor{statusorange}{\textsf{proved model + interface theorem}}
& One common \(h\) and one register family generate the response order,
  integer record capacities, and the phase action.  A joint discrepancy
  controls Lorentzian metric--measure convergence and the weak
  Einstein--matter Euler limit on the geometry reconstructed from those same
  records.  FLRW and four-dimensional TT refinements pass independent split
  controls.
& The four-dimensional topology, Einstein--Hilbert--scalar action, capacity
  unit, and \(G_h=G\) are supplied.  General nonlinear inhomogeneous
  convergence, autonomous phase/action selection, quantum fluctuations, and
  renormalization remain open.\\
\addlinespace
Same-link induced quantum TT gravity
& \textcolor{statusgreen}{\textsf{proved model}}
& A compiled anonymous pair constructor produces an unlabeled
  \(P_L\times T_L^3\) record lattice.  Its active links supply directed order,
  integer record volume, and scalar hopping.  With no bare curvature action,
  the scalar determinant has a strictly positive TT two-derivative
  coefficient; the saddle-selected, rescaled fixed-band TT Schwinger
  functions have a controlled Gaussian large-\(N_s\) limit.  One joint discrepancy controls
  constructor, order--volume, determinant, and quantum errors.
& The four-dimensional graph language, program/clock state, fair scheduler,
  one-cell volume unit, scalar species, fixed spacing, and a local \(q=0\)
  TT saddle measure are inputs.  The result is fixed-volume and plus-polarized, leading order
  at large \(N_s\), and has no graph backreaction, cutoff removal, universal
  \(G\), or nonlinear gauge-complete graviton measure.\\
\addlinespace
Finite relational order--volume history
& \textcolor{statusgreen}{\textsf{proved model}}
& A stationary local propagation Hamiltonian has an exact conditional
  Page--Wootters history.  Legal rank-increasing append records give a
  schedule-independent DAG, and the same accepted records carry an additive,
  blocking-stable information volume with \(O(L^{-2})\) weak convergence.
& The low-entropy clock endpoint, update program, parent-port law,
  four-dimensional branch, and microscopic volume unit are inputs.  The
  result is finite and does not select an eternal clock or generic
  Hamiltonian basin.\\
\addlinespace
Quadratic Einstein constraint completion
& \textcolor{statusgreen}{\textsf{proved model}}
& The induced TT normalization uniquely fixes the ten-component
  Fierz--Pauli kernel.  One common lattice derivative gives four exact gauge
  zero modes, de Donder ghosts, a nilpotent linear BRST differential,
  gauge-independent conserved-source response, and a full fixed-band
  continuum limit.
& Linear gauge symmetry and the conformal contour are model contracts.
  Nonlinear constraints, cubic Slavnov--Taylor identities, the BV master
  equation, and a nonperturbative metric measure remain open.\\
\addlinespace
Regulator-independent Newton flow
& \textcolor{statusblue}{\textsf{proved regulator-class theorem}}
& Positive determinant weights are obstructed from a cutoff-free Newton
  limit.  Vanishing zeroth and mass-squared supertrace moments remove the
  power, logarithmic, and scheme-local terms; three inequivalent regulators
  converge to one positive logarithmic mass invariant and fixed physical
  Gaussian band.
& The signed supertrace is an input to this regulator theorem.  BRIDGE-PH1
  derives it in a separate free positive-Hilbert model.  The absolute value
  is spectrum dependent, and vacuum plus higher-curvature moments are not
  balanced.\\
\addlinespace
Positive-Hilbert spectral balance
& \textcolor{statusgreen}{\textsf{proved model}}
& Exterior parity on \(\Lambda^\bullet\mathbb C^2\) and an additive
  number-operator mass ladder derive \((1,-2,1)\),
  \(\Str\mathbf1=\Str M^2=0\), a positive finite logarithmic invariant,
  transfer-form positivity, geometry-by-geometry persistence under any
  supplied finite positive graph sum, and exact finite-move response.
& The construction is free; its graph support and prior are supplied.
  Exact all-level pairing is proved to erase the finite invariant, while
  generic interactions break the balance.  Finite interacting protection,
  projective positive growth, and a covariant rank-one coherent branch are
  supplied separately below; a Lorentzian continuum, spin--statistics, and
  nonlinear gravity are not derived.\\
\addlinespace
Selective Ward, projective, covariant, and non-Abelian histories
& \textcolor{statusgreen}{\textsf{proved models}}
& A shifted CAR Ward identity protects the positive response for arbitrary
  finite interacting physical operators.  Conditional causal births extend
  it to exact deletion-projective cylinder laws, an infinite past-finite
  labelled-history measure, CPTP channels, and path-response identities.  A
  contact-response CSG subfamily additionally has exact discrete general
  covariance and Bell causality; a one-arity \(r\ge2\) complex branch extends
  to a countably additive measure and a strongly positive rank-one
  decoherence functional with nonzero finite-volume interference.  On a
  separate fixed two-port strip, fresh-ancilla noncommuting decisions give
  an extendible rank-four strongly positive history measure.  On arbitrary
  CSG causal histories, intrinsic fresh-event Pauli holonomies define a
  covariant countably additive non-Abelian operator measure and strongly
  positive grade-2 decoherence functional; a late-cylinder bound proves that
  integration does not abelianize its range.
& The commutant, maximal-birth language, contact motifs, response multiplet,
  parameters, coherent arity, phase, and two-port strip are supplied.  The
  arbitrary-causet non-Abelian construction uses a classical CSG density:
  it does not satisfy an operator Markov sum or square-operator
  TOBC/NTOBC/CPOBC, and has no quantum geometry backreaction.  No tight
  Lorentzian continuum or nonlinear constraint dynamics is derived.\\
\addlinespace
Central operator-valued causal growth
& \textcolor{statusgreen}{\textsf{proved model + classification}}
& Positive central CSG couplings give exact parentwise POVMs, covariant
  L\"uders path effects, central Bell cross-products, and a strongly
  countably additive infinite-history POVM.  Every such law is a direct
  integral of scalar CSG laws, so geometry statistics are blind to
  coupling-sector coherence.  A positive noncentral normalized control
  fails DGC.
& The central branch is state controlled but not coherent quantum
  backreaction.  No noncentral covariant operator Bell/Markov law,
  discarded-record infinite instrument, tight Lorentzian continuum, or
  nonlinear constraint closure is derived.\\
\addlinespace
Bell-causal operator-growth rigidity
& \textcolor{statusgreen}{\textsf{exact boundary theorem}}
& Every finite-dimensional self-adjoint nonsingular CPOBC transition
  system is commutative without a sign assumption, including
  complementary-spectrum blocks.  At the
  \(t=1/4\) scalar CSG point, exact four-event deformation theory
  independently forces all leading matrix tangents to commute; every
  \(2\times2\) upper-triangular extension between distinct strictly
  positive scalar CSG coupling sequences is split, at arbitrary
  first-difference order, while every coincident-character
  \(2\times2\) self-extension is a commuting coupling tangent.  An
  irreducible nonsingular \(2\times2\) branch with non-scalar
  \(R_2=Q_1^{-1}Q_2\) is excluded by exact next-stage minors.  On the
  remaining scalar branch, the two-chain condition excludes \(h=1\).
  Within the trace-balanced subchart \(d=a,\ x=-1\) of the normalized
  first-simple-spectrum chart at \(R_3\), exact factorization leaves three
  components, all excluded at the next antichain stage.
& Non-self-adjoint nonsingular and singular transitions, state-only
  covariance, higher-dimensional noncommutative nilpotent thickenings,
  the remainder of the normalized first-simple-spectrum chart,
  delayed-onset irreducible nontriangular representations with scalar
  \(R_2=h\mathbf1\), \(h\ne0,1\), and Jordan/later onset,
  and a projectively consistent infinite noncommutative CPOBC measure remain
  open.  The result is a self-adjoint nonsingular no-go theorem, not a
  construction or a no-go theorem for every operator-history formulation.\\
\addlinespace
Autonomous geometry and gravity
& \textcolor{statusred}{\textsf{open}}
& The intrinsic bridge supplies finite geometry and a generic
  metric--measure endpoint.  Three additional model chains supply a
  finite relational clock, quadratic BRST completion, and a
  regulator-independent balanced Newton limit.
& One autonomous background-free dynamics must still select adjacency,
  dimension, and the capacity-density law; protect the balanced multiplet
  under interacting Lorentzian RG; and produce nonlinear gravitational
  dynamics and the quantum metric measure.\\
\addlinespace
Vacuum, matter, and cosmology
& \textcolor{statusgray}{\textsf{conditional}}
& Fixed-sector normalization cancels a pure-volume term; standard tensor
  category, index, and effective-field-theory constructions identify precise
  targets.
& The central-volume bridge, gauge-category attractor, chiral-index selection,
  response-horizon coefficients, and dark equation of state are not derived.\\
\bottomrule
\end{longtable}
\end{landscape}

\subsection{Relation to the companion manuscripts}

The completed finite-model results are maintained as independent manuscripts.
Table~\ref{tab:companion-correspondence} fixes their current titles and
roles.  The labels PRL-1--PRL-5 are internal indices, not serial titles or
logical implications between adjacent entries.  The Boolean/metastability
Letter is denoted PRL-0 because it supplies the finite-system reference model
for the collision, memory, and central-charge studies.

\begin{longtable}{@{}P{0.10\textwidth}P{0.37\textwidth}P{0.45\textwidth}@{}}
\caption{Relation between the outstanding assumptions and the companion
manuscripts.  Each manuscript contains an independent
theorem--model--observable result; open compositions are stated separately
in Sec.~\ref{sec:updated-core-boundary}.}
\label{tab:companion-correspondence}\\
\toprule
Index & Canonical title & Unique role in the program\\
\midrule
\endfirsthead
\toprule
Index & Canonical title & Unique role in the program\\
\midrule
\endhead
\midrule
\multicolumn{3}{r}{\small Continued on the next page}\\
\endfoot
\bottomrule
\endlastfoot
PRA & \emph{Microscopic Derivation of Quantum Response Kinetics from Local
Unitary Dynamics}
& Derives the dephasing--random-SWAP kinetics, fixed-response hierarchy, and
finite-step error from fresh one-pass unitary collisions.\\
PRL-0 & \emph{Metastable Classical Dynamics of Noise-Selected Charge Sectors}
& Solves the exact Boolean algebra, full gap, weak-perturbation reduction, and
topology-dependent tolerance of that kinetic generator.\\
PRL-1 & \emph{Causal Order as a Dissipative Phase of Local Quantum Channels}
& Generates DAG order from cycle-defect extinction on a block cactus, without
supplying a Lyapunov function in advance.\\
PRL-2 & \emph{Dynamical Transition between Classical Records and Non-Abelian
Quantum Memory}
& Gives an exact late-write tail-algebra and quantum-capacity transition under
an absorbing noise field.\\
PRL-3 & \emph{Lorentzian Curvature from Sparse Quantum-Field Correlations}
& Infers an unknown de Sitter scale and blind invariants inside a prescribed
field--state--geometry family.\\
PRL-4 & \emph{Fixed-Band Modular-to-Boost Response in a Critical Quantum
Lattice}
& Replaces a shrinking fixed-mode window by a state-independent fixed physical
band and a finite-cutoff response certificate.\\
PRL-5 & \emph{Dynamical Formation and Controlled Decay of Charge
Superselection}
& Derives the complete charge center, sectorwise CPTP restrictions, exact
lumpability, and controlled breaking, while leaving charge-to-volume open.\\
BRIDGE-OV1 & \emph{An Intrinsic Order--Volume Bridge from Quantum Response to
Lorentzian Metric--Measure Geometry}
& Constructs the finite volume clock and generic Lorentzian metric--measure
limit, proves the continuum proper-time identity and smooth-limit error
budget, and turns nonmanifold behavior into explicit rejection certificates;
it consumes rather than dynamically generates the input order and
capacity-density law.\\
BRIDGE-1 & \emph{Order--Volume Reconstruction and Linearized Einstein Closure
in a Quantum Null-Wire Lattice}
& Closes the order--volume and multi-null Einstein interfaces in one affine,
decoupled-wire regulator, with an exact rank certificate and finite-error
ledger; it leaves generic and nonlinear completion open.\\
BRIDGE-NL1 & \emph{Finite Quantum-Information Remainders and Nonlinear Null
Focusing in a Quantum Record Lattice}
& Replaces the finite-exchange entropy first law and linearized optical law by
exact identities, proves curved rank-nine null inversion, and isolates the
remaining nonlinear Einstein content as a measurable Ward-defect condition.\\
BRIDGE-REF1 & \emph{Local Unitary Selection of a Capacity-Balanced Area
Reference}
& Selects the capacity-balanced reference in a factorized
capacity--complement token model, proves a size-independent entropy
contraction and bias test, and preserves the modular coefficient under
additive blocking.\\
BRIDGE-WB1 & \emph{Quantitative Ward--Bianchi Completion on Finite
Lorentzian Record Lattices}
& Converts local null, symmetry, matter-equation, and measurement residuals
into a certified full-tensor estimate, while exposing the still-open
constitutive identification with the geometric Euler derivative.\\
BRIDGE-SU1 & \emph{Same-Update Quantum Response and Einstein--Matter Ward
Closure on Finite Lattices}
& Derives the field and divergence constitutive budgets from one declared
local action-phase process through independent Ramsey and analytic pipelines;
one volume score fixes the missing constant metric mode.  This is a finite
model theorem, not dynamical selection of the action or a universal
continuum limit.\\
BRIDGE-UC1 & \emph{A Unified Operational Continuum Limit for Lorentzian
Geometry and Einstein Response}
& Uses one register family and one refinement parameter for the
order--volume and action-phase branches, proves a joint
metric--measure/Einstein--Euler limit, and rejects split geometry or split
action by a nonzero finite discrepancy.  Its explicit FLRW--TT realization
is model-specific and supplies rather than selects the gravitational action,
dimension, capacity unit, and Newton coefficient.\\
BRIDGE-IQ1 & \emph{From Relational Order--Volume Records to Induced Quantum
Gravity on the Same Links}
& Replaces the prescribed Einstein--Hilbert phase of BRIDGE-UC1 by a
strictly positive two-derivative TT kernel induced by linked scalar matter
on the constructor output, and gives a fixed-band large-\(N_s\) quantum
limit for that same kernel.  It compiles rather than selects dimension,
topology, clock, volume unit, and matter content, and it does not construct
the gauge-complete nonlinear theory or a regulator-independent Newton flow.\\
BRIDGE-AT1 & \emph{Relational Time and Autonomous Order--Volume Growth in a
Closed Quantum Network}
& Replaces the external update parameter by an exact finite stationary
history sector and proves that conditional rank-increasing growth, causal
incidence, capacity, and weak volume share one record refinement.  The
clock/program boundary and four-dimensional port language remain inputs.\\
BRIDGE-QC1 & \emph{Constraint-Complete Quantum Einstein Response from an
Induced Transverse--Traceless Coefficient}
& Extends the induced TT coefficient uniquely to the ten-component
Fierz--Pauli kernel and closes its finite lattice Ward, ghost, BRST, and
conserved-source Gaussian response at quadratic order.\\
BRIDGE-RG1 & \emph{Spectral Balance and a Regulator-Independent Induced
Newton Coupling}
& Proves the positive-spectrum obstruction, the necessity and sufficiency
of two spectral moments, and a common fixed-band limit across sharp
momentum, sharp proper-time, and smooth proper-time schemes; the signed
balance itself remains microscopic input.\\
BRIDGE-PH1 & \emph{A Positive-Hilbert Origin for Spectral Balance in
Induced Gravity}
& Derives the balanced weights and masses from exterior parity and an
additive internal number operator on a positive Fock space, proves
geometry-by-geometry persistence in any supplied finite graph ensemble and
exact finite-move response, and exposes both the nonlinear-mass defect and
the exact-pairing obstruction.\\
BRIDGE-SLW1 & \emph{Selective Ladder Ward Protection and Positive Dynamics
of Finite Quantum Geometries}
& Characterizes the exact shifted-CAR Ward manifold for an arbitrary finite
positive interacting physical operator, bounds approximate Ward and heat
defects, and realizes the surviving positive response as the unique
stationary law of a local unitary-dilatable finite-geometry channel.  The
internal commutant and finite geometry language remain supplied, and no
projective continuum measure or nonlinear gauge closure is claimed.\\
BRIDGE-PCG1 & \emph{Projectively Consistent Quantum-Response Growth of
Causal Histories}
& Extends the protected positive response to conditionally normalized
causal births, exact deletion-projective cylinder laws, a probability
measure on infinite past-finite labelled histories, and exact path-response
identities.  Natural labels and the growth language remain supplied; no
discrete general covariance, coherent quantum measure, or Lorentzian
continuum limit is claimed.\\
BRIDGE-CQG1 & \emph{Covariant Quantum-Response Growth and an Extendible
Quantum Measure on Causal Histories}
& Converts response-generated contact-star weights into a nonvanishing
classical sequential-growth law with exact deletion projectivity, discrete
general covariance, and Bell causality.  A separate one-arity \(r\ge2\)
complex response has bounded variation, a unique countably additive
extension, a normalized strongly positive rank-one decoherence functional,
and nonzero finite-volume quotient interference.  The growth language,
contact family, arity, and phase remain supplied; no local-unitary
non-rank-one history dynamics or Lorentzian continuum limit is claimed.\\
\addlinespace
BRIDGE-NQH1 & \emph{Extendible Non-Abelian Quantum History Measures from
Local Unitary Decisions}
& On a fixed ranked two-port causal strip, two fresh-ancilla decisions obey
exact operator Markov and Kraus identities, generate a noncommuting
transition algebra, and have a countably additive matrix-measure extension
whenever the response-angle tail is summable.  The induced decoherence
functional is strongly positive, has nonzero interference, and reaches the
full rank-four Hilbert--Schmidt bond space after two cells.  The port
language and foliation remain supplied, and this theorem does not solve
unlabeled operator Bell causality on arbitrary causal sets or a Lorentzian
continuum limit.\\
\addlinespace
BRIDGE-CNH1 & \emph{Covariant Non-Abelian Quantum Measures from
Spectator-Local Causal Holonomies}
& On arbitrary response-generated CSG histories, each birth rotates only a
fresh event qubit by an intrinsic precursor-dependent Pauli gate.  An
exponential seed-tail theorem gives an operator-norm infinite product, a
countably additive bounded-variation \(B(\mathcal H)\)-valued measure, and
a normalized strongly positive grade-2 decoherence functional.  Exact
relabeling equivariance, spectator invariance, and a late-cylinder lower
bound prove a noncommuting integrated range.  The scalar CSG law remains
external to the operator density; this is isometric Bell locality, not a
square-operator CPOBC solution, and no continuum or constraint theorem is
claimed.\\
\addlinespace
BRIDGE-OCG1 & \emph{Central Operator-Valued Causal-Set Growth and the
Boundary of Coherent Backreaction}
& Solves the commutative operator-valued CSG branch: central positive
couplings give exact finite POVMs, DGC, Bell cross-products, and an infinite
history POVM, but a direct-integral theorem proves coherence blindness of
the geometry marginal.  A positive noncentral normalized control violates
DGC, and an independent fresh-event instrument shows that non-Abelian
records can coexist without geometry backreaction.\\
\addlinespace
BRIDGE-BCGR1 & \emph{Self-Adjoint and All-Order Triangular Rigidity in
Bell-Causal Quantum Sequential Growth}
& Attacks the remaining nonsingular square-operator CPOBC equations rather
than adding a separate noncommuting density.  It proves commutativity of
every finite-dimensional self-adjoint nonsingular CPOBC system without a
sign assumption, an exact rank obstruction to every noncommuting leading
tangent at the four-event \(t=1/4\) scalar point, and an all-order
determinant proving splitting of every \(2\times2\) upper-triangular
extension between distinct strictly positive scalar CSG coupling
sequences and classifies every coincident-character \(2\times2\)
self-extension as a commuting dual-number coupling tangent.  It also
excludes irreducible nonsingular \(2\times2\) onset at non-scalar \(R_2\);
the two-chain condition removes the scalar value \(h=1\), and exact
certificates obstruct all three components in the trace-balanced
\(d=a,\ x=-1\) subchart of the normalized delayed first-simple-spectrum
chart at the following antichain stage.
It leaves non-self-adjoint nonsingular, zero/singular-coupling, state-only,
higher-dimensional noncommutative nilpotent, delayed-onset irreducible
nontriangular components with scalar
\(R_2=h\mathbf1\), \(h\ne0,1\), Jordan/later onset, and infinite branches
explicitly open; the remainder of the normalized first-simple-spectrum
chart is also open.\\
\end{longtable}

The original exact interface between two manuscripts is
\[
 \text{PRA collision circuit}
 \longrightarrow
 \text{PRL-0 dephasing--exchange generator}.
\]
BRIDGE-OV1 closes the mathematical composition from any certified
order--record input to a Lorentzian metric--measure output; it does not
identify the PRL-1 block-cactus dynamics with a universal capacity-density
law.  BRIDGE-1 is not a relabeling of PRL-1+PRL-3.  It introduces an affine
regulator whose Hodge order, Boolean volume, finite null frame, transverse
modular response, and lattice focusing are proved to compose internally.  It
reuses the fixed-band certificate of PRL-4 but does not identify the PRL-1
block-cactus generator with the PRL-3 de Sitter detector model.  The general
autonomous order--capacity generation, autonomous Einstein-phase selection,
and the BRIDGE-OV1+PRL-5 central-volume limit remain open.  BRIDGE-REF1 closes
reference selection only for factorized tokens and is not identified with
the nonlinear record--geometry generator of BRIDGE-NL1.  BRIDGE-WB1 closes
only the finite tensor-completion step.  BRIDGE-SU1 supplies the missing
same-update response--Euler theorem and constant-mode volume row in two
declared regulators; it does not derive their Einstein--matter phase action,
geometry, or universal continuum limit from the preceding record dynamics.
BRIDGE-UC1 then closes the compatibility of the order--volume and
same-update limits for one common FLRW--TT refinement family.  It proves that
the two continuum outputs refer to the same records and metric, while leaving
autonomous selection of that refinement family and action open.
BRIDGE-IQ1 closes a different missing interface: on one compiled anonymous
record lattice it removes the bare Einstein--Hilbert phase and derives a
positive TT kinetic coefficient from scalar matter carried by the same
links.  Its large-\(N_s\) limit quantizes only that fixed-volume infrared
sector.  It neither upgrades the compiled language to the generic
background-free target of BRIDGE-OV1 nor supplies the nonlinear tensor
completion demanded by P9.
BRIDGE-AT1 replaces its external scheduler by a finite relational history
and makes causal rank and information volume records of that same history.
BRIDGE-QC1 then closes the full quadratic gauge complex normalized by the
induced TT response.  BRIDGE-RG1 removes the cutoff only after a signed
two-moment balance is imposed.  BRIDGE-PH1 derives that balance from a
positive free graph multiplet.  BRIDGE-SLW1 supplies exact and approximate
selective protection for an arbitrary finite positive interacting physical
operator and a positive local dynamics for the induced weights on a supplied
finite geometry language.  BRIDGE-PCG1 then removes the independent
finite-size normalization defect by constructing exactly projective laws and
an infinite positive measure on a supplied naturally labelled causal-birth
language.  BRIDGE-CQG1 then removes the natural-label path-weight and
Bell-spectator defects for a response-generated contact CSG subfamily and
constructs an extendible rank-one coherent quantum measure with nonzero
finite-volume interference.  BRIDGE-NQH1 separately removes the rank-one
and local-realizability restrictions on a fixed two-port strip.
BRIDGE-CNH1 then combines arbitrary-down-set CSG covariance with intrinsic
fresh-event noncommuting holonomies and a countably additive operator-valued
measure.  It does not solve the square-operator Bell equations or make the
geometry probability law quantum and backreacting.  BRIDGE-OCG1 then closes
and classifies the central operator-valued branch: its infinite POVM is
exactly covariant, but is only a direct integral of scalar laws, while a
positive noncentral normalized control fails DGC.  It therefore rules out
centralization as the missing coherent-backreaction mechanism rather than
closing that mechanism.  BRIDGE-BCGR1 then proves that three natural
noncentral escape routes are rigid: every finite-dimensional self-adjoint
nonsingular system commutes without a sign assumption, noncommuting leading
tangents at the tested scalar point are obstructed, and the complete
distinct positive scalar CSG \(2\times2\) triangular family splits at
arbitrary first-difference order; the coincident-character branch consists
only of commuting coupling tangents, and irreducible nonsingular
\(2\times2\) onset at non-scalar \(R_2\) is impossible.  It does not
exclude a non-self-adjoint nonsingular, zero/singular-coupling, state-only,
higher-dimensional noncommutative nilpotent, or the remaining delayed-onset
irreducible nontriangular representations with scalar
\(R_2=h\mathbf1\), \(h\ne0,1\), including Jordan or later onset, and it
does not provide
an infinite extension.  These
manuscripts still do not derive the
protected internal commutant, maximal-birth/port language, contact motifs,
response multiplet, coherent arity or phase, a covariant non-Abelian history
law satisfying operator Markov/Bell constraints, a tight Lorentzian
continuum, continuum spin--statistics, or a nonlinear BV identity.  The
interfaces therefore compose at finite or quadratic
regulator level, but they are not evidence that the clock program, gauge law,
matter multiplet, growth language, or interacting continuum spectrum is
dynamically selected.

\subsection{Microscopic, kinetic, and continuum limits}

The methodological standard is close in spirit to the modern treatment of
Hilbert's sixth problem: microscopic dynamics, kinetic evolution, and
macroscopic equations are different limits, and each arrow needs its own
error estimate and order of limits.  Long-time derivations of Boltzmann
dynamics from hard spheres and subsequent fluid limits make this separation
explicit \cite{DengHaniMa2024,DengHaniMa2025}.  RQCP-QG has a different
microscopic object and does not inherit those theorems, but it should inherit
their discipline.  We therefore distinguish
\begin{equation}
 \text{local unitary process}
 \ \longrightarrow\
 \text{quantum response kinetics}
 \ \longrightarrow\
 \text{hydrodynamic or geometric field}.
 \label{eq:three-scale-standard}
\end{equation}
The first arrow asks for a microscopic dilation, a controlled memory
assumption, and a finite-time channel bound.  The second asks for the relevant
slow variables, a kinetic or spectral gap, a scaling regime, and a continuum
error.  A successful calculation at the middle level does not validate the
last arrow.  This is why the collision model below is a foundational advance
even though it does not generate a metric, and why the modular lattice result
is not called an Einstein derivation.

\subsection{Boundary of the present results}

The present results do not show that a single \emph{background-independent}
microscopic model yields
\[
\begin{gathered}
\text{quantum process}\to\text{generic curved Lorentzian spacetime}
\to\text{nonlinear Einstein dynamics}\\
\to\text{Standard Model and cosmology}.
\end{gathered}
\]
In particular:
\begin{enumerate}
\item the causal-order model generates order on a restricted graph family,
not the graph itself or a metric;
\item the inverse-geometry model infers curvature inside a specified
geometry--state class and does not establish spacetime emergence;
\item the intrinsic order--volume bridge constructs geometry only after a
robust order and positive additive record measure pass their gates; it does
not derive generic adjacency, the capacity-density conversion, or an
autonomous runtime from one closed dynamics;
\item the modular lattice model ends with a response certificate and does not
derive a gravitational field equation;
\item the affine null-wire bridge derives a linearized field equation only
after fixing the four-complex, response cone, decoupled-wire regulator, and
linear-response regime; it does not derive their generic selection or
nonlinear completion;
\item the nonlinear finite-regulator bridge derives exact information and
optical remainder identities before caustics, but not their required Ward
matching, an interacting modular limit, covariant stress/Bianchi evolution,
or a regulator-independent Einstein equation;
\item capacity-reference selection is proved only for factorized tokens.  It
does not derive the record--area map, correlated-record reference, or
observed Newton coefficient;
\item the finite Ward--Bianchi theorem controls tensor completion under a
constitutive identification; it does not derive that identification or
promote a holonomy identity into an exact finite diffeomorphism symmetry;
\item the same-update bridge derives the constitutive identification and
constant-mode score only for a prescribed local Einstein--matter phase
action; it does not derive that action, its four-dimensional regulator, or a
universal Newton normalization from the preceding record dynamics;
\item the unified-continuum bridge proves that order, volume, response, and
the Euler limit use one common FLRW--TT refinement, but it does not derive
the four-dimensional manifold phase, the Einstein--Hilbert action, the
capacity unit, \(G\), or a general quantum probability-law limit;
\item the same-link induced-gravity bridge removes the bare curvature action
only inside a compiled \(3+1\)-dimensional record language and a fixed-volume
TT sector.  It retains the program/clock, fair scheduler, volume unit,
microscopic spacing, scalar matter content, and local TT measure; its
large-\(N_s\) infrared limit is not a cutoff-free or gauge-complete quantum
gravity measure;
\item the relational-history model replaces external runtime only within a
prepared finite clock and program.  It does not select those inputs, the port
language, four dimensions, or the physical record-volume unit;
\item the constraint-complete model is gauge complete only at quadratic
order after the linear gauge law is declared; nonlinear BV/BRST closure and
a nonperturbative metric measure are absent;
\item the regulator-independent Newton theorem requires a signed
two-moment determinant balance.  The companion positive-Hilbert model
derives that balance for a free affine mass ladder on every member of a
supplied finite graph ensemble and rules out exact all-level pairing as the
sole protection.  The selective-Ward companion protects the same moments
for arbitrary finite positive interacting physical operators and samples
their positive weights on a supplied finite geometry language.  The
projective-growth companion further turns those weights into compatible
finite cylinder laws and an infinite positive history measure on a supplied
naturally labelled causal-birth language.  The covariant-growth companion
adds a response-generated contact CSG subfamily with exact discrete general
covariance and Bell causality and an extendible rank-one coherent measure.
The fixed-strip companion separately supplies a local-unitary
noncommutative rank-four measure with a countable extension.  The
arbitrary-causet holonomy companion combines the classical CSG density
with intrinsic fresh-event noncommuting gates to obtain a
relabeling-covariant countably additive operator measure whose integrated
range is provably non-Abelian.  The central operator-growth companion proves
that a commutative operator promotion has an exact infinite POVM but is
only a direct mixture of scalar laws; its positive noncentral normalized
control fails DGC.  These results do not derive the internal
commutant, maximal-birth/port language, contact motifs, response multiplet,
coherent arity, or phase.  In particular, the arbitrary-causet construction
retains an external classical CSG law and does not satisfy an operator
Markov sum or the square-operator TOBC/NTOBC/CPOBC conditions, while the
central branch is coherence blind; neither supplies a quantum-backreacting
geometry law.  The companions do not
prove Lorentzian continuum refinement,
spin--statistics or anomaly freedom, absolute Newton constant, or vacuum and
higher-curvature renormalization;
\item the charge-center model derives a superselection structure but not its
identification with metric volume;
\item the response-horizon, finite spectral triple, generation-number, and
response-superfluid sectors are research targets rather than consequences of
the foundational theorems.
\end{enumerate}
These exclusions are part of the theory's falsifiability.  They prevent an
application of a conditional synthesis theorem from being mistaken for a
microscopic derivation.

\section{Axiom-reduced formulation}\label{sec:v62-axiom-reduction}

Earlier versions collected a number of sufficient conditions under the label
``axioms.''  That presentation obscured which conditions were operational
definitions, which were model hypotheses, and which were open continuum
bridges.  We retain one foundational operational input and record the
remaining assumptions at the arrows where they are used.

\subsection{Foundational operational input}

\begin{assumption}[QI1: completely positive local processes]\label{ass:qi1}
For every finite operational context $C$, $\Omega_C$ is normal, multilinear, positive on completely positive local operations, and normalized on deterministic normal unital completely positive operations.  This is the minimal operational validity condition for the process calculus.
\end{assumption}

Complete positivity and normalization do not imply Markovianity, classical
records, order, geometry, or gravity.  Those are separate dynamical
questions.  In particular, the following interfaces are not axioms of the
operational calculus:
\begin{enumerate}[label=(R\arabic*)]
\item a local microscopic unitary process has a controlled Markovian response
limit;
\item the infrared information algebra is commutative rather than a noiseless
matrix algebra;
\item cycle-free directed response and a Lyapunov ranking are generated
without a preordered vertex set;
\item order and information number converge to an unknown Lorentzian
metric-measure space;
\item a local modular generator converges to boost energy in a fixed physical
band and supplies enough null directions for a tensor equation;
\item a local dynamics selects the capacity-balanced boundary reference and
derives the map from record number to boundary area;
\item an additive central charge converges to metric four-volume;
\item residual sectors dynamically realize gauge, chiral, and cosmological
phases.
\end{enumerate}
Items (R1) and (R2) are removed in the dephasing--exchange and memory model
classes described in Section~\ref{sec:closed-loops}.  Item (R3) is removed
for the stated block-cactus family, but not for generic graphs.  Restricted
finite-regulator instances of (R4)--(R5) are provided by the affine and
nonlinear bridge theorems; their background-independent and interacting
limits remain open.  The factorized capacity--complement model removes the
reference-selection part of (R6), but not the record--area map or a universal
coefficient.  Items (R7)--(R8) remain conditional or open.  The split
property used for
field-theoretic laboratories is an external AQFT theorem under its standard
nuclearity hypotheses, not an RQCP dynamical prediction.

\subsection{Microscopic status of Markovian response kinetics}

The companion foundational work \emph{Microscopic Derivation of Quantum
Response Kinetics from Local Unitary Dynamics} replaces the earlier
Markovian-response postulate by a constructive statement for a nontrivial
local model class.  On any finite connected graph,
each dephasing or exchange event is implemented by a unitary collision with a
fresh one-pass bath qubit.  A symmetric edge-colored collision round
\(\Phi_h\) then has, at every nonzero step size, the exact Heisenberg fixed
algebra
\[
 \operatorname{Fix}(\Phi_h)
 =\operatorname{span}\{P_0,\ldots,P_N\},
\]
where \(P_q\) projects onto total-charge sector \(q\).  The full Pauli
response hierarchy is a killed colored-interchange process, and for every
fixed Pauli order \(k\) its kinetic approximation obeys a system-size-uniform
bound
\[
 \left\|\Phi_h^n-e^{t\mathcal L}\right\|_{(k),p}
 \leq \frac{t h^2}{3}B_k^3e^{hB_k},
 \qquad
 B_k\leq \frac{2\chi k\gamma}{d},
 \qquad t=nh,\qquad p\in\{1,\infty\}.
\]
Thus a local unitary circuit, its Markov generator, its Boolean charge
records, and its operational response kernel are derived within one model
rather than postulated independently.  This construction removes a
phenomenological Markov-channel input for that class, but not in general:
product bath
input, one-pass cells, an external collision clock, and calibrated
step-dependent angles remain assumptions.  In particular it does not yet
derive a stationary reservoir spectrum, correlated-bath kinetics, or a
self-generated time variable.

This result establishes existence of a nontrivial microscopic realization of
(R1) and of one Boolean record mechanism.  It does not turn every item in the
former S1--S10 list into derived structure.  In particular, central volume,
matter categories, autonomous clocks, and geometric locality remain distinct
problems.

\subsection{From CP fixed ranges to Boolean facts}

Let $\Gamma$ be a normal UCP channel and suppose the Cesaro mean
\[
E_N(A)=\frac1N\sum_{k=0}^{N-1}\Gamma^k(A)
\]
converges ultraweakly to a normal completely positive projection $E_\infty$.  Its range
\[
\A_{\rm IR}=E_\infty(\A)
\]
is the infrared accessible algebra.  Choi--Effros theory
\cite{ChoiEffros1976,ChoiEffros1977} gives a product
\[
X\circ Y=E_\infty(XY)
\]
under which $\A_{\rm IR}$ is a $C^*$-algebra.  However, this alone does not imply classicality: $\A_{\rm IR}$ may still be noncommutative.

The missing dynamical condition is asymptotic abelianness on the invariant
range:
\[
\lim_{n\to\infty}\|[\Gamma^n(X),\Gamma^n(Y)]\|_\omega=0,
\qquad X,Y\in E_\infty(\A).
\]
Because \(E_\infty(\A)\subset\Fix(\Gamma)\), this condition says precisely
that the retained observables commute in the faithful GNS seminorm; it is
stronger than decay of equal-time commutators on a microscopic local
subalgebra.  Under it, the Choi--Effros product is commutative and any
faithful representation has a commutative von Neumann closure.  Its
projection lattice is therefore Boolean.

\begin{theorem}[IR Boolean facts from CP coarse-graining]\label{thm:v62-ir-bool}
Let $\Gamma$ be a normal UCP channel preserving a faithful normal state
$\omega$.  Suppose the normal mean-ergodic projection $E_\infty$ exists and
for every \(X,Y\in E_\infty(\A)\),
\[
\lim_{n\to\infty}
\|[\Gamma^n(X),\Gamma^n(Y)]\|_\omega=0.
\]
Then the Choi--Effros algebra \(E_\infty(\A)\) is commutative.  Consequently,
for any faithful normal representation \(\pi\),
\[
\Proj\bigl(\pi(E_\infty(\A))''\bigr)
\]
is a complete Boolean algebra of infrared facts.
\end{theorem}

\begin{proof}
The mean-ergodic identities give
\(\Gamma E_\infty=E_\infty\Gamma=E_\infty\).  Hence
\(\Gamma^n(X)=X\) and \(\Gamma^n(Y)=Y\) for \(X,Y\) in the range.  The
assumed limit is therefore \(\|[X,Y]\|_\omega=0\).  Faithfulness of
\(\omega\) implies \([X,Y]=0\).  Thus
\[
 X\circ Y=E_\infty(XY)=E_\infty(YX)=Y\circ X .
\]
The faithful von Neumann closure is commutative, and its projections form a
complete Boolean algebra.
\end{proof}

\subsection{Central superselection and vacuum-response sequestering}

\textbf{Status: \StatusModel\ for charge formation; \StatusConditional\ for
vacuum response.}
The local number-conserving model in Section~\ref{sec:loop-center} derives a
complete additive-charge center rather than attaching a global register:
\[
 \ker\mathcal L_0^*=\operatorname{span}\{P_0,\ldots,P_N\}.
\]
Every fixed-charge block is dynamically invariant and the restriction of the
semigroup is CPTP.  If a separately specified continuum map identifies the
charge \(v\) with metric four-volume and one works with a normalized
fixed-\(v\) functional, a common factor \(e^{-cv}\) cancels:
\[
 \frac{Z_v[J;c]}{Z_v[0;c]}
 =\frac{e^{-cv}Z_v[J;0]}{e^{-cv}Z_v[0;0]}
 =\frac{Z_v[J;0]}{Z_v[0;0]}.
\]
The leading heat-kernel \(a_0\) term of a matter loop has this pure-volume
form.  Curvature, anomaly, boundary, nonlocal, and state-dependent terms do
not.  This is a selective normalized-response identity, not yet a
gravitational sequestering theorem.  No present result derives the map
\(v\to\int\sqrt{-g}\), an independent central variable \(\lambda_R\), or a
small cosmological integration constant.

\subsection{Entanglement free-energy attractor for the response horizon}

\textbf{Status: \StatusOpen.}
The earlier long-form draft considered the phenomenological ansatz
\[
\rho_\Lambda^{\rm ren}\simeq \xi M_{\rm Pl}^2L_R^{-2}.
\]
A candidate thermodynamic selection principle defines, for a causal diamond
of radius \(L\),
\[
F_R(L)=E_{\rm mod}(L)-T_HS_{\rm gen}(L),
\qquad
T_H=\frac{\hbar H}{2\pi}.
\]
If one assumes the late-time expansions
\[
E_{\rm mod}(L)=\alpha M_{\rm Pl}^2H^2L^3+\cdots,
\]
\[
T_HS_{\rm gen}(L)=\beta M_{\rm Pl}^2HL^2+\cdots.
\]
then
\[
F_R(L)=\alpha M_{\rm Pl}^2H^2L^3-
\beta M_{\rm Pl}^2HL^2+\cdots.
\]
The nonzero stationary point is
\[
\frac{\partial F_R}{\partial L}=0
\quad\Longrightarrow\quad
L_R=\frac{2\beta}{3\alpha}H^{-1}.
\]
The condition \(3\alpha=2\beta\) was previously imposed by causal-diamond
first-law matching at \(L=H^{-1}\).  It is not an independent microscopic
calculation of both coefficients and therefore cannot be used to derive
\(L_R=H^{-1}\) without a circularity concern.  A closed project must compute
\(\alpha\) and \(\beta\) in the same model before locating the extremum, prove
\(F_R''(L_R)>0\), and determine the relaxation rate and matter corrections.
Until then neither \(L_R=H^{-1}\) nor
\(\rho_\Lambda^{\rm ren}\sim M_{\rm Pl}^2H^2\) is a result of the framework.

\subsection{Noncommutative-geometric matter from influence algebras}

\textbf{Status: \StatusOpen.}
The residual noncommutative part of an infrared influence algebra is a
candidate input to a finite real spectral triple.  Let
\[
\I_{\rm IR}^{\rm rel}/Z(\I_{\rm IR})
\]
be the residual noncentral influence algebra after the Boolean center has been separated.  Assume its finite semisimple real form is
\[
\A_F=\bigoplus_iM_{n_i}(\mathbb K_i),
\qquad
\mathbb K_i\in\{\mathbb R,\mathbb C,\mathbb H\}.
\]
A finite real spectral triple is then
\[
(\A_F,\Hh_F,D_F,J_F,\gamma_F).
\]
The product triple
\[
\A=C^\infty(M)\otimes\A_F,
\qquad
D=D_M\otimes1+\gamma_5\otimes D_F
\]
has inner fluctuations
\[
D\mapsto D_A=D+A+JAJ^{-1},
\]
which generate gauge bosons, scalar fields, and fermion mass/Yukawa data.
Under a particular set of finite spectral-triple representation and
consistency conditions---including a real structure, first-order condition,
orientability, Poincare duality, unimodularity, anomaly cancellation, and a
minimal faithful charge representation---the familiar minimal
Standard-Model-like algebra is \cite{ConnesMarcolli,ChamseddineConnes}
\[
\A_F\simeq \mathbb C\oplus\mathbb H\oplus M_3(\mathbb C),
\]
with gauge group
\[
G_F\simeq S(U(2)\times U(3))
\simeq \frac{SU(3)\times SU(2)\times U(1)}{\mathbb Z_6}.
\]
The generation number is the index pairing
\[
N_{\rm gen}=\Index(D_A,p)=\langle \operatorname{ch}(D_A),[p]\rangle.
\]
These are consequences of imposing the spectral-triple axioms on a chosen
residual algebra.  RQCP-QG has not proved that its influence sectors form this
algebra, that the relevant tensor category reconstructs the Standard Model
gauge group, or that dynamics selects an index equal to three.  The first
closed matter project is therefore the reconstruction of one non-Abelian
gauge group from an explicit influence-sector category; the chiral-index
problem is separate.

\section{Status of three conditional mechanisms}\label{sec:v62-referee-hardening}

The following mechanisms have different logical roles.  Block-primitive
mixing is realized exactly in the charge-sector models, although not proved
for generic channels.  The
free-energy coefficient relation remains a matching condition rather than an
independent microscopic calculation.  The finite-real-spectral-triple route
is a conditional classification problem.  Keeping these roles separate is
more informative than presenting them as parallel derivations.

\subsection{Physical origin of asymptotic abelianness}

The condition
\[
\lim_{n\to\infty}\|[\Gamma^n(A),\Gamma^n(B)]\|_\omega=0
\]
characterizes a particular universality class of coarse-graining channels:
local, information-scrambling channels whose microscopic off-diagonal
information is discarded by the infrared observer.  In a causal graph or
lattice regularization, Lieb--Robinson locality gives an approximate light
cone for commutators \cite{LiebRobinson1972}, while ETH-type mixing
\cite{Deutsch1991,Srednicki1994} can suppress off-diagonal matrix
elements of coarse observables inside high-dimensional energy or response
shells.  The channel $\Gamma$ then acts like a conditional expectation onto
long-lived macroscopic labels plus a primitive mixing channel inside each
label sector.

A useful sufficient condition is the following block-primitive form.  Suppose there are mutually orthogonal projections $P_a$ with $\sum_aP_a=1$, $\Gamma(P_a)=P_a$, and block states $\omega_a$ such that
\[
\|P_a\Gamma^n(X)P_b\|_\omega\le C_Xe^{-\gamma n}\quad (a\ne b),
\]
and
\[
\|P_a\Gamma^n(X)P_a-\omega_a(P_aXP_a)P_a\|_\omega\le C_Xe^{-\gamma n}.
\]
Then
\[
\Gamma^n(X)\longrightarrow \sum_a\omega_a(P_aXP_a)P_a
\]
in the GNS seminorm.  The limiting range is the abelian pointer algebra $\oplus_a\mathbb C P_a$, and commutators decay exponentially:
\[
\|[\Gamma^n(A),\Gamma^n(B)]\|_\omega\le C_{A,B}e^{-\gamma n}.
\]
Thus asymptotic abelianness is not an independent metaphysical postulate; it follows in this channel class from local propagation, internal scrambling, and loss of microscopic phase information.

\begin{theorem}[Block-primitive coarse-graining implies Boolean IR records]\label{thm:block-primitive-abelian}
Let $\Gamma$ be a normal UCP channel preserving a faithful normal state
$\omega$.  Assume the block-primitive estimates above hold for a finite
family of invariant projections $\{P_a\}$.  Then the fixed observables of
$\Gamma$ are the abelian algebra generated by the \(P_a\).  In particular,
the infrared record projections form a Boolean algebra, and the stated
commutator defect decays exponentially.
\end{theorem}

\begin{proof}
The off-block estimate kills coherences between different sectors.  The primitive estimate inside each sector sends any local observable to a scalar multiple of $P_a$.  Summing over $a$ gives the stated GNS limit.  Scalar multiples of mutually orthogonal projections commute, hence the represented asymptotic range is abelian.  The commutator estimate follows by applying the triangle inequality to the difference between $\Gamma^n(A),\Gamma^n(B)$ and their abelian limits.
\end{proof}

Integrable systems, many-body-localized sectors, strictly topological nonabelian memories, or channels harboring nonabelian noiseless subsystems generally bypass these block-primitive estimates, predicting macroscopic persistent quantum records or separated topological properties.

\subsection{\texorpdfstring{Why $2\beta/3\alpha=1$ remains a matching condition}{Why 2 beta / 3 alpha = 1 remains a matching condition}}

The response-horizon free energy was written as
\[
F_R(L)=\alpha M_{\rm Pl}^2H^2L^3-\beta M_{\rm Pl}^2HL^2+\cdots .
\]
Earlier versions treated $\alpha$ and $\beta$ as two normalizations of a
common causal-diamond first law.  If that identification is imposed at the
de Sitter scale, the local matching condition is
\[
\left.\frac{\diff E_{\rm mod}}{\diff L}\right|_{L=H^{-1}}
=
\left.\frac{\diff (T_HS_{\rm gen})}{\diff L}\right|_{L=H^{-1}}.
\]
Since
\[
\frac{\diff E_{\rm mod}}{\diff L}=3\alpha M_{\rm Pl}^2H^2L^2,
\qquad
\frac{\diff(T_HS_{\rm gen})}{\diff L}=2\beta M_{\rm Pl}^2HL,
\]
the first law at the de Sitter response horizon $L=H^{-1}$ gives
\[
3\alpha=2\beta,
\qquad
\frac{2\beta}{3\alpha}=1.
\]
This calculation is internally consistent, but it does not independently
predict $L_R=H^{-1}$: the scale at which the first-law matching is imposed is
also the scale subsequently recovered from the extremum.  The relation
$3\alpha=2\beta$ is therefore a diagnostic condition on a candidate
microscopic model, not a present theorem of RQCP-QG.  A closed derivation must
compute $\alpha$ and $\beta$ from the same process without using the desired
horizon location, establish $F_R''(L_R)>0$, and derive the relaxation rate.
If that calculation gives $2\beta/3\alpha\ne1$, the stationary radius is
$L_R=(2\beta/3\alpha)H^{-1}$; if it gives equality, the equality is then an
output rather than a calibration.

\subsection{What is, and is not, proved about the NCG matter algebra}

The finite spectral-triple completion is a selection theorem.  An arbitrary
residual influence algebra $\I_{\rm rel}$ is not guaranteed to flow to
\[
\mathbb C\oplus\mathbb H\oplus M_3(\mathbb C).
\]
The claim instead takes a conditional form.  If there is a stable local
matter sector, if its residual noncommutative influence algebra has a finite
semisimple real form, and if that form satisfies the finite real
spectral-triple consistency conditions---reality, the first-order condition,
orientability, Poincare duality, unimodularity, anomaly cancellation, and a
minimal faithful charge representation---then the admissible algebra is
constrained to the Connes--Chamseddine Standard-Model-like class.  RQCP-QG
supplies a possible operator-algebraic origin for the finite internal
algebra; spectral-triple classification and the spectral action supply the
conditional matter-selection constraint.

\begin{remark}[Status of the $\I_{\rm rel}\to\A_F$ step]\label{rem:irel-af-status}
This analysis provides a conditional result rather than a universal basin-of-attraction: if the infrared residual matter algebra is finite, semisimple, anomaly-free, and compatible with the axioms of a real finite spectral triple, then its minimal Standard-Model-like realization is $\mathbb C\oplus\mathbb H\oplus M_3(\mathbb C)$ with gauge group $S(U(2)\times U(3))$.  Proving that a broad CSR-TPN universality class dynamically selects this algebra is a separate mathematical problem.
\end{remark}

\section{Exact finite-model results}
\label{sec:closed-loops}

The abstract influence-algebra construction identifies the kind of
information that may survive coarse graining, but it does not by itself
provide a microscopic dilation, a mixing rate, a phase transition, or a
geometric estimator.  This section records independent models that supply
these missing pieces.  They are presented together because they
constrain the architecture of RQCP-QG; they are not combined into one
generator.  In particular, a solid implication below is always local to the
stated model.

\subsection{Local unitaries to quantum-response kinetics}
\label{sec:loop-unitary-kinetics}

\paragraph{Assumption removed.}
The original process description began with a UCP coarse-graining channel.
That is too high a starting point for a microscopic theory: complete
positivity may be postulated while the Markov property and the kinetic rates
remain unexplained.  A collision circuit supplies an exact dilation of
the channel used in the Boolean sector model
\cite{XuResponseKinetics2026}.  The continuum generator belongs to the
standard Gorini--Kossakowski--Sudarshan--Lindblad class
\cite{Gorini1976,Lindblad1976}; the point here is that its rates and its
finite-step error follow from an explicit one-pass quantum network, rather
than being postulated at the semigroup level \cite{Chiribella2009}.

For an involution \(U=U^\dagger=U^{-1}\), let a system interact once with a
fresh qubit bath cell through
\begin{equation}
 V_\theta=\exp(-i\theta\,U\otimes X_b),
 \qquad
 \rho_b=\tfrac12\mathbf1_b .
 \label{eq:loop-collision-unitary}
\end{equation}
Tracing out the cell gives
\begin{equation}
 \Phi_{\theta,U}(X)
 =\cos^2\theta\,X+\sin^2\theta\,UXU
 =\exp\!\left[h r(\operatorname{Ad}_U-\operatorname{id})\right]X
 \label{eq:loop-one-collision}
\end{equation}
when \(\sin^2\theta=(1-e^{-2rh})/2\).  The equality is finite-step and
contains no weak-coupling expansion.  Taking \(U=Z_i\) produces local
dephasing; taking \(U=S_{ij}\), the SWAP on an edge, produces symmetric
exchange.  An edge coloring of a bounded-degree graph arranges the exchange
collisions into a constant-depth symmetric round.  Every collision uses a
new cell, so the global reduced circuit is exactly CPTP at every step.

Let \(G=(V,E)\) be connected and \(d\)-regular, and write
\begin{align}
 \mathcal D&=\kappa\sum_i(\operatorname{Ad}_{Z_i}-\operatorname{id}),
 \nonumber\\
 \mathcal A&=\frac{\gamma}{d}\sum_{\{i,j\}\in E}
   (\operatorname{Ad}_{S_{ij}}-\operatorname{id}),
 \qquad
 \mathcal L_{\rm kin}=\mathcal D+\mathcal A .
 \label{eq:loop-kinetic-generator}
\end{align}
For a Pauli string \(P_{\boldsymbol\sigma}\), let \(k\) be its number of
nonidentity labels and \(q\) its number of \(X\) or \(Y\) labels.  The
Heisenberg generator closes exactly:
\begin{equation}
 \mathcal L_{\rm kin}P_{\boldsymbol\sigma}
 =-2\kappa q\,P_{\boldsymbol\sigma}
 +\frac{\gamma}{d}\sum_{\{i,j\}\in E}
  \left(P_{s_{ij}\boldsymbol\sigma}-P_{\boldsymbol\sigma}\right).
 \label{eq:loop-pauli-hierarchy}
\end{equation}
The second term is a colored interchange process and the first kills every
transverse label.  Thus the hierarchy is not a low-moment closure: all Pauli
orders are invariant sectors and together span the full operator algebra.

\begin{theorem}[Finite-step charge algebra and fixed-order kinetic limit]
\label{thm:loop-unitary-kinetics}
Let \(\Phi_h\) be the symmetric collision round constructed from
Eqs.~\eqref{eq:loop-collision-unitary}--\eqref{eq:loop-kinetic-generator}.
For every \(h,\kappa,\gamma>0\),
\begin{equation}
 \operatorname{Fix}(\Phi_h)
 =\operatorname{span}\{P_0,\ldots,P_N\},
 \qquad
 P_m=\sum_{|x|=m}|x\rangle\langle x|.
 \label{eq:loop-finite-fixed}
\end{equation}
On an order-\(k\) Pauli orbit and for coefficient
\(\ell^p\) norm, \(p\in\{1,\infty\}\), \(t=nh\),
\begin{align}
 \|\Phi_h^n-e^{t\mathcal L_{\rm kin}}\|_{(k),p}
 &\le \varepsilon_k(t,h),\nonumber\\
 \varepsilon_k(t,h)
 &=\frac{t h^2}{3}B_k^3e^{hB_k},
 \qquad
 B_k\le\frac{2\chi k\gamma}{d},
 \label{eq:loop-kinetic-bound}
\end{align}
where \(\chi\) is the edge-chromatic number of the chosen schedule.  At
fixed \(k\) the bound is independent of \(N\).
\end{theorem}

\begin{proof}[Proof outline]
Each random-unitary factor is a positive Hilbert--Schmidt contraction.
Equality in the product contraction forces simultaneous commutation with
all \(Z_i\) and all edge SWAPs.  The first condition makes an observable
diagonal, while connected edge transpositions make its diagonal value
depend only on Hamming weight, proving
Eq.~\eqref{eq:loop-finite-fixed}.  On an order-\(k\) orbit, a matching has at
most \(k\) active edges.  The symmetric product and
\(e^{h\mathcal A}\) have equal derivatives through second order; a third
derivative remainder gives the one-step error
\(h^3B_k^3e^{hB_k}/3\).  Contractivity and a telescoping sum over \(n=t/h\)
steps prove Eq.~\eqref{eq:loop-kinetic-bound}.
\end{proof}

The response is operational.  Preparing
\(\rho_{\boldsymbol\sigma,\eta}
=2^{-N}(\mathbf1+\eta P_{\boldsymbol\sigma})\) and measuring
\(P_{\boldsymbol\tau}\) gives a matrix element of \(\Phi_h^n\).  In
particular,
\begin{align}
 R^Z_{j\leftarrow i}(t)
 &=\left[e^{-(\gamma/d)L_Gt}\right]_{ji},\nonumber\\
 R^X_{j\leftarrow i}(t)
 &=e^{-2\kappa t}
   \left[e^{-(\gamma/d)L_Gt}\right]_{ji},
 \label{eq:loop-local-response}
\end{align}
up to the controlled finite-step error
\eqref{eq:loop-kinetic-bound}.  On a ring, taking
\(t=N^2s/\gamma\) and a compatible \(h(N)\) yields the heat kernel on the
unit circle, with separately bounded product-formula and lattice errors.
This realizes the microscopic--kinetic--hydrodynamic separation advocated
in Sec.~\ref{sec:status-map}.

\paragraph{Boundary of the result.}
Fresh one-pass bath cells are a physical memory assumption.  Reusing the
same cell gives coherent recurrences, for example \(\cos(2n\theta)\), rather
than exponential decay.  The circuit also uses an external schedule and a
calibrated collision strength.  Equation~\eqref{eq:loop-kinetic-bound}
therefore derives the response kinetics under a controlled Markovian
dilation; it does not derive an autonomous clock, a thermal reservoir, or a
metric.

\subsection{Boolean charge records and metastable classical motion}
\label{sec:loop-boolean}

\paragraph{Assumption removed.}
The first finite open-system loop removes asymptotic abelianness as a
separate dynamical hypothesis.  For the Lindbladian
\(\mathcal L_G=\mathcal D+\mathcal A\) in
Eq.~\eqref{eq:loop-kinetic-generator}, the invariant algebra, its complete
gap, and the weakly perturbed sector dynamics are all calculable
\cite{XuBoolean2026}.  General fixed-algebra theory identifies stationary
operator structures of quantum dynamical semigroups
\cite{Frigerio1978}; the result below additionally fixes the complete
many-body gap and the topology dependence of the metastable window.

Let \(\mathcal E\) be the trace-preserving conditional expectation
\begin{equation}
 \mathcal E(X)=\sum_{m=0}^{N}
 \frac{\operatorname{Tr}(P_mX)}{\binom Nm}P_m .
 \label{eq:loop-boolean-expectation}
\end{equation}

\begin{theorem}[Noise-selected Boolean algebra and complete gap]
\label{thm:loop-boolean}
For connected \(d\)-regular \(G\),
\begin{align}
 \operatorname{Fix}(e^{t\mathcal L_G})
 &=\bigoplus_{m=0}^{N}\mathbb CP_m,\nonumber\\
 \lambda_G&=\min\!\left\{2\kappa,
 \frac{\gamma}{d}\lambda_2(L_G)\right\},
 \label{eq:loop-boolean-gap}\\
 \|[e^{t\mathcal L_G}(A),e^{t\mathcal L_G}(B)]\|_2
 &\le 2e^{-\lambda_Gt}
 \bigl[\|B\|_\infty\|A-\mathcal EA\|_2
 +\|A\|_\infty\|B-\mathcal EB\|_2\bigr].
 \label{eq:loop-boolean-comm}
\end{align}
The norm \(\|\cdot\|_2\) is the normalized Hilbert--Schmidt norm.
\end{theorem}

Local dephasing removes all occupation-basis off-diagonal matrix units.
On the diagonal subspace, exchange is the symmetric exclusion process.
Aldous' spectral-gap theorem identifies its gap with the one-particle graph
gap; a uniform one-excitation coherence attains the off-diagonal value
\(2\kappa\).  Thus neither the algebra nor the slow scale is inferred from a
small-system diagonalization.

The theorem preserves the full charge distribution:
\begin{equation}
 \rho_\infty=\sum_{m=0}^{N}
 p_m\,\frac{P_m}{\binom Nm},
 \qquad p_m=\operatorname{Tr}(P_m\rho).
 \label{eq:loop-boolean-state}
\end{equation}
All intrasection quantum information is erased, while the atoms \(P_m\)
remain exactly distinguishable.  The finite-time Boolean defect obeys
\begin{equation}
 \mathfrak b(t)=
 \sup_{\|A\|_\infty,\|B\|_\infty\le1}
 \|[e^{t\mathcal L_G}(A),e^{t\mathcal L_G}(B)]\|_2
 \le4e^{-\lambda_Gt}.
 \label{eq:loop-boolean-defect}
\end{equation}

Now perturb by
\(\mathcal L_{G,\epsilon}=\mathcal L_G+\epsilon\mathcal V\), where
\(\mathcal V\) generates a bistochastic Hilbert--Schmidt contraction, and
put \(v=\|\mathcal V\|_{2\to2}\),
\(\delta=\epsilon v/\lambda_G\).  The compressed operator
\(\mathcal K=\mathcal E\mathcal V\mathcal E\) is a backward classical Markov
generator on the atoms \(P_m\).  Standard spectral perturbation and a
Duhamel expansion give, on the metastable time window,
\begin{equation}
 \left\|
 e^{t\mathcal L_{G,\epsilon}}
 -e^{t\epsilon\mathcal K}\mathcal E
 \right\|_{2\to2}
 \le C\!\left(e^{-\lambda_Gt}+\delta
 +\epsilon^2 v^2t/\lambda_G\right).
 \label{eq:loop-metastable}
\end{equation}
Uniform bit flips preserve the sector manifold and give an exact Ehrenfest
chain.  A one-site flip does not commute with \(\mathcal E\), produces
genuine leakage, and is controlled by the same ratio \(\delta\).

Topology enters twice.  On rings
\(\lambda_2(L_G)=O(N^{-2})\), so a fixed isolation quality requires
\(\epsilon v/\gamma=O(N^{-2})\).  On bounded-degree expanders,
\(\lambda_2(L_G)=O(1)\), and a size-independent perturbation is tolerated.
The result therefore ties a graph property to the formation and robustness
of a classical information manifold.

\paragraph{Boundary of the result.}
The uniform theorem is in normalized Hilbert--Schmidt norm.  It controls
the stated operator-space defect and mixed-state families of bounded
normalized purity, but it is not a dimension-free diamond-norm or
worst-case local operator-norm theorem.  Nor does the model decide which
charge should encode a spacetime record.

\subsection{A transition between full quantum memory and Boolean records}
\label{sec:loop-memory}

\paragraph{Assumption removed.}
A fixed Boolean algebra does not show whether non-Abelian information and
classical records can be separated by a genuine dynamical phase transition.
This is realized by a contact-process flag field that gates local memory
dephasing \cite{XuMemory2026}.  The flag quotient is the classical contact
process \cite{Harris1974,Liggett1999}, while the object that crosses the
transition is an operational quantum memory channel.

Let \(n_i(t)\in\{0,1\}\) be contact-process occupation variables and let a
memory qudit at \(i\) dephase at rate \(\kappa\) only when \(n_i=1\).
After waiting to \(t_w\), an arbitrary state is written into a finite window
\(W\).  For matrix units
\(E_{\boldsymbol a\boldsymbol b}\), define
\(D(\boldsymbol a,\boldsymbol b)=\{i:a_i\ne b_i\}\).  The reduced channel
obeys the exact operator identity
\begin{align}
 \Phi^*_{s,t}(E_{\boldsymbol a\boldsymbol b})
 &=q_D(s,t)E_{\boldsymbol a\boldsymbol b},\nonumber\\
 q_D(s,t)&=\mathbb E\exp\!\left[
 -\kappa\int_s^t\!du\sum_{i\in D}n_i(u)\right].
 \label{eq:loop-memory-fk}
\end{align}
This is a Feynman--Kac representation of the entire memory channel, not a
formula for one chosen coherence.

Define the late-write algebra by taking the infinite-volume or
quasistationary flag limit before the unbounded storage interval:
\begin{equation}
 \mathcal A_{\rm tail}(W)=
 \left\{A:\lim_{t_w\to\infty}\sup_{\tau\ge0}
 \|\Phi^*_{t_w,t_w+\tau}(A)-A\|=0\right\}.
 \label{eq:loop-tail-definition}
\end{equation}

\begin{theorem}[Tail-algebra and capacity transition]
\label{thm:loop-memory}
For the contact process on \(\mathbb Z^r\), away from criticality,
\begin{equation}
 \mathcal A_{\rm tail}(W)=
 \begin{cases}
 M_{d^{|W|}},&\lambda<\lambda_c,\\
 \operatorname{diag}(M_{d^{|W|}})
 \simeq\mathbb C^{d^{|W|}},&\lambda>\lambda_c.
 \end{cases}
 \label{eq:loop-memory-algebra}
\end{equation}
Below threshold there are \(C,\alpha>0\) such that
\begin{equation}
 1-q_D(t_w,\infty)
 \le\frac{\kappa|D|C}{\alpha}e^{-\alpha t_w}.
 \label{eq:loop-memory-bound}
\end{equation}
For one memory qubit entangled with an untouched reference,
\begin{equation}
 \mathcal N_{R:M}=\frac q2,
 \qquad
 Q_{\rm mem}(q)=
 1-h_2\!\left(\frac{1-q}{2}\right)
 \label{eq:loop-memory-capacity}
\end{equation}
for independently reinitialized channel uses.
\end{theorem}

Subcritical sharpness makes the late integrated activity finite and proves
\eqref{eq:loop-memory-bound}.  In the upper invariant supercritical process,
ergodicity makes the occupation integral diverge for every nonempty \(D\),
so only diagonal matrix units survive.  The one-qubit map is a phase-flip
channel with \(p=(1-q)/2\); its Choi state has negativity \(q/2\), and
degradability gives Eq.~\eqref{eq:loop-memory-capacity}.  The same scalar
therefore controls noncommutativity, entanglement transmission, and
memoryless-use quantum capacity.

Quantum absorbing-state models can support coherence in the activity sector
itself and can depart from directed-percolation criticality
\cite{WamplerCooper2025}.  That is not the mechanism used here.  The flags
form a classical quotient; the additional result is the exact matrix-unit
Feynman--Kac action, which determines the full tail algebra and both
reference-entanglement and capacity witnesses without a moment closure.

\paragraph{Boundary of the result.}
On every finite connected graph the flags eventually absorb and all memories
become dark.  The phase transition belongs to the stated thermodynamic or
quasistationary order of limits.  The activity quotient is classical and
does not receive backaction from the memories.  A nonzero background
dephasing rate rounds the ideal infinite-time distinction into a lifetime
crossover.  These qualifications do not alter the exact channel identity,
but they delimit the phase claim.

\subsection{Causal order as an absorbing dissipative phase}
\label{sec:loop-causal}

\paragraph{Assumption removed.}
The earlier causal-order criterion assumed a vertex function \(T\) satisfying
\(u\to v\Rightarrow T(v)>T(u)\).  A local defect model produces the DAG
first and constructs \(T\) afterward \cite{XuCausalOrder2026}.

The underlying undirected graph is a chain, or more generally a locally
finite tree, of triangular blocks.  Each edge carries a qubit whose monitored
record is an orientation.  In block \(r\), the two cyclic words \(000\) and
\(111\) define
\begin{equation}
 C_r=|000\rangle\langle000|_r+|111\rangle\langle111|_r,
 \qquad Q_r=\mathbf1-C_r .
 \label{eq:loop-cycle-projector}
\end{equation}
Local Lindblad jumps heal a cyclic block at total rate \(g\), while a cyclic
neighbor infects a transitive target at rate \(1/2\).  The jump set is
covariant under global arrow reversal, so no global direction is encoded in
the generator.

\begin{theorem}[Order from defect extinction]
\label{thm:loop-causal-order}
For every recorded orientation \(\omega\) of a block cactus, the threshold
response graph is acyclic if and only if \(C_r(\omega)=0\) for every block.
When this holds,
\begin{equation}
 T_\omega(v)=\max_{p:\,p\to v}|p|
 \label{eq:loop-longest-path}
\end{equation}
is finite, obeys \(u\to v\Rightarrow T_\omega(v)>T_\omega(u)\), and the
reachability relation is a locally finite strict partial order.
\end{theorem}

The block-cut tree is essential: every simple undirected cycle belongs to one
block.  Hence the local projectors certify all directed cycles.  On a generic
lattice a triangle test can miss a longer chordless loop.

The commuting algebra generated by \(\{C_r\}\) is invariant and has the exact
classical restriction
\begin{align}
 (\mathcal K_gf)(\boldsymbol n)
 ={}&g\sum_r n_r
 [f(\boldsymbol n^{r,0})-f(\boldsymbol n)]\nonumber\\
 &+\frac12\sum_r(1-n_r)(n_{r-1}+n_{r+1})
 [f(\boldsymbol n^{r,1})-f(\boldsymbol n)] .
 \label{eq:loop-contact-generator}
\end{align}
After rescaling time, this is the one-dimensional contact process with
\(\lambda=1/g\) \cite{Harris1974,Liggett1999}.  Thus
\begin{equation}
 g_c=\lambda_c^{-1}=0.3032280(18)
 \label{eq:loop-causal-gc}
\end{equation}
using the accepted numerical contact-process threshold.  In the ordered
phase, every fixed operational window \(W\) satisfies
\begin{equation}
 1-P_{\rm DAG}(W,t)
 \le A(g)|W|e^{-\alpha(g)t},
 \qquad
 P_{\rm DAG}(W,t)=
 \operatorname{Tr}\!\left[\rho(t)\prod_{r\in W}Q_r\right].
 \label{eq:loop-causal-window}
\end{equation}
The vanishing rate \(\alpha(g)\) is an order-formation gap, and the exact
quotient transfers the directed-percolation critical exponents to cycle
density and causal correlation scales.

\paragraph{Boundary of the result.}
The model assumes the block-cactus adjacency, the edge registers, the local
laboratory runtime, monitoring, and a thresholded response readout.  It
generates neither the undirected graph nor a volume measure.  Its theorem
removes an assumed Lyapunov ranking on a nontrivial graph family; it is not
yet a background-free spacetime theorem.

\subsection{Quantum-field response as an inverse Lorentzian observable}
\label{sec:loop-lorentz}

\paragraph{Assumption removed.}
The first inverse-geometry model encoded proper time into a calibrated
depolarizing parameter.  That construction has been replaced by a physical
field response: a Ramsey ratio isolates a Wightman covariance, whose
rank defect determines an unknown de Sitter radius
\cite{XuLorentz2026}.  Scalar propagators and localized quantum probes are
known to carry metric and curvature information
\cite{SaravaniAslanbeigiKempf2016,PercheMartinMartinez2022}, while
Lorentzian multidimensional scaling embeds supplied interval data
\cite{CloughEvans2017}.  The distinct claim tested here is partial-data
completion from a physical covariance kernel: anchor responses fix an
unknown curvature scale and predict target pairs never used in the fit.

For real spacelike field smearings \(\Phi_i,\Phi_j\), two qubit probes undergo
\begin{equation}
 U=\exp[-ig(\sigma_z^i\Phi_i+\sigma_z^j\Phi_j)] .
 \label{eq:loop-ramsey-unitary}
\end{equation}
In a zero-mean quasifree state their same- and opposite-sign coherences are
\(\chi_\pm=\exp\{-2g^2[V_i+V_j\pm2W_{ij}]\}\), so
\begin{equation}
 W_{ij}=\frac{1}{8g^2}\log\frac{\chi_-}{\chi_+}.
 \label{eq:loop-ramsey-ratio}
\end{equation}
The local self-variances cancel exactly.

For a conformally coupled massless scalar in the Bunch--Davies state on
\(\mathrm{dS}_4(R)\),
\begin{equation}
 W_{ij}=\frac{1}{8\pi^2R^2(1-Z_{ij})}
 \label{eq:loop-ds-kernel}
\end{equation}
on the spacelike branch.  For six generic anchors define
\(A_{ii}=0\), \(A_{ij}=(8\pi^2W_{ij})^{-1}\).

\begin{theorem}[Finite-anchor curvature and blind completion]
\label{thm:loop-lorentz}
If \(A\) is invertible and the six ambient anchor vectors span
\(\mathbb R^{1,4}\), then
\begin{equation}
 R^2=(\mathbf1^TA^{-1}\mathbf1)^{-1}.
 \label{eq:loop-radius}
\end{equation}
The candidate is admissible only if
\(G=R^2\mathbf1\mathbf1^T-A\) has rank five and Lorentzian inertia
\((-++++0)\).  With \(K\ge6\) anchors and
\(b_x=R^2\mathbf1-A_{Ax}\), every withheld target--target invariant is
\begin{equation}
 \widehat Z_{xy}=\frac{b_x^TG^+b_y}{R^2}.
 \label{eq:loop-blind-completion}
\end{equation}
\end{theorem}

The determinant lemma applied to
\(G=X\eta X^T\) proves Eq.~\eqref{eq:loop-radius}; the pseudoinverse identity
proves Eq.~\eqref{eq:loop-blind-completion}.  Redundant anchor sextuples,
Lorentzian inertia, target hyperboloid norms, and unopened target pairs
overidentify the model.  Perturbation bounds propagate Ramsey shot noise,
finite smearing, and Gram conditioning into errors on \(R\) and
\(\widehat Z_{xy}\).  The blind pairs can therefore falsify the assumed
geometry--state kernel rather than merely reproduce data used in a fit.

\paragraph{Boundary of the result.}
The field species, Bunch--Davies state, de Sitter model class, spacelike
support, canonical normalization, probe coupling, and a nondegenerate anchor
frame are inputs.  Equation~\eqref{eq:loop-ramsey-ratio} is physical, but the
kernel \eqref{eq:loop-ds-kernel} supplies the geometric model.  The theorem is
an inverse problem inside that class; it does not derive a Lorentzian
manifold from the causal-order model or determine temporal orientation of a
timelike held-out pair.

\subsection{A fixed physical modular band}
\label{sec:loop-modular}

\paragraph{Assumption removed.}
Earlier lattice tests retained a fixed number of modular modes.  As the
cutoff is removed, that prescription narrows the physical band and can make
convergence tautological.  A state-independent band selected by the local
lattice-boost operator removes this loophole
\cite{XuModular2026}.

For the half-filled free-fermion interval, let \(C_N\) be the restricted
correlation matrix and
\begin{equation}
 h_N=\log[(\mathbf1-C_N)C_N^{-1}]
 \label{eq:loop-modular-h}
\end{equation}
the exact one-particle modular Hamiltonian.  The known tridiagonal commuting
operator determines a local boost matrix \(B_N\), with
\([B_N,h_N]=0\).  This construction is rooted in the
Bisognano--Wichmann form and its lattice implementations
\cite{BisognanoWichmann1975,Giudici2018}.  The operational band
\begin{equation}
 Q_\Lambda=\mathbf1_{[-\Lambda,\Lambda]}(B_N)
 \label{eq:loop-modular-band}
\end{equation}
is fixed before \(h_N\) is diagonalized; its rank grows with \(N\) at fixed
\(\Lambda\).

For any real Borel function \(g\), put \(G_N=g(B_N)\) and
\(E_N=h_N-G_N\).  Every Hermitian correlation perturbation satisfies
\begin{align}
 |\operatorname{Tr}(\delta C E_N)|
 \le{}&\delta_\Lambda[G_N]\,
 \|Q_\Lambda\delta C Q_\Lambda\|_1\nonumber\\
 &+\delta_\Lambda^\perp[G_N]\,
 \|\bar Q_\Lambda\delta C\bar Q_\Lambda\|_1,
 \label{eq:loop-modular-certificate}\\
 \delta_\Lambda[G_N]
 &=\|Q_\Lambda E_NQ_\Lambda\|_\infty,\qquad
 \delta_\Lambda^\perp[G_N]
 =\|\bar Q_\Lambda E_N\bar Q_\Lambda\|_\infty .
\end{align}
The in-band constant is minimax optimal for trace-neutral Gaussian tangents.
It therefore certifies all linear modular responses in the stated physical
band, rather than only one prepared wave packet.

Common-spectrum analysis gives the sparse correction
\begin{equation}
 \widetilde B_N=B_N+\frac1{N^2}
 \left[-\frac14B_N+\frac{1}{4\pi^2}B_N^3\right].
 \label{eq:loop-range-three}
\end{equation}
It contains only first- and third-neighbor hopping.  The cubic coefficient is
analytic; the linear coefficient was selected on \(N\le128\), frozen, and
tested without refitting on \(N=192,\ldots,768\).  For \(\Lambda=8\), the
archived deterministic data give
\begin{equation}
 \delta_8[B_N]\propto N^{-1.992},
 \qquad
 \delta_8[\widetilde B_N]\propto N^{-3.984},
 \label{eq:loop-modular-scaling}
\end{equation}
with mode-entry-aware fourth-order envelopes on the held-out sizes.  The
first relation is an analytic finite-band response certificate; the
\(N^{-4}\) corrected law is controlled numerical evidence, not an all-size
theorem.

\paragraph{Boundary of the result.}
This loop concerns a one-dimensional free state and a modular-to-lattice-
boost response.  It contains no transverse area, multiple null directions,
Newton constant, focusing equation, conservation completion, or dynamical
metric.  Consequently it is not an Einstein derivation.  Interacting and
higher-dimensional versions require new operator estimates.

\subsection{Dynamical charge center and selective fixed-sector response}
\label{sec:loop-center}

\paragraph{Assumption removed.}
A central variable should not be attached by hand.  A local
number-conserving channel derives a complete additive-charge fixed
center, its formation gap, and its lifting under weak symmetry breaking
\cite{XuCentral2026}.

On a connected graph let \(n_i\) be hard-core occupations,
\(\widehat V=\sum_i n_i\), and use local dephasing together with symmetric
number-conserving hops \(L_{ij}=\sigma_i^-\sigma_j^+\).  The reversible
Lindbladian has Dirichlet form
\begin{align}
 -\langle X,\mathcal L^*(X)\rangle
 ={}&\frac{\gamma}{2}\sum_i\|[n_i,X]\|_2^2\nonumber\\
 &+\frac{\eta}{2}\sum_{\{i,j\}\in E}
 \left(\|[L_{ij},X]\|_2^2+\|[L_{ji},X]\|_2^2\right).
 \label{eq:loop-center-dirichlet}
\end{align}

\begin{theorem}[Dynamical additive-charge center]
\label{thm:loop-center}
Let \(P_v\) project onto \(\widehat V=v\).  Then
\begin{align}
 \ker\mathcal L^*&=\operatorname{span}\{P_0,\ldots,P_N\},\nonumber\\
 \mathcal E_V(X)&=\sum_{v=0}^{N}
 \frac{\operatorname{Tr}(P_vX)}{\binom Nv}P_v,\nonumber\\
 \|e^{t\mathcal L^*}X-\mathcal E_VX\|_2
 &\le e^{-\Delta t}\|X-\mathcal E_VX\|_2,\qquad
 \Delta\ge\min\{\gamma/2,\eta\lambda_2(L_G)\}.
 \label{eq:loop-center-gap}
\end{align}
The one-point density profile closes exactly:
\begin{equation}
 \mathcal L^*(n_i)=\eta\sum_{j\sim i}(n_j-n_i).
 \label{eq:loop-center-heat}
\end{equation}
\end{theorem}

The fixed projections are dynamically produced, not an appended register.
They are central in the fixed algebra and for charge-preserving
laboratories, but not in the microscopic algebra
\(\mathcal B(\mathcal H)\).  Homogeneous weak flips give an exact Ehrenfest
chain on \(v=0,\ldots,N\); in the balanced stationary state,
\begin{equation}
 \langle\delta\widehat V(t)\delta\widehat V(0)\rangle
 =\frac N4e^{-2\epsilon t}.
 \label{eq:loop-center-correlation}
\end{equation}
This separates the graph-controlled formation time from the
symmetry-breaking charge-memory time.

There is a precise, but conditional, field-theory consequence.  If \(v\) is
identified with fixed Euclidean metric volume and the sector response is
normalized as \(Z_v[J]/Z_v[0]\), a pure volume counterterm \(c v\) cancels:
\begin{equation}
 \frac{Z_v[J;c]}{Z_v[0;c]}
 =\frac{e^{-cv}Z_v[J;0]}{e^{-cv}Z_v[0;0]}
 =\frac{Z_v[J;0]}{Z_v[0;0]} .
 \label{eq:loop-fixed-volume-cancellation}
\end{equation}
A coupling to an operator \(Q\) that is not constant inside the sector does
not cancel; its derivative is the connected covariance
\(-\operatorname{Cov}(O,Q)\).  Thus the statement is selective: it removes
only the identity or pure-volume part from conditional response and leaves
curvature-dependent counterterms to ordinary effective-field-theory
renormalization.  Fixed-total-volume and unimodular formulations provide
the closest gravitational precedent \cite{FiolGarriga2010}, while global
sequestering uses a different set of constraints \cite{KaloperPadilla}.
The microscopic result here is the local charge-center theorem; identifying
that charge with metric volume remains a separate bridge.

\paragraph{Boundary of the result.}
The microscopic \(U(1)\) law is assumed.  No theorem identifies
\(\widehat V\) with metric four-volume, derives sectorwise gravitational
normalization, predicts a cosmological constant, or explains why the
appropriate laboratory is confined to one volume sector.  Equation
\eqref{eq:loop-fixed-volume-cancellation} is therefore a conditional
interface attached to an exact charge-center theorem, not a solution of the
cosmological-constant problem.

\subsection{What these loops do and do not compose}
\label{sec:loop-noncomposition}

The completed loops remove several assumptions independently:
\begin{equation}
\begin{array}{rcl}
\text{postulated UCP kinetics}
&\Longrightarrow&\text{fresh-cell unitary dilation},\\
\text{assumed asymptotic abelianness}
&\Longrightarrow&\text{exact Boolean algebra and gap},\\
\text{static algebraic dichotomy}
&\Longrightarrow&\text{memory-capacity transition},\\
\text{assumed Lyapunov ranking}
&\Longrightarrow&\text{DAG followed by derived rank},\\
\text{calibrated proper-time oracle}
&\Longrightarrow&\text{field-covariance inverse geometry},\\
\text{shrinking fixed-mode audit}
&\Longrightarrow&\text{fixed physical modular band},\\
\text{attached central label}
&\Longrightarrow&\text{dynamical additive-charge center}.
\end{array}
\label{eq:loop-removals}
\end{equation}
No arrow in Eq.~\eqref{eq:loop-removals} licenses the replacement
\[
\text{collision graph}\equiv\text{causal graph}\equiv
\text{de Sitter anchors}\equiv\text{fermion interval}\equiv
\text{metric-volume sectors}.
\]
Those identifications would be additional assumptions.  A future unified
model must define common microscopic degrees of freedom and prove that its
distinct scaling regimes reduce to the models above, with an error budget
that remains controlled across both arrows
\[
\text{microscopic process}\longrightarrow\text{kinetic response}
\longrightarrow\text{order--volume geometry}.
\]
This noncomposition statement is as important as the positive theorems: it
prevents evidence obtained in one solvable model from being counted twice as
support for a different bridge.

\section{Operational and algebraic process foundations}
\label{sec:process-algebra}

A finite process matrix assigns probabilities to local operations by
\[
p(\vec a|\vec x)=\Tr\left[W_C\bigotimes_{i\in C}M_i^{a_i|x_i}\right],
\]
where $M_i^{a_i|x_i}$ are CP-instrument Choi operators and $W_C$ is a positive process matrix normalized on deterministic operations.  Process-matrix theory allows local quantum laboratories without assuming a fixed global causal order \cite{OCB2012}.

For infinite-dimensional and field-theoretic settings one replaces $W_C$ by a normal positive multilinear process functional.

\begin{definition}[Algebraic process functional]
Let $C=\{\N_i\}_{i=1}^n$ be a finite family of local von Neumann algebras.  An algebraic process functional is a positive multilinear map
\[
\Omega_C:\prod_i \mathrm{CP}_\sigma(\N_i)\to[0,1]
\]
which is normalized on deterministic normal UCP operations:
\[
\Omega_C(\Phi_1,\ldots,\Phi_n)=1
\]
whenever every $\Phi_i$ is deterministic.
\end{definition}

This object is operational: it maps local interventions to probabilities or states without first postulating a classical spacetime background.

\section{Background-relative causal influence and exact facts}

Fix a target algebra $\M_Y$ and source $X$.  A background strategy class $\B$ acts on the complement of $X\cup Y$.  For $\beta\in\B$ and source operation $\Phi_X$, define
\[
r_Y^{C;\beta}(\Phi_X)(A)=\Omega_C(\Phi_X,A,\beta),\qquad A\in\M_Y.
\]
Response differences are
\[
\delta_Y^{C;\beta}(\Phi_X,\Phi_X')
=r_Y^{C;\beta}(\Phi_X)-r_Y^{C;\beta}(\Phi_X').
\]
Let $\D_Y^{C;\B}$ be the family of all such differences.

\begin{definition}[Influence algebra]
\[
\Stab(\D)=\{U\in\mathcal U(\M_Y):\delta(U^*AU)=\delta(A),\ \forall A,\delta\},
\]
and
\[
\boxed{\I_Y^{C;\B}=\Stab(\D_Y^{C;\B})'\cap\M_Y.}
\]
\end{definition}

In finite dimension, if $\delta_k(A)=\Tr(\Delta_k A)$, then
\[
\I_Y^{C;\B}=\vN\{\Delta_k\}.
\]

\begin{definition}[Exact event]
An exact RQCP event is a tuple
\[
e=(C,Y,\B,P),\qquad P\in\Ev_0(Y|C;\B),
\]
where
\[
\boxed{\Ev_0(Y|C;\B)=\Proj(Z(\I_Y^{C;\B})).}
\]
\end{definition}

\begin{theorem}[Boolean facts]
For any von Neumann algebra $N$, $\Proj(Z(N))$ is a complete Boolean algebra with
\[
P\wedge Q=PQ,\quad P\vee Q=P+Q-PQ,
\quad \neg P=1-P.
\]
\end{theorem}

\begin{proof}
$Z(N)$ is commutative and monotone complete; hence its projections form a complete Boolean projection lattice.
\end{proof}

If $P,Q$ are noncommuting projections, $PQ$ is not an orthogonal projection.  Therefore approximate records are effects, not exact Boolean facts.

\section{Split-record laboratories}

QFT local algebras are typically Type III.  RQCP does not demand atomic local projections.  Instead it uses split inclusions and record conditional expectations.  Under standard nuclearity/split hypotheses, one has
\[
\N\subset\F\subset\widetilde\N,
\]
with $\F$ a Type I factor.  This is an external AQFT input, connected to the Doplicher-Longo split framework \cite{DoplicherLongo}.

A Type I split buffer alone has trivial center.  Classical records arise from a record decomposition
\[
\Hh_{\rm lab}=\bigoplus_a \Hh_a,
\qquad P_a:\Hh_{\rm lab}\to\Hh_a.
\]
The record expectation is
\[
E_R(A)=\sum_a P_aAP_a,
\]
with range
\[
\R=\bigoplus_a\mathcal B(\Hh_a),
\qquad Z(\R)=\left\{\sum_a\lambda_aP_a\right\}.
\]
If $\I^{C;\B_R}=\R$, then exact facts are exactly record-center projections.

\section{Causal renormalization and multiplicative-domain stability}

Let $\Gamma:\A_{\rm coarse}\to\A_{\rm fine}$ be a normal UCP coarse-graining map in Heisenberg picture.  A UCP map need not preserve products.  Its multiplicative domain is
\[
\MD(\Gamma)=\{a:\Gamma(a^*a)=\Gamma(a)^*\Gamma(a),\ \Gamma(aa^*)=\Gamma(a)\Gamma(a)^*\}.
\]

If $P=P^*=P^2$ and $P\in\MD(\Gamma)$, then $\Gamma(P)$ is again a projection.  Hence exact facts are RG-stable if
\[
Z(\I_{\ell+1})\subset\MD(\Gamma_\ell),
\qquad
\Gamma_\ell(Z(\I_{\ell+1}))\subset Z(\I_\ell).
\]

\begin{theorem}[Boolean RG homomorphism]
Under the two conditions above,
\[
\Gamma_\ell:\Proj(Z(\I_{\ell+1}))\to\Proj(Z(\I_\ell))
\]
is a Boolean algebra homomorphism.
\end{theorem}

\begin{proof}
For central $P,Q$ in the multiplicative domain,
\[
\Gamma(PQ)=\Gamma(P)\Gamma(Q),\qquad \Gamma(1-P)=1-\Gamma(P).
\]
Linearity gives joins; center covariance places images in $Z(\I_\ell)$.
\end{proof}

\section{Transfer-channel spectral gap and threshold-edge stability}

A transfer-channel gap controls when threshold decisions stop changing; it
does not by itself make the limiting graph acyclic.  Let $\Gamma$ be a UCP
transfer channel with fixed-point algebra
\[
\Z=\Fix(\Gamma)=Z(\R),
\]
and let $E_\Z$ be the conditional expectation onto $\Z$.  Suppose
\[
r(\Gamma|_{\ker E_\Z})=r<1.
\]
Then for any $\rho\in(r,1)$ there is $C_\rho$ such that
\[
\boxed{\|\Gamma^k(A)-E_\Z(A)\|\le C_\rho\rho^k\|A\|.}
\]

Let response edge weights be
\[
w_k(e)=\ell_e(\Gamma^k(B_e))
\]
for bounded predual functionals $\ell_e$.  Then
\[
|w_k(e)-w_\infty(e)|\le C_\rho\rho^k\|\ell_e\|\|B_e\|.
\]
If a finite observation window $W$ has threshold margin
\[
m_W=\min_{e\in W}|w_\infty(e)-\tau|>0,
\]
then for
\[
k\ge \frac{\log(2C_\rho L_W/m_W)}{-\log\rho},
\]
where $L_W=\max_{e\in W}\|\ell_e\|\|B_e\|$, the edge set in $W$ is stable.
If the limiting edge set in \(W\) is already acyclic, its transitive closure
is stable as well.  Acyclicity is a separate dynamical conclusion in the
block-cactus defect model of Section~\ref{sec:loop-causal}; it is not a
consequence of spectral convergence alone.


\section{The information-to-order core and its present boundary}
\label{sec:updated-core-boundary}

Figure~\ref{fig:rqcp-core-mechanism} is the organizing picture for the
foundational part of RQCP-QG.  It contains two logically distinct layers.
Panels 1--5 describe the passage
\begin{equation}
 \text{local quantum observables}
 \longrightarrow
 \text{UCP coarse graining}
 \longrightarrow
 \text{infrared algebra}
 \longrightarrow
 \text{Boolean records},
 \label{eq:updated-information-core}
\end{equation}
while panel 6 asks when response records define a directed acyclic graph.
The abstract statements concern normal process functionals, influence
algebras, Choi--Effros ranges, and center projections.  The dynamical
statements require the explicit models of
Sec.~\ref{sec:closed-loops}.  Conflating those two levels would turn an
existence result into a false universality claim.

A previously used sufficient condition guaranteed acyclicity by prescribing
a Lyapunov function:
\begin{equation}
 i\to j,\quad T(j)>T(i)
 \quad\Longrightarrow\quad
 \text{no directed cycle}.
 \label{eq:updated-lyapunov-sufficient}
\end{equation}
It is no longer the strongest available statement.  In the block-cactus
model, every directed cycle is detected by a bounded local projector, the
defect algebra has an exact contact-process quotient, and extinction of all
defects makes the recorded graph acyclic.  Only after this event is the rank
\begin{equation}
 T_\omega(v)=\max_{p:\,p\to v}|p|
 \label{eq:updated-derived-rank}
\end{equation}
defined.  Thus the model realizes
\[
 \text{local dissipative dynamics}
 \longrightarrow
 \text{cycle extinction}
 \longrightarrow
 \text{DAG}
 \longrightarrow
 \text{derived Lyapunov ranking}.
\]
The distinction is substantive: Eq.~\eqref{eq:updated-derived-rank} is a
property of a realized response record, whereas the semigroup parameter is
the laboratory duration used to generate the record.

The result does not yet compose with the preceding Boolean model.  The
dephasing--exchange graph and the oriented block-cactus graph use different
degrees of freedom and different generators.  A unified model would need an
invariant subalgebra carrying both the Boolean record atoms and the cycle
projectors, plus a proof that the gap controlling record formation is
compatible with the critical closing of the order-formation gap.  That
compatibility is an open theorem, not a matter of relabeling one graph as the
other.

\section{A closed affine order--volume and linearized-gravity bridge}
\label{sec:affine-null-wire-closure}

\textbf{Status: \StatusModel\ for the finite algebraic statements and
\StatusNumerical\ for the archived regulator calculations.}

The affine regulator provides one common-model realization of the
order--volume and gravity interfaces.  The generic order--volume map is
treated separately in Sec.~\ref{sec:updated-order-volume}; the
fully autonomous nonlinear same-update bridge remains open, although
Sec.~\ref{sec:same-update-einstein} closes its finite action-phase
subproblem and Sec.~\ref{sec:three-decisive-chains} closes the relational
clock, quadratic constraint, and balanced-regulator subproblems under their
own visible inputs.  The companion bridge manuscript
\cite{XuOrderVolumeEinstein2026} constructs an affine null-wire regulator in
which order, information volume, metric reconstruction, transverse modular
response, focusing, and tensor completion are interfaces of the same model.
It is not obtained by identifying the block-cactus process of PRL-1 with the
de Sitter detector model of PRL-3.

\subsection{Order and information volume}

Let \(X\subset\mathbb Z^4\) be a finite contractible cubical complex.  Edge
response modes \(a_e\) are acted on by plaquette-local jumps
\begin{equation}
 L_p=\sqrt\kappa\sum_e(d_1)_{pe}a_e ,
 \qquad
 \dot z=-\frac{\kappa}{2}d_1^\dagger d_1z ,
 \label{eq:affine-hodge-flow}
\end{equation}
where \(z_e=\langle a_e\rangle\) and \(d_1d_0=0\).  A fresh-bath
beam-splitter collision realizes each jump; a local product formula gives the
summed semigroup.  Since \(H^1(X)=0\),
\begin{equation}
 z_\infty\in\ker d_1=\operatorname{im}d_0,
 \qquad z_\infty=d_0T .
 \label{eq:affine-exact-form}
\end{equation}
Thresholded responses with a positive limiting margin are therefore oriented
by differences \(T_y-T_x\).  Their directed graph is acyclic because the sum
of strictly positive differences around a cycle would vanish.  The laboratory
semigroup parameter, the recovered rank \(T\), and the affine coordinate
\(x^0\) are distinct.

At vertex \(x\), local dephasing leaves a Boolean record algebra
\(\mathcal A_x^{\rm rec}\simeq\mathbb C^{q_x}\).  Its zero-error classical
capacity gives
\begin{equation}
 \nu(R)=\sum_{x\in R}\log q_x,\qquad
 \mu(R)=\ell_I^4\nu(R).
 \label{eq:affine-information-volume}
\end{equation}
Thus the model produces both an order and a positive additive measure.  The
single dimensionful conversion \(\ell_I^4\) is an explicit calibration.

\subsection{Finite null-frame metric}

The response-active unoriented displacement classes consist of six vectors
of type \((1,\pm1,0,0)\) and permutations, and 24 vectors whose time
component is \(3\) and whose spatial absolute values are
\(\{1,2,2\}\).  For a symmetric tensor \(g\), define
\begin{equation}
 (Ag)_r=v_r^\mu g_{\mu\nu}v_r^\nu .
 \label{eq:affine-design}
\end{equation}
An integer \(9\times9\) minor of \(A\) has determinant \(-32768\).  Hence
\begin{equation}
 \operatorname{rank}A=9,\qquad
 \ker A=\operatorname{span}\{\eta\},\qquad
 \eta=\operatorname{diag}(-1,1,1,1).
 \label{eq:affine-rank-nine}
\end{equation}
The normalized frame has smallest positive singular value \(0.684225689\)
and condition number \(3.635381686\).  The response cone therefore fixes a
conformal Lorentzian metric, while the information volume fixes its scale by
\begin{equation}
 \sqrt{|\det g|}=\mu_0 .
 \label{eq:affine-scale}
\end{equation}
This is a finite cell reconstruction, not a generic curved
metric--measure-limit theorem.  In particular, the affine complex and the
30-class response cone are microscopic inputs.

\subsection{Area response and tensor completion}

Along each of the same 30 directions place a bundle of critical
free-fermion wires.  The fixed-physical-band modular certificate of
Sec.~\ref{sec:loop-modular} is additive over \(N_\perp\) wires.  Since the
transverse cut has \(A_\perp=N_\perp a^2\), division by area cancels
\(N_\perp\) exactly at finite regulator size.  One maximally entangled
\(q_g\)-record pair per cut cell gives
\begin{equation}
 S_\partial=A_\perp\frac{\log q_g}{a^2},
 \qquad
 G_{\rm eff}=\frac{a^2}{4\hbar\log q_g}.
 \label{eq:affine-geff}
\end{equation}
A stated local capacity-exchange channel is primitive and unital within each
conserved capacity sector, so its full sector entropy is stationary to first
order.  The split relation
\(\delta S_\partial+\delta S_{\rm bulk}=0\) additionally uses a full-rank
product reference on the cut and adjacent wire algebras.  Fixed-sector
primitivity does not select this product preparation; its failure is retained
as a measurable stationarity-interface error.

The lattice interval boost and the Dirichlet Green function share the
parabolic weight.  Discrete focusing therefore gives, for every translated
subdiamond and direction \(r\),
\begin{equation}
 \left(
 G_{\mu\nu}^{\rm L}-8\pi G_{\rm eff}T_{\mu\nu}
 \right)v_r^\mu v_r^\nu=\varepsilon_r .
 \label{eq:affine-null-equations}
\end{equation}
The family of translated kernels separates lattice sites.  With
\(E_{\mu\nu}=G_{\mu\nu}^{\rm L}-8\pi G_{\rm eff}T_{\mu\nu}\), the same
rank-nine frame yields
\begin{equation}
 \|P_g^\perp E\|_2
 \le\frac{\|\boldsymbol\varepsilon\|_2}{\sigma_9(A)}
 =1.46151\,\|\boldsymbol\varepsilon\|_2 .
 \label{eq:affine-tensor-error}
\end{equation}
Central finite differences commute, so the linearized lattice Einstein
tensor obeys an exact discrete Bianchi identity.  Conservation of the
number-carrying wire stress leaves only a constant metric mode:
\begin{equation}
 G_{\mu\nu}^{\rm L}
 =8\pi G_{\rm eff}T_{\mu\nu}
 +\Lambda_{\rm int}g_{\mu\nu}.
 \label{eq:affine-einstein}
\end{equation}
A fixed reference boundary may set \(\Lambda_{\rm int}=0\); otherwise it is
retained as an integration constant.

\begin{theorem}[Model-specific order--volume and Einstein closure]
\label{thm:affine-model-closure}
On the contractible affine four-complex, with the stated response cone,
Boolean record cells, decoupled critical-wire bundles, primitive
capacity-exchange channel, stationary product cut--wire reference, linearized
focusing law, and conserved lattice stress,
Eqs.~\eqref{eq:affine-hodge-flow}--\eqref{eq:affine-einstein} form a closed
finite-regulator implication from response dynamics to an acyclic order,
additive information volume, Lorentzian metric, and linearized Einstein tensor
equation.  The nonmetric tensor error is bounded by
Eq.~\eqref{eq:affine-tensor-error}.
\end{theorem}

The theorem closes \(P8_{\rm model}\) and \(P9_{\rm lin,model}\).  It does not
close the general P8 or P9: adjacency and dimension selection, generic curved
convergence, interacting \(3+1\)-dimensional modular universality, nonlinear
focusing, graviton self-interaction, and regulator independence of
\(G_{\rm eff}\) remain open.

\section{A local capacity-reference selection bridge}
\label{sec:capacity-reference-selection}

\textbf{Status: \StatusModel\ for a factorized finite
capacity--complement regulator; \StatusOpen\ for correlated records, the
record--area map, and regulator universality.}

The nonlinear information balance below requires a full-rank boundary
reference whose eigenvalue in an \(n\)-record payload block compensates the
payload multiplicity \(q^n\).  That state need not be installed by an
\(n\)-dependent reset.  The companion reference-selection model
\cite{XuCapacityReference2026} derives it as the unique attractive boundary
state of a sector-blind local random-unitary dynamics.

\subsection{Capacity--complement tokens and local collisions}

Each of \(M\) tokens has a location and a \(q\)-level payload,
\begin{equation}
 \widetilde{\mathcal H}_j
 =\mathbb C^2_{\ell,j}\otimes\mathbb C^q_{p,j},
 \qquad
 \ell\in\{\partial,r\}.
 \label{eq:reference-token}
\end{equation}
Moving a token transfers a fixed information capacity between the visible
boundary and its complement.  Let \(X_d,Z_d\) be generalized Weyl operators.
The product-Weyl twirl and its local generator are
\begin{align}
 \mathcal T_j(\widetilde\rho)
 &=\frac{1}{4q^2}
 \sum_{\substack{a,b=0,1\\c,d=0,\ldots,q-1}}
 W_{abcd}^{(j)}\widetilde\rho W_{abcd}^{(j)\dagger}
 =\frac{I_j}{2q}\otimes\Tr_j\widetilde\rho ,
 \label{eq:reference-twirl}\\
 \widetilde{\mathcal L}^{*}
 &=\gamma\sum_{j=1}^{M}(\mathcal T_j-\mathrm{id}),\qquad
 W_{abcd}=X_2^aZ_2^b\otimes X_q^cZ_q^d .
 \label{eq:reference-generator}
\end{align}
Every collision rate is independent of the total visible-record number.  A
fresh uniformly prepared control cell and the controlled unitary
\(U_j=\sum_\mu|\mu\rangle\langle\mu|\otimes W_\mu^{(j)}\) give a one-pass
unitary dilation, so the twirl is not postulated as a sector reset.

The boundary algebra retains a payload only when its location is
\(\partial\):
\begin{equation}
 \mathcal A_{\partial,j}\simeq\mathbb C\oplus M_q,\qquad
 \mathcal R_j(\widetilde\rho)
 =\Tr_p\langle r|\widetilde\rho|r\rangle
 \oplus\langle\partial|\widetilde\rho|\partial\rangle .
 \label{eq:reference-readout}
\end{equation}
The physical and boundary semigroups intertwine.  On the boundary, the local
expectation replaces token \(j\) by
\(\omega_j=\frac12\oplus I_q/(2q)\), and
\(\mathcal L_\partial^*=\gamma\sum_j(\mathcal E_j^*-\mathrm{id})\).

\begin{theorem}[Capacity-complement reference selection]
\label{thm:capacity-reference-selection}
For every finite \(M\) and \(q\ge2\), the boundary semigroup is primitive and
has the unique stationary state
\begin{equation}
 \omega_\partial
 =\bigotimes_{j=1}^{M}\omega_j
 =2^{-M}\bigoplus_{S\subseteq[M]}q^{-|S|}I_S .
 \label{eq:selected-capacity-reference}
\end{equation}
For every faithful boundary state,
\begin{equation}
 D\!\left(e^{t\mathcal L_\partial^*}\rho\middle\Vert\omega_\partial\right)
 \le e^{-\gamma t}D(\rho\Vert\omega_\partial),
 \label{eq:reference-entropy-contraction}
\end{equation}
with a rate independent of \(M\) and \(q\).  If
\(\widehat N=\sum_j\widehat n_j\) and
\(\widehat A=a_0^2\widehat N\), then
\begin{equation}
 -\log\omega_\partial
 =\frac{\log q}{a_0^2}\widehat A+M\log2 ,
 \qquad
 \operatorname{spec}(-\mathcal L_\partial)
 =\{0,\gamma,\ldots,M\gamma\}.
 \label{eq:selected-modular-area}
\end{equation}
The coefficient \((\log q)/a_0^2\) is unchanged by additive blocking of
microscopic tokens.
\end{theorem}

The stationary state follows by reducing the uniform physical state
\((I_{2q}/2q)^{\otimes M}\).  For a visible set \(S\), tracing the
\(q^{M-|S|}\) complement payload states gives eigenvalue
\(2^{-M}q^{-|S|}\) to every visible payload vector.  Hence the total
\(n\)-record probability is
\(\binom{M}{n}2^{-M}\): payload multiplicity is canceled without appearing
in a collision rate.  Convexity of relative entropy and the quantum Shearer
inequality give Eq.~\eqref{eq:reference-entropy-contraction}; the commuting
local expectations give the spectrum in
Eq.~\eqref{eq:selected-modular-area}
\cite{MullerHermes2016,Capel2021,Dong2026}.

The unbiased location law is falsifiable.  If a visible token has stationary
probability \(p\), the modular increment is
\begin{equation}
 \Delta K(p)=\log\!\left[\frac{q(1-p)}{p}\right],
 \label{eq:reference-bias-control}
\end{equation}
so a pure capacity slope \(\log q\) occurs only at \(p=1/2\).  The finite
model therefore removes a chosen-reference input for the factorized family
\(g_n=\binom{M}{n}\).  It does not derive the location/payload
factorization, unbiased Weyl symmetry, the identification
\(\widehat A=a_0^2\widehat N\), an arbitrary correlated geometric
degeneracy \(g_n\), or a regulator-independent Newton coefficient.

\section{A finite nonlinear information--focusing bridge}
\label{sec:nonlinear-information-focusing}

\textbf{Status: \StatusModel\ for the finite information, optical, and
curved-frame theorems; \StatusConditional\ for nonlinear Einstein closure.}

The affine bridge ends at first order in state exchange and curvature.  The
companion nonlinear bridge \cite{XuNonlinearInformationFocusing2026} asks a
more limited but sharper question: which terms can be retained exactly
before a continuum nonlinear gravitational equation is assumed?  It removes
the entropy-first-law and linearized-focusing truncations in one finite
record regulator and leaves the missing dynamics as a measured Ward defect.
It neither supplies the autonomous order/capacity inputs of the generic
bridge nor identifies a finite-difference equation with general relativity
by definition.

\subsection{Record capacity and the modular area coefficient}

Let the boundary-record Hilbert space be graded by elementary record number,
\begin{equation}
 \mathcal H_{\partial}
 =\bigoplus_{n=0}^{n_{\max}}
 \mathcal G_n\otimes(\mathbb C^{q_g})^{\otimes n},
 \qquad
 \widehat A=a^2\sum_{n=0}^{n_{\max}}nP_n ,
 \label{eq:nonlinear-graded-boundary}
\end{equation}
where \(g_n=\dim\mathcal G_n\), \(P_n\) is the sector projector, and
\(\operatorname{Tr}P_n=g_nq_g^n\).  The capacity-balanced reference is
\begin{equation}
 \rho_{\partial}^0
 =Z_g^{-1}\bigoplus_{n=0}^{n_{\max}}q_g^{-n}I_n,
 \qquad Z_g=\sum_n g_n .
 \label{eq:nonlinear-capacity-reference}
\end{equation}
Its total weight in sector \(n\) is \(g_n/Z_g\): the eigenvalue compensates
the \(q_g^n\) payload multiplicity rather than erasing the geometric
degeneracy \(g_n\).  Blockwise functional calculus gives the exact identity
\begin{equation}
 K_{\partial}^0=-\log\rho_{\partial}^0
 =\frac{\log q_g}{a^2}\widehat A+\log Z_g .
 \label{eq:nonlinear-capacity-modular}
\end{equation}
Consequently,
\begin{equation}
 \frac{1}{4G_{\rm eff}\hbar}=\frac{\log q_g}{a^2},
 \qquad
 G_{\rm eff}=\frac{a^2}{4\hbar\log q_g}.
 \label{eq:nonlinear-geff}
\end{equation}
Unlike a fit of an area law, Eq.~\eqref{eq:nonlinear-geff} is fixed once the
reference \eqref{eq:nonlinear-capacity-reference} is selected.
Theorem~\ref{thm:capacity-reference-selection} supplies exactly such a unique
attractive state for the factorized family
\(g_n=\binom{M}{n}\), with a size-independent contraction rate.  The
nonlinear theorem here allows arbitrary \(g_n\), however, and it does not
couple the selection semigroup to the record--bulk--geometry update.
Selection for correlated record geometry, the record--area identification,
and regulator universality therefore remain open.

\subsection{Exact finite information balance}

Let \(\rho_b^0\) be any full-rank bulk reference and
\(\rho_0=\rho_{\partial}^0\otimes\rho_b^0\).  For an arbitrary unitary
record--bulk exchange \(U\), write
\(\rho=U\rho_0U^\dagger\) and
\(K_b^0=-\log\rho_b^0\).  Entropy invariance and the product-reference
chain rule give
\begin{align}
 \frac{\Delta A}{4G_{\rm eff}\hbar}
 +\Delta\langle K_b^0\rangle
 &=D(\rho\Vert\rho_0)\nonumber\\
 &=D(\rho_{\partial}\Vert\rho_{\partial}^0)
 +D(\rho_b\Vert\rho_b^0)
 +I_\rho(\partial{:}b)\nonumber\\
 &:=\mathcal R_{\rm info}.
 \label{eq:nonlinear-finite-information}
\end{align}
No exchange-angle expansion enters.  All three terms in
\(\mathcal R_{\rm info}\) are nonnegative.  For
\(U_\theta=e^{-i\theta V}\), the first derivative at \(\theta=0\) vanishes
and the Hessian is the Bogoliubov--Kubo--Mori information metric:
\begin{equation}
 \mathcal R_{\rm info}(\theta)
 =\frac{\theta^2}{2}
 g_{\rm BKM,\rho_0}
 \!\left(-i[V,\rho_0],-i[V,\rho_0]\right)+O(\theta^3).
 \label{eq:nonlinear-bkm}
\end{equation}
The entropy first law is therefore recovered as the tangent statement of an
exact finite balance; the information metric is the first omitted term, not
an optional correction \cite{Petz1996,Lashkari2016}.

\subsection{Exact discrete nonlinear optics}

For affine steps \(\lambda_j=ja\), let \(J_j\) be the transverse Jacobi map
and define the finite regulator by
\begin{equation}
 J_{j+1}-2J_j+J_{j-1}=-a^2\mathcal R_jJ_j ,
 \label{eq:nonlinear-jacobi}
\end{equation}
where \(\mathcal R_j\) is the symmetric optical tidal matrix.  Before the
first caustic, define
\begin{equation}
 F_j=J_jJ_{j-1}^{-1}=\mathbf1+aB_j .
\end{equation}
Equation~\eqref{eq:nonlinear-jacobi} is then exactly equivalent to
\begin{align}
 B_{j+1}
 &=B_j(\mathbf1+aB_j)^{-1}-a\mathcal R_j,
 \label{eq:nonlinear-riccati}\\
 \frac{B_{j+1}-B_j}{a}
 &=-B_j^2(\mathbf1+aB_j)^{-1}-\mathcal R_j ,
 \label{eq:nonlinear-discrete-raychaudhuri}\\
 \frac{A_j}{A_{j-1}}
 &=\det(\mathbf1+aB_j).
 \label{eq:nonlinear-area-update}
\end{align}
The quadratic expansion, shear, and twist terms arise from the Riccati
transform; they are not appended to a scalar area law.  The continuum
expansion of Eq.~\eqref{eq:nonlinear-discrete-raychaudhuri} gives the full
Raychaudhuri equation.  At a caustic the Riccati chart fails while the doubled
transfer relation for \((J_{j+1},J_j)\) remains regular.  A global extension
requires a Lagrangian-Grassmannian or Maslov-index description and is kept
open.

\subsection{Curved null frames and tensor stability}

The same 30 rational null directions used in the affine construction may be
transported by an arbitrary invertible coframe.  Let
\begin{equation}
 g=e^T\eta e,\qquad k_r=e^{-1}v_r .
 \label{eq:nonlinear-curved-frame}
\end{equation}
For a symmetric tensor \(S\), let \(A_gS\) be the vector of row-normalized
quadratic forms \(S_{\mu\nu}k_r^\mu k_r^\nu\), with the symmetric-tensor
vectorization chosen isometrically in the Frobenius norm.  Congruence by
\(e^{-1}\) is invertible, so the affine rank certificate implies
\begin{equation}
 \operatorname{rank}A_g=9,\qquad
 \ker A_g=\operatorname{span}\{g\}.
 \label{eq:nonlinear-curved-rank}
\end{equation}
If \(\sigma_9(A_g)\) is the smallest nonzero singular value and
\(A_gS=\epsilon\), then
\begin{equation}
 \inf_{\phi\in\mathbb R}\|S-\phi g\|_F
 \le\frac{\|\epsilon\|_2}{\sigma_9(A_g)} .
 \label{eq:nonlinear-curved-stability}
\end{equation}
Thus the finite frame remains informationally complete for the trace-free
tensor after arbitrary local coframe deformation.  Curvature changes its
conditioning, not its rank.  Spacetime dimension and the reference
direction set are still regulator inputs.

\subsection{The finite Ward-defect equation}

For a small null diamond in direction \(r\), use a fixed Green weight to
define the integrated contractions \(Q_R^{(r)}\) and \(Q_T^{(r)}\).  The
exact Jacobi area and the bulk modular response define separately measurable
remainders:
\begin{align}
 \Delta A^{(r)}
 &=-A_{\perp,r}Q_R^{(r)}+\varepsilon_{\rm foc}^{(r)},\label{eq:nonlinear-foc-def}\\
 \Delta\langle K_b^0\rangle^{(r)}
 &=\frac{2\pi}{\hbar}A_{\perp,r}Q_T^{(r)}
 +\varepsilon_{\rm mod}^{(r)}.\label{eq:nonlinear-mod-def}
\end{align}
Substitution into Eq.~\eqref{eq:nonlinear-finite-information} yields the exact
finite identity
\begin{equation}
 \boxed{\begin{aligned}
 Q_R^{(r)}-8\pi G_{\rm eff}Q_T^{(r)}
 ={}&\frac{\varepsilon_{\rm foc}^{(r)}}{A_{\perp,r}}\\
 &+\frac{4G_{\rm eff}\hbar}{A_{\perp,r}}
 \left(\varepsilon_{\rm mod}^{(r)}
 -\mathcal R_{\rm info}^{(r)}\right).
 \end{aligned}}
 \label{eq:nonlinear-null-defect}
\end{equation}
Variation of the coframe and Green weight across a shrinking diamond adds a
separately defined localization defect
\(\varepsilon_{\rm loc}^{(r)}\) to the projection data in
Eq.~\eqref{eq:nonlinear-curved-stability}.  The zero-defect target at finite
resolution is therefore the Ward relation
\begin{equation}
 \mathcal R_{\rm info}^{(r)}
 =\varepsilon_{\rm mod}^{(r)}
 +\frac{\varepsilon_{\rm foc}^{(r)}}
 {4G_{\rm eff}\hbar},
 \label{eq:nonlinear-ward-target}
\end{equation}
together with a controlled localization limit.  Equation
\eqref{eq:nonlinear-ward-target} is not imposed in the companion model.  It
is the microscopic identity required of a common
record--matter--geometry update.

\begin{theorem}[Finite nonlinear remainder bridge]
\label{thm:finite-nonlinear-remainder}
For the capacity-balanced reference
\eqref{eq:nonlinear-capacity-reference}, a full-rank product bulk reference,
an arbitrary unitary record--bulk exchange, the discrete Jacobi regulator
\eqref{eq:nonlinear-jacobi} before its first caustic, and the coframe-
transported 30-direction null design, Eqs.
\eqref{eq:nonlinear-capacity-modular},
\eqref{eq:nonlinear-finite-information},
\eqref{eq:nonlinear-riccati}--\eqref{eq:nonlinear-area-update},
\eqref{eq:nonlinear-curved-rank}--\eqref{eq:nonlinear-curved-stability}, and
\eqref{eq:nonlinear-null-defect} hold exactly at finite regulator size.
\end{theorem}

The theorem closes \(P9_{\rm nonlinear\ remainder,model}\), not P9.  If the
normalized right-hand side of Eq.~\eqref{eq:nonlinear-null-defect} and the
localization defect vanish for all frame directions, then the curvature--
stress difference is a cell-dependent multiple \(\phi_xg_{\mu\nu}(x)\).
Turning that local scalar freedom into one constant is a separate finite
completion problem.  It can be controlled quantitatively, but only after the
response tensor is identified with the geometric Euler derivative of the
same process.

\section{A quantitative finite Ward--Bianchi completion bridge}
\label{sec:ward-bianchi-completion}

\textbf{Status: \StatusAbstract\ for the finite completion theorem and
\StatusNumerical\ for the periodic four-lattice certificate;
\StatusOpen\ for the constitutive Ward identification and its uniform
interacting continuum limit.}

Let \(\mathcal S_h\) and \(\mathcal V_h\) be finite Hilbert spaces of
symmetric tensor and covector fields.  A metric field \(g_x\) defines
\[
 M_g\phi=\{\phi_xg_x\}_{x\in C_h}.
\]
Let \(A_h:\mathcal S_h\to\mathcal Y_h\) collect normalized contractions with
a local null frame and assume the exact finite completeness condition
\begin{equation}
 \ker A_h=\operatorname{ran}M_g,\qquad
 A_h^+A_h=P_g^\perp .
 \label{eq:finite-null-kernel}
\end{equation}
Thus null data determine every nonmetric component but are blind to an
independent scalar \(\phi_x\) at each cell.  Let
\(B_h:\mathcal S_h\to\mathcal V_h\) be the covariant divergence constructed
from the same finite geometry, with
\begin{equation}
 B_hM_g=D_h,\qquad
 \|\phi-\bar\phi\|
 \le\lambda_h^{-1/2}\|D_h\phi\|,
 \quad \lambda_h>0 .
 \label{eq:finite-compatible-poincare}
\end{equation}
The second relation is a scalar Poincar\'e estimate after removal of the
constant mode.

For noisy direction-resolved data
\[
 y=A_hW+n,\qquad \|n\|\le\nu,\qquad
 \widehat E=A_h^+y,
\]
define the two reconstructed budgets
\begin{equation}
 \eta_0=\|\widehat E\|+\|A_h^+\|\nu,\qquad
 \eta_1=\|B_h\widehat E\|+\|B_hA_h^+\|\nu .
 \label{eq:finite-reconstructed-budgets}
\end{equation}

\begin{theorem}[Noisy finite Ward--Bianchi completion]
\label{thm:ward-bianchi-completion}
Assume Eqs.~\eqref{eq:finite-null-kernel} and
\eqref{eq:finite-compatible-poincare}.  If
\(\|B_hW\|\le\varepsilon_{\rm W}\), then there is one constant
\(\bar\phi\) such that
\begin{equation}
 \boxed{
 \|W-M_g\bar\phi\|
 \le
 \eta_0+\frac{\|M_g\|}{\sqrt{\lambda_h}}
 \left(\varepsilon_{\rm W}+\eta_1\right).}
 \label{eq:ward-bianchi-completion}
\end{equation}
\end{theorem}

\begin{proof}
Write \(W=M_g\phi+E\), with
\(E=P_g^\perp W=A_h^+A_hW\).  The noisy reconstruction gives
\(\|E\|\le\eta_0\) and \(\|B_hE\|\le\eta_1\).  Compatibility then yields
\[
 \|D_h\phi\|
 =\|B_hM_g\phi\|
 \le\|B_hW\|+\|B_hE\|
 \le\varepsilon_{\rm W}+\eta_1.
\]
Apply the Poincar\'e estimate and
\(\|M_g(\phi-\bar\phi)\|\le\|M_g\|\|\phi-\bar\phi\|\).
\end{proof}

The Ward residual can itself be tied to one finite process functional.
Let \(\Gamma_h(z,m)\) depend differentiably on geometric and matter
variables.  For the tangent generators \(R_z\xi,R_m\xi\), define Euler
covectors \(\mathcal E_z=D_z\Gamma_h\),
\(\mathcal E_m=D_m\Gamma_h\) and the finite symmetry defect
\(\delta_{\rm sym}\) by differentiation along the joint transformation.
The chain rule gives the exact identity
\begin{equation}
 R_z^*\mathcal E_z+R_m^*\mathcal E_m
 =\delta_{\rm sym}.
 \label{eq:finite-noether-defect}
\end{equation}
Consequently, if
\begin{equation}
 W=\mathcal E_z,\qquad B_h=R_z^*,
 \label{eq:constitutive-ward-identification}
\end{equation}
then
\begin{equation}
 \varepsilon_{\rm W}
 \le\|\delta_{\rm sym}\|
 +\|R_m^*\|\,\|\mathcal E_m\|.
 \label{eq:finite-ward-budget}
\end{equation}
More generally, retain the constitutive defect
\[
 C_h=W-\mathcal E_z,\qquad
 \|C_h\|\le\varepsilon_{\rm cons}^{(0)},\qquad
 \|B_hC_h\|\le\varepsilon_{\rm cons}^{(1)} .
\]
The same proof gives the explicit certificate
\begin{align}
\inf_{\bar\phi\in\mathbb R}
\|\mathcal E_z-M_g\bar\phi\|
\le{}&\eta_0+\varepsilon_{\rm cons}^{(0)}
+\frac{\|M_g\|}{\sqrt{\lambda_h}}\bigl[
\eta_1+\varepsilon_{\rm cons}^{(1)}
\nonumber\\[-2pt]
&+\|\delta_{\rm sym}\|
+\|R_m^*\|\,\|\mathcal E_m\|\bigr].
\label{eq:finite-constitutive-certificate}
\end{align}
Equations~\eqref{eq:ward-bianchi-completion} and
\eqref{eq:finite-constitutive-certificate} put null tomography, constitutive
mismatch, finite symmetry breaking, off-shell matter, measurement noise, and
scalar conditioning in one error budget \cite{XuWardBianchi2026}.

The irreducible open arrow is the simultaneous decay of the two
constitutive defects in Eq.~\eqref{eq:finite-constitutive-certificate};
Eq.~\eqref{eq:constitutive-ward-identification} is their exact endpoint.
Gauge covariance cannot prove that the tensor assembled from the information
and optical sides of Eq.~\eqref{eq:nonlinear-null-defect} approaches the Euler
covector of the same functional, either before or after applying the spatial
divergence.  That is the microscopic information--geometry constitutive Ward
relation.  Nor may a holonomy identity be substituted for it.  Products of
face holonomies obey exact nonlinear lattice Bianchi identities
\cite{Batrouni1982,HamberKagel2004}, and finite Noether laws have exact
discrete formulations \cite{Skopenkov2023,Arrighi2023}; nevertheless, a
generic curved Regge lattice on a fixed triangulation does not have exact
vertex-displacement symmetry \cite{BahrDittrich2009}.  Kinematic holonomy
closure and dynamical diffeomorphism symmetry are different statements.

The companion certificate uses a periodic physical-side-one \(L^4\)
lattice, backward-incidence divergence, and the 30 rational null directions
of the earlier bridges.  Its row-normalized design has rank nine,
\(\ker A_h=\mathbb R\eta\), and
\(\sigma_9=0.8498365856\).  All 96 independently generated noisy smooth
fields on \(L=4,6,8,10\) satisfy
Eq.~\eqref{eq:ward-bianchi-completion}; the certified-bound to
actual-error ratios range from \(1.586\) to \(6.776\).  A frozen refinement
family over
\(L=4,6,8,10,12,16\) has actual and certified powers \(2.000\) and \(2.057\).
Two exact controls expose nonredundancy:
\[
 W(x)=\phi(x)\eta
 \quad\Longrightarrow\quad
 A_hW=0\ \text{but}\ B_hW\ne0,
\]
whereas a constant trace-free anisotropy has
\(B_hW=0\) but \(A_hW\ne0\).  Only a constant metric mode makes both
residuals vanish.

For a refinement family, the finite theorem gives a continuum-use
corollary.  If \(\|M_g\|\) and \(\|A_h^+\|\) remain bounded,
\(\lambda_h\) stays bounded below in the physical normalization, and
\[
\|A_h^+\|\nu_h+\|B_hA_h^+\|\nu_h
+\|\widehat E_h\|+\|B_h\widehat E_h\|
+\|\delta_{{\rm sym},h}\|
+\|R_{m,h}^*\|\,\|\mathcal E_{m,h}\|
\longrightarrow0 ,
\]
then, under the exact constitutive identification,
\[
 \inf_{\bar\phi\in\mathbb R}\|W_h-M_{g,h}\bar\phi\|\longrightarrow0.
\]
For an approximate identification, the two constitutive terms in
Eq.~\eqref{eq:finite-constitutive-certificate} must vanish as well.
This does not prove those uniform hypotheses.  It identifies exactly which
conditioning or dynamical defect would obstruct the limit and leaves the
constant metric mode to boundary or cosmological dynamics.

\section{A same-update finite Einstein--matter closure}
\label{sec:same-update-einstein}

\textbf{Status: \StatusModel\ for the local Ramsey action-phase identity and
its central-difference theorem; \StatusNumerical\ for the nonlinear FLRW and
four-dimensional tensor refinements; \StatusOpen\ for dynamical selection of
the action and a regulator-independent interacting limit.}

The constitutive quantities in
Eq.~\eqref{eq:finite-constitutive-certificate} need not remain independently
specified.  Let one finite regulator have variables \(q=(z,m)\), positive
cell weights \(\omega_{a,\alpha}\), and one declared local phase functional
\begin{equation}
 \Gamma_a(q)=\sum_{\alpha\in\Lambda_a}
 \omega_{a,\alpha}\gamma_{a,\alpha}(q_{\mathcal N(\alpha)}),
 \qquad
 \widehat\Gamma_a=\sum_q\Gamma_a(q)|q\rangle\langle q|.
 \label{eq:same-update-local-action}
\end{equation}
A fresh ancilla coherently controls the two interventions
\(q\mapsto q\pm\varepsilon e_\alpha\), applies
\(U_a(s)=e^{-is\widehat\Gamma_a}\), and undoes the intervention.  Its
\(X/Y\) coherence gives the operational score
\begin{equation}
 W^\varepsilon_{a,\alpha}
 =\frac{\Gamma_a(q+\varepsilon e_\alpha)
              -\Gamma_a(q-\varepsilon e_\alpha)}
        {2\varepsilon\omega_{a,\alpha}} ,
 \label{eq:same-update-operational-score}
\end{equation}
without evaluating the analytic derivative
\(\mathcal E_a=D\Gamma_a\).

\begin{theorem}[Same-update response--Euler certificate]
\label{thm:same-update-einstein}
Suppose the normalized third-derivative densities of \(\Gamma_a\) and the
first finite differences of their central-difference remainder field are
bounded by \(M_{0,a}\) and \(M_{1,a}\).  Then
\begin{equation}
 \|W_a^\varepsilon-\mathcal E_a\|_a
 \le\frac{\varepsilon^2}{6}M_{0,a},
 \qquad
 \|B_a(W_a^\varepsilon-\mathcal E_a)\|_a
 \le\varepsilon^2M_{1,a}.
 \label{eq:same-update-two-defects}
\end{equation}
For a quadratic phase action both defects vanish exactly.  Let \(r_a\)
denote the right-hand side of
Eq.~\eqref{eq:finite-constitutive-certificate} after substituting these two
budgets.  If one additional operational volume covector \(v_a\) satisfies
\(v_a(M_g1)=1\), \(\|v_a\|\le c_{v,a}\), and
\(|v_a(W_a^\varepsilon)|\le\eta_{V,a}\), then the complete geometric Euler
covector obeys
\begin{equation}
 \boxed{\;
 \|\mathcal E_{z,a}\|_a
 \le (1+c_{v,a})r_a+\eta_{V,a}
       +c_{v,a}\frac{\varepsilon^2M_{0,a}}{6}\; .}
 \label{eq:same-update-full-certificate}
\end{equation}
\end{theorem}

The first two inequalities are Taylor's formula with its integral remainder,
retained as a lattice field before applying \(B_a\).  For the last inequality,
the null--Ward argument gives
\(\mathcal E_{z,a}=M_g\bar\phi_a+\rho_a\) with
\(\|\rho_a\|_a\le r_a\).  Applying \(v_a\) controls
\(|\bar\phi_a|\), including its response and central-difference errors, and
then the triangle inequality proves
Eq.~\eqref{eq:same-update-full-certificate}.  Without \(v_a\), null and
divergence data still leave the constant metric mode undetermined.

The companion calculation \cite{XuSameUpdateEinstein2026} keeps the
operational and analytic pipelines independent.  In a midpoint discretization
of the Einstein--massless-scalar FLRW action, the response--Euler field and
derivative defects converge with fitted powers \(2.994\) and \(2.915\);
the finite symmetry and full Einstein--scalar residuals converge with powers
\(2.104\) and \(1.962\).  On a periodic \(L^4\) transverse-traceless
lattice the quadratic Ramsey--Euler defect is below
\(2.8\times10^{-14}\), the null and null-plus-volume frames have ranks nine
and ten with
\(\sigma_9=\sigma_{10}=0.8498365856\), and the complete tensor residual has
power \(1.888\).  A quartic FLRW phase term and a quadratic TT term inserted
only in the Ramsey branch leave respective defects above
\(3.37\times10^{-2}\) and \(1.27\times10^{-3}\).  These controls exclude
agreement obtained by projecting the response into an Einstein sector after
measurement.

The result removes the finite-model sub-assumption \(A11S\), not the general
gravitational assumption \(A11\).  Equation
\eqref{eq:same-update-local-action} already contains an Einstein--matter phase
action on a four-dimensional Lorentzian regulator.  The Ramsey theorem
certifies the Euler equation of that action; it does not explain why a less
structured autonomous quantum process selects the action, its variables, or
its normalization.  A general continuum claim additionally requires uniform
control of \(M_{0,a},M_{1,a},A_a^+,\lambda_a^{-1/2}\), phase-unwrapping
resources, finite symmetry and matter residuals, and the regulator response
coefficient.

The remaining steps are:
\begin{enumerate}[label=(N\arabic*)]
\item couple the selected factorized reference to the same
record--matter--geometry update, extend its contraction theorem to correlated
record geometry, and derive the record--area map;
\item prove an interacting fixed-physical-band modular theorem for local
wave packets;
\item replace the prescribed phase action in
Theorem~\ref{thm:same-update-einstein} by one autonomous
record--matter--Lorentz-holonomy dynamics whose infrared phase is proved to
approach the Einstein--Hilbert form, while retaining independent response and
variation pipelines;
\item bound the same update's finite symmetry defect and matter
equation-of-motion residual, while keeping the exact cell-holonomy Bianchi
identity as a distinct geometric consistency check;
\item continue the optical and record data through caustics with a
patch-independent Maslov index; and
\item prove a uniform shrinking-diamond and regulator-independence limit.
\end{enumerate}
These are stronger requirements than discretizing Einstein's equation.  They
ask the quantum information cost, optical nonlinearity, and conservation
identities to be consequences of one dynamics.

\section{A generic order--volume bridge to Lorentzian metric--measure geometry}
\label{sec:generic-order-volume}
\label{sec:updated-order-volume}

\textbf{Status: \StatusAbstract\ for the finite construction, intrinsic
compactness theorem, and continuum volume-clock identity;
\StatusConditional\ for smooth-manifold consistency; \StatusOpen\ for
autonomous generation of the input order and capacity-density law.}

The affine regulator above demonstrates one finite composition, but it is no
longer the only bridge from order and volume to geometry.  The companion
manuscript \cite{XuGenericOrderVolume2026} gives a coordinate-free
construction whose generic output is a Lorentzian metric--measure space,
not a fitted member of a prescribed manifold family:
\begin{equation}
 \boxed{\begin{gathered}
 \text{margin-certified response order}
 +\text{positive additive record measure}\\
 \longrightarrow
 (X,\tau,\mu)\quad\text{or an explicit rejection}.
 \end{gathered}}
 \label{eq:updated-order-volume-target}
\end{equation}
The theorem closes the mathematical middle of
Eq.~\eqref{eq:updated-order-volume-target}.  It does not assert that every
local channel passes the input gates.

\subsection{Operational order and information volume}

For measured intervention strengths \(\kappa_{i\to j}\), define the
antisymmetric score
\begin{equation}
 A_{ij}=\log\frac{\kappa_{i\to j}+\kappa_0}
 {\kappa_{j\to i}+\kappa_0},
 \qquad A_{ji}=-A_{ij}.
 \label{eq:generic-response-score}
\end{equation}
An edge is retained only when its simultaneous confidence interval clears a
declared margin.  A full directed-cycle test is performed before transitive
closure; antisymmetry alone excludes only two-cycles.  If a robust cycle
remains, the output is \emph{order rejected}.  Thus no latent time coordinate
is fitted to repair the data.

Let \(\Z_i=Z(\A_i^{\rm IR})\simeq\mathbb C^{q_i}\) be a record-bearing
Boolean center, \(q_i\ge2\).  Its zero-error classical capacity is
\(\log q_i\), additive under independent tensor composition.  Assign
\begin{equation}
 w_i=v_\star\frac{\log\dim\Z_i}{c_\star},
 \qquad
 \widehat\mu_N=\sum_iw_i\delta_i,
 \qquad
 \widehat V_N(i,j)=\widehat\mu_N[I_N(i,j)] .
 \label{eq:generic-record-volume}
\end{equation}
The ratio \(v_\star/c_\star\) is one retained microscopic conversion between
record capacity and physical volume.  Nonpositive weights, nonadditivity, or
uncontrolled drift of the renormalized capacity density
\(\rho\,\mathbb E w_i\) reject the volume interpretation.  In this sense the
bridge uses an operational measure but does not claim to derive its absolute
unit or refinement law from dimensionless order data.

\subsection{Dimension and the finite volume clock}

For a flat \(d\)-dimensional Alexandrov interval, the comparable-pair
fraction is
\begin{equation}
 p_d=\frac{\Gamma(d+1)\Gamma(d/2)}{2\Gamma(3d/2)} .
 \label{eq:generic-myrheim-meyer}
\end{equation}
A separated plateau of the weighted local ordering fraction selects
\(\widehat d_N\); an unresolved runner-up margin produces abstention.  With
\(\zeta_d=\omega_{d-1}/(d2^{d-1})\), define
\begin{equation}
 \ell_{V,N}(i,j)
 =\zeta_{\widehat d_N}^{-1/\widehat d_N}
 \widehat V_N(i,j)^{1/\widehat d_N}
 \label{eq:generic-local-volume-clock}
\end{equation}
on a pair-independent local interval window
\(\mathcal E_N^{\rm loc}\), containing empty-interval saturated links and
pairs with \(v_N^-\le\widehat V_N\le v_N^+\).  Its canonical completion is
\begin{equation}
 \widehat\tau_N(i,j)
 =\max_{\substack{i=i_0\prec_N\cdots\prec_N i_m=j\\
                  (i_{r-1},i_r)\in\mathcal E_N^{\rm loc}}}
 \sum_{r=1}^{m}\ell_{V,N}(i_{r-1},i_r).
 \label{eq:generic-volume-clock}
\end{equation}
This object is permutation covariant and uses no coordinate, tangent frame,
signature matrix, curvature radius, or manifold label.

\begin{theorem}[Finite order--volume geometry]
\label{thm:generic-finite-order-volume}
If the robust response graph is acyclic and the record weights are positive,
then Eq.~\eqref{eq:generic-volume-clock} obeys
\begin{equation}
 \widehat\tau_N(i,k)\ge
 \widehat\tau_N(i,j)+\widehat\tau_N(j,k)
 \qquad(i\prec_N j\prec_N k).
 \label{eq:generic-reverse-triangle}
\end{equation}
If the admissible edge set is unchanged and every interval mass changes by a
relative amount at most \(\eta<1\), the clock changes multiplicatively by at
most
\(\max\{(1+\eta)^{1/d}-1,1-(1-\eta)^{1/d}\}\).
\end{theorem}

The reverse triangle is not imposed as a fitting penalty.  It follows because
maximizing paths from \(i\) to \(j\) and \(j\) to \(k\) concatenate to an
admissible path from \(i\) to \(k\).  The numerical implementation uses a
cross-fitted estimator to remove finite-sample maximization bias; a single
cross-fit need not obey the exact identity and is not substituted for the
canonical theorem object.

\subsection{Generic compactness without a manifold}

Define the strong profile pseudometric
\begin{equation}
 D_N(x,y)=\max_z\left\{
 |\widehat\tau_N(x,z)-\widehat\tau_N(y,z)|,\,
 |\widehat\tau_N(z,x)-\widehat\tau_N(z,y)|
 \right\}.
 \label{eq:generic-strong-profile}
\end{equation}
After quotienting coincident profiles, it is a metric computable before any
embedding.

\begin{theorem}[Intrinsic Lorentzian metric--measure compactness]
\label{thm:generic-order-volume-compactness}
Let the normalized finite order--volume geometries have uniformly bounded
time diameter, uniformly bounded \(D_N\)-covering numbers at every fixed
resolution, and one margin-certified dimension for all sufficiently large
\(N\).  Then a subsequence converges to a compact measured space
\((X,D,\mu)\) carrying a continuous time separation \(\tau\).  The reverse
triangle passes to the limit, so the profile quotient is a bounded
Lorentzian metric--measure space.  No smooth manifold is assumed.
\end{theorem}

Metric Gromov compactness supplies \((X,D)\), Prokhorov compactness supplies
\(\mu\), and
\[
 |\widehat\tau_N(x,y)-\widehat\tau_N(x',y')|
 \le D_N(x,x')+D_N(y,y')
\]
passes the time separation uniformly to the limit.  The generic limit may be
smooth, singular, lower dimensional, or rejected when covering numbers
diverge.

\subsection{Why local interval volume determines proper time}

For a future-timelike curve \(\gamma\) and partition \(\Pi\), let
\begin{equation}
 L_V(\gamma,\Pi)=\sum_r
 \zeta_d^{-1/d}
 \mu_g[I(\gamma(s_{r-1}),\gamma(s_r))]^{1/d}.
 \label{eq:generic-continuum-clock}
\end{equation}
Small causally convex diamonds obey
\[
 \mu_g[I(p,q)]
 =\zeta_d\tau_g(p,q)^d
 \left[1+O(\tau_g(p,q)^2/\lambda^2)\right].
\]
The local clock error is therefore cubic in the segment time.  Its partition
sum vanishes when the mesh does, even if the full end-to-end diamond is large
or curved.

\begin{theorem}[Continuum volume-clock identity]
\label{thm:generic-volume-clock-identity}
On every smooth strongly causal spacetime,
\begin{equation}
 \tau_g(x,y)=
 \sup_{\gamma:x\to y}\lim_{\|\Pi\|\to0}L_V(\gamma,\Pi).
 \label{eq:generic-clock-identity}
\end{equation}
\end{theorem}

This is the step absent from a single use of
\(\tau=\zeta_d^{-1/d}V^{1/d}\): subdivision removes the integrated curvature
bias without fitting a curvature model or assuming a longest-chain/proper-
time calibration.

\subsection{Smooth-limit consistency, topology, and uniqueness}

Suppose only for the consistency statement that the unknown target is a
globally hyperbolic bounded-geometry \(d\)-spacetime, records sample its
physical volume at intensity \(\rho\), capacity weights satisfy
\(\rho\,\mathbb E(w_i\mid x)=1\) with uniform relative interval error
\(\eta_w(r)\), and the capacity-weighted interval-membership symmetric
difference of the robust and chronological orders is bounded by
\(\varepsilon_{\rm ord}\).  At a scale
\(h_\rho\ll r\ll\lambda\),
\[
 h_\rho=C_d\left[\frac{\log(\rho\mu_gK)}{\rho}\right]^{1/d},
\]
use a guarded positive window
\(v_N^-\asymp\zeta_d(r/2)^d\),
\(v_N^+\asymp\zeta_dr^d\).  Then
the companion theorem gives
\begin{align}
 |\widehat\tau_N(x,y)-\tau_g(x,y)|
 &\le C_d(\tau_g(x,y)+r)\,\mathcal E_\rho(r),\label{eq:updated-lor-error}\\
 \mathcal E_\rho(r)
 &=\left(\frac r\lambda\right)^2
 +\sqrt{\frac{\log(\rho\mu_gK)}{\rho\zeta_dr^d}}
 +\frac{h_\rho}{r}+\eta_w(r)+\varepsilon_{\rm ord}.
 \label{eq:generic-error-budget}
\end{align}
The terms are respectively curvature bias, interval-count fluctuation,
event-net resolution, capacity calibration, and response-order error.
Response uncertainty and sampling density are therefore not merged into one
fit residual.

Finite Alexandrov intervals in a stable mass window form a nerve.  When they
are certified as a causally convex good cover, the nerve lemma identifies its
homotopy type with that of the smooth region.  Without a good-cover
certificate, only the finite nerve is reported.  Finally, Lorentzian
metric--measure reconstruction makes two admissible smooth limits with the
same limiting distance laws measure-preserving isometric
\cite{BraunSaemann2025}; chronological order--number reconstruction gives the
same uniqueness endpoint under its stated smooth hypotheses
\cite{BraunOrderNumber2026}.  Quantitative finite-embedding work instead
retains a separate longest-chain/proper-time condition
\cite{MadsenCausalEmbedding2026}; the volume-clock identity above constructs
the time separation rather than assuming that condition.  These published
results enter only at the comparison and uniqueness endpoints.

Blind simulations in \(d=2,3,4\) select the correct dimension in every
certified trial, with one underresolved trial abstaining.  At \(N=1050\), the
median flat-clock errors are \(4.6\%\), \(9.8\%\), and \(9.4\%\).
In a curved conformal diamond, fine subdivision reduces the median error from
\(6.8\%\) to \(4.3\%\), and an independent Alexandrov cover of
\([0,1]\times S^1\) gives \((b_0,b_1)=(1,1)\).  These are synthetic
calculations, not experiments.

\begin{remark}[What has and has not closed]
The generic order--volume problem separates into two steps.  Once a robust
order and a positive additive record measure pass their operational gates,
the finite geometry, generic metric--measure limit, smooth consistency,
topology test, and identifiability chain are closed.  What remains open is a
single autonomous microscopic dynamics that generates generic adjacency and
acyclicity, derives the capacity-density conversion \(v_\star\), and supplies
its own runtime reference.  ``Order plus number'' does not remove those
upstream tasks.
\end{remark}

\section{A unified operational continuum limit}
\label{sec:unified-operational-continuum}

\textbf{Status: \StatusAbstract\ for the joint-interface theorem;
\StatusNumerical\ for its FLRW and four-dimensional tensor realization;
\StatusOpen\ for autonomous selection of the continuum phase, the
Einstein--Hilbert action, and a universal Newton coefficient.}

The order--volume and same-update results close two different interfaces.
The first reconstructs a Lorentzian metric--measure geometry from certified
response order and additive record capacity.  The second compares a Ramsey
response with the Euler derivative of one declared finite action.  Neither
statement alone proves that the geometry in the field equation is the one
reconstructed from the microscopic records.  Two separately convergent
sequences can approach inequivalent metrics, or can silently use different
refinement maps and actions.

The companion work \cite{XuUnifiedContinuum2026} closes this compatibility
problem for one explicit refinement family.  Its main requirement is
stronger than simultaneous convergence but weaker than autonomous action
selection: the order, capacity, and phase branches are evaluated
independently, while all three consume the same finite registers and the
same lattice spacing.

\subsection{One update record and two independent pipelines}

Let \(h\downarrow0\) label finite event sets \(X_h\).  A refinement record is
\begin{equation}
 {\cal R}_h=(X_h,q_h,\kappa_h,K_h,\Gamma_h,{\cal I}_h),
 \label{eq:unified-record}
\end{equation}
where \(q_h\) contains the finite geometry and matter registers,
\(\kappa_h(i\to j)\) is an intervention response, \(K_h(i)\) is an integer
Boolean-record capacity, \(\Gamma_h(q_h)\) is a local real phase action, and
\({\cal I}_h\) is the declared interpolation into continuum test fields.
The order--volume pipeline produces
\[
 {\cal X}_h=(X_h,\prec_h,\widehat\tau_h,\widehat\mu_h),
\]
whereas the response pipeline uses a fresh Ramsey ancilla to measure
\begin{equation}
 W_{h,\alpha}^{\varepsilon}
 =\frac{\Gamma_h(q_h+\varepsilon e_\alpha)
              -\Gamma_h(q_h-\varepsilon e_\alpha)}
        {2\varepsilon\,\omega_{h,\alpha}} .
 \label{eq:unified-ramsey-score}
\end{equation}
An independent analytic variation gives
\begin{equation}
 {\cal E}_{h,\alpha}
 =\omega_{h,\alpha}^{-1}
  \frac{\partial\Gamma_h}{\partial q_{h,\alpha}} .
 \label{eq:unified-euler-score}
\end{equation}
The comparison is therefore falsifiable: the measured response is not
defined by calling it an Euler tensor.

\begin{definition}[Single-update refinement family]
\label{def:unified-single-update}
A family \(\{{\cal R}_h\}\) is single-update when the response order,
capacity measure, and phase action are functions of the same configuration
\(q_h\), use one common refinement and interpolation, and every
response--Euler comparison refers to that same \(\Gamma_h\).
\end{definition}

For an attempted split construction, let \(q_h^{\rm OV}\) denote the
registers supplied to the order--volume branch and \(q_h^\Gamma\) those
supplied to the phase branch.  The mismatch
\begin{equation}
 \Delta_h^{\rm same}
 =\|{\cal I}_hq_h^{\rm OV}-{\cal I}_hq_h^\Gamma\|_{\cal Q}
 +d_{\mathrm{mL}}\!\left[
  {\cal G}_h(q_h^{\rm OV}),{\cal G}_h(q_h^\Gamma)\right]
 \label{eq:unified-same-mismatch}
\end{equation}
is retained in the certificate.  It vanishes identically for a
single-update family and cannot be fitted away when the two branches
approach different continuum geometries.

\begin{landscape}
\begin{figure}[p]
 \centering
 \includegraphics[width=0.97\linewidth]
 {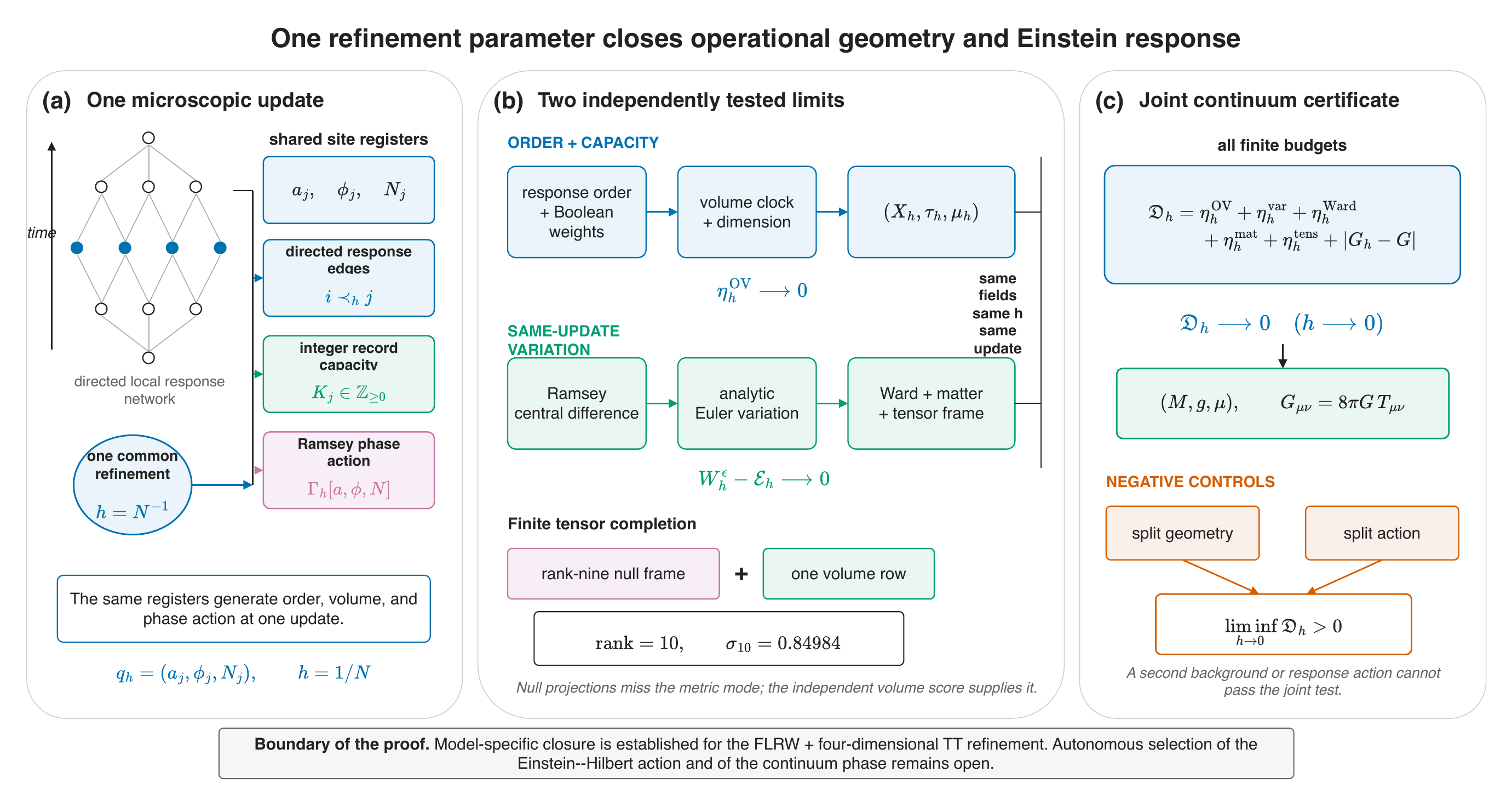}
 \caption{\textbf{Unified operational continuum limit.}
 Panel (a) records the common microscopic data: the same
 \(q_h=(a_j,\phi_j,N_j)\) and \(h=N^{-1}\) generate directed response
 records, integer Boolean capacities, and a Ramsey phase action.
 Panel (b) keeps geometry reconstruction and analytic variation independent;
 rank-nine null data plus one volume row determine all ten components of a
 symmetric four-tensor.  Panel (c) collects every finite defect before the
 limit is taken.  Split-geometry and split-action controls retain a nonzero
 floor.  The gray boundary box lists the inputs that the theorem does not
 select.  The figure is vector and its editable source is distributed with
 the companion manuscript.}
 \label{fig:unified-continuum-mechanism}
\end{figure}
\end{landscape}

\subsection{The joint finite certificate}

Let \(\eta_h^{\rm OV}\) bound the measured Lorentzian reconstruction error,
\(\eta_h^{\rm var}\) the error made by interchanging variation and
refinement, \(\eta_h^{\rm Ward}\) the finite Noether defect,
\(\eta_h^{\rm mat}\) the independently evaluated matter equation residual,
and \(\eta_h^{\rm tens}\) the error in completing the symmetric metric Euler
tensor from operational projections.  Let \(B_h\) be the compatible finite
divergence.  Taylor's formula applied to
Eq.~\eqref{eq:unified-ramsey-score} gives
\begin{equation}
 \|W_h^\varepsilon-{\cal E}_h\|_h
 \le\frac{\varepsilon^2}{6}M_{0,h},
 \qquad
 \|B_h(W_h^\varepsilon-{\cal E}_h)\|_h
 \le\varepsilon^2M_{1,h}.
 \label{eq:unified-central-difference}
\end{equation}
The second bound is imposed on the differentiated remainder field itself;
the cruder estimate \(\|B_h\|M_{0,h}\) can lose a power of \(h\).

At each event, contract a symmetric tensor with a finite set of normalized
null rows
\(A_\ell[X]=X_{\mu\nu}\ell^\mu\ell^\nu\).  These rows have rank at most nine
in four dimensions because they annihilate the metric direction.  Adding
one normalized volume row \(A_V[X]=g^{\mu\nu}X_{\mu\nu}\) gives a full frame
\(A_h\).  If
\(\sigma_{\min}(A_h)\ge\sigma_*>0\), then
\begin{equation}
 \|X\|_2\le\sigma_*^{-1}\|A_hX\|_2 .
 \label{eq:unified-frame-bound}
\end{equation}
The volume score is thus information-theoretically nonredundant; no further
null direction can replace it.

Define
\begin{align}
 {\mathfrak D}_h={}&
 \eta_h^{\rm OV}
 +\Delta_h^{\rm same}
 +\eta_h^{\rm var}
 +\frac{\varepsilon_h^2}{6}M_{0,h}
 +\varepsilon_h^2M_{1,h}\nonumber\\
 &+\eta_h^{\rm Ward}
 +\eta_h^{\rm mat}
 +\sigma_*^{-1}\eta_h^{\rm tens}
 +|G_h-G|.
 \label{eq:unified-certificate}
\end{align}
Every term is a finite-resolution quantity.  In particular, the Newton
coefficient is not hidden in the tensor norm: if it is supplied rather than
derived, the term \(|G_h-G|\) records that fact.

\begin{theorem}[Unified operational continuum limit]
\label{thm:unified-operational-continuum}
Let \(\{{\cal R}_h\}\) be a single-update refinement family on a compact
globally hyperbolic test region.  Suppose:
\begin{enumerate}[label=(U\arabic*)]
\item the order--volume reconstruction obeys
\(d_{\mathrm{mL}}[{\cal X}_h,({\cal M},\tau_g,\mu_g)]
\le\eta_h^{\rm OV}\to0\);
\item sampling and interpolation are stable and
\(\|{\cal I}_h{\cal E}_h-{\cal E}(g,\psi)\|_{H^{-1}}
\le\eta_h^{\rm var}\to0\);
\item Eq.~\eqref{eq:unified-central-difference} holds and
\(\varepsilon_h^2(M_{0,h}+M_{1,h})\to0\);
\item the null-plus-volume frame satisfies
\(\sigma_{\min}(A_h)\ge\sigma_*>0\); and
\item the Ward, matter, tensor, and coupling errors in
Eq.~\eqref{eq:unified-certificate} vanish.
\end{enumerate}
Then the operational geometry and response converge jointly:
\begin{align}
 d_{\mathrm{mL}}[{\cal X}_h,({\cal M},\tau_g,\mu_g)]
 &\le {\mathfrak D}_h,\label{eq:unified-geometry-limit}\\
 \|{\cal I}_hW_h^{\varepsilon_h}-{\cal E}(g,\psi)\|_{H^{-1}}
 &\le C_{\cal I}{\mathfrak D}_h .
 \label{eq:unified-dynamics-limit}
\end{align}
If a compactly convergent sequence also has vanishing operational score,
\(\|{\cal I}_hW_h^{\varepsilon_h}\|_{H^{-1}}\to0\), then
\({\cal E}(g,\psi)=0\) on the same metric--measure geometry selected by
Eq.~\eqref{eq:unified-geometry-limit}.  For the Einstein--matter action,
\begin{equation}
 G_{\mu\nu}[g]+\Lambda g_{\mu\nu}
 =8\pi G\,T_{\mu\nu}[g,\psi],
 \qquad {\cal E}_{\psi}(g,\psi)=0
 \label{eq:unified-einstein-limit}
\end{equation}
in the weak sense.
\end{theorem}

\begin{proof}
Insert \({\cal E}_h\) between the Ramsey score and the continuum Euler
covector.  Stable interpolation, the two estimates in
Eq.~\eqref{eq:unified-central-difference}, and variational consistency give
the projected response bound.  The rank-nine null frame reconstructs the
nonmetric tensor sector; Eq.~\eqref{eq:unified-frame-bound} and the volume
row add the missing metric component.  The Ward and matter terms control the
finite conservation and off-shell matter defects.  The order--volume term
fixes the continuum geometry, and
\(\Delta_h^{\rm same}\) prevents a second geometry from entering the phase
branch.  The triangle inequality yields
Eq.~\eqref{eq:unified-dynamics-limit}.  Passing to a compact limit of
vanishing operational scores gives
Eq.~\eqref{eq:unified-einstein-limit}.  No unproved interchange of
variation and limit occurs; its entire cost is \(\eta_h^{\rm var}\).
\end{proof}

\begin{corollary}[Split-limit rejection]
\label{cor:unified-split-rejection}
If \(\liminf_{h\to0}\Delta_h^{\rm same}>0\), or if the Ramsey branch uses
\(\Gamma_h+Q_h\) while the analytic branch varies \(\Gamma_h\) and
\(\liminf\|{\cal I}_hDQ_h\|>0\), then
\(\liminf_{h\to0}{\mathfrak D}_h>0\).
\end{corollary}

\subsection{Explicit FLRW--tensor realization}

The theorem is realized on
\begin{equation}
 {\cal M}=[1,2]\times\mathbb T^3,\qquad
 \mathrm ds^2=-\mathrm dt^2+a(t)^2\mathrm d\bm x^2 ,
 \label{eq:unified-flrw-metric}
\end{equation}
with \(8\pi G=1\) and
\begin{equation}
 a(t)=t^{1/3},\qquad
 \phi(t)=\sqrt{\frac23}\log t,\qquad
 N(t)=1 .
 \label{eq:unified-flrw-solution}
\end{equation}
For \(h=N^{-1}\), the same samples
\(q_h=(a_j,\phi_j,N_j)\) determine all three branches.  With
\(\eta(t)=\tfrac32t^{2/3}\), the response order is
\begin{equation}
 (j,\bm n)\prec_h(k,\bm m)
 \quad\Longleftrightarrow\quad
 j<k,\qquad
 h\,d_{\mathbb T^3}(\bm n,\bm m)
 <\eta(t_k)-\eta(t_j).
 \label{eq:unified-flrw-order}
\end{equation}
Each four-cell carries the Boolean capacity and weight
\begin{equation}
 C_j=\operatorname{round}\!\left[(N^3+1)a(t_{j+1/2})^3\right],
 \qquad
 w_j=h^4\frac{C_j}{N^3+1}.
 \label{eq:unified-capacity}
\end{equation}
The weighted counting measure converges to
\(a^3\,\mathrm dt\,\mathrm d^3x\).  Applying the intrinsic local
volume clock of Sec.~\ref{sec:generic-order-volume} on pieces of proper size
\(r_h\) gives
\begin{equation}
 |\widehat\tau_h-\tau_g|
 \le C\left(r_h^2+\frac{h}{r_h}+h^2+h^3\right).
 \label{eq:unified-order-volume-rate}
\end{equation}
The terms arise from small-diamond curvature, the lattice boundary shell,
midpoint quadrature, and capacity quantization.  Choosing
\(r_h=h^{1/3}\) gives the conservative complete clock rate
\(\eta_h^{\rm OV}=O(h^{2/3})\); the weak measure and local homogeneous metric
converge faster.

The phase branch uses the same samples in the midpoint
Einstein--massless-scalar functional
\begin{align}
 \Gamma_h[a,\phi,N]
 =\sum_{j=0}^{N-1}h\bigg[
 -\frac{3\bar a_j}{N_j}
  \left(\frac{a_{j+1}-a_j}{h}\right)^2
 +\frac{\bar a_j^3}{2N_j}
  \left(\frac{\phi_{j+1}-\phi_j}{h}\right)^2
 \bigg],
 \label{eq:unified-discrete-action}\\
 \bar a_j=\tfrac12(a_j+a_{j+1}).
 \nonumber
\end{align}
It generates the diagonal Ramsey phase and is varied independently to
obtain \({\cal E}_h\).  The variational, finite Ward, and scalar Euler
residuals converge at second order.  A periodic \(N^4\)
transverse-traceless sector supplies genuinely four-dimensional tensor
data.  Thirty rational null directions give rank nine; adding the volume
row gives rank ten with
\begin{equation}
 \sigma_{10}=0.8498365856 .
 \label{eq:unified-frame-sigma}
\end{equation}
The TT Ramsey--Euler defect is below \(3.01\times10^{-14}\), and the tensor
reconstruction error converges with fitted power \(1.951\).

\begin{figure}[htbp]
 \centering
 \includegraphics[width=\linewidth]
 {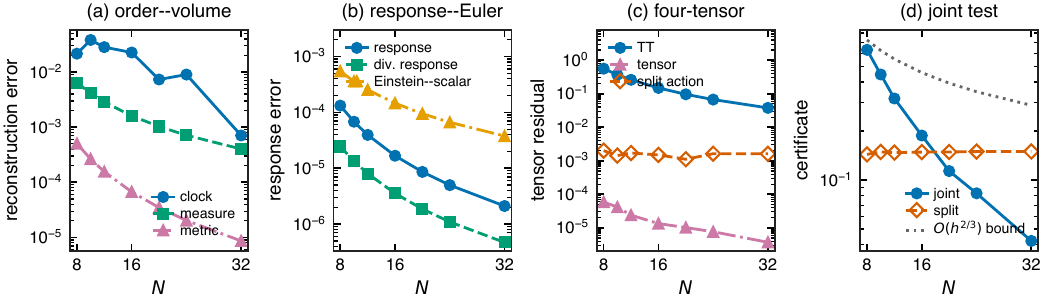}
 \caption{\textbf{Deterministic synthetic refinement tests for the unified
 continuum certificate.}
 All panels use the common sizes \(N=8,10,12,16,20,24,32\).
 The order--volume observables, independent Ramsey/variation defects, and
 four-dimensional tensor residuals decrease under the same refinement.
 The complete certificate falls from \(0.630515\) to \(0.042254\), with
 fitted power \(1.949\) over the displayed sizes.  The gray
 \(O(h^{2/3})\) line is the conservative analytic envelope for the complete
 volume clock, not a fitted exponent.  Deliberately split geometry and
 action branches remain between \(0.1437\) and \(0.1498\).  These are
 reproducible synthetic calculations, not experimental data.}
 \label{fig:unified-continuum-data}
\end{figure}

The numerical archive contains 467 independent algebraic, refinement,
conditioning, split-control, and figure-data checks with no failed test.
The measured joint certificate decreases from \(0.630515\) at \(N=8\) to
\(0.042254\) at \(N=32\).  The split control remains between \(0.1437\) and
\(0.1498\), the minimum theorem-bound/actual-tensor-error ratio is \(1.7405\),
and the fitted powers of the Ward, Einstein--scalar, TT, tensor, and joint
errors are respectively \(1.898\), \(1.927\), \(1.945\), \(1.951\), and
\(1.949\).  These fitted powers diagnose this smooth family; the proved
order--volume statement remains the conservative
\(O(h^{2/3})\) estimate in
Eq.~\eqref{eq:unified-order-volume-rate}.

\subsection{Boundary of the result}

The model removes the possibility that the order--volume and
Einstein-response limits are mutually incompatible merely because they use
different discrete fields or different refinement maps.  It establishes a
nonempty, falsifiable realization in which one \(h\), one register family,
one capacity measure, and one phase action pass a common finite certificate.
It does \emph{not} remove the general assumptions \(A9V\) or \(A11\).
Specifically:
\begin{enumerate}
\item the four-dimensional manifold and \(\mathbb T^3\) topology are inputs;
\item the Einstein--Hilbert plus massless-scalar action is supplied, not
selected by a less structured microscopic channel;
\item the integer capacity unit and \(G_h=G\) are fixed microscopically;
\item the proof covers nonlinear homogeneous geometry and local
four-dimensional TT refinement, not arbitrary nonlinear inhomogeneous
spacetimes;
\item quantum fluctuations, renormalization, autonomous clock selection, and
a probability-law continuum limit remain open.
\end{enumerate}
The separate same-link construction in
Sec.~\ref{sec:same-link-induced-gravity} removes the supplied-curvature-action
input and controls large-\(N_s\) quantum fluctuations only for one
fixed-volume TT polarization at fixed microscopic spacing.  It does not
upgrade the present FLRW--Einstein--matter theorem to a general quantum
probability-law continuum limit.
The result is therefore a \emph{model-specific unified continuum closure}.
The general target remains one autonomous local quantum dynamics that
selects its manifoldlike phase, capacity-density law, gravitational action,
and regulator-independent Newton coefficient.

\section{Open bridge II: modular response and gravitational dynamics}
\label{sec:updated-gravity-bridge}

\textbf{Status: \StatusAbstract\ for the relative-entropy identity;
\StatusNumerical\ for the free-lattice certificate;
\StatusModel\ for the affine linearized closure, finite nonlinear
remainder theorem, and balanced quadratic regulator limit;
\StatusModel\ for the separate free positive-Hilbert spectral-balance
realization on supplied finite graph ensembles; \StatusOpen\ for a partial
interacting Ward identity, dynamically generated refinement measure,
nonlinear Einstein dynamics, Lorentzian RG protection, and a common
nonperturbative measure.}

Let \(\mathcal M_D\) be a local algebra and
\(\omega_D^0\) a faithful reference state.  In a Type-I regulator,
\(K_D^0=-\log\rho_D^0\), and
\begin{equation}
 S(\rho_s\|\rho_0)
 =\Delta\langle K_D^0\rangle_s-\Delta S_s .
 \label{eq:updated-relative-entropy}
\end{equation}
The first variation at \(s=0\) vanishes:
\begin{equation}
 \delta S=\delta\langle K_D^0\rangle .
 \label{eq:updated-first-law}
\end{equation}
For Type-III algebras the corresponding statement uses Araki relative
entropy and the relative modular operator \cite{Araki}.  Equation
\eqref{eq:updated-first-law} is an algebraic theorem.  It does not identify
\(K_D^0\) with stress-energy or curvature.

The lattice result in Sec.~\ref{sec:loop-modular} asks a narrower,
operationally well-posed question.  For a free critical interval, can the
linear modular response in a fixed physical frequency band be replaced by a
sparse local boost generator with a certified error?  The answer is yes in
that model.  The band is chosen from the local operator \(B_N\), its retained
dimension grows with the cutoff, and the response error is bounded in trace
norm.  The range-three correction has a frozen no-refit numerical validation.

To reach a gravitational equation, a common microscopic model must still
establish:
\begin{enumerate}[label=(E\arabic*)]
\item local diamond algebras on the response-generated geometry;
\item a modular-to-boost limit at fixed physical bandwidth in more than one
spatial dimension;
\item several independent null directions, so that an integrated scalar
projection does not masquerade as a tensor equation;
\item a transverse information density whose normalization determines
\(1/(4G_{\rm eff}\hbar)\), rather than importing \(G\);
\item a nonlinear focusing or area-response identity with controlled lattice,
localization, and state errors;
\item same-update decay of both constitutive defects in
\eqref{eq:finite-constitutive-certificate}, controlled finite symmetry and
matter-equation residuals, and the uniform conditioning required by
Theorem~\ref{thm:ward-bianchi-completion}.
\end{enumerate}
The affine null-wire regulator of
Sec.~\ref{sec:affine-null-wire-closure} realizes model versions of
(E1)--(E6): translated lattice subdiamonds, a fixed-band transverse wire
bundle, 30 null directions, boundary record entropy, a discrete focusing
recurrence, conserved wire stress, and an exact central-difference Bianchi
identity.  The qualifications are essential.  Its ``higher-dimensional''
matter is a direct sum of decoupled wires; the response cone and area
calibration are specified; and the equation is linearized around the
reconstructed constant metric.

The finite nonlinear bridge of
Sec.~\ref{sec:nonlinear-information-focusing} removes two further
approximations without claiming the remaining implication.  It replaces the
entropy first law by Eq.~\eqref{eq:nonlinear-finite-information}, the scalar
linear focusing law by the matrix identities
\eqref{eq:nonlinear-riccati}--\eqref{eq:nonlinear-area-update}, and the flat
rank certificate by Eq.~\eqref{eq:nonlinear-curved-rank}.  Its output is the
defect equation \eqref{eq:nonlinear-null-defect}.  The missing content of
(E6) is precisely the microscopic Ward relation
\eqref{eq:nonlinear-ward-target}, not an unspecified collection of
``higher-order corrections.''

The finite completion theorem of
Sec.~\ref{sec:ward-bianchi-completion} sharpens (E6) further.  Once the tensor
defined by the null defect is the geometric Euler covector, its nonmetric
part, divergence, local scalar mode, measurement noise, symmetry defect, and
off-shell matter residual enter the explicit bound
\eqref{eq:ward-bianchi-completion}.  The theorem therefore removes a
qualitative tensor-completion assumption but not the constitutive
identification or its uniform dynamical suppression.

The same-update theorem of Sec.~\ref{sec:same-update-einstein} closes that
identification for one declared local phase action and fixes the remaining
constant metric mode with a volume score.  Thus (E6) is complete in the FLRW
and \(L^4\) transverse-traceless regulators at finite resolution.  What
remains open in the generic program is not another tensor inversion: it is
the dynamical selection of this Einstein--matter phase from the preceding
record--geometry update, together with interacting modular input and uniform
regulator control.

Only after (E1)--(E6) are available would the familiar conditional chain
\begin{equation}
 \delta S_{\rm alg}
 =\delta\langle K\rangle
 \longrightarrow
 \delta A+4G_{\rm eff}\hbar\,\delta S_{\rm bulk}=0
 \longrightarrow
 \delta G_{\mu\nu}=8\pi G_{\rm eff}\delta T_{\mu\nu}
 \label{eq:updated-einstein-chain}
\end{equation}
become an internal derivation.  It is internal to the affine regulator under
the hypotheses of Theorem~\ref{thm:affine-model-closure}; it remains
conditional in the generic program.  The last two arrows are closely related
to thermodynamic and entanglement-equilibrium routes to the Einstein equation
\cite{Jacobson1995,Jacobson2015}; citing those routes alone would not supply
the microscopic area density, separating null probes, or focusing identity
constructed in the regulator.

At finite exchange and curvature, the corresponding nonlinear chain is not
obtained by replacing every \(\delta\) in
Eq.~\eqref{eq:updated-einstein-chain} with \(\Delta\).  The positive
\(\mathcal R_{\rm info}\), nonlinear optical remainder, modular defect, and
localization error must instead satisfy
Eqs.~\eqref{eq:nonlinear-null-defect} and
\eqref{eq:nonlinear-ward-target}.  This distinction is the main advance of
the second bridge and the reason the full nonlinear claim remains open.

\section{de Sitter response: benchmark, inverse theorem, and non-emergence}
\label{sec:updated-desitter}

\textbf{Status: \StatusNumerical/\StatusConditional.}
Two de Sitter constructions appeared in the earlier manuscript and should
not be conflated.

First, a Poisson sprinkling of a de Sitter causal diamond, equipped with
response gates and record dephasing, is a consistency benchmark.  If the
continuum geometry is used to generate the graph, recovering de Sitter
relations shows that the model class is nonempty and that the discretization
can preserve the intended symmetries.  It does not show that the geometry
emerges.

Second, the quantum-field inverse theorem of
Sec.~\ref{sec:loop-lorentz} treats the radius as unknown.  The data are Ramsey
coherences, not coordinate distances, and the rank-five condition fixes
\[
 R^2=(\mathbf1^TA^{-1}\mathbf1)^{-1}.
\]
This is a genuine output within a specified background and state family.
Its strongest tests are redundant sextuple radii, Lorentzian inertia,
target hyperboloid norms, and target pairs withheld until the estimator is
frozen.  Failure of any of those tests rejects the assumed kernel--geometry
pair.  This result builds on the established use of field propagators and
localized quantum measurements as geometric probes
\cite{SaravaniAslanbeigiKempf2016,PercheMartinMartinez2022}; it is not a
claim to have invented response-based geometry.

The two constructions can support a future closed loop only if a single
background-free process produces both the laboratories and the field
response used by the inverse theorem.  Until then, the correct claim is
\[
 \text{physical quantum response can infer an unknown de Sitter scale},
\]
not
\[
 \text{de Sitter spacetime has been derived from the RQCP dynamics}.
\]

\section{Central charge, metric volume, and vacuum counterterms}
\label{sec:updated-vacuum}

\textbf{Status: \StatusModel\ for the charge center;
\StatusConditional\ for fixed-volume cancellation;
\StatusOpen\ for cosmology.}

The local model of Sec.~\ref{sec:loop-center} proves
\begin{equation}
 \ker\mathcal L^*
 =\operatorname{span}\{P_0,\ldots,P_N\},
 \qquad
 \Delta\ge\min\{\gamma/2,\eta\lambda_2(L_G)\}.
 \label{eq:updated-charge-center}
\end{equation}
This is the part that is derived.  The projectors \(P_v\) are the complete
fixed observable algebra of a local Lindbladian, and weak homogeneous flips
lift the degeneracy into an exact Ehrenfest spectrum.  No global center
variable is appended to obtain this result.

\subsection{The conditional fixed-volume identity}

Suppose, in addition, that \(v\) converges to Euclidean four-volume and that
one computes normalized response in a fixed sector:
\[
 \mathcal W_v[J]=\log Z_v[J]-\log Z_v[0].
\]
If a matter loop adds a source-independent term \(c_0m^4v\), then
\begin{equation}
 Z_v[J]\mapsto e^{-c_0m^4v}Z_v[J]
 \quad\Longrightarrow\quad
 \mathcal W_v[J]\mapsto\mathcal W_v[J].
 \label{eq:updated-a0-cancel}
\end{equation}
Every connected \(J\)-derivative is unchanged.  This is more informative
than cancellation of a global state-vector phase, because it concerns a
normalized generating functional.

The selectivity is equally important.  For a regulated local matter theory,
\begin{equation}
 \Gamma_{\rm eff}[g]
 =\int\!\sqrt g\,
 \left(c_0m^4+c_1m^2R+c_2R^2
 +c_3R_{\mu\nu}R^{\mu\nu}+\cdots\right).
 \label{eq:updated-heat-kernel}
\end{equation}
Only the \(c_0m^4v\) term is constant on a fixed metric-volume sector.
Curvature, anomaly, boundary, and field-dependent terms remain in connected
response and undergo ordinary effective-field-theory renormalization.

\subsection{The missing central-volume bridge}

Equation~\eqref{eq:updated-a0-cancel} becomes relevant to gravity only after
proving all of the following in one model:
\begin{enumerate}[label=(V\arabic*)]
\item an additive microscopic charge with the exact center property
\eqref{eq:updated-charge-center};
\item convergence of that charge to metric four-volume, including units and
finite-size error;
\item sectorwise normalization derived from process composition rather than
imposed as a choice of ensemble;
\item a gravitational response map that is local within a sector and cannot
measure cross-sector weights;
\item preservation of (V1)--(V4) under coarse graining and matter
counterterms.
\end{enumerate}
The present work establishes (V1) for a hard-core channel and the algebraic
consequence of (V2)--(V3) if they hold.  It does not establish (V2)--(V5).

The earlier action with a hand-introduced variable \(\lambda_R\) is therefore
best read as a comparison with sequestering or unimodular constructions, not
as a microscopic result.  Likewise,
\(\rho_\Lambda\propto M_{\rm Pl}^2L_R^{-2}\) is dimensional scaling unless
the central stiffness, the response length \(L_R\), and their coupling to
curvature are all derived.  The present construction makes no prediction for
\(\Lambda\) \cite{FiolGarriga2010,KaloperPadilla,WeinbergCC}.

\section{Matter reconstruction as a sequence of open classification
problems}
\label{sec:updated-matter}

\textbf{Status: \StatusOpen.}
The earlier long form attached gauge registers to causal links and imposed a
finite spectral triple.  Both are useful target constructions, but neither
derives the Standard Model from the foundational channel.

\subsection{Gauge symmetry from influence sectors}

The nontrivial target is to begin with localized endomorphisms or residual
influence sectors \(\{\rho_a\}\) and prove that they form a rigid symmetric
tensor \(C^*\)-category,
\begin{equation}
 \rho_a\otimes\rho_b
 \simeq\bigoplus_cN_{ab}^{\ \ c}\rho_c ,
 \label{eq:updated-fusion}
\end{equation}
with conjugates, intertwiners, locality, and a faithful fiber functor.
Tannaka--Krein or Doplicher--Roberts reconstruction would then identify a
compact group \(G\) with
\(\mathcal C_{\rm infl}\simeq\operatorname{Rep}(G)\)
\cite{DoplicherRoberts}.  The theorem to be
proved is the emergence of the category from one local channel, not the
known reconstruction result once the category is supplied.

The first credible milestone is one finite non-Abelian group or \(SU(2)\) in
an explicit model, together with an operational holonomy or charge-fusion
observable.  Attaching \(U_{xy}\) by hand and obtaining
\(U_p\simeq\exp(iF_{\mu\nu}\Sigma^{\mu\nu})\) shows how a given gauge register
has a continuum limit; it does not explain why the register or group exists.

\subsection{Chiral matter on causal defects}

A domain-wall Dirac operator can carry an index, and that index can protect
chiral zero modes.  The required RQCP result is a discrete causal or
non-Hermitian Dirac operator whose spectral flow, anomaly inflow, and index
remain stable under the same coarse graining that produces the causal
records.  Selecting topological data with \(\operatorname{Index}=3\) does
not explain three generations.  The first closed result should instead prove
index stability under disorder, weak nonunitarity, and record RG; a later
dynamical selection problem may ask why a particular index is preferred.

\subsection{Finite noncommutative geometry}

The algebra
\[
 \mathcal A_F=\mathbb C\oplus\mathbb H\oplus M_3(\mathbb C)
\]
is highly constrained once the axioms of a finite real spectral triple,
unimodularity, anomaly cancellation, and a representation are imposed.  In
RQCP-QG these are selection conditions on a candidate residual algebra.
They are not presently consequences of the influence-channel flow.  The
correct dependency is
\[
 \text{derived residual sector category}
 \longrightarrow
 \text{classification constraints}
 \longrightarrow
 \text{candidate matter algebra},
\]
with the first arrow still open.

\section{Response horizon and dark condensate: phenomenological branches}
\label{sec:updated-cosmology}

\textbf{Status: \StatusOpen.}
The response-horizon and response-superfluid constructions are retained as
testable long-horizon branches, not as support for the foundational chain.

\subsection{Response-horizon attractor}

For
\begin{equation}
 F_R(L)=\alpha M_{\rm Pl}^2H^2L^3
 -\beta M_{\rm Pl}^2HL^2+\cdots ,
 \label{eq:updated-horizon-free-energy}
\end{equation}
stationarity gives
\[
 L_R=\frac{2\beta}{3\alpha}H^{-1}.
\]
Imposing first-law matching at \(L=H^{-1}\) and then deriving
\(L_R=H^{-1}\) is circular.  A microscopic result must compute
\(\alpha\) and \(\beta\) independently of the desired radius, show
\(F_R''(L_R)>0\), derive the relaxation rate, and include matter and anomaly
corrections.  Until that calculation is performed,
\(3\alpha=2\beta\) is a model check, not a prediction.

\subsection{Response condensate}

The hydrodynamic ansatz
\[
 S_R=\int\!P_R(X)-\alpha_b\int\!\varphi\rho_b+\cdots
\]
can reproduce baryonic scaling if
\(P_R(X)\propto X\sqrt{|X|}\), as in the superfluid-dark-matter mechanism
\cite{BerezhianiKhoury}.  That exponent is currently an assumption.
A closed theory must derive the infrared singularity from a transfer spectrum
or response-charge critical point and produce at least one joint prediction
that cannot be adjusted independently, for example
\[
 \text{BTFR normalization}
 \longleftrightarrow
 \eta_{\rm slip}(k,z)
 \longleftrightarrow
 \text{cluster melting scale}.
\]
Rotation curves alone do not identify the mechanism.  Weak lensing,
clusters, the CMB, and linear growth must be treated by one parameter set,
with the condensate/normal-fluid boundary included.

\section{A status-preserving conditional synthesis}
\label{sec:updated-synthesis}

The value of a synthesis theorem is bookkeeping: it identifies a sufficient
set of bridges and prevents an omitted input from reappearing in the
conclusion.  It is not evidence that the bridges hold.

\begin{theorem}[Conditional composition with visible external inputs]
\label{thm:updated-conditional-synthesis}
Suppose:
\begin{enumerate}[label=(C\arabic*)]
\item the process-functional and influence-algebra hypotheses of
Sec.~\ref{sec:process-algebra} hold;
\item a local microscopic dilation has a controlled kinetic limit and an
infrared Boolean center;
\item a single model produces a locally finite causal order and a positive
information-volume measure with a Lorentzian metric--measure limit;
\item its local diamond algebras satisfy a modular-to-boost limit in a fixed
physical band, with a curved rank-nine family of local null probes;
\item the transverse information density fixes \(G_{\rm eff}\); the finite
information, modular, focusing, and localization remainders define the
geometric Euler covector of the same process; its symmetry and matter
residuals vanish uniformly; and the hypotheses of
Theorem~\ref{thm:ward-bianchi-completion} hold in a shrinking-diamond limit;
\item an additive central charge converges to metric four-volume and
sectorwise process normalization is stable under matter loops.
\end{enumerate}
Then the exact record events form a Boolean algebra; the continuum order and
measure determine the stated Lorentzian geometry; the relative-entropy first
law gives the tangent response, while the exact finite balance and
assumptions (C4)--(C5) imply the corresponding nonlinear Einstein dynamics;
and pure-volume counterterms cancel from normalized
fixed-sector response under (C6).
\end{theorem}

\begin{proof}
The Boolean conclusion follows from the center-projection theorem and the
assumed infrared convergence.  Assumption (C3) supplies, rather than derives,
the hypotheses of the relevant order--measure reconstruction theorem.
Equation~\eqref{eq:updated-first-law} gives the linear tangent statement.
At finite exchange, Theorem~\ref{thm:finite-nonlinear-remainder}, the
modular-to-boost limit, the constitutive Ward identification, and
Theorem~\ref{thm:ward-bianchi-completion} give the nonlinear gravitational
response under (C4)--(C5).  Finally, a pure volume counterterm
is constant in a fixed sector, so its multiplicative factor cancels between
\(Z_v[J]\) and \(Z_v[0]\).
No step proves (C3)--(C6); the theorem only records their consequences.
\end{proof}

\begin{remark}[Use of Theorem~\ref{thm:updated-conditional-synthesis}]
Theorem~\ref{thm:updated-conditional-synthesis} may be used to check that a
future microscopic model closes the whole chain.  It must not be cited as
evidence that RQCP-QG already derives a generic or nonlinear Einstein
equation, a cosmological constant, or a matter spectrum.
\end{remark}

The affine null-wire regulator supplies a model-specific linearized instance
of portions of (C3)--(C5).  The nonlinear finite-regulator theorem supplies
the exact information and optical identities and the curved null-frame
certificate.  Theorem~\ref{thm:capacity-reference-selection} additionally
selects the capacity-balanced reference for factorized tokens, but not from
the same nonlinear update and not for arbitrary correlated record geometry.
Theorem~\ref{thm:ward-bianchi-completion} additionally supplies the finite
tensor estimate and propagates nonzero response--Euler field and divergence
defects.  Theorem~\ref{thm:same-update-einstein} derives those two
constitutive budgets for one declared phase update, and
Theorem~\ref{thm:unified-operational-continuum} places that update and the
order--volume geometry on one common FLRW--TT refinement.  These model
theorems still do not select the Einstein--matter action, dimension,
capacity-density law, interacting modular dynamics, or universal \(G\).
They also do not supply (C6) or turn the generic continuum hypotheses into
background-independent theorems.  Accordingly, the warning in the preceding
remark continues to apply to autonomous quantum gravity, the cosmological
constant, and the matter spectrum.

\begin{landscape}
\section{Outstanding assumptions and decisive tests}
\label{sec:updated-ledger}

The remaining assumptions are ordered by logical dependence.  The last
column identifies an outcome that would falsify the proposed mechanism,
rather than merely leave its proof incomplete.

\small
\begin{longtable}{@{}P{0.17\linewidth}P{0.25\linewidth}
P{0.34\linewidth}P{0.16\linewidth}@{}}
\caption{Outstanding assumptions, the results needed to remove them, and
decisive failure tests.}
\label{tab:updated-assumption-ledger}\\
\toprule
Assumption & Required new result & Decisive observable or test & Failure
criterion\\
\midrule
\endfirsthead
\toprule
Assumption & Required new result & Decisive observable or test & Failure
criterion\\
\midrule
\endhead
One-pass memoryless bath
& Repeated-interaction limit with quantified initial correlations or finite
  bath reset
& Response-kernel error and information-backflow witness
& Recurrence survives on the claimed kinetic time\\
\addlinespace
Restricted block-cactus adjacency
& Bounded local cycle certificate or controlled nonlocal defect theory on a
  broader graph family
& Window DAG probability and critical order gap
& Long chordless cycles evade the certificate\\
\addlinespace
Autonomous adjacency and capacity-density law
& One closed dynamics selects the interaction support, enters a robust
  acyclic response phase, and derives the stable conversion from Boolean
  record capacity to the measure consumed by
  Theorem~\ref{thm:generic-order-volume-compactness}
& Adjacency susceptibility, cycle margin, capacity-density flow, and blind
  passage through the dimension, clock, covering, and topology certificates
& The operational gates fail persistently, or inequivalent capacity
  normalizations survive the same infrared dynamics\\
\addlinespace
Specified field state and de Sitter class
& Joint state--geometry identifiability or independent state calibration
& Redundant radii, inertia, hyperboloid norms, blind pairs
& State deformation passes every geometric test\\
\addlinespace
Free or decoupled-wire modular regulator
& Fixed-band bound in an interacting higher-dimensional model
& Multi-direction modular response and transverse density
& Error remains finite at fixed physical band\\
\addlinespace
Capacity-balanced boundary reference
& Extend Theorem~\ref{thm:capacity-reference-selection} from factorized tokens
  to correlated record geometry and couple it to the nonlinear update
& Correlated sector weights, modular area coefficient, contraction rate, and
  same-update Ward residual
& Selection fails without a record-number-dependent reset, or contraction
  vanishes under refinement\\
\addlinespace
Record count equals boundary area
& Derive \(\widehat A=a_0^2\widehat N\) from order--volume geometry and prove
  a universal refinement flow of \(a_0^2/\log q\)
& Blind area reconstruction and regulator flow of \(G_{\rm eff}\)
& Different microscopic tokenizations give inequivalent infrared
  coefficients\\
\addlinespace
Dynamical phase, action, and continuum selection
& Extend Theorems~\ref{thm:same-update-einstein} and
  \ref{thm:unified-operational-continuum} so that one
  record--matter--holonomy dynamics derives its effective phase action,
  rather than receiving the Einstein--matter form, selects its dimension and
  capacity law, and makes the finite symmetry and matter-equation residuals
  vanish uniformly in an interacting quantum refinement
& Direction-resolved information, modular, optical, constitutive, symmetry,
  matter, holonomy, order--volume, and coupling residuals, inserted into
  Eq.~\eqref{eq:unified-certificate}
& The infrared phase is not Einstein--Hilbert, the independent response and
  variation pipelines disagree, split controls pass, or the joint
  certificate stays finite under refinement\\
\addlinespace
Pre-caustic and integrated-diamond restriction
& Maslov-complete optical propagation and a uniform shrinking-diamond
  localization theorem
& Conjugate-point index, patch independence, and localization defect
& Area or tensor reconstruction becomes chart dependent at a caustic\\
\addlinespace
Regulator-dependent Newton normalization
& Renormalization or universality theorem for the boundary information
  coefficient
& The same infrared coefficient survives changes of record dimension and
  regulator
& Independent continuum extrapolations disagree\\
\addlinespace
Charge equals metric volume
& Scaling theorem from additive charge to four-volume
& Sector volume calibration and finite-size residual
& Charge fluctuations fail geometric extensivity\\
\addlinespace
Sectorwise vacuum normalization
& Derive the fixed-sector functional from process composition and perform a
  regulated matter-loop test
& \(a_0\) cancels while \(a_1,a_2,\ldots\) remain
& Source-dependent or curvature terms are spuriously removed\\
\addlinespace
Gauge category
& Residual influence sectors form a rigid symmetric tensor category
& Fusion, conjugates, intertwiners, and holonomy
& No stable category survives RG\\
\addlinespace
Chiral index and three generations
& RG-stable causal-defect Dirac index, followed by a separate selection
  mechanism
& Spectral flow and anomaly inflow under disorder
& Zero modes gap or index changes under allowed RG\\
\addlinespace
Horizon coefficient matching
& Compute \(\alpha,\beta,F_R''\), and relaxation from one model
& \(L_R/H^{-1}\) and attraction rate
& Equality requires imposing the target radius\\
\addlinespace
Dark equation of state
& Derive \(P_R(X)\) and a cross-scale parameter relation
& Rotation, lensing, clusters, CMB, and growth
& No common parameter region exists\\
\bottomrule
\end{longtable}
\end{landscape}

The order of work follows the dependency graph rather than the apparent
phenomenological reach:
\begin{enumerate}
\item submit and experimentally anchor the Boolean and memory loops;
\item unite Boolean records with dissipative causal order in one local
process;
\item extend the finite relational clock/order/volume model beyond its
compiled four-dimensional port language, then test the intrinsic inverse
geometry gates blindly;
\item extend and couple the factorized capacity-reference selection theorem,
derive the record--area map, and establish an interacting fixed-physical-band
modular response on that geometry;
\item couple record exchange, nonlinear Jacobi optics, Lorentz holonomy, and
matter transport so that the response tensor is the geometric Euler covector
of the same update and its finite symmetry and matter residuals vanish;
\item extend the common-refinement theorem from FLRW plus TT perturbations to
general nonlinear inhomogeneous four-dimensional configurations, quantum
fluctuations, and renormalized observables;
\item prove a caustic-complete, shrinking-diamond, regulator-independent
interacting limit and derive rather than impose the physical spectral
balance and Newton coefficient;
\item derive a metric-volume center and perform the selective loop
cancellation;
\item address gauge, chirality, horizon dynamics, and dark phenomenology
after their upstream inputs exist.
\end{enumerate}

\section{Conclusion}
\label{sec:updated-conclusion}

RQCP begins with an operational statement: response differences generate an
influence algebra, and its central projections are the jointly readable exact
events.  The explicit models studied here show that this statement has
nontrivial dynamics behind it.  Fresh-cell collisions produce the
dephasing--exchange kinetics microscopically; the resulting Lindbladian has
an exact Boolean charge algebra, a complete graph-dependent gap, and a
controlled metastable Markov reduction.  An absorbing noise field drives a
transition between quantum memory and Boolean records, and a
reversal-covariant defect dynamics produces a partial order before any
Lyapunov rank is introduced.

The order--volume result identifies the next mathematical step.  A certified
partial order and a positive additive record measure define an intrinsic
finite volume clock and, under compactness and identifiability conditions, a
Lorentzian metric--measure subsequential limit.  Smooth admissible limits
recover proper time, dimension, measure, and local topology with explicit
error and rejection criteria.  The theorem does not generate its order or
its absolute record-density unit.  The de Sitter Ramsey construction is
complementary: it determines an unknown curvature scale and blind invariants
from a physical field response, but only within a specified state and
geometry class.

The gravitational constructions establish several finite compatibility
results.  They control a physical modular band, null-frame inversion,
finite information balance, nonlinear optical propagation, Ward completion,
and the comparison between a measured response and the variation of the same
phase functional.  A common-refinement model places order, record measure,
and Einstein--matter response on one family of registers.  A same-link
determinant model derives a positive transverse-traceless two-derivative
kernel without a bare curvature term; relational histories, quadratic
constraint completion, and spectral balance extend this to a
regulator-independent Gaussian response.  These results demonstrate that
the required ingredients are mutually compatible in restricted models.
They do not select the dimension, microscopic program, nonlinear action,
interacting matter content, or continuum normalization.

The remaining foundational problem can be stated without reference to the
auxiliary regulators:
\begin{equation}
 \boxed{\begin{gathered}
 \text{one autonomous background-free process}
 \dashrightarrow
 \text{certified order + capacity-density law}\\
 \xrightarrow{\ \text{intrinsic volume clock}\ }
 (X,\tau,\mu)\\
 \dashrightarrow
 \text{selected Einstein phase + interacting quantum limit}.
 \end{gathered}}
 \label{eq:updated-open-core}
\end{equation}
The middle implication is proved under operational compactness and
identifiability conditions.  The first remains open because no general local
dynamics in the manuscript selects both adjacency and the absolute
capacity-density conversion.  The last is realized only in prescribed
finite or quadratic model families.  An autonomous completion must also
derive the gauge and clock structure, control nonlinear constraints and
renormalized observables, and show that the relevant limits are independent
of the microscopic regulator.

Projective causal histories and operator-growth rigidity constrain possible
solutions of the first implication, but they do not yet yield a covariant
non-Abelian geometry dynamics with backreaction.  Likewise, fixed-sector
normalization identifies a selective pure-volume cancellation, but a
metric-volume center and a regulated matter-loop calculation remain absent.
Gauge reconstruction, chirality, horizon dynamics, and dark phenomenology
therefore remain downstream questions.

The value of the framework at its present stage is the separation of these
claims.  Exact models establish the information and finite-response
mechanisms; conditional theorems specify how geometry and gravity would
follow under additional hypotheses; and the unresolved hypotheses are
accompanied by failure tests.  The framework will approach a microscopic
theory of quantum gravity only when the two dashed implications in
Eq.~\eqref{eq:updated-open-core} are obtained from the same dynamics.

\section{A same-link induced quantum-gravity chain}
\label{sec:same-link-induced-gravity}

The unified-continuum theorem of
Sec.~\ref{sec:unified-operational-continuum} solves a compatibility problem:
one register family can supply order, volume, and an Einstein--matter
response on the same limiting geometry.  It does not solve the origin of the
Einstein--Hilbert coefficient, because the action phase is part of that
model's microscopic update.  This section closes that narrower missing
interface in a different explicit regulator.  No curvature action is placed
in the update.  Scalar matter propagating on the same relational links that
carry order and volume induces a positive two-derivative transverse-
traceless (TT) kernel, and that kernel has a controlled large-species quantum
infrared limit \cite{XuSameLinkInduced2026}.

The result is model-specific.  Four dimensions, a colored graph language, a
low-entropy constructor/clock state, a fair pair scheduler, one physical
volume unit per record, the scalar species, a fixed microscopic spacing, and
the fixed-volume TT sector are declared inputs.  The theorem therefore
removes neither the generic-order assumption A7 nor the universal
capacity-density assumption A9V.  It removes the bare-curvature-action input
only in the displayed TT branch.

\subsection{Anonymous constructor and common record support}

Let
\begin{equation}
 {\cal G}_L=P_L\times T_L^3
 \label{eq:same-link-language}
\end{equation}
be a path of \(L\) clock layers times a periodic cubic spatial graph.
Constant-size gadgets encode its clock and spatial edge colors in a simple
unlabeled graph.  The resulting permutation-invariant language is decidable
in logarithmic workspace: a verifier checks gadget type, consistent layer
increments, periodic three-axis incidence, and the graph-size relation.
Universal stable-network-constructor results then compile an anonymous
finite-state pair protocol whose absorbing output is isomorphic to
\({\cal G}_L\) \cite{MichailSpirakis2016}.  This invokes constructor
universality; it is not a new optimal constructor theorem and does not
dynamically select \(3+1\) dimensions.

For the final relational verification stage, let \(M=n(n-1)/2\) be the
number of unordered pair registers and \(W_r\) the number of wrong edge
bits after \(r\) random pair inspections.  Overwriting the inspected bit by
its relational certificate gives
\begin{equation}
 {\mathbb E}(W_{r+1}\mid W_r)
 =W_r\left(1-\frac1M\right),
 \qquad
 \eta_{\rm con}(s)
 :=\frac{{\mathbb E}W_{sM}}{M}\leq
 \eta_{\rm con}(0)e^{-s}.
 \label{eq:same-link-constructor-error}
\end{equation}
Every stochastic pair transition \(P(c'|c)\) has the exact fresh-cell
isometry
\begin{equation}
 V|c\rangle|0\rangle_E
 =\sum_{c'}\sqrt{P(c'|c)}\,|c'\rangle|c,c'\rangle_E ,
 \label{eq:same-link-stinespring}
\end{equation}
which extends to a unitary and recovers the pair map after the cell is
traced out \cite{Stinespring1955}.  A history-state construction can encode
the finite collision word \cite{Feynman1985,PageWootters1983}; it does not
make the program or clock state dynamical.

The directed clock links generate a partial order after transitive closure.
Each accepted vertex carries one Boolean record cell.  For any record subset
\(A\),
\begin{equation}
 \mu_L(A)=a_*^4\,\# A ,
 \label{eq:same-link-volume}
\end{equation}
which is positive and exactly additive.  For a macroscopic profile
\(F_R(x)=f(x/R)\), \(R=La_*\), midpoint summation yields
\begin{equation}
 \left|
 \frac{a_*^4}{R^4}\sum_{x\in{\cal G}_L}F_R(x)
 -\int_{[0,1]^4}f(u)\,\diff^4u
 \right|
 \leq C_fL^{-2}.
 \label{eq:same-link-volume-rate}
\end{equation}
The same active links define the scalar graph Laplacian.  Its spatial
dispersion is
\begin{equation}
 \widehat{\bm k}^{\,2}
 =\frac4{a_*^2}\sum_{i=1}^{3}
 \sin^2\!\left(\frac{a_*k_i}{2}\right)
 =\bm k^2+O(a_*^2|\bm k|^4).
 \label{eq:same-link-dispersion}
\end{equation}
Thus order, integer volume, and long-wave matter response are not evaluated
on three separately supplied graphs.

\begin{figure}[p]
 \centering
 \includegraphics[width=\linewidth]
 {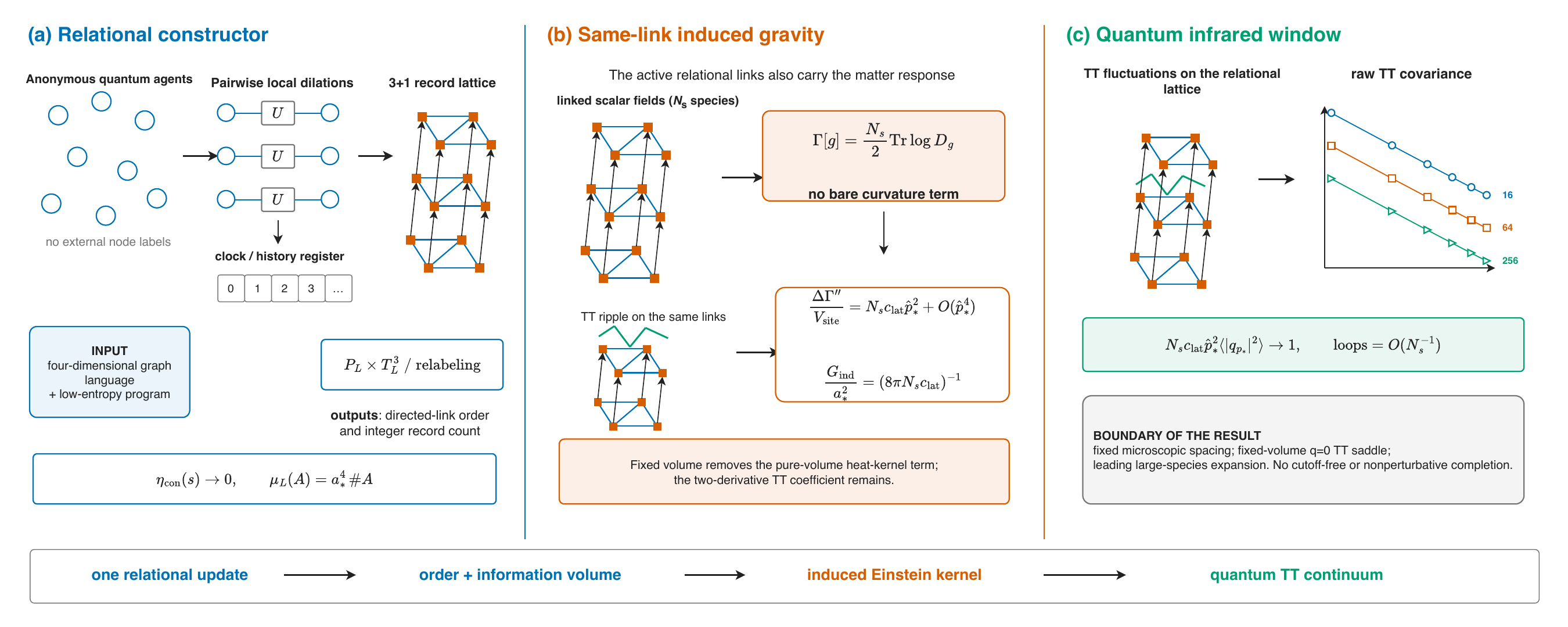}
 \caption{\textbf{Model-specific same-link chain.}
 (a) Anonymous local pair interactions compile an unlabeled
 \(P_L\times T_L^3\) record lattice.  The four-dimensional language,
 program/clock state, and scheduler are inputs; order and integer record
 volume are outputs on the absorbed support.
 (b) Linked scalar species propagate on those active edges.  Their
 determinant, with no bare curvature term, induces the positive
 two-derivative TT kernel.
 (c) At fixed microscopic spacing, the rescaled TT covariance approaches a
 Gaussian massless infrared propagator as \(N_s\) increases.  The lower box
 states the boundary that is part of the theorem.  The plotted covariance is
 calculated from the archived regulator, not illustrative data.}
 \label{fig:same-link-mechanism}
\end{figure}

\subsection{Positive Einstein kernel without a bare curvature action}

On the absorbed Euclidean support, take \(N_s\) identical real scalar fields.
Put \(\theta_\alpha=a_*k_\alpha\), \(\mathfrak m=ma_*\), and define the
dimensionless nearest-neighbor operator
\begin{equation}
 D_0(\theta)=\mathfrak m^2+\sum_{\alpha=1}^4\ell_\alpha(\theta),
 \qquad
 \ell_\alpha(\theta)=4\sin^2\frac{\theta_\alpha}{2}.
 \label{eq:same-link-D0}
\end{equation}
The physical operator is \(a_*^{-2}D_0\); its overall power of \(a_*\)
adds only a \(q\)-independent constant to the determinant.
The same link weights carry one plus-polarized, unimodular TT register,
\begin{equation}
 g_{11}=e^{2q\cos(p_*n_4)},\qquad
 g_{22}=e^{-2q\cos(p_*n_4)},\qquad
 g_{33}=g_{44}=1 .
 \label{eq:same-link-TT}
\end{equation}
Here \(p_*=a_*p\) and \(n_4=x_4/a_*\).
Because \(\det g=1\), this perturbation leaves
Eq.~\eqref{eq:same-link-volume} fixed.  Its weighted scalar operator is
\begin{equation}
 D(q)=e^{-2q\cos(p_*n_4)}\ell_1
     +e^{2q\cos(p_*n_4)}\ell_2+\ell_3+\ell_4+\mathfrak m^2 ,
\end{equation}
and the only TT action is the matter determinant
\begin{equation}
 \Gamma[q]=\frac{N_s}{2}\Tr\log D(q).
 \label{eq:same-link-determinant}
\end{equation}

Write \(D=D_0(\theta)\), \(A=\ell_2-\ell_1\), and
\(S=\ell_1+\ell_2\).  For a nonzero cosine mode, direct
differentiation gives, per species and lattice site,
\begin{equation}
 \frac{\Gamma''(p_*)}{N_sV_{\rm site}}
 =\frac12\int_{[-\pi,\pi]^4}\frac{\diff^4\theta}{(2\pi)^4}
 \left[
 \frac{2S}{D}
 -\frac{A^2}{D}
 \left(\frac1{D_{\theta+p_*e_4}}+\frac1{D_{\theta-p_*e_4}}\right)
 \right].
 \label{eq:same-link-exact-kernel}
\end{equation}
The \(p_*=0\) subtraction is taken with the limiting normalization of the
nonzero cosine.  Expanding the shifted resolvents and integrating by parts
in \(\theta_4\) yields
\begin{align}
 \frac{\Gamma''(p_*)-\Gamma''(0)}{N_sV_{\rm site}}
 &=c_{\rm lat}(\mathfrak m)\widehat p_*^{\,2}
   +O(\widehat p_*^{\,4}),\nonumber\\
 c_{\rm lat}(\mathfrak m)
 &=2\int_{[-\pi,\pi]^4}\frac{\diff^4\theta}{(2\pi)^4}
 \frac{(\ell_2-\ell_1)^2\sin^2\theta_4}
 {\left[\mathfrak m^2+\sum_\alpha\ell_\alpha(\theta)\right]^4}>0 .
 \label{eq:same-link-positive-c}
\end{align}
Here \(\widehat p_*^{\,2}=4\sin^2(p_*/2)\) is dimensionless.
Strict positivity follows because the integrand is nonnegative and is
positive on a set of nonzero measure whenever matter hops in both transverse
directions.  It is not inferred from a fit.  Deleting intersite hopping makes
\(\ell_2-\ell_1=0\) and gives the exact zero control.

For the normalization \eqref{eq:same-link-TT}, the Euclidean convention
\begin{equation}
 S_E=-\frac{M_{\rm P}^2}{2}\int\sqrt g\,R
\end{equation}
has quadratic TT response
\(\tfrac12M_{\rm P}^2a_*^2V_{\rm site}
\widehat p_*^{\,2}|q_{p_*}|^2\).
Consequently
\begin{equation}
 M_{\rm P,ind}^2=\frac{N_sc_{\rm lat}}{a_*^2},
 \qquad
 G_{\rm ind}=\frac{a_*^2}{8\pi N_sc_{\rm lat}} .
 \label{eq:same-link-newton}
\end{equation}
This is a determined coefficient of the specified cutoff theory, not a
universal low-energy number.  The heat-kernel interpretation is the usual
Sakharov mechanism: the linked matter determinant produces the curvature
term \cite{Sakharov1968,Adler1982,Visser2002,Vassilevich2003}.  The novelty
claimed here is the exact positive lattice coefficient and its composition
with the relational record support, not induced gravity by itself.
The pure-volume heat-kernel coefficient cancels from the fixed-volume
difference; curvature-squared and nonlocal \(O(\widehat p_*^4)\) terms remain.
This cancellation is not a vacuum-energy sequestering theorem.

\begin{theorem}[Same-link induced Einstein kernel]
\label{thm:same-link-induced-kernel}
For the scalar operator \eqref{eq:same-link-D0} on an absorbed relational
lattice, the fixed-volume TT determinant has no \(p_*^0\) response and has the
strictly positive leading coefficient
\eqref{eq:same-link-positive-c}.  Hence the induced Newton coefficient
\eqref{eq:same-link-newton} is finite and positive at fixed \(a_*\).
Heavy linked matter decouples as \(ma_*\to\infty\), whereas an ultralocal
matter control gives \(c_{\rm lat}=0\).
\end{theorem}

\subsection{Quantum infrared limit of the induced field}

Integrate the TT link register on the perturbative basin of the stable
\(q=0\) saddle with a translation-invariant local reference measure and the
determinant action \eqref{eq:same-link-determinant}.  Around that saddle,
\begin{equation}
 \Gamma[q]
 =N_s\sum_{n\geq2}\frac1{n!}\Gamma_n[q^n],
 \qquad
 \Gamma_2(p_*)=c_{\rm lat}\widehat p_*^{\,2}
 \left[1+O(\widehat p_*^{\,2})\right].
 \label{eq:same-link-vertices}
\end{equation}
With \(\varphi_{p_*}=\sqrt{N_sc_{\rm lat}}\,q_{p_*}\),
\begin{equation}
 N_sc_{\rm lat}\widehat p_*^{\,2}
 \langle|q_{p_*}|^2\rangle
 =1+O(\widehat p_*^{\,2})+O(N_s^{-1}).
 \label{eq:same-link-covariance}
\end{equation}
The \(n\)-point vertex of \(\varphi\) scales as
\(N_s^{1-n/2}\), and every additional metric loop costs \(N_s^{-1}\).
Thus connected non-Gaussian cumulants vanish for \(n>2\) on the declared
fixed infrared band.  Reflection positivity of the leading
nearest-neighbor Gaussian kernel gives its standard
Osterwalder--Schrader continuation \cite{OsterwalderSchrader1973}; the
statement is only for this TT field, not for a gauge-complete graviton.

Let \({\cal P}_{\Lambda,L}\) be a finite set of nonzero modes with
\(\max_{p\in{\cal P}_{\Lambda,L}}a_*|p|\to0\), and define
\begin{align}
 \epsilon_{\rm det}(L,\Lambda)
 &=\sup_{p\in{\cal P}_{\Lambda,L}}
 \left|
 \frac{\Gamma_2(p_*)}
 {c_{\rm lat}\widehat p_*^{\,2}}-1
 \right|,\nonumber\\
 {\mathfrak Q}_{L,s,N_s}
 &=\eta_{\rm con}(s)+C_{\rm vol}L^{-2}
 +C_{\rm disp}L^{-2}
 +\epsilon_{\rm det}(L,\Lambda)
 +\epsilon_{\rm BZ}+N_s^{-1}.
 \label{eq:same-link-joint-certificate}
\end{align}
Here \(\epsilon_{\rm BZ}\) is the deterministic Brillouin-zone quadrature
error.  Every summand is either a finite observable or a remainder with a
declared norm.

\begin{theorem}[Composed relational quantum TT limit]
\label{thm:same-link-quantum-limit}
Run the constructor so that \(\eta_{\rm con}\to0\), take \(L\to\infty\) at
fixed \(a_*\) with \(a_*|p|\to0\) on the probed band, refine the
Brillouin-zone quadrature, and take \(N_s\to\infty\).  Then the normalized
integer record measure and graph response converge on their declared
macroscopic tests; the scalar determinant on those same links converges in
the TT sector to the Einstein two-derivative kernel with coefficient
\eqref{eq:same-link-newton}; and every saddle-selected finite-dimensional
Schwinger function of \(\varphi\) on the probed band converges to that of a
massless Gaussian field.  Connected \(n>2\) cumulants are
\(O(N_s^{1-n/2})\), and the composed error is
\(O({\mathfrak Q}_{L,s,N_s})\).
\end{theorem}

The order of limits is part of the theorem.  It is an infrared scale-
separation limit at fixed microscopic spacing, followed by the controlled
large-\(N_s\) expansion.  It is not the removal of a UV cutoff at fixed
physical \(G\).

\subsection{Independent finite certificate and controls}

The canonical calculation checks every interface separately before applying
the composition theorem.  Eight independent relabelings give a maximum
graph-spectrum residual \(6.04\times10^{-14}\); the wrong-link fraction after
ten sweeps is \(2.05\times10^{-5}\), compared with the exact expectation
\(2.28\times10^{-5}\).  The explicit Stinespring isometry and induced channel
have zero residual at numerical precision.  The normalized volume and
dispersion errors scale as \(L^{-1.99997}\) and \(L^{-1.98997}\).

Independent positive-integrand quadrature gives
\begin{equation}
 c_{\rm lat}(ma_*=10^{-3})=0.0029784498 .
 \label{eq:same-link-numerical-c}
\end{equation}
The finite-\(p\) joint extrapolation gives \(0.0029784389\), while the first
periodic momentum gives \(0.00297180\); these are cross-checks rather than
definitions of the coefficient.  At \(N_s=64\) and \(a_*=1\),
\(G_{\rm ind}/a_*^2=0.2087324\).  Raising the lattice mass to \(ma_*=8\) suppresses
\(c_{\rm lat}\) to \(5.07\times10^{-5}\) of its light-mass value.  In the
declared \(L=48\) band, the largest finite-momentum covariance-collapse
error for \(N_s=64\) is \(6.80\%\); the expansion parameter for metric loops
is \(1/64\).

\begin{figure}[p]
 \centering
 \includegraphics[width=\linewidth]
 {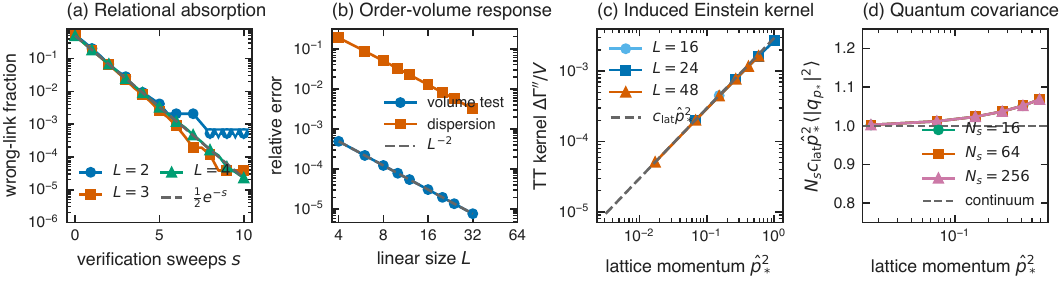}
 \caption{\textbf{Independent certificate for the composed limit.}
 (a) Random-pair verification after vertex relabeling.
 (b) Macroscopic record-volume quadrature and linked scalar dispersion.
 (c) The subtracted determinant and its positive
 \(c_{\rm lat}\widehat p_*^2\) limit.
 (d) Collapse of the rescaled TT covariance for three species numbers.
 Open downward markers in panel (a) denote zero observed errors plotted at
 half the aggregate detection floor.}
 \label{fig:same-link-joint-certificate}
\end{figure}

The coefficient has two sharp controls.  It decreases with the linked
matter mass, and it vanishes exactly when intersite hopping is removed.
These distinguish a genuine link-metric response from a fitted positive
constant.  The numerical tables and the scripts used to reproduce the
calculation are provided with the supporting material.

\subsection{What this closes, and what remains}

The model closes the following chain:
\begin{equation}
 \boxed{\begin{gathered}
 \text{anonymous local pair process}\\
 \Downarrow\\
 \text{one unlabeled record support carrying order, integer volume,
 and matter response}\\
 \Downarrow\\
 \text{strictly positive induced TT Einstein kernel}\\
 \Downarrow\\
 \text{controlled saddle-selected fixed-band large-\(N_s\) quantum TT field}.
 \end{gathered}}
 \label{eq:same-link-closed-chain}
\end{equation}
The action and quantum covariance are therefore no longer validated on
microscopic supports unrelated to the order--volume records.

Five boundaries prevent this theorem from being promoted to a general
quantum-gravity result.
\begin{enumerate}
\item The constructor compiles a four-dimensional graph language; it does
not select dimension or topology as a thermodynamic phase.
\item The program/clock state, fair scheduler, orientation color, and record
volume unit \(a_*^4\) remain microscopic inputs.
\item Matter is integrated out after the record support forms.  Its
backreaction on graph formation and a fully autonomous arrow of time are
absent.
\item Only the perturbative basin of the stable \(q=0\), fixed-volume TT
polarization is quantized.  Competing metric saddles, lapse, shift,
scalar/vector constraints, ghosts, nonlinear diffeomorphism symmetry,
caustics, and realistic gauge/fermion matter are not included.
\item The continuum statement keeps \(a_*\) fixed.  Higher-curvature
operators remain, and no regulator-independent trajectory for \(G\) or
cutoff-free quantum measure is established.
\end{enumerate}

Accordingly, A11I is marked \StatusModel, whereas A7, A9V, A11, and A19
remain open.  The next mathematical target is a phase-selection theorem that
replaces the compiled language and clock, followed by a constraint-complete
induced metric measure whose renormalization makes a regulator-independent
prediction.  The present chain is the model-level bridge on which that
stronger theorem must improve.

\section{Three decisive model chains beyond the fixed-regulator TT sector}
\label{sec:three-decisive-chains}

The same-link construction of
Sec.~\ref{sec:same-link-induced-gravity} ended at three visible boundaries:
the update word was still prepared relative to an external schedule, the
induced metric response was restricted to its transverse--traceless
quadratic component, and its Newton coefficient was defined at fixed
microscopic spacing.  Three companion results address these interfaces at
the model-specific level stated below
\cite{XuRelationalOrderVolume2026,XuConstraintCompleteEinstein2026,
XuRegulatorIndependentNewton2026}.  Their composition is a
regulator-independent \emph{quadratic} quantum Einstein band.  It is not a
nonlinear, background-independent completion of A11.

\subsection{Why the three arrows must be separated}

The three problems are logically independent:
\begin{equation}
\begin{split}
&\text{external update word}
 \longrightarrow \text{finite relational history},\\
&\text{two TT polarizations}
 \longrightarrow \text{ten components + constraints + ghosts},\\
&\text{fixed cutoff coefficient}
 \longrightarrow \text{regulator-independent Newton limit}.
\end{split}
\label{eq:three-chain-separation}
\end{equation}
A history-state representation does not generate a gauge identity.  A
Fierz--Pauli completion does not remove the cutoff.  A balanced determinant
does not explain the clock or the causal records.  Combining the conclusions
is legitimate only after the record support, metric derivative, and
determinant response are identified through explicit interfaces.

\subsection{Chain I: finite relational time, order, and information volume}

\paragraph{The obstruction.}
Let \(H=H^\dagger\) and \(T=T^\dagger\) act in finite dimension.  If the
Heisenberg velocity is nonnegative in every state, then
\begin{equation}
 i[H,T]\geq0,\qquad \Tr i[H,T]=0
 \quad\Longrightarrow\quad [H,T]=0.
\label{eq:three-clock-no-go}
\end{equation}
Thus no bounded observable is a state-independent strictly increasing clock
for a finite closed system.  The appropriate finite replacement is a
stationary state with conditional clock correlations, not a monotone
operator.

\paragraph{Exact history sector.}
For local reversible record updates \(U_0,\ldots,U_{M-1}\), introduce
\begin{align}
 H_{\rm prop}
 =\frac12\sum_{r=0}^{M-1}\big(&
 |r\rangle\langle r|+|r+1\rangle\langle r+1|\nonumber\\[-1mm]
 &-|r+1\rangle\langle r|\otimes U_r
 -|r\rangle\langle r+1|\otimes U_r^\dagger\big).
\label{eq:three-hprop}
\end{align}
Writing
\(|\psi_r\rangle=U_{r-1}\cdots U_0|\psi_0\rangle\), the state
\begin{equation}
 |\Psi_{\rm hist}\rangle
 =\frac1{\sqrt{M+1}}\sum_{r=0}^{M}
 |r\rangle_C|\psi_r\rangle
\label{eq:three-history}
\end{equation}
has exactly zero propagation energy.  Conditioning on
\(|r\rangle\langle r|\) returns \(|\psi_r\rangle\), and the first positive
eigenvalue is
\begin{equation}
 \Delta_M=1-\cos\frac{\pi}{M+1}.
\label{eq:three-history-gap}
\end{equation}
This follows by conjugating Eq.~\eqref{eq:three-hprop} to one half the path
Laplacian.  The result replaces an external runtime by a finite internal
clock correlation.  The low-entropy clock endpoint and the program remain
state-preparation inputs.

\paragraph{Causal rank and the common volume record.}
An accepted event \(v\) is appended only after the event records in its
parent-port set \(P(v)\) are present.  The update writes
\begin{equation}
 r(v)=
 \begin{cases}
 0,&P(v)=\varnothing,\\
 1+\max_{u\in P(v)}r(u),&P(v)\neq\varnothing .
 \end{cases}
\label{eq:three-causal-rank}
\end{equation}
Every accepted edge \(u\to v\) raises \(r\), so a directed cycle is
impossible.  Append operations at incomparable targets use disjoint
registers and commute.  Adjacent exchanges of such operations connect all
topological schedules, proving that their final incidence record is the
same even when the schedule register is coherently superposed.

On the explicit \(P_L\times T_L^3\) port branch, the accepted events also
carry \(q\) stable record qubits.  If a local coupling \(J\) fixes the
distinguishable clock tick \(\tau_*=\pi/(2J)\), define
\begin{equation}
 \mu_L(A)=\tau_*^4\#A,\qquad
 C_L(A)=q\#A,\qquad
 \frac{C_L(A)}{\mu_L(A)}=\frac{q}{\tau_*^4}.
\label{eq:three-record-volume}
\end{equation}
The ratio is exactly preserved by additive \(b^4\)-cell blocking.  For
midpoint tests \(f\in C^2([0,1]^4)\),
\begin{equation}
 \left|L^{-4}\sum_{\bm n}f\!\left(
 \frac{\bm n+\bm1/2}{L}\right)-\int_{[0,1]^4}f\,\diff^4x\right|
 \leq\frac1{24L^2}
 \sum_{\alpha=1}^4\|\partial_\alpha^2f\|_\infty+O(L^{-4}).
\label{eq:three-weak-volume}
\end{equation}
Hence conditional history, order, capacity, and weak volume use one record
family and one refinement.  Exact diagonalization and 96 randomized
spacelike schedules give zero history, rank, cycle, activation, and blocking
defects within the archived precision; the fitted weak-volume exponent is
\(-1.99996\).

\begin{theorem}[Finite relational order--volume chain]
\label{thm:three-relational-order-volume}
For the Hamiltonian \eqref{eq:three-hprop} and the declared rank-increasing
append program, every conditional history branch has a schedule-independent
strict partial order.  Its accepted-event records define the positive
additive measure \eqref{eq:three-record-volume}; order and measure share the
same blocking maps, and the normalized weak measure converges at
\(O(L^{-2})\).
\end{theorem}

This proves A19M.  It does not prove the general A19: the finite clock
boundary, parent rule, four-dimensional port language, and volume unit are
not dynamically selected.

\subsection{Chain II: the constraint-complete quadratic Einstein field}

The same-link determinant previously fixed one positive TT coefficient,
\(Z_{\rm TT}\).  To determine what follows for the remaining eight tensor
components, consider the most general parity-even isotropic two-derivative
quadratic form of a symmetric tensor at nonzero momentum:
\begin{align}
 Q[h;q]={}&a_1q^2h_{\mu\nu}h_{\mu\nu}
 +a_2(q_\mu h_{\mu\nu})(q_\rho h_{\rho\nu})\nonumber\\
 &+a_3(q_\mu h_{\mu\nu}q_\nu)h
 +a_4q^2h^2.
\label{eq:three-quadratic-ansatz}
\end{align}
The declared linear gauge generator is
\begin{equation}
 (R\xi)_{\mu\nu}=q_\mu\xi_\nu+q_\nu\xi_\mu .
\label{eq:three-gauge-generator}
\end{equation}
Solving \(KR=0\) gives a one-dimensional coefficient space,
\begin{equation}
 (a_1,a_2,a_3,a_4)=\zeta(1,-2,2,-1).
\label{eq:three-fp-unique}
\end{equation}
TT matching fixes \(\zeta=Z_{\rm TT}/2\).  Thus no second coupling is
introduced when extending the determinant response to
\begin{align}
 S_E^{(2)}=\frac{Z_{\rm TT}}2\sum_q\big[&
 q^2h_{\mu\nu}h_{\mu\nu}
 -2(q_\mu h_{\mu\nu})^2\nonumber\\
 &+2(q_\mu h_{\mu\nu}q_\nu)h-q^2h^2\big].
\label{eq:three-fierz-pauli}
\end{align}

On the lattice, every occurrence of \(q_\mu\) is replaced by the same
\(\widehat q_\mu=2a^{-1}\sin(aq_\mu/2)\).  Define the de Donder map
\begin{equation}
 (Bh)_\nu=\widehat q^\mu h_{\mu\nu}
 -\frac12\widehat q_\nu h .
\label{eq:three-de-donder}
\end{equation}
Then the finite identities are
\begin{equation}
 K_E R=0,\qquad
 B R=\widehat q^{\,2}\mathbb 1_4,\qquad
 \operatorname{rank}K_E=6,\qquad
 \operatorname{rank}R=4.
\label{eq:three-finite-ward}
\end{equation}
The gauge-fixed kernel
\(K_\alpha=K_E+\alpha^{-1}B^TB\) has rank ten.  The
Faddeev--Popov operator is \(\widehat q^2\mathbb 1_4\), and the linear BRST
differential
\begin{equation}
 sh=Rc,\qquad sc=0,\qquad
 s\bar c=b,\qquad sb=0
\label{eq:three-brst}
\end{equation}
is nilpotent.

For a source satisfying \(R^TJ=0\), the Gaussian response is independent of
the gauge parameter:
\begin{equation}
 \frac{\partial}{\partial\alpha}
 \left(J^TK_\alpha^{-1}J\right)=0 .
\label{eq:three-source-gauge}
\end{equation}
A source with a gauge component supplies a required failure control and
changes by order unity across the tested \(\alpha\)-range.  On every fixed
finite set of nonzero physical Fourier modes,
\begin{equation}
 \frac{\|K_E(\widehat q)-K_E(q)\|}{\|K_E(q)\|}
 =O(a^2).
\label{eq:three-full-band}
\end{equation}
The deterministic archive gives a maximum normalized Ward residual
\(1.98\times10^{-16}\), BRST-square residual zero, conserved-source gauge
spread \(3.47\times10^{-18}\), and continuum exponent \(-1.988\).

\begin{theorem}[Constraint-complete induced Gaussian response]
\label{thm:three-constraint-complete}
Given a positive induced TT normalization and the linear gauge law
\eqref{eq:three-gauge-generator}, locality, isotropy, parity, and
two-derivative scaling uniquely determine the full Fierz--Pauli kernel.
Equations~\eqref{eq:three-finite-ward}--\eqref{eq:three-brst} form an exact
finite lattice gauge--ghost complex, and all conserved-source Gaussian
Schwinger functions have a gauge-independent fixed-band continuum limit.
\end{theorem}

This proves A11Q.  Linear gauge symmetry and a conformal-contour or
Lorentzian \(i0\) prescription are model inputs.  The theorem contains no
cubic vertex, nonlinear constraint algebra, Slavnov--Taylor hierarchy, or
BV master equation.

\subsection{Chain III: a regulator-independent induced Newton coefficient}

Let determinant sector \(i\) have fixed mass \(m_i>0\), signed weight \(w_i\),
and a common curvature normalization \(\kappa_g\).  For the declared local
regulator class, its dimension-two response has
\begin{align}
 I_r(m,\Lambda)
={}&A_r\Lambda^2
 -m^2\ln\frac{\Lambda^2}{\mu^2}
 +B_rm^2\nonumber\\
 &+m^2\ln\frac{m^2}{\mu^2}
 +\rho_r(m,\Lambda),\qquad \rho_r\to0 ,
\label{eq:three-regulator-class}
\end{align}
where \(A_r>0\) and \(A_r,B_r\) may be scheme dependent.  The induced
normalization is
\begin{equation}
 Z_r(\Lambda)=\kappa_g\sum_iw_iI_r(m_i,\Lambda).
\label{eq:three-running-z}
\end{equation}
Writing
\begin{equation}
 S_0=\sum_iw_i,\qquad
 S_1=\sum_iw_im_i^2,\qquad
 F=\sum_iw_im_i^2\ln\frac{m_i^2}{\mu^2},
\label{eq:three-spectral-moments}
\end{equation}
gives
\begin{equation}
 \frac{Z_r}{\kappa_g}
 =A_rS_0\Lambda^2
 -S_1\ln\frac{\Lambda^2}{\mu^2}
 +B_rS_1+F+o(1).
\label{eq:three-running-decomposition}
\end{equation}

If all \(w_i\geq0\) and the spectrum is nonzero, then \(S_0>0\), so a
cutoff-free coefficient is impossible without a bare counterterm or an
additional signed sector.  If \(S_0=0\) but \(S_1\neq0\), a logarithmic
divergence remains.  The two conditions
\begin{equation}
 S_0=0,\qquad S_1=0
\label{eq:three-balance}
\end{equation}
remove the power, logarithmic, and scheme-dependent finite local terms and
leave
\begin{equation}
 Z_*=\kappa_gF
 =\kappa_g\sum_iw_im_i^2\ln\frac{m_i^2}{\mu^2}.
\label{eq:three-zstar}
\end{equation}
The reference-scale derivative is
\(\partial F/\partial\ln\mu^2=-S_1=0\).

The explicit example
\begin{equation}
 w=(1,-2,1),\qquad
 m^2=M_0^2(1,2,3)
\label{eq:three-balanced-example}
\end{equation}
has
\begin{equation}
 \frac{Z_*}{\kappa_gM_0^2}
 =3\ln3-4\ln2
 =0.523248143764548\ldots>0 .
\label{eq:three-balanced-value}
\end{equation}
A sharp momentum cutoff, sharp proper-time cutoff, and smooth proper-time
cutoff agree to a spread \(1.03\times10^{-5}\) at
\(\Lambda/M_0=1200\).  On the unit lattice, with
\begin{equation}
 C_{r,L}(n)=
 \frac1{Z_r(LM_0)[2L\sin(\pi n/L)]^2},
\label{eq:three-balanced-covariance}
\end{equation}
every fixed nonzero mode tends to
\([Z_*(2\pi n)^2]^{-1}\), with
\(O(L^{-2}\ln L)\) error for the three schemes.

\begin{theorem}[Two-moment Newton universality]
\label{thm:three-newton-universality}
Within the regulator class \eqref{eq:three-regulator-class}, the two
conditions \eqref{eq:three-balance} are sufficient for a finite,
scheme-independent induced Newton coefficient and a common fixed physical
Gaussian band.  They are necessary when the class contains independent
power and finite-local scheme coefficients and the spectrum is not tuned as
a function of the cutoff.
\end{theorem}

This proves A11RG.  Taken alone, the negative weight in
\eqref{eq:three-balanced-example} is only a formal determinant supertrace
or constrained determinant ratio, not an ordinary negative-multiplicity
particle.  Section~\ref{sec:positive-hilbert-balance} supplies a separate
free positive-Hilbert realization of exactly this spectrum, proves that it
persists before averaging over any supplied finite positive graph
ensemble, and thereby closes the model-level node A11BM.  It also proves
that exact all-level pairing would erase the desired finite invariant.  It
is followed by the selective-Ward closure in
Sec.~\ref{sec:selective-ward-closure}, which protects the moments for
arbitrary finite positive interacting physical operators and supplies a
positive finite-geometry channel on a declared language.  The
projective-growth extension in
Sec.~\ref{sec:projective-response-growth} then constructs exact compatible
cylinder laws and an infinite positive measure on a supplied naturally
labelled causal-birth language.  The covariant-growth extension in
Sec.~\ref{sec:covariant-response-growth} then maps contact-star responses to
a nonvanishing classical sequential-growth law with exact discrete general
covariance and Bell causality.  Its one-arity \(r\ge2\) complex branch has
bounded variation and therefore a countably additive history measure,
strongly positive rank-one decoherence functional, and nonzero
finite-volume interference.  The fixed-strip construction in
Sec.~\ref{sec:nonabelian-history-measure} separately supplies a
fresh-ancilla, noncommuting, rank-four, countably extendible history
measure.  The arbitrary-causet holonomy construction in
Sec.~\ref{sec:covariant-nonabelian-holonomy} then combines the classical CSG
density with intrinsic fresh-event noncommuting gates and proves a
relabeling-covariant countably additive operator measure with a
non-Abelian integrated range.  Next,
Sec.~\ref{sec:central-operator-growth} solves the central operator-valued
CSG branch and classifies it as a direct integral of scalar laws: its
infinite geometry POVM is exactly covariant but cannot detect
coupling-sector coherence, while a positive noncentral normalized control
fails DGC.  Section~\ref{sec:bell-causal-rigidity} then proves exact
rigidity of three natural neighborhoods of the remaining nonsingular CPOBC
branch without constructing or globally excluding that branch.  These
results do not prove the stronger A11B
requirement of deriving the protected commutant, maximal-birth/port language,
contact motifs, response multiplet, coherent arity and phase, an
operator-Markov or square-operator Bell law, or quantum geometry
backreaction; nor do they produce a tight
Lorentzian continuum, establish spin--statistics and anomaly freedom, or
couple the result to nonlinear BV constraints and graph backreaction.  The
cosmological and higher-curvature heat-kernel moments are not cancelled by
Eq.~\eqref{eq:three-balance}.

\subsection{The composed quadratic limit}

The three theorems compose under the following compatibility conditions.
Let the \(P_L\times T_L^3\) record support of
Theorem~\ref{thm:three-relational-order-volume} carry the determinant
sectors of Eq.~\eqref{eq:three-running-z}.  Assume that, away from the two
open clock boundaries, their local heat-kernel response belongs uniformly
to Eq.~\eqref{eq:three-regulator-class}.  Use the same local incidence
derivative in the scalar determinants, the symmetric-tensor kernel, and the
de Donder operator.  Boundary-insensitive fixed-band observables are
selected before \(L\to\infty\).  These are interface hypotheses, not
additional conclusions.

Define a joint discrepancy
\begin{align}
 {\cal D}_{M,L,r}={}&
 \varepsilon_{\rm hist}+\varepsilon_{\rm ord}
 +\varepsilon_{\rm block}
 +\varepsilon_{\rm vol}(L)+\varepsilon_{\rm link}(L)\nonumber\\
 &+\varepsilon_{\rm Ward}+\varepsilon_{\rm BRST}
 +\varepsilon_{\rm src}+\varepsilon_{Z,r}(L)
 +\varepsilon_{\rm band}(L),
\label{eq:three-joint-discrepancy}
\end{align}
where:
\begin{itemize}
\item \(\varepsilon_{\rm hist}\) is the zero-energy plus
conditional-state residual;
\item \(\varepsilon_{\rm ord}\) is the integer sum of rank, cycle, and causal
activation defects;
\item \(\varepsilon_{\rm block}\) and \(\varepsilon_{\rm vol}\) are the
capacity-density blocking and weak-measure errors;
\item \(\varepsilon_{\rm link}\) measures disagreement between the record support
and the links used in the determinant and tensor derivatives;
\item \(\varepsilon_{\rm Ward}\), \(\varepsilon_{\rm BRST}\), and
\(\varepsilon_{\rm src}\) measure \(K_ER\), \(s^2\), and
conserved-source gauge spread;
\item \(\varepsilon_{Z,r}=|Z_r-Z_*|/Z_*\), and
\(\varepsilon_{\rm band}\) is the lattice-dispersion/resolvent error.
\end{itemize}
Every term is a finite observable or a remainder in a stated norm.
In the explicit archives, all algebraic terms vanish to roundoff,
\(\varepsilon_{\rm vol}=O(L^{-2})\), and
\(\varepsilon_{Z,r}+\varepsilon_{\rm band}=O(L^{-2}\ln L)\).

\begin{theorem}[Relational, constraint-complete, regulator-independent
quadratic composition]
\label{thm:three-chain-composition}
Suppose the interface contract above holds, the two spectral moments vanish,
and \(Z_*>0\).  Condition on a finite internal clock history, take the common
record refinement \(L\to\infty\) on a fixed boundary-insensitive physical
band, and remove the regulator along \(\Lambda=LM_0\).  Then:
\begin{enumerate}
\item the accepted records converge weakly to their additive
order--volume measure while retaining an exact finite causal rank;
\item the ten-component Gaussian metric complex obeys the finite Ward and
BRST identities and has four gauge directions before fixing;
\item for every finite collection of conserved sources, the connected
two-point generating functional converges to the same gauge- and
regulator-independent Fierz--Pauli response with normalization \(Z_*\);
\item the total error is bounded by a constant times
\({\cal D}_{M,L,r}\) on the chosen band.
\end{enumerate}
\end{theorem}

\begin{proof}
The first conclusion is
Theorem~\ref{thm:three-relational-order-volume}.  Determinant additivity
reduces the total TT coefficient to Eq.~\eqref{eq:three-running-z}, and
Theorem~\ref{thm:three-newton-universality} replaces it by \(Z_*\) uniformly
on the fixed band.  The unique coefficient extension
\eqref{eq:three-fp-unique} then fixes the complete quadratic kernel without
introducing another coupling.  Exact finite Ward/BRST identities and
Eq.~\eqref{eq:three-source-gauge} remove gauge-parameter dependence.
The resolvent identity bounds the covariance change by the coefficient,
lattice-momentum, and source-projection errors.  Adding the independent
record/interface errors gives Eq.~\eqref{eq:three-joint-discrepancy}.
\end{proof}

\subsection{Boundary after the three chains}

Theorem~\ref{thm:three-chain-composition} is the strongest quantum-gravity
statement presently justified by these three additions.  It closes a
finite-clock, common-record, gauge-complete, regulator-independent
\emph{linear} quantum response under a supplied four-dimensional program,
linear gauge law, and balanced signed determinant spectrum.  It does not
close the following arrows:
\begin{enumerate}
\item dynamical selection of the clock/program state, adjacency language,
dimension, signature, or absolute record-volume unit;
\item dynamical derivation of the protected internal commutant, the
maximal-birth graph language/topology/move law, contact motifs, response
multiplet, coherent arity and phase, a square-operator
Markov/Bell-complete quantum-backreacting geometry law, and a tight
Lorentzian anomaly-free continuum realization of
the finite selective-Ward/projective/covariant-growth models in
Secs.~\ref{sec:selective-ward-closure}
--\ref{sec:bell-causal-rigidity};
\item nonlinear gauge symmetry, the BV master equation, graviton
self-interactions, or a nonperturbative metric probability measure;
\item graph backreaction, interacting matter, and a joint continuum limit
away from the declared fixed physical band;
\item simultaneous renormalization of vacuum and higher-curvature
coefficients;
\item the matter, central-volume, and cosmological branches.
\end{enumerate}
Accordingly A19M, A11Q, A11RG, A11BM, the finite selective-Ward node
A11BW, the labelled projective-history node A11PG, and the covariant
rank-one history node A11CG, together with the fixed-strip non-Abelian node
A11NH, arbitrary-causet holonomy node A11CH, and central operator-growth
classification A11OC, together with the CPOBC boundary-rigidity node A11BR,
are marked proved at model or regulator-class level.
General A7,
A9V, A11, A11B, and A19 remain open.  This distinction
is the reason the
three chains
strengthen the program without turning the conditional synthesis theorem
into a claim of completed quantum gravity.

\section{A positive-Hilbert origin of the balanced determinant spectrum}
\label{sec:positive-hilbert-balance}

\textbf{Status: \StatusModel\ for a balanced positive-Hilbert multiplet and
its exact or approximate selective Ward protection over an arbitrary
positive interacting operator on every member of a supplied finite geometry
language, including a positive finite geometry channel; \StatusOpen\ for the
microscopic origin of the protected commutant, a dynamically selected
geometry language and projective refinement measure, continuum Lorentz
spin--statistics and anomaly freedom, nonlinear BV closure, and a
background-independent continuum quantum measure.}

The regulator theorem in
Theorem~\ref{thm:three-newton-universality} identifies exactly which
spectral moments must vanish, but it does not explain why a physical
microscopic system should have the signed weights needed for that
cancellation.  Entering \((1,-2,1)\) as a list of ``multiplicities'' is not
an answer: a negative multiplicity is not a positive state space.  The
companion construction \cite{XuPositiveHilbertBalance2026} closes this
narrower free-model problem by separating two notions that must not be
conflated:
\[
 \text{sign of a closed Gaussian loop}
 \quad\ne\quad
 \text{sign of a state norm}.
\]

\subsection{Exterior parity and an additive mass ladder}

Let \(G\) be any finite weighted graph and let \(K_G\geq0\) be one real
symmetric local kernel on it.  Introduce the positive internal Hilbert space
\begin{equation}
 \Hh_{\rm int}=\Lambda^\bullet\mathbb C^k
 =\bigoplus_{n=0}^{k}\Lambda^n\mathbb C^k,
 \qquad
 N=\sum_{a=1}^{k}c_a^\dagger c_a ,
 \label{eq:ph-exterior}
\end{equation}
with its standard exterior-product inner product.  The degree-\(n\)
subspace has dimension \(\binom{k}{n}\) and parity
\((-1)^F=(-1)^N\).  Give every component the same graph operator and the
additive mass-squared operator
\begin{equation}
 M^2=x\mathbf 1+\Delta N,\qquad x,\Delta>0,\qquad
 h_n=K_G+(x+n\Delta)\mathbf 1 .
 \label{eq:ph-mass}
\end{equation}
Even components are represented by commuting Gaussian variables and odd
components by Grassmann variables.  Integrating out the complete multiplet
therefore gives
\begin{equation}
 \Gamma_k[K_G]
 =\sum_{n=0}^{k}(-1)^n\binom{k}{n}
   \Tr\log\!\left[K_G+(x+n\Delta)\mathbf 1\right].
 \label{eq:ph-gamma}
\end{equation}
Both the multiplicity and the loop sign in this expression are fixed by
\(\Lambda^\bullet\mathbb C^k\); they are not independently tuned
coefficients.

\begin{theorem}[Exterior-algebra spectral balance]
\label{thm:ph-balance}
For Eqs.~\eqref{eq:ph-exterior} and \eqref{eq:ph-mass},
\begin{equation}
 \Str_{\Hh_{\rm int}}e^{-sM^2}
 =e^{-sx}(1-e^{-s\Delta})^k .
 \label{eq:ph-heat}
\end{equation}
Consequently,
\begin{equation}
 \Str(M^2)^j=0\quad(0\leq j<k),\qquad
 \Str(M^2)^k=(-1)^k k!\Delta^k .
 \label{eq:ph-moments}
\end{equation}
The identities hold at every graph size and for every common kernel
\(K_G\).
\end{theorem}

\begin{proof}
Resolving the supertrace by exterior degree gives
\[
 \Str e^{-sM^2}
 =e^{-sx}\sum_{n=0}^{k}(-1)^n\binom{k}{n}e^{-sn\Delta}.
\]
The binomial theorem yields Eq.~\eqref{eq:ph-heat}.  Since
\(1-e^{-s\Delta}=s\Delta+O(s^2)\), its zero at \(s=0\) has order exactly
\(k\).  Taking the first \(k\) derivatives at the origin proves
Eq.~\eqref{eq:ph-moments}.
\end{proof}

For the induced Newton term the minimal choice is \(k=2\):
\begin{equation}
 (w_0,w_1,w_2)=(1,-2,1),\qquad
 m_n^2=x+n\Delta,\qquad
 \Str\mathbf 1=\Str M^2=0 .
 \label{eq:ph-k2}
\end{equation}
The logarithmic invariant that survives the two cancellations is
\begin{align}
 F_2(x,\Delta)
 &:=\sum_{n=0}^{2}w_n m_n^2
          \log\frac{m_n^2}{\mu^2}\nonumber\\
 &=\int_0^\Delta\!\diff u\int_0^\Delta\!\diff v\,
     \frac{1}{x+u+v}>0 .
 \label{eq:ph-positive-invariant}
\end{align}
The integral is the second finite-difference identity for
\(f(y)=y\log(y/\mu^2)\).  It proves both positivity and independence of
\(\mu\), without subtracting large cutoff terms.  At \(x=\Delta=M^2\),
\begin{equation}
 F_2=M^2(3\log3-4\log2)
 =0.523248143764548\ldots\,M^2 .
 \label{eq:ph-value}
\end{equation}

\subsection{Where positivity is proved}

The determinant supertrace is not used as the physical inner product.
Associate a bosonic Fock factor \(\F_+(h_n)\) with each even component and a
fermionic Fock factor \(\F_-(h_n)\) with each odd component:
\begin{equation}
 \F_{\rm phys}
 =\bigotimes_{n\ {\rm even}}
   \F_+(h_n)^{\otimes\binom{k}{n}}
 \otimes
 \bigotimes_{n\ {\rm odd}}
   \F_-(h_n)^{\otimes\binom{k}{n}} .
 \label{eq:ph-fock}
\end{equation}
Every factor has its standard positive norm.  On the finite graph,
\begin{equation}
 \widehat H=\sum_{n,\alpha}
 \diff\Gamma_\pm(h_n)_{n,\alpha}\geq0,\qquad
 T_a=e^{-a\widehat H},\quad 0<T_a\leq\mathbf1 .
 \label{eq:ph-transfer}
\end{equation}

\begin{proposition}[Finite transfer-form positivity]
\label{prop:ph-positive}
For every positive density operator \(\rho\) and finite operator family
\(\{A_\alpha\}\), the matrix
\begin{equation}
 (G_T)_{\alpha\beta}
 =\Tr\!\left[
 (T_a^{1/2}A_\alpha\rho^{1/2})^\dagger
 (T_a^{1/2}A_\beta\rho^{1/2})\right]
 \label{eq:ph-gram}
\end{equation}
is positive semidefinite.  This statement is independent of the loop
weights \(w_n\).
\end{proposition}

\begin{proof}
For every vector \(z\),
\[
 z^\dagger G_Tz
 =\left\|T_a^{1/2}
 \left(\sum_\alpha z_\alpha A_\alpha\right)\rho^{1/2}\right\|_2^2
 \geq0.
\]
The spectral theorem gives \(\widehat H\geq0\) and \(0<T_a\leq\mathbf1\).
\end{proof}

Thus the minus signs in Eq.~\eqref{eq:ph-gamma} are ordinary closed
fermion-loop signs.  Proposition~\ref{prop:ph-positive} is a finite-graph
transfer statement, not an interacting Osterwalder--Schrader
reconstruction \cite{OsterwalderSchrader1973}.

\subsection{Regulator and same-link interfaces}

In the regulator class used in
Theorem~\ref{thm:three-newton-universality}, every scheme-dependent term is
proportional either to \(\Str\mathbf1\) or to \(\Str M^2\).  Hence the
two-mode multiplet gives the common limit
\begin{equation}
 Z_*=\kappa_g F_2(x,\Delta)>0 .
 \label{eq:ph-zstar}
\end{equation}
For \(x=\Delta=1\) and \(\Lambda=1200\), the archived sharp-momentum,
sharp-proper-time, and smooth-proper-time evaluations are respectively
\[
 0.5232467549,\qquad 0.5232474487,\qquad 0.5232388701,
\]
with a total scheme spread \(8.58\times10^{-6}\).  The exact target is
Eq.~\eqref{eq:ph-value}; these values are deterministic quadratures rather
than fits.

The same graph links can be varied rather than merely used as a common
spectral label.  If \(K_G(\eta)\) is differentiable, then
\begin{equation}
 \frac{\diff\Gamma_2}{\diff\eta}
 =\sum_{n=0}^{2}w_n
 \Tr\!\left[
 (K_G+m_n^2\mathbf1)^{-1}\frac{\diff K_G}{\diff\eta}\right].
 \label{eq:ph-same-link}
\end{equation}
On the archived inhomogeneous \(17\)-cycle, the analytic value is
\(0.022621032016492876\).  A central difference converges with measured
exponent \(1.99996\) and reaches absolute error
\(9.77\times10^{-11}\) before floating-point saturation.  This verifies
the weight--mass--link interface that is needed when the multiplet is
placed on the record support of Sec.~\ref{sec:same-link-induced-gravity}.

\subsection{Finite geometry sums and why exact pairing is insufficient}

Because the graph operator and internal mass act on different tensor
factors, the balance is pointwise in geometry.  Let \(\Omega\) be any finite
family of positive graph kernels \(K_\omega\), of possibly different
dimensions, with a supplied positive prior \(p_\omega\).  Define
\begin{equation}
 \Gamma_\omega
 =\sum_{n=0}^{k}w_n\Tr_\omega\log(K_\omega+m_n^2\mathbf1),
 \qquad
 \mathcal Z_\Omega
 =\sum_{\omega\in\Omega}p_\omega e^{-\Gamma_\omega}.
 \label{eq:ph-graph-sum}
\end{equation}

\begin{proposition}[Geometry-by-geometry spectral balance]
\label{prop:ph-graph-sum}
For every \(\omega\) and \(s>0\),
\begin{align}
 &\sum_{n=0}^{k}w_n
 \Tr_\omega e^{-s(K_\omega+m_n^2\mathbf1)}
 \nonumber\\
 &\hspace{1.5cm}
 =\Tr_\omega e^{-sK_\omega}\,
 e^{-sx}(1-e^{-s\Delta})^k .
 \label{eq:ph-graph-heat}
\end{align}
Thus the first \(k\) internal ultraviolet moments vanish on each graph
before geometry averaging, and \(\mathcal Z_\Omega\) is finite and strictly
positive.  For every finite move
\(K(t)=K_0+t\,\delta K\geq0\),
\begin{align}
 \Gamma[K_0+\delta K]-\Gamma[K_0]
 &=
 \int_0^1\!\diff t\sum_n w_n
 \nonumber\\[-2mm]
 &\quad\times
 \Tr[(K(t)+m_n^2\mathbf1)^{-1}\delta K] .
 \label{eq:ph-finite-move}
\end{align}
\end{proposition}

\begin{proof}
The heat operator factorizes on graph and internal tensor factors, and
Theorem~\ref{thm:ph-balance} evaluates the internal supertrace.  Every
shifted determinant is strictly positive, so every term in
Eq.~\eqref{eq:ph-graph-sum} has positive real weight.  The final identity
is the fundamental theorem of calculus applied to the log determinant.
\end{proof}

This extension is nonperturbative in the amplitude of a finite graph move
and uses no cancellation between geometries.  It is not yet a dynamical
sum over causal complexes: \(\Omega\), \(p_\omega\), and the refinement maps
are inputs.  In the supporting data, 64 connected weighted graphs with
7--15 vertices give a maximum pointwise heat-factorization residual
\(1.07\times10^{-14}\); an order-one chord insertion agrees with
Eq.~\eqref{eq:ph-finite-move} to \(2.98\times10^{-15}\).

The most obvious interaction-protection proposal also has a precise
failure mode.  Exact supersymmetry can be realized in interacting lattice
models and on arbitrary triangulations
\cite{CatterallKaramov2002,BerensteinCatterall2025}, but exact all-level
pairing cancels too much.

\begin{theorem}[Exact-pairing obstruction]
\label{thm:ph-pairing-obstruction}
Let \(\mathcal V=\mathcal V_+\oplus\mathcal V_-\) be a finite positive
graded Hilbert space, let \(D=D^\dagger\) be odd, and set
\(\mathcal H_D=D^2\).  Then for every continuous spectral function \(f\),
\begin{equation}
 \Str f(\mathcal H_D)=\operatorname{ind}(D_+)\,f(0).
 \label{eq:ph-pairing-index}
\end{equation}
In particular,
\(\Str[\mathcal H_D\log(\mathcal H_D/\mu^2)]=0\), with
\(0\log0:=0\).
\end{theorem}

\begin{proof}
On each positive \(\lambda\)-eigenspace,
\(D/\sqrt{\lambda}\) is an isometric bijection between even and odd
states, so their contributions cancel.  The zero-mode difference is
\(\operatorname{ind}(D_+)\).  The logarithmic function vanishes at zero.
\end{proof}

The required interacting symmetry is therefore selective: it must protect
\(\Str\mathbf1=\Str M^2=0\) but leave positive-energy spectra sufficiently
misaligned that \(F_2>0\).  Broken-supersymmetry mass-sum rules provide
continuum examples of partial rather than exact pairing
\cite{ItoyamaMaru2015,BenaEtAl2016}; they do not yet derive the required
record-lattice Ward identity.

\subsection{Selective Ward closure on a supplied finite geometry language}
\label{sec:selective-ward-closure}

The companion selective-Ward construction
\cite{XuSelectiveLadderWard2026} closes the finite interacting part of the
preceding problem without imposing exact spectral pairing.  Let
\(\K\) be any finite positive Hilbert space, let
\(\F_k=\Lambda^\bullet\mathbb C^k\), and let \(c_a^\dagger\), \(N\), and
\(\Pi=(-1)^N\) be the canonical creation, number, and parity operators on
\(\F_k\).  All physical interactions and all degrees of freedom of one
fixed geometry are contained in an arbitrary self-adjoint operator on
\(\K\).

\begin{theorem}[Selective ladder characterization and stability]
\label{thm:selective-ladder-closure}
For a self-adjoint \(H\) on \(\K\otimes\F_k\) and a real
\(\Delta\), the identities
\begin{equation}
 [H,\mathbf1\otimes c_a^\dagger]
 =\Delta(\mathbf1\otimes c_a^\dagger),
 \qquad a=1,\ldots,k,
 \label{eq:selective-ladder-ward}
\end{equation}
hold if and only if
\begin{equation}
 H=A\otimes\mathbf1+\Delta\,\mathbf1\otimes N
 \label{eq:selective-ladder-normal}
\end{equation}
for one self-adjoint \(A\) on \(\K\).  Consequently,
\begin{equation}
 \Str e^{-sH}=\Tr_\K e^{-sA}(1-e^{-s\Delta})^k,
 \qquad
 \Str H^j=0\quad(0\leq j<k).
 \label{eq:selective-heat}
\end{equation}
For an arbitrary self-adjoint \(H\), define
\[
 \varepsilon=\max_a
 \|[H,c_a^\dagger]-\Delta c_a^\dagger\|.
\]
Its Hilbert--Schmidt projection \(H_0\) onto the affine manifold
in Eq.~\eqref{eq:selective-ladder-normal} satisfies
\begin{equation}
 \|H-H_0\|\leq4k\varepsilon.
 \label{eq:selective-distance}
\end{equation}
If \(H,H_0\geq m\mathbf1\), \(D=2^k\dim\K\), and \(s>0\), then
\begin{equation}
 \left|\Str(e^{-sH}-e^{-sH_0})\right|
 \leq4kD\,s e^{-sm}\varepsilon .
 \label{eq:selective-heat-bound}
\end{equation}
\end{theorem}

\begin{proof}
After subtracting \(\Delta N\), Eq.~\eqref{eq:selective-ladder-ward}
commutes with the irreducible CAR algebra and therefore acts trivially on
\(\F_k\), proving the characterization.  Equation
\eqref{eq:selective-heat} follows by tensor factorization.  For the
approximate statement, average \(H-\Delta N\) over the finite Majorana
conjugation group.  The average is the conditional expectation onto the
CAR commutant, and a telescoping product of at most \(2k\) Majoranas bounds
the distance by \(4k\varepsilon\).  The Duhamel formula and
\(\|e^{-tH}\|\leq e^{-tm}\) give
Eq.~\eqref{eq:selective-heat-bound}.
\end{proof}

This theorem is nonperturbative in \(A\): \(A\) may contain arbitrary
finite interacting physical and geometric terms.  For \(k=2\), the finite
response is
\begin{align}
 \I_\Delta(A)
 &:=
 \Tr\!\left[
 A\log A-2(A+\Delta)\log(A+\Delta)\right.\nonumber\\[-1mm]
 &\hspace{29mm}\left.
 +(A+2\Delta)\log(A+2\Delta)\right]\nonumber\\
 &=\int_0^\Delta\!\diff u\int_0^\Delta\!\diff v\,
 \Tr(A+u+v)^{-1}>0 .
 \label{eq:selective-invariant}
\end{align}
Thus the protected low moments coexist with a nonzero response.  A fresh-cell
unitary of the factorized form
\[
 U_h=U_{\rm phys}\otimes e^{-ih\Delta N}
\]
preserves the shifted ladder eigenoperator exactly even when
\(U_{\rm phys}\) is strong.  This realizes the Ward law microscopically, but
also exposes its present boundary: the ladder factor is installed in the
commutant rather than derived from a Lorentzian continuum theory.

The positive invariant also defines a genuine finite geometry dynamics.
For a supplied finite language \(\Omega\), positive base weights
\(p_\omega\), and positive interacting operators \(A_\omega\), set
\begin{equation}
 \pi_\omega=
 \frac{p_\omega e^{-g\I_\Delta(A_\omega)}}
 {\sum_{\nu\in\Omega}p_\nu e^{-g\I_\Delta(A_\nu)}} .
 \label{eq:selective-geometry-law}
\end{equation}
On any connected symmetric local move graph, the Metropolis transition
\[
 P_{\omega\nu}=q_{\omega\nu}
 \min\{1,\pi_\nu/\pi_\omega\}
\quad(\nu\neq\omega)
\]
has unique stationary law \(\pi\), obeys exact detailed balance, and admits
the positive-Hilbert isometry
\begin{equation}
 V|\omega\rangle
 =\sum_\nu\sqrt{P_{\omega\nu}}\,
 |\nu\rangle\otimes|e_{\omega\nu}\rangle .
 \label{eq:selective-stinespring}
\end{equation}
Extending \(V\) to a unitary and tracing a fresh environment gives a CPTP
geometry channel.  Hence the finite weights are dynamically sampled rather
than only written as a formal partition sum.

The archived calculation exhausts all 728 connected labelled simple graphs
on five vertices, using
\(A_\omega=0.8\mathbf1+L_\omega\) and \(\Delta=0.72\).
The maximum heat-factorization, detailed-balance, stationary-law, and
Stinespring-normalization residuals are respectively
\(7.00\times10^{-15}\), \(2.71\times10^{-20}\),
\(4.34\times10^{-19}\), and \(2.22\times10^{-16}\).
The exact geometry-channel gap is \(0.1474055064\).  An independent
\(2.0\times10^5\)-sample retained trajectory has total-variation distance
\(0.0248875\) and mean-edge error \(7.24\times10^{-4}\); random Ward
perturbations never exceed one quarter of the analytic distance bound.
These computations test all supplied geometries rather than a selected
saddle.

The finite closure removes the old ``partial interacting Ward protection''
and ``positive graph-weight dynamics'' gaps only at fixed regulator and for
a supplied language.  It does not select \(\Omega\), its topology, or its
move set; it does not give deletion-consistent measures as the regulator
grows; and it does not establish continuum spin--statistics, anomaly
freedom, nonlinear BV closure, or background-independent graph
backreaction.

\subsection{Projective response-weighted causal growth}
\label{sec:projective-response-growth}

The next companion construction
\cite{XuProjectiveCausalGrowth2026} removes one of those qualifications:
the positive response can generate compatible laws at every finite event
number and hence a genuine infinite-history probability measure.  The result
is exact, but it is deliberately stated on a naturally labelled causal-birth
language; quotienting by labellings and taking a Lorentzian continuum limit
remain separate tasks.

Let \(\Omega_n^{(r)}\) be the naturally labelled \(n\)-event partial orders
obtained by adjoining one maximal event at a time.  A birth selects a
down-set \(S\subseteq C\) whose set of maximal elements has size
\[
 p_C(S):=|\max_C S|\leq r
\]
and writes \(C\oplus S\) for the child.  Deleting the largest label gives a
map \(\rho_n:\Omega_{n+1}^{(r)}\to\Omega_n^{(r)}\), and every labelled child
has exactly one parent under this deletion.  On the undirected Hasse graph
of \(C\), define
\begin{equation}
 A_C=m^2\mathbf1+\kappa L_C,\qquad
 \I_\Delta(C)=\Tr f_\Delta(A_C)>0 ,
 \label{eq:pcg-response}
\end{equation}
where
\begin{align}
 f_\Delta(y)
 &=y\log y-2(y+\Delta)\log(y+\Delta)\nonumber\\
 &\quad +(y+2\Delta)\log(y+2\Delta)
 \label{eq:pcg-f}
\end{align}
and, equivalently,
\begin{equation}
 \I_\Delta(C)=
 \int_0^\Delta\!\diff u\int_0^\Delta\!\diff v\,
 \Tr(A_C+u+v)^{-1}.
 \label{eq:pcg-positive}
\end{equation}
This is the \(k=2\) selective-ladder response of
Eq.~\eqref{eq:selective-invariant}, now evaluated on each causal history.

The response of one birth is available from the parent without recomputing
the child spectrum.  Put \(B_C=A_C\oplus m^2\) and
\[
 U_S=\sqrt\kappa\,[e_a-e_{n+1}]_{a\in\max_C S}.
\]
Because a maximal birth changes only its incident Hasse links,
\(A_{C\oplus S}=B_C+U_SU_S^\dagger\).  Woodbury's identity then gives
\begin{align}
 \Delta\I_C(S)
 &:=\I_\Delta(C\oplus S)-\I_\Delta(C)\nonumber\\
 &=f_\Delta(m^2)
 -\int_0^\Delta\!\diff u\int_0^\Delta\!\diff v\,
 \mathcal W_{u+v}(C,S),
 \label{eq:pcg-increment}\\
 \mathcal W_t(C,S)
 &=\Tr\!\left[
 R_tU_S(\mathbf1+U_S^\dagger R_tU_S)^{-1}
 U_S^\dagger R_t\right],
 \qquad R_t=(B_C+t\mathbf1)^{-1}.
 \label{eq:pcg-woodbury}
\end{align}
Thus a \(p\)-parent birth requires only a \(p\times p\) inverse.  Positivity
also yields the nonperturbative parent-number bound
\begin{equation}
 f_\Delta(m^2)-\frac{2\kappa p_C(S)\Delta^2}{m^4}
 \leq\Delta\I_C(S)\leq f_\Delta(m^2).
 \label{eq:pcg-increment-bound}
\end{equation}

For real \(g,\alpha,\beta\), conditionally normalize the allowed births:
\begin{equation}
 G_{n,g}(C,S)=
 \frac{
 \exp[-g\Delta\I_C(S)-\alpha p_C(S)+\beta|S|/n]}
 {\displaystyle\sum_{S'\in\mathcal D_r(C)}
 \exp[-g\Delta\I_C(S')-\alpha p_C(S')+\beta|S'|/n]} .
 \label{eq:pcg-kernel}
\end{equation}
Starting from the one-event order, set
\begin{equation}
 \mu_{n+1,g}(C\oplus S)=\mu_{n,g}(C)G_{n,g}(C,S).
 \label{eq:pcg-recursion}
\end{equation}

\begin{theorem}[Projective response law and positive dilation]
\label{thm:pcg-projective-law}
For every finite \(n\), the laws in
Eq.~\eqref{eq:pcg-recursion} are normalized and satisfy
\begin{equation}
 (\rho_n)_*\mu_{n+1,g}=\mu_{n,g}.
 \label{eq:pcg-projectivity}
\end{equation}
They therefore define a unique probability measure on infinite
past-finite naturally labelled causal histories.  Moreover,
\begin{equation}
 V_{n,g}|C\rangle=
 \sum_{S\in\mathcal D_r(C)}
 \sqrt{G_{n,g}(C,S)}\,
 |C\oplus S\rangle\otimes|e_{C,S}\rangle
 \label{eq:pcg-isometry}
\end{equation}
is an isometry.  Extending it to a unitary and tracing a fresh environment
implements each birth by a CPTP channel.
\end{theorem}

\begin{proof}
Every child has the unique labelled parent obtained by deleting \(n+1\).
Summing Eq.~\eqref{eq:pcg-recursion} over that parent's children and using
\(\sum_SG_{n,g}(C,S)=1\) proves
Eq.~\eqref{eq:pcg-projectivity}.  Kolmogorov extension on the finite
discrete cylinder spaces gives the infinite law.  Orthogonality of the
environment labels and conditional normalization imply
\(V_{n,g}^\dagger V_{n,g}=\mathbf1\).
\end{proof}

The same conditional structure supplies exact observables along a path.
For an \(N\)-event history, let
\[
 X_n=\Delta\I_{C_n}(S_n),\qquad
 \bar X_n=\mathbb E_g[X_n\mid\mathcal F_n].
\]
Its response score is the martingale
\begin{equation}
 \mathcal S_N(g)=\partial_g\log\mathbb P_g(C_1,\ldots,C_N)
 =-\sum_{n=1}^{N-1}(X_n-\bar X_n),
 \label{eq:pcg-score}
\end{equation}
and therefore
\begin{align}
 \mathbb E_g\mathcal S_N&=0,\nonumber\\
 \mathcal F_N(g):=\mathbb E_g\mathcal S_N^2
 &=\sum_{n=1}^{N-1}
 \mathbb E_g\operatorname{Var}_g(X_n\mid\mathcal F_n),\label{eq:pcg-fisher}\\
 \partial_g\mathbb E_g O&=\mathbb E_g[O\mathcal S_N]
 \label{eq:pcg-linear-response}
\end{align}
for every bounded path observable \(O\).  Path relative entropy likewise
obeys the exact chain rule
\begin{equation}
 D(\mathbb P_g^{(N)}\Vert\mathbb P_{g'}^{(N)})
 =\sum_{n=1}^{N-1}
 \mathbb E_gD(G_{n,g}(C_n,\cdot)\Vert
 G_{n,g'}(C_n,\cdot)).
 \label{eq:pcg-kl-chain}
\end{equation}
These identities distinguish a response-generated path law from a
postselected terminal partition sum.

The archived \(r=2\) calculation enumerates every allowed state through
seven events:
\[
 1,\ 2,\ 7,\ 39,\ 322,\ 3730,\ 58113.
\]
The largest normalization, deletion-projectivity, and Stinespring residuals
are \(2.22\times10^{-16}\), \(5.55\times10^{-17}\), and
\(4.44\times10^{-16}\).  The Fisher-chain residual is zero at stored
precision, the relative-entropy-chain residual is
\(2.58\times10^{-17}\), and independent differentiation verifies
Eq.~\eqref{eq:pcg-linear-response} to \(1.84\times10^{-12}\).  Forty-seven
independent birth checks bound the Woodbury and integrated-response errors
by \(2.89\times10^{-16}\) and \(8.89\times10^{-15}\).  A
\(2.5\times10^5\)-history sample reproduces a predeclared 131-bin
trajectory coarse graining with total-variation distance \(0.00677\);
the complete \(58113\)-state histogram is retained rather than hidden by
coarse graining.  As a failure control, independently normalized terminal
Gibbs weights have deletion defect as large as \(0.08635\).

This closes A11PG only at model level.  Natural labels, the precursor
language, \(m^2,\kappa,\Delta,g,\alpha,\beta\), the Hasse response operator,
and the external birth index are inputs.  No equality of different linear
extensions, Rideout--Sorkin discrete general covariance, Bell causality,
coherent quantum measure, dimensional selection, tight Lorentzian
metric--measure continuum limit, nonlinear BV master equation, continuum
matter, graph backreaction, spin--statistics theorem, or anomaly
cancellation is derived.  The result is an exact positive projective
history interface, not a completed nonperturbative gravity measure.

\subsection{Covariant response growth and an extendible quantum measure}
\label{sec:covariant-response-growth}

The preceding conditional kernel solves deletion consistency but not the
natural-label gauge problem.  A second companion construction
\cite{XuCovariantQuantumResponseGrowth2026} identifies a response-generated
subfamily for which the stronger CSG constraints are exact.  The
construction also isolates a coherent subfamily whose cylinder amplitudes
extend countably additively.  These two branches should not be conflated:
the positive branch has a local fresh-environment CPTP dilation, whereas the
coherent branch has a strongly positive rank-one decoherence functional but
no derived local unitary transfer operator.

For each \(k\geq0\), let \(S_k=K_{1,k}\) be the \(k\)-contact star, with
\(S_0=K_1\), and use the same positive logarithmic response as in
Eq.~\eqref{eq:pcg-response}.  Define
\begin{align}
 \chi_k&=\I_\Delta(S_k)-\I_\Delta(S_0),&
 t_k(g)&=e^{-g\chi_k}.
 \label{eq:cqg-contact-couplings}
\end{align}
The star spectrum gives
\begin{equation}
 \chi_0=0,\qquad
 \chi_k=(k-1)f_\Delta(m^2+\kappa)
 +f_\Delta[m^2+\kappa(k+1)]>0
 \label{eq:cqg-star-response}
\end{equation}
for \(k\geq1\).  Thus \(t_k(g)>0\) for every \(k\).  If a precursor
down-set \(S\subseteq C\) has \(v=|S|\) elements and \(m=|\max_C S|\)
maxima, put
\begin{align}
 \lambda_g(v,m)
 &=\sum_{k=m}^{v}\binom{v-m}{k-m}t_k(g),\label{eq:cqg-lambda}\\
 q_{n,g}(C,S)
 &=\frac{\lambda_g(|S|,|\max_C S|)}{\lambda_g(n,0)}.
 \label{eq:cqg-transition}
\end{align}
This transition has a microscopic seed sampler.  Choose
\(K\subseteq C\) with weight
\begin{equation}
 \mathbb P_g(K\mid C)=\frac{t_{|K|}(g)}{\lambda_g(n,0)}
 \label{eq:cqg-seed}
\end{equation}
and place the new maximal event above \(S=\downarrow K\).  The seeds with
down-closure \(S\) are exactly
\(\max S\subseteq K\subseteq S\).  Grouping them by cardinality gives the
binomial coefficients in Eq.~\eqref{eq:cqg-lambda} and proves the Markov
identity
\begin{equation}
 \sum_{S\in\mathcal D(C)}
 \lambda_g(|S|,|\max S|)=\lambda_g(n,0).
 \label{eq:cqg-markov}
\end{equation}

\begin{theorem}[Covariant positive response growth]
\label{thm:cqg-positive-growth}
The transition law in Eq.~\eqref{eq:cqg-transition} is normalized and
deletion-projective.  Its path probability has the intrinsic form
\begin{equation}
 p_g(C)=
 \frac{\displaystyle\prod_{x\in C}
 \lambda_g(|J^-(x)|,|\max J^-(x)|)}
 {\displaystyle\prod_{j=0}^{|C|-1}\lambda_g(j,0)},
 \label{eq:cqg-intrinsic-path}
\end{equation}
and is therefore independent of the natural labelling.  For any two
precursors \(S_1,S_2\) and \(B=S_1\cup S_2\),
\begin{equation}
 q_C(S_1)q_B(S_2)=q_C(S_2)q_B(S_1).
 \label{eq:cqg-bell}
\end{equation}
Thus the law satisfies discrete general covariance and the zero-safe
cross-multiplied Bell condition.  Moreover,
\begin{equation}
 V_{n,g}|C\rangle=
 \sum_{K\subseteq C}
 \sqrt{\frac{t_{|K|}(g)}{\lambda_g(n,0)}}\,
 |C\oplus\downarrow K\rangle|C,K\rangle_E
 \label{eq:cqg-isometry}
\end{equation}
is an isometry and gives a fresh-environment CPTP birth channel.
\end{theorem}

\begin{proof}
Equation~\eqref{eq:cqg-markov} proves normalization and the unique labelled
parent proves deletion projectivity.  Multiplication along a path gives
Eq.~\eqref{eq:cqg-intrinsic-path}; its numerator is a product of intrinsic
past invariants and its denominator depends only on cardinality.  Removing
the spectators \(C\setminus B\) leaves the two precursor numerators
unchanged, while the row normalizer cancels in
Eq.~\eqref{eq:cqg-bell}.  Finally,
\(\sum_{K\subseteq C}t_{|K|}=\lambda_g(n,0)\) proves
\(V_{n,g}^\dagger V_{n,g}=\mathbf1\).
\end{proof}

The response remains operational along the covariant paths.  Define
\begin{equation}
 R_g(v,m)=-\partial_g\log\lambda_g(v,m)
 =\frac{\sum_{k=m}^{v}\binom{v-m}{k-m}t_k(g)\chi_k}
 {\lambda_g(v,m)}.
 \label{eq:cqg-response-mean}
\end{equation}
The score increment of a birth is
\[
 Y_n=R_g(n,0)-R_g(v_n,m_n).
\]
Differentiating Eq.~\eqref{eq:cqg-markov} gives
\(\mathbb E_g[Y_n\mid\mathcal F_n]=0\), so the Fisher, linear-response, and
path-relative-entropy chain rules in
Eqs.~\eqref{eq:pcg-fisher}--\eqref{eq:pcg-kl-chain} hold with \(Y_n\) in
place of the global action increment.

The coherent branch retains one response-generated contact amplitude:
\begin{equation}
 \widetilde t_0=1,\qquad
 \widetilde t_r=z_r=e^{-g\chi_r/2+i\theta\chi_r},\qquad
 \widetilde t_k=0\quad(k\neq0,r).
 \label{eq:cqg-complex-coupling}
\end{equation}
The complex version of Eqs.~\eqref{eq:cqg-lambda} and
\eqref{eq:cqg-transition} has row sum one, the same intrinsic product, and
the cross-multiplied Bell identity even when some transitions vanish.
Let
\[
 \zeta_n(C)=\sum_{S\in\mathcal D(C)}
 |\widetilde q_n(C,S)|-1,\qquad
 \zeta_n^{\max}=\max_{|C|=n}\zeta_n(C).
\]
For \(s=|z_r|\) and \(\phi=\arg z_r\), the maximum occurs at the
\(n\)-antichain \cite{SuryaZalel2020} and equals
\begin{equation}
 \zeta_n^{\max}
 =\frac{1+\binom nr s}
 {\sqrt{1+2s\binom nr\cos\phi+s^2\binom nr^2}}-1.
 \label{eq:cqg-zeta}
\end{equation}

\begin{theorem}[Extendible coherent response history]
\label{thm:cqg-complex-extension}
If \(r\geq2\), \(g>0\), and
\(\phi=\theta\chi_r\in(0,\pi/2]\), the complex cylinder measure generated by
Eq.~\eqref{eq:cqg-complex-coupling} has bounded variation and extends
uniquely to a countably additive complex measure \(\mu\) on the cylinder
sigma algebra.  The functional
\begin{equation}
 \mathfrak D(A,B)=\mu(A)\overline{\mu(B)}
 \label{eq:cqg-decoherence}
\end{equation}
is normalized, countably additive, and strongly positive; its diagonal
\(\mathfrak q(A)=|\mu(A)|^2\) obeys the grade-2 quantum-measure sum rule.
For \(r=1\) and nonreal \(z_1\), bounded variation fails.
\end{theorem}

\begin{proof}
Put \(x=s\binom nr\).  Rationalizing Eq.~\eqref{eq:cqg-zeta} gives
\[
 \zeta_n^{\max}
 =\frac{2x(1-\cos\phi)}
 {\sqrt{1+2x\cos\phi+x^2}
 [1+x+\sqrt{1+2x\cos\phi+x^2}]}
 \leq\frac{1-\cos\phi}{s\binom nr}.
 \]
The upper bound is summable for \(r\geq2\), so the finite cylinder
variations are bounded by
\(\prod_n(1+\zeta_n^{\max})<\infty\).  The complex-measure extension
criterion of Ref.~\cite{SuryaZalel2020} gives \(\mu\).  Every finite matrix
\([\mathfrak D(A_i,A_j)]\) is the outer product
\([\mu(A_i)][\mu(A_j)]^\dagger\), hence is positive semidefinite.  Ordinary
additivity of \(\mu\) yields the grade-2 sum rule
\cite{Sorkin1994QuantumMeasure}.  For nonreal \(r=1\),
\(\zeta_n^{\max}\sim(1-\cos\phi)/(sn)\), and the variation diverges.
\end{proof}

The coherent measure is not merely a complex rephrasing of a probability
law.  For \(r=2\), the two allowed three-event unlabelled endpoint classes
are the antichain \(A_3\) and the two-parent order \(V_3\), with amplitudes
\[
 \mu_3(A_3)=\frac1{1+z_2},\qquad
 \mu_3(V_3)=\frac{z_2}{1+z_2}.
 \]
Their finite-volume quotient interference is
\begin{equation}
 I_3=
 \mathfrak q(A_3\cup V_3)-\mathfrak q(A_3)-\mathfrak q(V_3)
 =\frac{2\Re z_2}{|1+z_2|^2}.
 \label{eq:cqg-interference}
\end{equation}
For the archived benchmark, \(I_3=0.3312668\).

Complete enumeration gives
\[
 1,\ 2,\ 7,\ 40,\ 357,\ 4824
\]
naturally labelled histories and
\[
 1,\ 2,\ 5,\ 16,\ 63,\ 318
\]
unlabelled classes through six events.  Markov, seed-grouping, intrinsic
path, DGC, Bell, Fisher, and relative-entropy residuals are all below
\(1.8\times10^{-15}\).  The \(r=2,3\) accumulated row excess stabilizes
with the predicted \(N^{-r}\) tail through \(N=200000\), while the
\(r=1\) control grows logarithmically.

There is also a decisive falsification.  If the full Hasse response is used
as a conditionally normalized global potential,
\[
 q^{\rm glob}_{n,g}(C,S)\propto
 e^{-g[\I_\Delta(C\oplus S)-\I_\Delta(C)]},
 \]
two natural labellings of one five-event causal set differ in log path
weight by \(1.624733\), and removing one spectator changes a competing log
transition ratio by \(0.136424\).  Positivity and projectivity therefore do
not imply DGC or Bell causality; the response must enter through a
covariant contact sequence (or another structure satisfying the same
constraints).

This closes A11CG at model level.  It removes the natural-label and Bell
defects for the response-generated CSG subfamily and supplies a
countably-additive strongly positive coherent history measure.  It does not
select the maximal-birth language, contact stars, response multiplet,
parameters, arity, or phase.  The coherent decoherence functional is rank
one and has no derived local unitary complex transfer.  No Lorentzian
continuum tightness, dimension selection, nonlinear gravitational
constraint algebra, continuum matter, anomaly freedom, or graph
backreaction is proved.

\subsection{A local-unitary non-Abelian history measure on a fixed causal
strip}
\label{sec:nonabelian-history-measure}

The rank-one boundary above can be removed without sacrificing countable
extension, but presently only after fixing a smaller causal language
\cite{XuNonAbelianQuantumHistory2026}.  Let
\[
 \Omega_2=\{00,01,10,11\}^{\mathbb N}
\]
be the record space of a ranked two-port strip.  A letter
\(x_n=(a_n,b_n)\) tells the two ports at rank \(n\) which of the two
ports at rank \(n-1\) is their parent.  Every generated edge raises rank,
so acyclicity is kinematic.  This is not an unlabeled causal-set growth law.

On a two-dimensional bond space, for any Pauli involution \(H\), define
\[
 F_H^{0,1}(\epsilon)=\frac{\mathbf 1\pm e^{i\epsilon H}}{2},
 \qquad
 A_n^{ab}=F_{H_a}^{b}(\epsilon_n)F_X^a(\epsilon_n),
 \qquad (H_0,H_1)=(Z,Y).
\]
The angle is supplied by the same positive contact response,
\[
 \epsilon_n=\epsilon_\star e^{-g\chi_n}/n^p,
 \qquad 0<\epsilon_\star<\pi.
\]
The two binary decisions have a Hadamard--controlled-unitary--Hadamard
dilation with fresh ancillas.  Coherently conditioning the second axis on
the first record gives all four blocks.

\begin{theorem}[Local-unitary non-Abelian history extension]
\label{thm:nonabelian-history-extension}
For every cell,
\begin{equation}
 \sum_{a,b}A_n^{ab}=\mathbf 1,
 \qquad
 \sum_{a,b}A_n^{ab\dagger}A_n^{ab}=\mathbf 1.
 \label{eq:nonabelian-history-two-sums}
\end{equation}
The transition algebra is noncommutative:
\begin{equation}
 [A^{00},A^{10}]=\frac{i}{2}\epsilon^2Y+O(\epsilon^3).
 \label{eq:nonabelian-history-commutator}
\end{equation}
For a cylinder word \(w=x_1\cdots x_N\), set
\[
 M([w])=A_N^{x_N}\cdots A_1^{x_1}.
\]
If \(\sum_n\epsilon_n<\infty\), this operator cylinder content extends
uniquely to a countably additive regular
\(M_2(\mathbb C)\)-valued measure on the Borel sigma algebra of
\(\Omega_2\).  For every density matrix \(\rho\),
\begin{equation}
 D(E,F)=\Tr[\rho\,M(E)^\dagger M(F)]
 \label{eq:nonabelian-history-decoherence}
\end{equation}
is normalized and strongly positive, and \(D(E,E)\) is a grade-2 quantum
measure.  Two independent strips obey exact spectator factorization.
\end{theorem}

\begin{proof}
For either binary decision,
\[
 F_H^0+F_H^1=\mathbf 1,\qquad
 F_H^{0\dagger}F_H^0+F_H^{1\dagger}F_H^1=\mathbf 1.
\]
Summing first over \(b\) and then \(a\) gives
Eq.~\eqref{eq:nonabelian-history-two-sums}; the circuit above proves local
unitary realizability.  The Pauli expansion gives
Eq.~\eqref{eq:nonabelian-history-commutator}.  Cylinder refinement is exact
because the first identity in
Eq.~\eqref{eq:nonabelian-history-two-sums} is an operator Markov sum.
Moreover,
\begin{align}
 \sum_{|w|=N}\|A_w\|
 &\le
 \prod_{n=1}^{N}
 \left(\cos\frac{\epsilon_n}{2}
      +\sin\frac{\epsilon_n}{2}\right)^2\nonumber\\
 &\le \exp\!\left(\sum_n\epsilon_n\right).
 \label{eq:nonabelian-history-variation}
\end{align}
Thus every matrix entry defines a uniformly bounded functional on the
locally constant functions of the compact product space.  Density in
\(C(\Omega_2)\) followed by Riesz--Markov gives the unique countably
additive matrix measure.  Equation
\eqref{eq:nonabelian-history-decoherence} is a Hilbert--Schmidt Gram
functional, which proves strong positivity and the grade-2 identity.
Tensor-product instruments and product bond states give spectator
factorization after summing the complete spectator outcome partition.
\end{proof}

The response sequence is summable for every \(p\ge0\) when \(g>0\),
because \(\chi_n\ge(n-1)f_\Delta(m^2+\kappa)>0\); at \(g=0\), \(p>1\)
is sufficient and necessary for this power-law tail.  At the archived
benchmark \(\epsilon=0.7\),
\[
 2\Re D([00],[10])=-\frac14\sin^3\epsilon=-0.06684023094.
\]
The history Gram rank is three after one cell and reaches the full
Hilbert--Schmidt bond rank four after two; the transition commutator is
\(0.1949275542\).  The largest Markov, Kraus, dilation, and spectator
residual is \(6.3\times10^{-16}\).  A commuting-axis control gives zero
commutator, while the \(g=0,p=1\) variation envelope grows to
\(1.7604\times10^4\) by \(N=200000\).  The latter is failure of the
sufficient extension certificate, not a proof of nonextension.

This proves A11NH at model level.  It removes the scalar/rank-one and
local-realizability restrictions for one fixed causal strip.  It does not
combine that non-Abelian transfer with the unlabeled DGC and Bell-causal
arbitrary-down-set growth of A11CG.  The strip, rank foliation, port
language, bond dimension, response sequence, and summable tail are inputs.
This restriction is substantive: natural operator orderings of quantum
Bell causality on unrestricted sequential growth can force the transition
algebra to commute, and a general causal-past-ordered representation remains
unknown \cite{SrivastavaSurya2026}.
No fluctuating generic adjacency, Lorentzian continuum tightness, nonlinear
constraint algebra, continuum matter, or graph backreaction is proved.

\subsection{A covariant non-Abelian operator measure on arbitrary causal
histories}
\label{sec:covariant-nonabelian-holonomy}

The fixed-strip restriction can be removed without giving up countable
extension if the noncommuting operator is treated as an event-local density
over the covariant positive growth law, rather than as a square transition
ratio \cite{XuCovariantNonAbelianHolonomy2026}.  Use the response-generated
CSG probability \(\mathbb P_g\) of
Theorem~\ref{thm:cqg-positive-growth}.  If an event \(x\) is born above a
precursor whose volume and number of maximal elements are
\[
 v_x=|J^-(x)|,\qquad m_x=|\max J^-(x)|,
\]
attach a fresh qubit and apply
\begin{equation}
 u_x=
 \exp\!\left[
 i\epsilon_\star e^{-h\chi_{v_x}}H_{m_x\bmod3}
 \right],
 \qquad (H_0,H_1,H_2)=(X,Y,Z).
 \label{eq:cnh-gate}
\end{equation}
Every factor acts only on the new event qubit.  The aggregate
\((v_x,m_x)\) is intrinsic to the precursor and is unchanged when spectators
outside a competing precursor union are deleted.

Let \(a=f_\Delta(m^2+\kappa)\),
\[
 r=e^{-ga},\qquad s=e^{-ha},\qquad
 \rho=\frac{1+rs}{1+r}<1.
\]
Equation~\eqref{eq:cqg-star-response} gives
\(\chi_k=(k-1)a+b_k\), with \(0<b_k\le a\).  Hence
\begin{equation}
 r^k\le t_k\le e^{ga}r^k,\qquad
 \epsilon_\star e^{-h\chi_k}
 \le \epsilon_\star e^{ha}s^k.
 \label{eq:cnh-geometric-comparison}
\end{equation}
At stage \(n\), the CSG seed cardinality has conditional distribution
\[
 \mathbb P_g(|K_n|=k\mid\mathcal F_n)
 =\binom nk t_k/\lambda_g(n,0).
\]
Since the precursor volume \(V_n=|\downarrow K_n|\) is at least
\(|K_n|\), Eq.~\eqref{eq:cnh-geometric-comparison} gives the
state-uniform tail
\begin{equation}
 \mathbb E_g(s^{V_n}\mid\mathcal F_n)
 \le e^{ga}\rho^n,\qquad
 \mathbb E_g\theta_{V_n}
 \le\epsilon_\star e^{(g+h)a}\rho^n.
 \label{eq:cnh-tail}
\end{equation}
Define
\[
 B_{\rm hol}=
 \frac{\epsilon_\star e^{(g+h)a}}{1-\rho}.
\]

\begin{theorem}[Covariant non-Abelian causal-holonomy measure]
\label{thm:covariant-nonabelian-holonomy}
Suppose \(B_{\rm hol}<1\).  On the incomplete event tensor product
\[
 \mathcal H_{\rm ev}
 =\bigotimes_{x\in\mathbb N}(\mathbb C^2,|0\rangle_x),
\]
the finite products \(U_N(\omega)=\bigotimes_{x\le N}u_x(\omega)\)
converge in operator norm for \(\mathbb P_g\)-almost every causal history to
a unitary \(U(\omega)\).  The Bochner integral
\begin{equation}
 \mathsf M(E)=\int_EU(\omega)\,\mathbb P_g(d\omega)
 \label{eq:cnh-operator-measure}
\end{equation}
is a countably additive \(B(\mathcal H_{\rm ev})\)-valued measure of total
variation at most one.  For every permutation-invariant density operator
\(\varrho\),
\begin{equation}
 D(E,F)=
 \frac{\Tr[\varrho\,\mathsf M(E)^\dagger\mathsf M(F)]}
 {\Tr[\varrho\,\mathsf M(\Omega)^\dagger\mathsf M(\Omega)]}
 \label{eq:cnh-decoherence}
\end{equation}
is normalized and strongly positive, and \(D(E,E)\) is grade 2.  Natural
relabelings act equivariantly,
\[
 U(\pi\omega)=\Pi_\pi U(\omega)\Pi_\pi^\dagger,
\]
and the range of \(\mathsf M\) is noncommutative.
\end{theorem}

\begin{proof}
Equation~\eqref{eq:cnh-tail} and Tonelli's theorem give
\[
 \mathbb E_g\sum_n\theta_{V_n}\le B_{\rm hol}<\infty,
\]
so the angle sum is finite almost surely.  Since
\(\|u_x-\mathbf1\|\le\theta_{V_x}\), telescoping makes \(U_N\)
operator-norm Cauchy.  Its finite-cylinder approximants have separable
range, which proves Bochner measurability.  A bounded Bochner density
defines an operator-norm countably additive vector measure and
\(|\mathsf M|(\Omega)\le\int\|U\|d\mathbb P_g=1\).  Furthermore,
\[
 \|\mathsf M(\Omega)-\mathbf1\|
 \le\mathbb E_g\|U-\mathbf1\|
 \le B_{\rm hol}<1,
\]
so the denominator in Eq.~\eqref{eq:cnh-decoherence} is at least
\((1-B_{\rm hol})^2\).  The Gram identity
\[
 \sum_{ij}\bar c_i c_jD(E_i,E_j)
 \propto
 \Tr\!\left[
 \varrho\left(\sum_i c_i\mathsf M(E_i)\right)^\dagger
 \left(\sum_j c_j\mathsf M(E_j)\right)
 \right]\ge0
\]
proves strong positivity, and biadditivity gives the grade-2 sum rule.
Intrinsic past data prove equivariance under natural relabeling.

For noncommutativity, fix a positive-probability \(n\)-event prefix and
compare a next birth with empty precursor to one with a singleton precursor.
Their fresh-factor gates \(u_{0,0}\) and \(u_{1,1}\) obey
\[
 c_{\rm NA}:=\|[u_{0,0},u_{1,1}]\|
 =2|\sin\theta_0\sin\theta_1|>0.
\]
After integrating unrestricted futures, both conditional tail operators
are within \(B_{\rm hol}\rho^{n+1}\) of the identity.  The corresponding
cylinder commutator is bounded below, apart from strictly positive path
probabilities, by
\begin{equation}
 c_{\rm NA}-4B_{\rm hol}\rho^{n+1}.
 \label{eq:cnh-range-bound}
\end{equation}
This is positive for all sufficiently large \(n\), so integration does not
abelianize the range.
\end{proof}

For the archived parameters
\[
 (m^2,\Delta,\kappa,g,h,\epsilon_\star)
 =(0.85,0.65,0.60,0.82,2.20,0.040),
\]
\(\rho=0.833981\) and \(B_{\rm hol}=0.446848\).
The lower bound~\eqref{eq:cnh-range-bound} becomes positive at \(n=36\).
Complete enumeration through six events gives 4824 naturally labelled
histories in 318 unlabeled classes.  The Markov, scalar Bell, and
spectator-gate residuals are below \(6.7\times10^{-16}\), while intrinsic
decorated-causet covariance is exact.  A deliberately stage-labelled gate
has covariance and spectator defects \(2.34\times10^{-2}\) and
\(1.00\times10^{-2}\).  At four events the normalized decoherence matrix
has interference \(1.10\times10^{-2}\), grade-2 residual
\(5.4\times10^{-19}\), and a nonzero operator-measure commutator.

This proves A11CH at model level.  It combines arbitrary-down-set CSG
covariance with a local-unitary non-Abelian, countably additive history
measure.  The exact spectator statement is
\begin{equation}
 \frac{q_C(S_1)}{q_C(S_2)}
 =\frac{q_B(S_1)}{q_B(S_2)},\qquad
 u_C(S_i)=u_B(S_i),
 \quad B=S_1\cup S_2.
 \label{eq:cnh-isometric-bell}
\end{equation}
It is \emph{isometric Bell locality}, not the square-operator TOBC, NTOBC,
or CPOBC cross-product condition.  The operator density is integrated
against a classical CSG law and does not obey an operator Markov sum at each
parent.  Geometry does not backreact on the gate, precursor locality is not
bounded-neighborhood locality, and no Lorentzian continuum, nonlinear
constraint algebra, matter dynamics, anomaly freedom, or graph backreaction
is proved.

\subsection{Central operator-valued growth and the no-coherence boundary}
\label{sec:central-operator-growth}

The preceding constructions leave a tempting shortcut: promote the scalar
CSG couplings to positive operators and interpret their state dependence as
quantum backreaction.  The commutative branch of that proposal can be solved
and classified completely \cite{XuOperatorValuedCausalGrowth2026}.  The
classification is useful precisely because it proves that central
operator-valued couplings do \emph{not} make the geometry marginal
coherent.

Let \(\mathcal Z\) be a commutative von Neumann algebra with separable
predual,
\[
 \mathcal Z\simeq L^\infty(Z,\nu),
 \qquad T_k=M_{t_k}\in\mathcal Z_+,\qquad T_0=\mathbf1 .
\]
For a birth above a down-set \(S\) of an \(n\)-event parent \(C\), put
\[
 v=|S|,\qquad m=|\max S|,
\]
and define
\begin{align}
 \Lambda(v,m)
 &=\sum_{k=m}^{v}\binom{v-m}{k-m}T_k,\nonumber\\
 \Lambda_n&=\Lambda(n,0),\qquad
 Q_C(S)=\Lambda(v,m)\Lambda_n^{-1}.
 \label{eq:ocg-transition}
\end{align}
The normalization is an operator identity rather than a fitted row sum.
Every seed \(K\subseteq C\) has the unique down-closure
\(\downarrow K=S\), and the seeds in that fiber are exactly
\(\max S\subseteq K\subseteq S\).  Hence
\begin{equation}
 \sum_{S\in\mathcal D(C)}
 \Lambda(|S|,|\max S|)
 =\sum_{K\subseteq C}T_{|K|}
 =\Lambda_n .
 \label{eq:ocg-seed-identity}
\end{equation}
Since \(T_0=\mathbf1\), \(\Lambda_n\geq\mathbf1\); therefore
\(Q_C(S)\geq0\) and \(\sum_SQ_C(S)=\mathbf1\).

\begin{theorem}[Central operator-valued growth and classification]
\label{thm:central-operator-growth-classification}
The effects in Eq.~\eqref{eq:ocg-transition} have the following
properties.
\begin{enumerate}
\item \(K_C(S)=Q_C(S)^{1/2}\) is a trace-preserving L\"uders instrument at
each finite parent.
\item For every finite history \(\omega\),
\begin{equation}
 K_\omega=\prod_jK_{C_j}(S_j)
 =M_{\sqrt{P_z(\omega)}} ,
 \label{eq:ocg-path}
\end{equation}
where \(P_z\) is the scalar CSG law with couplings \(t_k(z)\).  Thus
naturally labeled paths with isomorphic endpoints have identical path
operators.
\item For \(B=S_1\cup S_2\), the exact Bell cross-product is
\begin{equation}
 Q_C(S_1)Q_B(S_2)=Q_C(S_2)Q_B(S_1).
 \label{eq:ocg-bell}
\end{equation}
\item The finite effects extend to a normalized strongly countably
additive POVM on the infinite past-finite history event algebra,
\begin{equation}
 Q(E)=M_{P_z(E)}.
 \label{eq:ocg-infinite-povm}
\end{equation}
\item Conversely, every central operator-valued CSG has this direct-integral
form.  For every normal state \(\rho\),
\begin{equation}
 p_\rho(E)=\Tr[\rho Q(E)]
 =\int_ZP_z(E)\,\diff\mu_\rho(z),
 \label{eq:ocg-mixture}
\end{equation}
where \(\mu_\rho\) is fixed by the restriction of \(\rho\) to
\(\mathcal Z\).
\end{enumerate}
Consequently two states with the same central restriction give identical
geometry statistics on every history event.
\end{theorem}

\begin{proof}
Equation~\eqref{eq:ocg-seed-identity} proves Kraus completeness.  In the
multiplication representation,
\[
 Q_C(S)=M_{q_z(C,S)},\qquad
 q_z(C,S)=
 \frac{\sum_{k=m}^{v}\binom{v-m}{k-m}t_k(z)}
 {\sum_{k=0}^{n}\binom nk t_k(z)} .
\]
The scalar path product depends only on the intrinsic past
\((|J^-(x)|,|\max J^-(x)|)\) of each endpoint event, proving
Eq.~\eqref{eq:ocg-path} and DGC.  The two precursor numerators are unchanged
when spectators outside \(B\) are removed, while the parent normalizer is
common; commutativity then gives Eq.~\eqref{eq:ocg-bell}.  The scalar
kernels define probability measures \(P_z\) on infinite histories.
Measurability first holds for cylinders and extends by the monotone-class
theorem.  Scalar countable additivity and dominated convergence give strong
countable additivity of Eq.~\eqref{eq:ocg-infinite-povm}.  Finally, the
representation theorem for commutative von Neumann algebras makes every
central coupling a multiplication function, and every normal functional on
\(\mathcal Z\) is integration against \(\mu_\rho\), which proves
Eq.~\eqref{eq:ocg-mixture}.
\end{proof}

For two central sectors,
\[
 Q(E)=
 \begin{pmatrix}P_0(E)&0\\0&P_1(E)\end{pmatrix},\qquad
 \rho=
 \begin{pmatrix}w&c\\\bar c&1-w\end{pmatrix},
\]
so
\begin{equation}
 p_\rho(E)=wP_0(E)+(1-w)P_1(E)
 \label{eq:ocg-coherence-blind}
\end{equation}
is independent of \(c\).  Even when an external reference makes the
off-diagonal block operational in a larger algebra, the central geometry
POVM cannot detect it.  The geometry can be state controlled or classically
correlated with a hidden sector, but it cannot interfere between sectors.

Positivity and one-step normalization alone do not evade this conclusion.
For noncommuting \(T_k>0\), the symmetrically normalized effects
\begin{equation}
 E_C(S)=
 \Lambda_n^{-1/2}\Lambda(v,m)\Lambda_n^{-1/2}
 \label{eq:ocg-noncentral-effect}
\end{equation}
remain positive and sum to the identity, yet their square roots need not
compose covariantly.  An explicit two-dimensional control is
\begin{equation}
 T_k=e^{-\bar g\chi_k}
 \left[
 \mathbf1+\eta\cos(0.79k)X+\eta\sin(0.79k)Z
 \right],
 \qquad 0<\eta<1 ,
 \label{eq:ocg-noncentral-coupling}
\end{equation}
with the positive contact-star response \(\chi_k\) of
Eq.~\eqref{eq:cqg-star-response}.  For the archived parameters, the two
natural labelings
\[
 (0,1,0),\qquad (0,0,1)
\]
of a two-event chain plus one isolated event have sequential probabilities
\[
 0.14917364759401094,\qquad
 0.14919932629598825.
\]
Their DGC spread is \(2.5678702\times10^{-5}\), although the maximum
parentwise POVM residual is \(1.11\times10^{-15}\).  This is a failure
control, not a proposed noncentral growth law.

Non-Abelian records can nevertheless coexist with the central geometry
marginal.  For \(U_H=e^{i\epsilon H}\), set
\[
 F_H^0=\frac{\mathbf1+U_H}{2},\qquad
 F_H^1=\frac{\mathbf1-U_H}{2},\qquad
 A^{ab}=F_{H_a}^{\,b}F_X^{\,a},
 \quad(H_0,H_1)=(Z,Y).
\]
The two binary identities imply
\[
 \sum_{ab}A^{ab}=\mathbf1,\qquad
 \sum_{ab}A^{ab\dagger}A^{ab}=\mathbf1,
\]
while at \(\epsilon=0.22\),
\[
 \|[A^{00},A^{10}]\|=2.0800387\times10^{-2}.
\]
The combined fresh-event instrument
\begin{equation}
 L_C(S,a,b)=Q_C(S)^{1/2}\otimes A^{ab}
 \label{eq:ocg-combined-instrument}
\end{equation}
is trace preserving, but summing over \(a,b\) returns \(Q_C(S)\); the
record noncommutativity does not backreact on the geometry law.

Complete enumeration through six events covers 4824 naturally labeled
histories in 318 unlabeled classes.  The largest central POVM, DGC, and
Bell residuals are \(6.66\times10^{-16}\),
\(2.78\times10^{-17}\), and \(2.78\times10^{-17}\); direct-mixture and
coherence-visibility residuals vanish at stored precision.  The analytic
classification, rather than these finite residuals, is the infinite-history
result.

This proves A11OC at model level and sharpens the remaining cut.  The
infinite object established here is a POVM, not a countably additive
noncommutative instrument after unobserved future records are discarded.
No noncentral solution of causal-past-ordered Bell causality, coherent
geometry interference, geometry backreaction, tight Lorentzian continuum,
constraint algebra, anomaly theorem, or nonlinear Einstein limit is
obtained.  A viable quantum-growth replacement must therefore leave the
central algebra while preserving normalization, DGC, and operator locality,
or must use a different strongly positive history construction whose
relative phases survive both quotienting and infinite extension.

\subsection{Rigidity of the remaining square-operator Bell branch}
\label{sec:bell-causal-rigidity}

The preceding classification does not decide whether noncentral transition
operators can obey an operator Markov sum, operator discrete general
covariance (DGC), and causal-past-ordered Bell causality (CPOBC)
simultaneously.  CPOBC is the only nonsingular ordering in the current
square-operator formulation that is not already known to force immediate
commutativity \cite{SrivastavaSurya2026}.  The companion analysis
\cite{XuBellCausalGrowthRigidity2026} removes three natural neighborhoods
from that search.

Let \(a_n\) be the \(n\)-element antichain and set
\[
 Q_n=A(a_n,\varnothing).
\]
Pairing the gregarious transition with a singleton precursor and deleting
their common spectators gives the exact CPOBC identity
\begin{equation}
 A_n^{(1)}Q_1=(\mathbf1-Q_1)Q_n,
 \qquad
 A_n^{(1)}=(\mathbf1-Q_1)Q_nQ_1^{-1}.
\label{eq:bcgr-one-precursor}
\end{equation}
If the three displayed transition operators are self-adjoint, then
\begin{equation}
 Q_1\!\left(A_n^{(1)}-A_n^{(1)\dagger}\right)Q_1
 =
 [\,Q_1(\mathbf1-Q_1),Q_n\,].
\label{eq:bcgr-hermiticity}
\end{equation}
Thus self-adjointness makes \(Q_n\) commute with
\(Q_1(\mathbf1-Q_1)\).  The
identity alone has a complementary-spectrum loophole, because the map
\(f(\lambda)=\lambda(1-\lambda)\) identifies
\(\lambda\) and \(1-\lambda\).  The full CPOBC system closes that loophole.

\begin{theorem}[Self-adjoint and all-order triangular CPOBC rigidity]
\label{thm:bcgr-rigidity}
For nonsingular square-operator CPOBC:
\begin{enumerate}
\item if every transition is self-adjoint, then the complete
finite-dimensional transition algebra is commutative, with no positivity
or sign assumption and including all complementary-spectrum blocks;
\item through four events, at the scalar transitive-percolation point
\(t=1/4\), the 79 operator residuals in 30 marked-transition classes have
\begin{equation}
 \operatorname{rank}J=27,\qquad
 \operatorname{rank}[\,J\mid W_{12}\ W_{13}\ W_{23}\,]=30,
\label{eq:bcgr-ranks}
\end{equation}
so the three matrices multiplying every leading formal deformation obey
\([H_1,H_2]=[H_1,H_3]=[H_2,H_3]=0\);
\item for every pair of distinct strictly positive scalar CSG coupling
sequences \(\bm t_L,\bm t_R\in(0,\infty)^{\mathbb N}\), every
\(2\times2\) upper-triangular extension
\[
 A_e=
 \begin{pmatrix}
 a_e(\bm t_L)&x_e\\0&a_e(\bm t_R)
 \end{pmatrix},
\]
is globally similar to the diagonal direct sum; if \(t_k\) is the first
unequal coupling, the universal early obstruction has
\begin{equation}
\det\mathcal O_k=
\frac{t_2}{\Lambda_1\Lambda_2}
\prod_{j=2}^{k-1}\frac{t_1t_{j-1}}{\Lambda_j^2}>0
\qquad(k\ge3).
\label{eq:bcgr-allorder-det}
\end{equation}
Independently, the six-event \(9680\times707\) linear system has rank 706
on every distinct positive five-jet and its kernel is only
\[
 \bm x=\xi\,[\,\bm a(\bm t_L)-\bm a(\bm t_R)\,],
\]
the coboundary generated by a global similarity transform.
\item for coincident positive scalar characters, every \(2\times2\)
upper-triangular self-extension is uniquely
\begin{equation}
 A_e=a_e(\bm t)\mathbf1+
 \left(\sum_{\ell\ge1}c_\ell\partial_{t_\ell}a_e(\bm t)\right)E_{12},
\label{eq:bcgr-equal-character}
\end{equation}
where a finite transition contains only finitely many terms.  Its
transition algebra is commutative, and it splits if and only if all
\(c_\ell\) vanish.
\item if a nonsingular \(2\times2\) CPOBC representation is irreducible
and \(R_n=Q_1^{-1}Q_n\), then \(R_2\) is scalar.  A non-scalar \(R_2\)
reduces to an eigenline-swapping normal form, whose two exceptional
components are both eliminated by the following antichain stage.
\item write the surviving scalar ratio as \(R_2=h\mathbf1\).  The
two-element chain requires \(h\ne0,1\).  In the trace-balanced subchart
\(d=a,\ x=-1\) of the first-simple-spectrum chart
\begin{equation}
 Q_1=\begin{pmatrix}a&1\\q&d\end{pmatrix},\quad
 R_3=\operatorname{diag}(1,x),\quad
 R_4=\operatorname{diag}(1,y)
\label{eq:bcgr-delayed-component}
\end{equation}
the physically saturated stage-three equations have exactly
\[
C_0:y=1,\quad
E_1:a=\tfrac12,\ 4hq-3h+2=0,\quad
E_2:q=a-a^2,\ 2ah-2h+1=0.
\]
None has a nonsingular extension through the next antichain stage.  No
classification of the complementary part of the first-simple-spectrum
chart is asserted.
\end{enumerate}
\end{theorem}

\begin{proof}
Put \(F=Q_1(\mathbf1-Q_1)\).  Equation~\eqref{eq:bcgr-hermiticity}
gives \([F,Q_n]=0\).  A second exact CPOBC relation is
\begin{equation}
 [Q_1Q_2^{-1},Q_1^{-1}Q_2]=0.
\label{eq:bcgr-q1q2}
\end{equation}
The following elementary block argument removes the former sign
restriction.  Let \(A,B\) be invertible self-adjoint matrices satisfying
\[
 [A(\mathbf1-A),B]=0,\qquad [AB^{-1},A^{-1}B]=0 .
\]
On a spectral subspace of \(A(\mathbf1-A)\), the spectrum of \(A\) contains
at most the complementary pair \(q,1-q\).  In the nontrivial case write
\[
 A=\begin{pmatrix}q\mathbf1&0\\0&(1-q)\mathbf1\end{pmatrix},
 \qquad
 B=\begin{pmatrix}U&D\\D^\dagger&V\end{pmatrix}.
\]
With \(C=A^{-1}B\), the second commutator is equivalent to
\([C,ACA^{-1}]=0\).  Its upper diagonal block is
\[
 \bigl[q^{-2}-(1-q)^{-2}\bigr]DD^\dagger .
\]
Invertibility excludes \(q=0,1\); for \(q\ne1/2\) the coefficient is
nonzero and hence \(D=0\), while at \(q=1/2\) the two eigenvalues already
coincide.  Thus \([A,B]=0\).  Applying this lemma to \(Q_1,Q_2\) proves
\([Q_1,Q_2]=0\).

Set \(R=Q_2Q_1^{-1}\).  It is self-adjoint and commutes with \(Q_1\).
The three transitions from the two-element chain reduce to
\begin{equation}
 G_2=Q_2,\qquad
 V_2=(\mathbf1-Q_1)R,\qquad
 T_2=\mathbf1-R .
\label{eq:bcgr-chain}
\end{equation}
Nonsingularity of \(T_2\) gives \(1\notin\operatorname{spec}R\), while the
three-generator CPOBC relation gives \([R,Q_n]=0\).

It remains to exclude mixing of complementary \(Q_1\) eigenvalues.  The
two-element-antichain timid transition is
\[
 S=A_2^{(2)}=\mathbf1-(2\mathbf1-Q_1)R,
 \qquad
 A_n^{(2)}=S Q_nQ_2^{-1}.
\]
Resolve the commuting pair \((Q_1,R)\) into joint spectral projections
\(P_{q,r}\).  The relations \([F,Q_n]=[R,Q_n]=0\) permit an off-diagonal
block only between \((q,r)\) and \((1-q,r)\).  Hermiticity of
\(A_n^{(2)}\), after multiplication by the invertible scalar denominators,
yields
\begin{align}
 &\bigl[q\{1-(2-q)r\}
 -(1-q)\{1-(1+q)r\}\bigr]\nonumber\\
 &\hspace{23mm}\times
 P_{q,r}Q_nP_{1-q,r}=0 .
\label{eq:bcgr-separator}
\end{align}
The coefficient factors exactly as
\((2q-1)(1-r)\).  It is nonzero for \(q\ne1/2\), because \(r\ne1\);
for \(q=1/2\), \(Q_1\) is scalar on the block.  Thus
\([Q_1,Q_n]=0\) for every \(n\), and the nonsingular CPOBC
central-generator lemma propagates commutativity to every transition
\cite{SrivastavaSurya2026}.

For the second statement, exact marked-isomorphism enumeration gives
\[
 8\ {\rm Markov}+26\ {\rm DGC}
 +35\ {\rm unequal\ CPOBC}+10\ {\rm equal\ CPOBC}.
\]
At a scalar reference the linearization \(J\) acts entrywise.  Its
three-dimensional kernel supplies
\(X_e=\sum_{\alpha=1}^3(n_\alpha)_eH_\alpha\).  Symmetric quadratic terms
are tangent to the three-parameter scalar solution manifold.  Modulo
\(\operatorname{im}J\), the remaining second-order equation is
\[
 \sum_{\alpha<\beta}
 [W_{\alpha\beta}]\otimes[H_\alpha,H_\beta]=0.
\]
The augmented rank in Eq.~\eqref{eq:bcgr-ranks} makes the three wedge classes
independent, so all three commutators vanish.  A certified nonzero
\(30\times30\) minor is
\[
 -\frac{1099511627776}{582076609134674072265625}.
\]
For the triangular ansatz, all diagonal equations hold by scalar
covariance.  Through four events a scalar CSG character uses the complete
three-jet \(\bm t=(t_1,t_2,t_3)\), with \(t_0=1\):
\[
 \lambda(v,m;\bm t)
 =\sum_{k=m}^{v}\binom{v-m}{k-m}t_k,\qquad
 a(c,P;\bm t)
 =\frac{\lambda(|P|,m(P);\bm t)}
 {\lambda(|c|,0;\bm t)}.
\]
Write the first three triangular antichain generators as
\[
 Q_i=\begin{pmatrix}a_i&x_i\\0&b_i\end{pmatrix}.
\]
An exact lower CPOBC residual has upper-right entry
\begin{align}
0={}&
\frac{(a_1a_2-a_1b_2-a_2)(a_2a_3-a_2b_3-a_3)}
{a_1a_2b_1b_2}\nonumber\\
&\times\big[(a_1-b_1)x_2-(a_2-b_2)x_1\big].
\label{eq:bcgr-triangular-separator}
\end{align}
For positive scalar transitions, both numerator factors are strictly
negative:
\[
a_1(a_2-b_2)-a_2<0,\qquad
a_2(a_3-b_3)-a_3<0.
\]
Thus \([Q_1,Q_2]=0\).  If the characters first differ at \(t_1\) or
\(t_2\), one global triangular similarity diagonalizes this lower pair,
and an exact \(22\times22\) late block has determinant
\begin{align}
D_1={}&
\frac{t_{L2}t_{R1}^{5}t_{R2}^{2}
(t_{L1}-t_{R1})(t_{R1}+t_{R2})^2}
{(1+t_{L1})(1+t_{R1})^{10}
(1+2t_{L1}+t_{L2})(1+2t_{R1}+t_{R2})^{12}},
\label{eq:bcgr-minor-t1}\\
D_2={}&
-\frac{t_1^{7}t_{R2}^{2}(t_1+t_{R2})^2(t_{L2}-t_{R2})}
{(1+t_1)^{10}(1+2t_1+t_{L2})
(1+2t_1+t_{R2})^{13}},
\label{eq:bcgr-minor-t2}
\end{align}
respectively.  The first expression is used when
\(t_{L1}\ne t_{R1}\); the second when \(t_{L1}=t_{R1}=t_1\) and
\(t_{L2}\ne t_{R2}\).  All undisplayed factors are positive on the
declared domain.

When the first difference is \(t_3\), the two characters agree through
stage two and the early extension coordinates cannot be removed by
similarity.  Exact block elimination instead gives
\begin{align}
D_{3,\mathrm{late}}
&=\frac{t_1^5t_2^2(t_1+t_2)^2}
{(1+t_1)^{10}(1+2t_1+t_2)^{12}},\\
D_{3,\mathrm{Schur}}
&=-\frac{t_1^3t_2(t_{L3}-t_{R3})^2}
{(1+t_1)^2(1+2t_1+t_2)}
\nonumber\\
&\hspace{7mm}\times
\frac{1}
{(1+3t_1+3t_2+t_{L3})^2
(1+3t_1+3t_2+t_{R3})^2}.
\label{eq:bcgr-minor-t3}
\end{align}
Their product is a nonzero \(29\times29\) minor.  The three
first-difference strata exhaust all distinct positive three-jets.  Since
the global-similarity coboundary is always present, the rank is exactly 29
and the kernel is exactly its span.

At five events there are 131 marked transitions and 820 residual
equations.  Suppose first that the characters differ among
\(t_1,t_2,t_3\).  The four-event restriction just proved is a global
coboundary.  Subtract it from all stages; the first 30 off-diagonal
coordinates vanish.  The remaining \(820\times101\) late matrix has exact
\(101\times101\) minors
\begin{align}
E_1={}&
\frac{t_{L2}t_{R1}^{37}t_{R2}^{13}t_{R3}^{3}
(t_{L1}-t_{R1})(t_{R1}+t_{R2})^{13}}
{(1+t_{L1})(1+t_{R1})^{44}
(1+2t_{L1}+t_{L2})(1+2t_{R1}+t_{R2})^{56}}
\nonumber\\[-1mm]
&\times
\frac{(t_{R2}+t_{R3})^8(t_{R1}+2t_{R2}+t_{R3})^6}
{(1+3t_{R1}+3t_{R2}+t_{R3})^{60}},
\label{eq:bcgr-five-t1}\\
E_2={}&-
\frac{t_1^{39}t_{R2}^{13}t_{R3}^{3}
(t_{L2}-t_{R2})(t_1+t_{R2})^{13}}
{(1+t_1)^{44}(1+2t_1+t_{L2})
(1+2t_1+t_{R2})^{57}}
\nonumber\\[-1mm]
&\times
\frac{(t_{R2}+t_{R3})^8(t_1+2t_{R2}+t_{R3})^6}
{(1+3t_1+3t_{R2}+t_{R3})^{60}},
\label{eq:bcgr-five-t2}\\
E_3={}&-
\frac{t_1^{38}t_2^{14}t_{R3}^{3}
(t_1+t_2)^{13}(t_2+t_{R3})^8}
{(1+t_1)^{44}(1+2t_1+t_2)^{56}}
\nonumber\\[-1mm]
&\times
\frac{(t_{L3}-t_{R3})(t_1+2t_2+t_{R3})^6}
{(1+3t_1+3t_2+t_{L3})
(1+3t_1+3t_2+t_{R3})^{61}} .
\label{eq:bcgr-five-t3}
\end{align}
The first formula applies when \(t_1\) first differs, the second when
\(t_1\) agrees and \(t_2\) first differs, and the third when \(t_1,t_2\)
agree and \(t_3\) first differs.  Every factor other than the declared
coupling difference is strictly positive.  These minors are independent of
\(t_{L4},t_{R4}\), so no stage-four coupling can restore a late extension.

It remains to treat a first difference in \(t_4\).  Put
\[
\Lambda_3=1+3t_1+3t_2+t_3,\qquad
\Lambda_{4\alpha}=1+4t_1+6t_2+4t_3+t_{\alpha4}.
\]
Omit one nonzero stage-four coboundary column.  A 100-dimensional late
block has determinant
\[
D_{4,\mathrm{late}}=
\frac{t_1^{37}t_2^{13}t_3^3(t_1+t_2)^{13}
(t_2+t_3)^8(t_1+2t_2+t_3)^6}
{(1+t_1)^{44}(1+2t_1+t_2)^{56}\Lambda_3^{60}}.
\]
Solving this block exactly and substituting into 30 remaining rows gives a
Schur determinant
\[
D_{4,\mathrm{Schur}}=
\frac{t_1^9t_2^4(t_1+t_2)^2(t_{L4}-t_{R4})^3}
{(1+t_1)^{12}(1+2t_1+t_2)^{13}
\Lambda_{4L}^3\Lambda_{4R}^3}.
\]
The first block is removed by 100 degree-one row pivots; 23 further leaf
pivots reduce the Schur block to a \(7\times7\) core.  Their product is a
nonzero \(130\times130\) minor,
\begin{equation}
\begin{split}
D_4={}&
\frac{t_1^{46}t_2^{17}t_3^3(t_1+t_2)^{15}
(t_2+t_3)^8(t_1+2t_2+t_3)^6}
{(1+t_1)^{56}(1+2t_1+t_2)^{69}\Lambda_3^{60}}\\
&\times
\frac{(t_{L4}-t_{R4})^3}
{\Lambda_{4L}^3\Lambda_{4R}^3}.
\end{split}
\label{eq:bcgr-five-t4}
\end{equation}
Thus every first-difference stratum has rank at least 130.  The global
similarity coboundary makes the rank at most 130, proving the third
statement through five events.

At six events the catalog contains
\[
2+6+22+101+576=707
\]
marked transitions and 9680 residual equations: 87 operator Markov sums,
4826 covariance equations, 3483 unequal-size CPOBC equations, and 1284
equal-size CPOBC equations.  If the first unequal coupling is
\(t_j\), \(j\le4\), the five-event restriction is already split.
Subtracting its global-similarity coboundary leaves the 576 stage-five
coordinates.  Exact \(576\times576\) late minors close all four strata.
The \(t_2\) minor is completely row-leaf eliminable; the \(t_1,t_3,t_4\)
minors have 573 degree-one column pivots and a \(3\times3\) core.  For
example, the \(t_3\) core is
\[
\begin{pmatrix}
-t_1/\Lambda_1&1/\Lambda_1&0\\
-t_2/\Lambda_2&0&1/\Lambda_2\\
0&-t_2/\Lambda_{3L}&t_1/\Lambda_{3R}
\end{pmatrix},
\]
where
\[
\Lambda_1=1+t_1,\quad
\Lambda_2=1+2t_1+t_2,\quad
\Lambda_{3\alpha}=1+3t_1+3t_2+t_{\alpha3}.
\]
Its determinant is proportional to \(t_{L3}-t_{R3}\); the \(t_1\) and
\(t_4\) cores have the same staircase form.  Multiplication by the exact
leaf factors gives
\[
F_j=(t_{Lj}-t_{Rj})\,P_j/Q_j,\qquad j=1,\ldots,4,
\]
with \(P_j,Q_j\) products of strictly positive couplings, consecutive
coupling sums, and CSG normalization polynomials.  Their complete
factorizations and row lists are retained in the supporting material
\cite{XuBellCausalGrowthRigidity2026}.

For the only remaining stratum, put
\[
t_{L1}=t_{R1}=t_1,\ldots,t_{L4}=t_{R4}=t_4,
\qquad t_{L5}\ne t_{R5},
\]
and define
\begin{align*}
\Lambda_1&=1+t_1,&\Lambda_2&=1+2t_1+t_2,\\
\Lambda_3&=1+3t_1+3t_2+t_3,\\
\Lambda_4&=1+4t_1+6t_2+4t_3+t_4,\\
\Lambda_{5\alpha}&=1+5t_1+10t_2+10t_3+5t_4+t_{\alpha5}.
\end{align*}
After one nonzero coboundary coordinate is omitted, 575 row leaves
eliminate the remaining late columns.  Four obstruction rows close the
four early tangent directions.  After scaling out
\((t_{L5}-t_{R5})/(\Lambda_{5L}\Lambda_{5R})\), their determinant is
\[
\frac{t_1^4t_2^2t_3}
{\Lambda_1\Lambda_2^3\Lambda_3^2\Lambda_4^2}>0.
\]
Restoring the four scaled factors and the late determinant gives the
nonzero \(706\times706\) minor
\begin{equation}
\begin{split}
D_5={}&-\frac{t_1^{304}t_2^{105}t_3^{26}t_4^4
(t_1+t_2)^{98}(t_2+t_3)^{71}(t_3+t_4)^{19}}
{\Lambda_1^{269}\Lambda_2^{309}\Lambda_3^{355}\Lambda_4^{361}}\\
&\times
\frac{(t_1+2t_2+t_3)^{49}(t_2+2t_3+t_4)^{40}}
{\Lambda_{5L}^4\Lambda_{5R}^4}\\
&\times
(t_1+3t_2+3t_3+t_4)^{22}(t_{L5}-t_{R5})^4 .
\end{split}
\label{eq:bcgr-six-t5}
\end{equation}
Every nondifference factor is positive.  Thus each of the five
first-difference strata has rank at least 706.  The coboundary supplies a
null vector, so the rank is exactly 706 and its span is the full kernel.

The six-event result is the first complete four-tangent audit of an
all-order antichain obstruction.  To derive the latter, let \(t_k\) be the
least unequal coupling and subtract the global-similarity coboundary so that
the upper-right entry of \(Q_k\) vanishes.  Below stage \(k\), the two
diagonal characters coincide and every triangular transition belongs to
the dual-number algebra
\(\mathbb C[\varepsilon]/(\varepsilon^2)\).  The nonsingular scalar CSG
classification uses only sums, products, and inverses of gregarious
transitions, so it base-changes to this algebra.  Consequently every early
coordinate is uniquely
\[
y_e=\sum_{\ell=1}^{k-1}c_\ell\partial_{t_\ell}a_e(\bm t).
\]

For \(k\ge3\), set
\(\Lambda_n=1+\sum_{j=1}^n\binom njt_j\).  The antichain pivots
\[
A_k^{(r)}A_r^{(0)}=A_r^{(r)}A_k^{(0)}
\]
solve the needed stage-\(k\) coordinates.  Substitute them into
\[
[A_k^{(1)},A_2^{(1)}]=0,\qquad
A_k^{(2)}A_2^{(1)}=A_2^{(2)}A_k^{(1)},
\]
and
\[
A_k^{(r)}A_{r+1}^{(1)}
=A_{r+1}^{(r)}A_k^{(1)},\qquad2\le r\le k-2.
\]
After scaling by
\(\Lambda_{kL}\Lambda_{kR}/(t_{Lk}-t_{Rk})\), one row interchange makes
the resulting \((k-1)\)-square matrix lower triangular, with diagonal
\[
-\frac{t_2}{\Lambda_1\Lambda_2},\quad
\frac{t_1^2}{\Lambda_2^2},\quad
\frac{t_1t_2}{\Lambda_3^2},\ldots,
\frac{t_1t_{k-2}}{\Lambda_{k-1}^2}.
\]
This proves Eq.~\eqref{eq:bcgr-allorder-det} and removes all early
tangents.  The cases \(k=1,2\) follow from
Eq.~\eqref{eq:bcgr-triangular-separator}.  The gregarious identity
\[
Q_nQ_1^{-1}Q_k=Q_kQ_1^{-1}Q_n\qquad(1<k<n)
\]
then diagonalizes every later antichain generator for \(k\ge2\); the
\(k=1\) case uses the complementary pair
\([Q_n,Q_2Q_1^{-1}]=0\) and
\([A_n^{(1)},A_2^{(1)}]=0\).  Atomization, Bell causality against the
gregarious alternative, and MSR propagate diagonality to all transitions.
Restoring the coboundary proves the all-order statement.

The symbolic archive derives the complete \(\mathcal O_k\) through
\(k=12\) and reproduces the six-event \(\mathcal O_5\) entry by entry.  A
second implementation using only exact rational arithmetic verifies
normalization cancellation, triangularity, Eq.~\eqref{eq:bcgr-allorder-det},
and positivity through \(k=40\).  These checks audit the formula; the
all-\(k\) conclusion follows from the displayed triangular factorization.

For the coincident-character branch, identify the upper-triangular algebra
with the dual numbers
\(\mathbb D=\mathbb C[\varepsilon]/(\varepsilon^2)\) through
\(\varepsilon\mapsto E_{12}\).  Strict positivity makes every denominator
in the scalar CSG induction a unit of \(\mathbb D\).  Because that induction
uses only sums, products, and inverses of gregarious transitions, it
base-changes to \(\mathbb D\): through parent size \(n\), the unique dual
couplings are
\[
 t_\ell+\varepsilon c_\ell,\qquad1\le\ell\le n.
\]
Restriction from \(n+1\) to \(n\) preserves the coefficients, giving one
infinite sequence \((c_\ell)_{\ell\ge1}\), and differentiation gives
Eq.~\eqref{eq:bcgr-equal-character}.  The image of \(\mathbb D\) is
commutative.  If \(c_k\) is the first nonzero coefficient, the nilpotent
entry of the gregarious transition from the \(k\)-antichain is
\(-c_k/\Lambda_k^2\ne0\), so the Jordan extension is nonsplit; the converse
is immediate.  An independent exact-rational implementation evaluates
47\,904 dual-number transition types through parent size 64.  The complete
four- and five-event systems have nullities three and four, exactly their
coupling-tangent dimensions, and all five tangent columns annihilate the
9680 six-event residuals.

For the irreducible two-dimensional branch, the gregarious identity makes
all \(R_n=Q_1^{-1}Q_n\) commute.  If \(R_2\) is a non-scalar Jordan block,
\([Q_1Q_2^{-1},Q_1^{-1}Q_2]=0\) makes \(Q_1\) preserve its unique
eigenline.  If \(R_2\) has simple spectrum, irreducibility leaves only
\[
 Q_1=\begin{pmatrix}0&1\\p&0\end{pmatrix},\qquad
 R_2=\operatorname{diag}(r,s),
\quad prs(r-s)(p-1)\ne0 .
\]
Writing \(R_3=\operatorname{diag}(u,v)\), two rows of the next CPOBC
relation have determinant
\begin{equation}
 -\frac{p^2(r-s)^2(r+s)[p(r^2+s^2)-rs]}
 {rs^2(p-1)}.
\label{eq:bcgr-early-irred}
\end{equation}
The first exceptional component \(s=-r\) forces \(u=v\).  At the following
stage the required two-precursor transition has determinant
\(-u^2D/r^2\), \(D=9pr^2-4r^2+1\), while two response rows have determinant
\(48pu^2/D\); nonsingularity therefore makes the latter full rank and
forces the next generator to vanish.  The other component has
\(p=rs/(r^2+s^2)\), \(v=-su/r\); away from the overlap \(r+s=0\), the
following two-row determinant is
\[
 \frac{2u^2(r-s)(r+s)^2}{r^3}\ne0.
\]
It gives the same contradiction.  A symbolic certificate reconstructs
these factorizations exactly.  This excludes non-scalar onset at \(R_2\),
not a delayed-onset branch with scalar \(R_2\).

For the surviving scalar ratio, \(Q_2=hQ_1\).  The exact transitions from
the two-element chain are
\begin{equation}
 G_2=hQ_1,\qquad V_2=h(\mathbf1-Q_1),\qquad
 T_2=(1-h)\mathbf1 .
\label{eq:bcgr-delayed-chain}
\end{equation}
Thus nonsingularity gives \(h\ne0,1\), eliminating every antichain
candidate on the \(h=1\) branch.  We next restrict
Eq.~\eqref{eq:bcgr-delayed-component} to the trace-balanced slice
\(d=a,\ x=-1\).  Direct substitution into all twelve
determinant-cleared stage-three residuals gives the common factor
\((y-1)(q-a^2)\).  On \(y\ne1\), nonsingularity permits removal of
\(q-a^2\), and two quotient rows give
\[
(2a-1)(a^2-a+q)^2=0,\qquad
(2a-1)^2(a^2-a+q)=0.
\]
The remaining rows force exactly \(E_1\) or \(E_2\) as stated in the
theorem; \(y=1\) is \(C_0\).  Direct factorization verifies sufficiency, so
these components exhaust the trace-balanced subchart over
\(\mathbb Q\).  The attempted exact elimination of \(d-a\) and \(x+1\)
from the full saturated first-simple-spectrum ideal did not produce a
finite characteristic-zero certificate; the complementary part of that
chart therefore remains open.

On \(C_0\), a putative next relative generator commutes with \(R_3\); write
\(R_5=\operatorname{diag}(1,z)\).  For three stage-four residuals
\(f_1,f_2,f_3\), with
\begin{equation}
 \Delta_3=-a^2h^2+6ah^2-6ah+h^2q-8h^2+18h-9 ,
\end{equation}
exact characteristic-zero reduction gives
\begin{equation}
 \operatorname{NF}_{I}\!\left[z(q-a^2)\Delta_3\right]\ne0,\qquad
 \operatorname{NF}_{I}\!\left([z(q-a^2)\Delta_3]^2\right)=0,
 \quad I=\langle f_1,f_2,f_3\rangle .
\label{eq:bcgr-delayed-radical}
\end{equation}
Hence every common zero has
\(z(q-a^2)\Delta_3=0\).  Yet \(z\ne0\),
\(q-a^2=-\det Q_1\ne0\), and
\(\Delta_3/h^2=\det A_3^{(3)}\ne0\) on the nonsingular branch, a
contradiction.  SymPy and Singular independently reproduce both normal
forms; saturation by all transition determinants gives the unit ideal.

On \(E_1,E_2\), a next-stage residual first forces \(y=\pm1\); \(y=1\)
returns to \(C_0\).  At \(y=-1\), two residuals force
\[
L_1=(2h-3)z-4h+3=0,\qquad
L_2=(3a-2)z-3a+1=0,
\]
respectively.  Their resultants with a third residual core are
\[
-(2h-3)(9h^2-29h+18),\qquad
(3a-2)(26a^2-31a+8).
\]
The linear leading factors cannot vanish with \(L_1,L_2\), while the
quadratics are required stage-three transition determinants.  Both
components are impossible.  Independent full saturated stage-four ideals
are the unit ideal over \(\mathbb Q\).

Transitive percolation is the subfamily \(t_k=t^k\).  With
\(p=t/(1+t)\), one compact corollary is
\[
 p_L^3p_R^{13}(p_R-p_L)^2(1-p_R)^{22}
 =
 \frac{t_L^3t_R^{13}(t_R-t_L)^2}
 {(1+t_L)^5(1+t_R)^{37}} .
\]
At the archived rational check pair
\((t_L,t_R)=(1/4,1/2)\), this minor is
\[
 \frac{16777216}{1407137205909366759375}.
\]
\end{proof}

Equation~\eqref{eq:bcgr-hermiticity} alone cannot prove the first part.  For
\(Q_1=\operatorname{diag}(\lambda,1-\lambda)\), one positive off-diagonal
\(Q_n\) can make Eq.~\eqref{eq:bcgr-one-precursor} positive while
\([Q_1,Q_n]\ne0\).  This local example is a failure control for the
one-precursor argument, not a complete CPOBC growth law: the second
antichain-generator identity, nonsingularity of the chain timid transition,
and the two-precursor self-adjointness condition supply the missing
separators.  The second part excludes a noncommuting
\emph{leading} tangent through second order; it does not exclude
higher-order contact or a disconnected irreducible component in the
non-self-adjoint or singular branches.  The third part applies to every
pair of distinct strictly positive scalar CSG coupling sequences,
irrespective of the stage at which they first differ.  It does not cover
zero/singular coupling boundaries, irreducible nontriangular
representations, or higher-dimensional noncommutative nilpotent structure.
The fourth part closes the positive equal-character \(2\times2\)
upper-triangular branch.  The fifth closes non-scalar onset in an
irreducible nonsingular \(2\times2\) branch.  The sixth removes the scalar
value \(h=1\) and exhausts the trace-balanced \(d=a,\ x=-1\) subchart at
\(R_3\), but does not classify the rest of the first-simple-spectrum chart,
Jordan onset, or first onset after \(R_3\).

This proves A11BR as a boundary theorem.  It narrows the square-operator
search to non-self-adjoint nonsingular transitions, zero/singular
couplings, state-only covariance, higher-dimensional noncommutative
nilpotent thickenings, the unclassified part of the normalized
first-simple-spectrum chart, and the remaining delayed-onset irreducible
nontriangular components with scalar \(R_2=h\mathbf1\), \(h\ne0,1\),
including Jordan onset and later first onset.
It neither constructs nor excludes such a representation.  Nor does it
provide an infinite noncommutative CPOBC measure with consistent cylinder
marginals, a Lorentzian continuum, or geometry backreaction.

\subsection{Falsifying deformations and the remaining assumption}

Three controls isolate the mechanism.  With only one internal mode,
\(\Str M^2=-\Delta\), so the logarithmic cutoff dependence remains.
Erasing parity replaces the supertrace by a positive trace with
\(\Tr\mathbf1=4\), restoring the quadratic divergence.  Most importantly,
the lowest number-conserving nonadditive deformation,
\begin{equation}
 M_\eta^2=x+\Delta N+\eta N(N-1),
 \label{eq:ph-deformation}
\end{equation}
gives
\begin{equation}
 \Str M_\eta^2=2\eta .
 \label{eq:ph-deformation-moment}
\end{equation}
The free result is therefore structural but not generically
interaction-stable.  A partial Ward identity or renormalization-group
invariant manifold must forbid or compensate the \(N(N-1)\) term without
enforcing the complete pairing ruled out by
Theorem~\ref{thm:ph-pairing-obstruction}.

This distinction separates eight statements.  A11BM is proved by the
positive finite-graph multiplet.  A11BW is proved by
Theorem~\ref{thm:selective-ladder-closure} and the finite geometry channel:
the required moments are selectively protected for arbitrary finite
positive interacting \(A_\omega\), with a quantitative defect bound, and
the resulting weights have a local positive dynamics.  A11PG is proved by
Theorem~\ref{thm:pcg-projective-law}: on a supplied naturally labelled
causal-birth language, those weights generate compatible finite cylinder
laws and an infinite positive history measure.  A11CG is proved by
Theorems~\ref{thm:cqg-positive-growth}
and~\ref{thm:cqg-complex-extension}: the response-generated contact
sequence is DGC and Bell causal, and one \(r\geq2\) coherent subfamily has a
countably additive strongly positive rank-one quantum measure.  A11NH is
proved by Theorem~\ref{thm:nonabelian-history-extension}: on a supplied
two-port strip, a fresh-ancilla local unitary generates an extendible
rank-four non-Abelian quantum measure.  A11CH is proved by
Theorem~\ref{thm:covariant-nonabelian-holonomy}: on arbitrary CSG causal
histories, intrinsic fresh-event gates define a covariant countably
additive non-Abelian operator measure and strongly positive decoherence
functional.  A11OC is proved by
Theorem~\ref{thm:central-operator-growth-classification}: central
operator-valued CSG has an infinite POVM and exact covariance, but is
classified as a coherence-blind direct integral of scalar growth laws; the
positive noncentral control proves that one-step normalization alone does
not preserve DGC.  A11BR is proved by
Theorem~\ref{thm:bcgr-rigidity}: every finite-dimensional self-adjoint
nonsingular CPOBC system classicalizes, without a sign assumption and
including complementary-spectrum blocks; independently, the tested scalar
locus has no noncommuting leading tangent, and every \(2\times2\)
upper-triangular extension between distinct strictly positive scalar CSG
coupling sequences splits, at arbitrary first-difference order, while every
coincident-character \(2\times2\) self-extension is a commuting coupling
tangent; an irreducible nonsingular \(2\times2\) branch must have scalar
\(R_2=h\mathbf1\), with \(h\ne0,1\), and the trace-balanced
\(d=a,\ x=-1\) subchart of the normalized first-simple-spectrum chart at
\(R_3\) is exhausted and obstructed at the next antichain stage.
These are exact exclusions, not an existence theorem or a no-go
theorem for non-self-adjoint, singular, higher-dimensional noncommutative
nilpotent, the remainder of the first-simple-spectrum chart, the remaining
delayed-onset irreducible nontriangular, or state-only formulations.
A11B remains open and is stronger:
the \emph{same interacting Lorentzian
dynamics} must generate the relational records, growth language, topology,
contact response, and arity, replace the external classical CSG density by
a genuinely quantum backreacting geometry law satisfying the desired
operator locality/constraint conditions; derive rather than declare the
protected commutant; produce a
tight Lorentzian continuum measure; satisfy spin--statistics and anomaly
freedom; and couple consistently to nonlinear gauge constraints and graph
backreaction.

\subsection{Operational completion criterion for nonperturbative gravity}
\label{sec:nonperturbative-completion-criterion}

To prevent a growing list of model closures from being mistaken for a
nonperturbative definition, we record the minimum completion target used by
the present program.  A refining family \(\{\mathfrak Q_N\}\) will count as
a nonperturbative RQCP-QG completion only if one construction supplies all
of the following, with compatible orders of limits:
\begin{enumerate}[label=(NP\arabic*)]
\item a normalized positive microscopic process or state on each finite
  regulator, with relational rather than external time;
\item autonomous fluctuating adjacency, causal order, additive physical
  volume, and a tight nondegenerate probability law on the resulting
  geometries;
\item a projectively consistent continuum net of local observables with a
  positive Lorentzian reconstruction or an equivalent constructive
  continuation;
\item anomaly-free nonlinear diffeomorphism constraints, expressed by a
  closed BV/BRST master equation or an equivalent exact relational Ward
  algebra;
\item regulator-independent relational observables, including matter
  correlators and the Newton response, with graph topology changes,
  backreaction, and caustic sectors controlled rather than excluded;
\item a semiclassical regime in which the same observables recover
  nonlinear Einstein--matter dynamics with an error that vanishes
  uniformly on a stated physical domain.
\end{enumerate}
Every item is a simultaneous requirement, not a menu.  The present series
has explicit finite or quadratic model realizations of parts of (NP1),
(NP2), (NP4), and (NP6), and the regulator theorem plus the present
multiplet and selective-Ward channel close one finite interacting,
fixed-language contribution to (NP5).  The projective-growth channel closes
the positive finite-to-infinite probability part of (NP2), and the
covariant-growth companion removes the natural-label/Bell defect and
supplies one extendible rank-one coherent measure.  The fixed-strip
companion separately supplies an extendible local-unitary non-Abelian
measure.  The causal-holonomy companion combines arbitrary-causet CSG
covariance with an intrinsic noncommuting operator density and proves a
countably additive measure whose integrated range remains non-Abelian.
The central operator-growth classification supplies an infinite geometry
POVM but proves that its probabilities are only direct mixtures; its
positive noncentral control fails DGC.  The CPOBC rigidity theorem further
excludes the complete finite-dimensional self-adjoint nonsingular branch
without a sign assumption, every noncommuting leading tangent at one exact
four-event scalar point, and every nonsplit \(2\times2\)
upper-triangular extension between distinct strictly positive scalar CSG
coupling sequences; it also classifies every coincident-character
\(2\times2\) self-extension as a commuting coupling tangent.  It leaves
non-self-adjoint nonsingular, zero/singular-coupling, state-only,
higher-dimensional noncommutative nilpotent, delayed-onset irreducible
nontriangular components with scalar
\(R_2=h\mathbf1\), \(h\ne0,1\), Jordan/later onset, and infinite branches
open.
These results do not autonomously
select the language, contact sequence, or arity, produce a genuinely
noncentral operator-Markov/Bell quantum-backreacting law, or provide
continuum tightness.
The series has no single
construction satisfying (NP1)--(NP6), no projectively tight interacting
Lorentzian continuum measure, and no nonlinear anomaly-free continuum
constraint algebra.  Accordingly these bridges shorten the open cut but do
not constitute nonperturbative quantum gravity.

\appendix
\clearpage
\section{Finite-dimensional RQCP and projection logic}\label{app:finite}

This appendix records the finite-dimensional algebraic lemmas that motivate concluding exact events must correspond to center projections. Let $\Hh$ be a finite-dimensional Hilbert space and let $\A=B(\Hh)$. A sharp yes/no experimental proposition is represented by an orthogonal projection $P$ satisfying $P=P^*=P^2$. In a commutative projection algebra, logical conjunction corresponds naturally to operator multiplication. However, in a noncommutative algebra, this association breaks down.

\begin{lemma}[Product of projections]\label{lem:projection-product}
Let $P,Q\in B(\Hh)$ be orthogonal projections. Then $PQ$ is an orthogonal projection if and only if $PQ=QP$.
\end{lemma}

\begin{proof}
If $PQ$ is an orthogonal projection, it must be self-adjoint. This directly yields
\[
PQ=(PQ)^*=Q^*P^*=QP.
\]
Conversely, if $P$ and $Q$ commute ($PQ=QP$), their product is simultaneously self-adjoint:
\[
(PQ)^*=Q^*P^*=QP=PQ,
\]
and idempotent:
\[
(PQ)^2=PQPQ=P^2Q^2=PQ.
\]
This confirms $PQ$ is indeed an orthogonal projection.
\end{proof}

\begin{corollary}[No Boolean conjunction for noncommuting records]
If $[P,Q]\ne0$, then the formal product $PQ$ cannot represent the Boolean meet of $P$ and $Q$.
\end{corollary}

\begin{proof}
If $PQ$ were the Boolean meet within a projection algebra, it would necessarily be a projection itself. By Lemma~\ref{lem:projection-product}, this would strictly require $P$ and $Q$ to commute ($PQ=QP$), establishing a contradiction.
\end{proof}

The RQCP event definition is therefore deliberately conservative: a fact is not merely an arbitrary projection in the full local algebra. Instead, it must be a projection belonging to the center of the influence algebra.

\begin{proposition}[Center projections form a Boolean algebra]\label{prop:centerboolean}
Let $M$ be a von Neumann algebra. Then $\Proj(Z(M))$ is a complete Boolean algebra equipped with the standard lattice operations:
\[
P\wedge Q=PQ,\qquad P\vee Q=P+Q-PQ,\qquad \neg P=1-P.
\]
\end{proposition}

\begin{proof}
By definition, the center $Z(M)$ is a commutative von Neumann algebra; thus, all its internal projections naturally commute. Consequently, the specified operations correspond to the standard lattice operations for commuting projections, natively satisfying all required Boolean identities. Furthermore, arbitrary suprema of projections within any von Neumann algebra are guaranteed to exist in the strong operator topology. Because the center $Z(M)$ itself is strongly closed, these suprema inherently remain central, thereby establishing completeness.
\end{proof}

\subsection*{Finite Wigner-friend check}
Let $\Hh_L=\operatorname{span}\{|a\rangle,|b\rangle\}$ and define
\[
P_a=|a\rangle\langle a|,
\qquad
Q_+=|+\rangle\langle+|,
\qquad
|+\rangle=\frac{|a\rangle+|b\rangle}{\sqrt 2}.
\]
In the ordered basis $(|a\rangle,|b\rangle)$,
\[
P_a=\begin{pmatrix}1&0\\0&0\end{pmatrix},
\qquad
Q_+=\frac12\begin{pmatrix}1&1\\1&1\end{pmatrix}.
\]
Then
\[
[P_a,Q_+]=\frac12\begin{pmatrix}0&1\\-1&0\end{pmatrix}\ne0.
\]
Thus the friend's record proposition and the Wigner coherence proposition cannot be placed in a single exact Boolean algebra.  In RQCP language they may both be meaningful as contextual events, but they are not one global Boolean sample space.

\section{Influence algebras and stabilizers}\label{app:influence}

Let $\D$ be a family of normal linear functionals on a target von Neumann algebra $M$.  In RQCP, $\D$ is the family of response differences
\[
\delta_Y^{C;\beta}(\Phi_X,\Phi'_X)=r_Y^{C;\beta}(\Phi_X)-r_Y^{C;\beta}(\Phi'_X).
\]
Define
\[
\Stab(\D)=\{U\in\mathcal U(M):\delta(U^*AU)=\delta(A),\ \forall A\in M,\ \forall\delta\in\D\}.
\]
The influence algebra is
\[
\I(\D)=\Stab(\D)'\cap M.
\]

\begin{proposition}[Finite-dimensional density form]\label{prop:stabfinite}
Let $M=M_d(\mathbb C)$ and suppose $\delta_k(A)=\Tr(\Delta_k A)$ for matrices $\Delta_k$.  Let $S=\{\Delta_k,\Delta_k^*:k\}$.  Then
\[
\I(\D)=\vN(S).
\]
\end{proposition}

\begin{proof}
The stabilizer condition is
\[
\Tr(\Delta_kU^*AU)=\Tr(\Delta_kA)
\qquad\forall A,k.
\]
By cyclicity,
\[
\Tr((U\Delta_kU^*-\Delta_k)A)=0\qquad\forall A,k.
\]
Since the trace pairing is nondegenerate, $U\Delta_kU^*=\Delta_k$ for every $k$.  Applying the same condition to the adjoints gives $U\Delta_k^*U^*=\Delta_k^*$.  Hence
\[
\Stab(\D)=\mathcal U(\vN(S)').
\]
Taking commutants inside $M_d(\mathbb C)$ gives $\I(\D)=\vN(S)$.
\end{proof}

\begin{proposition}[Background monotonicity]
If $\B_1\subseteq\B_2$, then
\[
\I_Y^{C;\B_1}\subseteq \I_Y^{C;\B_2}.
\]
\end{proposition}

\begin{proof}
The response-difference family for $\B_1$ is contained in the response-difference family for $\B_2$:
\[
\D^{C;\B_1}_Y\subseteq\D^{C;\B_2}_Y.
\]
A unitary stabilizing the larger family stabilizes the smaller family, so
\[
\Stab(\D^{C;\B_2}_Y)\subseteq\Stab(\D^{C;\B_1}_Y).
\]
Taking commutants reverses inclusion and gives the claim.
\end{proof}

\begin{remark}
The centers need not be monotone under inclusion of von Neumann algebras.  Thus enlarging a background class can increase the influence algebra while either increasing, decreasing, or reorganizing the available central event algebra.  This is why the physical event is the quadruple $(C,Y,\B,P)$ rather than simply $P$.
\end{remark}

\section{Split-record laboratories and Type-III local algebras}\label{app:split}

Local algebras in relativistic QFT are typically Type-III factors.  If $M$ is a factor, then $Z(M)=\mathbb C1$ and
\[
\Proj(Z(M))=\{0,1\}.
\]
Thus nontrivial exact classical records cannot be obtained merely by taking projections in the local field algebra.  The RQCP construction uses split inclusions and record conditional expectations.

\begin{assumption}[Split laboratory]
For a double inclusion of local regions $O\Subset\widetilde O$, the corresponding von Neumann algebras obey a split inclusion
\[
\A(O)\subset\F(O,\widetilde O)\subset\A(\widetilde O),
\]
where $\F$ is a Type-I factor.
\end{assumption}

Since $\F$ is Type I, there exists a Hilbert-space factorization
\[
\Hh\simeq\Hh_{\rm lab}\otimes\Hh_{\rm env}
\]
such that
\[
\F\simeq B(\Hh_{\rm lab})\otimes1_{\rm env}.
\]
However $Z(\F)=\mathbb C1$.  A nontrivial exact record requires a further record algebra.

Let
\[
\Hh_{\rm lab}=\bigoplus_{a\in A}\Hh_a,
\qquad
P_a:\Hh_{\rm lab}\to\Hh_a
\]
be the record decomposition.  Define
\[
E_R(X)=\sum_{a\in A}P_aXP_a.
\]

\begin{proposition}[Record conditional expectation]\label{prop:recordexpectation}
The map $E_R:B(\Hh_{\rm lab})\to B(\Hh_{\rm lab})$ is a normal UCP idempotent conditional expectation onto
\[
\R=\bigoplus_{a\in A}B(\Hh_a).
\]
Moreover
\[
Z(\R)=\left\{\sum_{a\in A}\lambda_aP_a:\lambda_a\in\mathbb C\right\}
\]
and
\[
\Proj(Z(\R))=\left\{\sum_{a\in S}P_a:S\subseteq A\right\}.
\]
\end{proposition}

\begin{proof}
Complete positivity follows from the Kraus form $E_R(X)=\sum_aP_aXP_a$.  It is unital because $\sum_aP_a=1$.  It is idempotent because $P_aP_b=\delta_{ab}P_a$.  The range is precisely the block-diagonal algebra $\oplus_aB(\Hh_a)$.  An element of this direct sum is central exactly when each block is a scalar multiple of the identity on that block, so the center is $\oplus_a\mathbb CP_a$.  Projections in this commutative algebra are exactly sums over subsets of the minimal central projections $P_a$.
\end{proof}

\begin{theorem}[Split-record event theorem]
Assume a split laboratory $\F$ and a record algebra $\R\subset\F$.  If the background strategy class $\B_R$ is record-nondemolition,
\[
\delta(A)=\delta(E_R(A))\qquad\forall\delta\in\D_Y^{C;\B_R},
\]
and record-complete,
\[
\I_Y^{C;\B_R}=\R,
\]
then
\[
\Ev_0(Y|C;\B_R)=\Proj(Z(\R)).
\]
\end{theorem}

\begin{proof}
The event definition gives $\Ev_0=\Proj(Z(\I_Y^{C;\B_R}))$.  By record completeness $\I_Y^{C;\B_R}=\R$, so the result follows from Proposition~\ref{prop:recordexpectation}.
\end{proof}

\section{Multiplicative domains and transfer-channel spectral bounds}\label{app:md}

Let $\Gamma:A\to B$ be a unital completely positive map.  Its multiplicative domain is
\[
\MD(\Gamma)=\{a:\Gamma(a^*a)=\Gamma(a)^*\Gamma(a),\ \Gamma(aa^*)=\Gamma(a)\Gamma(a)^*\}.
\]

\begin{lemma}[Projection preservation]\label{lem:mdprojection}
Let $P=P^*=P^2$.  Then $P\in\MD(\Gamma)$ if and only if $\Gamma(P)$ is a projection.
\end{lemma}

\begin{proof}
Since $P^*P=PP^*=P$, the two multiplicative-domain identities become
\[
\Gamma(P)=\Gamma(P)^*\Gamma(P),\qquad
\Gamma(P)=\Gamma(P)\Gamma(P)^*.
\]
Because positive maps are $*$-preserving on self-adjoint elements, $\Gamma(P)^*=\Gamma(P)$.  Hence the identities are equivalent to $\Gamma(P)^2=\Gamma(P)$.
\end{proof}

\begin{lemma}[Multiplicativity on the multiplicative domain]
If $a\in\MD(\Gamma)$, then for all $b$,
\[
\Gamma(ab)=\Gamma(a)\Gamma(b),
\qquad
\Gamma(ba)=\Gamma(b)\Gamma(a).
\]
\end{lemma}

\begin{proof}
By Stinespring, $\Gamma(x)=V^*\pi(x)V$.  Equality in Kadison's inequality for $a$ implies $\pi(a)V=V\Gamma(a)$.  Thus
\[
\Gamma(ab)=V^*\pi(a)\pi(b)V=\Gamma(a)V^*\pi(b)V=\Gamma(a)\Gamma(b).
\]
The right multiplication identity is analogous.
\end{proof}

\begin{theorem}[Boolean homomorphism under RG]\label{thm:booleanrg}
Assume
\[
Z(\I_{\ell+1})\subset\MD(\Gamma_\ell),
\qquad
\Gamma_\ell(Z(\I_{\ell+1}))\subset Z(\I_\ell).
\]
Then
\[
\Gamma_\ell:\Proj(Z(\I_{\ell+1}))\to\Proj(Z(\I_\ell))
\]
is a Boolean algebra homomorphism.
\end{theorem}

\begin{proof}
For $P,Q\in\Proj(Z(\I_{\ell+1}))$, multiplicativity gives
\[
\Gamma(PQ)=\Gamma(P)\Gamma(Q).
\]
Unitality gives $\Gamma(1-P)=1-\Gamma(P)$, and linearity gives preservation of joins.  Center covariance ensures images are central; Lemma~\ref{lem:mdprojection} ensures images are projections.
\end{proof}

\subsection*{Transfer-channel convergence}
Let $\Gamma:\A\to\A$ be a normal UCP map on a finite-dimensional algebra or on a Banach space of quasi-local observables.  Suppose there is a normal conditional expectation $E_\Z:\A\to\Z$ onto the fixed-point algebra $\Z=\Fix(\Gamma)$ and a spectral gap on $\ker E_\Z$:
\[
r_0=r(\Gamma|_{\ker E_\Z})<1.
\]
For any $\rho$ with $r_0<\rho<1$, the spectral-radius formula gives a constant $C_\rho$ such that
\[
\|\Gamma^k|_{\ker E_\Z}\|\le C_\rho\rho^k.
\]
Therefore
\[
\|\Gamma^k(A)-E_\Z(A)\|\le C_\rho\rho^k\|A-E_\Z(A)\|
\le 2C_\rho\rho^k\|A\|.
\]
Absorbing the factor of $2$ into $C_\rho$ yields the bound used in the main text.

For response weights $w_k(e)=\ell_e(\Gamma^k(B_e))$ one obtains
\[
|w_k(e)-w_\infty(e)|\le C_\rho\rho^k\|\ell_e\|\|B_e\|.
\]
If in a finite window $W$ the limiting weights have margin
\[
m_W=\min_{e\in W}|w_\infty(e)-\tau|>0,
\]
then for
\[
k\ge \frac{\log(2C_\rho L_W/m_W)}{-\log\rho},
\qquad
L_W=\max_{e\in W}\|\ell_e\|\|B_e\|,
\]
the threshold decision $w_k(e)>\tau$ equals $w_\infty(e)>\tau$ for every edge in $W$.  Hence the finite-window directed graph and its transitive closure stabilize.

\section{Araki relative entropy and the causal-modular first law}\label{app:araki}

In finite dimensions, for density matrices $\rho$ and $\sigma$,
\[
S(\rho\|\sigma)=\Tr(\rho\log\rho)-\Tr(\rho\log\sigma).
\]
If $K_\sigma=-\log\sigma$, then
\[
S(\rho\|\sigma)=\Delta\langle K_\sigma\rangle-\Delta S,
\]
where $\Delta S=S(\rho)-S(\sigma)$ and $\Delta\langle K_\sigma\rangle=\Tr((\rho-\sigma)K_\sigma)$.  Let $\rho_s=\sigma+s\dot\rho+O(s^2)$ with $\Tr\dot\rho=0$.  Since relative entropy has a minimum at $s=0$,
\[
\frac{d}{ds}\bigg|_{0}S(\rho_s\|\sigma)=0.
\]
Hence
\[
\delta S=\delta\langle K_\sigma\rangle.
\]

For Type-III local algebras there is no density matrix.  The appropriate object is Araki relative entropy.  Let $M$ be a von Neumann algebra and $\omega,\phi$ faithful normal states.  The relative modular operator $\Delta_{\omega|\phi}$ defines
\[
S_A(\omega\|\phi)=-\langle\Omega_\omega,\log\Delta_{\phi|\omega}\,\Omega_\omega\rangle
\]
when represented in standard form.  The key properties needed in RQCP are positivity, monotonicity under restriction, and differentiability for smooth perturbations in a split-local chart.

\begin{theorem}[Split-local Araki first law]\label{thm:arakifirstlaw}
Let $M_D$ be a local diamond algebra approximated by a nested family of
split Type-I buffers $\F_{D,n}$.  Let $\omega_s$ be a differentiable family
of faithful normal states with $\omega_0$ as reference.  Suppose the
regulated relative-entropies converge in \(C^1\) on a neighborhood of
\(s=0\) to the Araki relative entropy, and suppose the corresponding
regulated modular-energy and entropy differences have finite first
variations.  Then
\[
\frac{d}{ds}\bigg|_0 S_A(\omega_s\|\omega_0)=0
\]
and
\[
\delta S_{\rm bulk}^{\rm alg}=\delta\langle K_D^0\rangle.
\]
\end{theorem}

\begin{proof}
On each Type-I split approximant, the density-matrix derivation applies and
the first variation at the reference state vanishes.  The assumed \(C^1\)
convergence permits the regulator limit to be interchanged with that first
variation.  Taking the limit in the regulated identity
\(S=\Delta\langle K\rangle-\Delta S_{\rm alg}\) gives the stated result.
Without the \(C^1\) hypothesis, convergence of the relative-entropies alone
would not justify differentiating through the split limit.
\end{proof}

\begin{remark}
The remaining nontrivial input is not the first law itself but the local geometric identification
\[
K_D^0\longrightarrow \frac{2\pi}{\hbar}\int_D\xi^\mu T_{\mu\nu}\,d\Sigma^\nu.
\]
In RQCP this is explicitly listed as a local modular-geometric convergence assumption.
\end{remark}

\section{Conditional central-variable construction and fixed-sector
cancellation}\label{app:sequestering}

\textbf{Status: conditional comparison, not a derived cosmological
mechanism.}
Two distinct statements appeared under the earlier term
``vacuum-response sequestering.''  They should be kept separate.

The first is a comparison with global sequestering.  If one \emph{introduces}
a central variable \(\lambda_R\) and the functional
\[
S_{\rm seq}[g,\lambda_R]
=\int_M\sqrt{-g}\left[
\frac{M_{\rm Pl}^2}{2}R-\rho_{\rm vac}^{\rm bare}
-\lambda_R+\mathcal L_{\rm rel}\right]
+\Sigma\left(\frac{\lambda_R}{\mu^4}\right),
\]
then the simultaneous shift
\[
\rho_{\rm vac}^{\rm bare}\mapsto\rho_{\rm vac}^{\rm bare}+\Delta\rho,
\qquad
\lambda_R\mapsto\lambda_R-\Delta\rho
\]
leaves the combination entering the metric equation unchanged.  Variation
with respect to \(\lambda_R\) imposes
\[
\int_M\sqrt{-g}
=\frac{1}{\mu^4}\Sigma'\left(\frac{\lambda_R}{\mu^4}\right).
\]
This algebra is exact once the global variable and its action are supplied.
It does not derive either object from the local response process.

The second statement uses the additive charge center obtained in
Sec.~\ref{sec:loop-center}.  Suppose an additional continuum theorem
identifies a sector label \(v\) with metric four-volume, and define normalized
fixed-sector response by
\[
\mathcal W_v[J]=\log Z_v[J]-\log Z_v[0].
\]
Any source-independent term \(cv\) multiplies numerator and denominator by
the same factor:
\[
\frac{Z_v[J;c]}{Z_v[0;c]}
=\frac{e^{-cv}Z_v[J;0]}{e^{-cv}Z_v[0;0]}
=\frac{Z_v[J;0]}{Z_v[0;0]}.
\]
For a heat-kernel expansion this removes only the pure-volume \(a_0\) term.
Curvature, anomaly, boundary, nonlocal, and source-dependent terms are not
constant on the sector and remain in connected response.

Neither statement presently implies
\[
\rho_\Lambda^{\rm ren}
=\xi M_{\rm Pl}^2L_R^{-2}
+O(L_R^{-4})+O(e^{-L_R/\ell_R}).
\]
That expression is a dimensional infrared ansatz until a common microscopic
model derives the metric-volume map, the central susceptibility, the scale
\(L_R\), and the coupling of the surviving sector to curvature.  The
required missing steps are listed as (V2)--(V5) in
Section~\ref{sec:updated-vacuum}.

\section{Conditional de Sitter CSR-TPN benchmark}\label{app:ds}

\textbf{Status: background-relative consistency construction.}
This appendix starts from a de Sitter diamond and therefore tests
compatibility with that geometry; it does not derive de Sitter space.

Let $(M,g_H)$ be a compact causally convex diamond in four-dimensional de Sitter spacetime,
\[
R_{\mu\nu}=3H^2g_{\mu\nu},\qquad R=12H^2.
\]
Let $a_n=H^{-1}2^{-n}$ and let $V_n$ be a nested Poisson sprinkling of intensity $a_n^{-4}$.  Each point $x\in V_n$ carries a finite split-record Hilbert space
\[
\Hh_x=\Hh_x^S\otimes\Hh_x^R\otimes\Hh_x^E,
\qquad \dim\Hh_x^R=q\ge2.
\]
The finite region algebra is
\[
\A_n(O)=B\left(\bigotimes_{x\in O}\Hh_x\right).
\]
The record algebra is generated by projectors $P_{x,m}$ on $\Hh_x^R$ and the conditional expectation
\[
E_R(A)=\sum_{\mathbf m}P_{\mathbf m}AP_{\mathbf m}.
\]
The causal parent graph is induced from the de Sitter causal order in thin slabs.  Its transitive closure converges to the de Sitter causal order on finite windows, assuming the usual causal-set sprinkling convergence.

The old notation
\[
\sin^2\theta_n(H)\sim \rho_\Lambda(H)a_n^4
\]
is a fixed-point matching coordinate.  If a separate vacuum-response
mechanism supplies
\[
3H^2M_{\rm Pl}^2=\rho_\Lambda^{\rm ren}.
\]
then the benchmark may be evaluated at the resulting \(H\).  No present
microscopic calculation selects either side of this equality.

\begin{theorem}[Conditional de Sitter benchmark]
Assume split-record laboratories, transfer-channel center stability, antichain boundary area capacity, causal-set curvature convergence, local modular geometry, and vacuum-response sequestering.  If the renormalized vacuum response satisfies
\[
\rho_\Lambda^{\rm ren}=3H^2M_{\rm Pl}^2,
\]
then the continuum vacuum sector satisfies
\[
G_{\mu\nu}+3H^2g_{\mu\nu}=0.
\]
\end{theorem}

\begin{proof}
The center-stable RG limit provides Boolean facts and a locally finite causal order.  Sprinkling convergence and information-volume convergence reconstruct $(M,g_H)$, while the causal-set curvature operator gives $R=12H^2$.  The modular first law and entropy equilibrium give Einstein's equation with source equal to the renormalized response stress.  In the normal-ordered vacuum sector the causal stress vanishes and the residual sequestered vacuum density is $3H^2M_{\rm Pl}^2$.  Hence the displayed equation follows.
\end{proof}

Every geometric and gravitational arrow used in this theorem appears among
its hypotheses.  Its purpose is to verify consistency and nonemptiness of
the conditional construction, not to close the order--volume or
modular-gravity bridges.

\section{Gauge holonomy and Yang-Mills stress}\label{app:gauge}

\textbf{Status: conditional continuum construction.}
Given a compact group \(G\), attach a gauge register to each causal edge
$e=(x,y)$,
\[
\Hh_e^G=L^2(G_N)
\]
for a finite cutoff approximation $G_N$ of a compact gauge group $G$.  Link variables transform as
\[
U_{xy}\mapsto g_xU_{xy}g_y^{-1}.
\]
The gauge-invariant algebra is obtained by conditional expectation
\[
E_G(A)=\int_{G_N^{V_n}}\alpha_g(A)\,dg.
\]
For a causal plaquette $p=(x;y,z;w)$ define
\[
U_p=U_{xy}U_{yw}U_{zw}^{-1}U_{xz}^{-1}.
\]
If
\[
U_{xy}=\mathcal P\exp\left(i\int_x^yA_\mu dx^\mu\right),
\]
then the small plaquette expansion is
\[
U_p=\exp\left(iF_{\mu\nu}(m_p)\Sigma_p^{\mu\nu}+O(a_n^3\|\nabla F\|+a_n^4\|F\|^2)\right).
\]
Thus holonomy defects converge to Yang-Mills curvature.

The discrete Yang-Mills action is
\[
S_{{\rm YM},n}=\frac{1}{2g_n^2}\sum_p\operatorname{ReTr}(1-U_p).
\]
For smooth fields and shape-regular causal plaquettes,
\[
S_{{\rm YM},n}\to\frac{1}{4g^2}\int_M\sqrt{-g}\,\operatorname{tr}(F_{\mu\nu}F^{\mu\nu})\,d^4x.
\]
Varying with respect to the metric yields
\[
T_{\mu\nu}^{\rm YM}=\frac{1}{g^2}\operatorname{tr}\left(F_{\mu\alpha}F_\nu{}^\alpha-\frac14g_{\mu\nu}F_{\alpha\beta}F^{\alpha\beta}\right).
\]
This calculation recovers Yang--Mills response from a supplied gauge
register.  It does not generate \(G\), the link Hilbert spaces, or their
fusion rules.

\section{Domain-wall and spectral-triple fermions}\label{app:fermions}

\textbf{Status: conditional matter constructions.}
Naive local lattice fermions in four dimensions face the doubling problem.  RQCP therefore treats chiral matter through either a higher-dimensional domain-wall construction or a noncommutative spectral triple.

\subsection*{Domain-wall route}
Let $\mathcal B^5=M^4\times[-L,L]$ and define
\[
D_5=\gamma^\mu(\nabla_\mu+iA_\mu)+\gamma^5\partial_s+M(s),
\qquad
M(s)=m\operatorname{sign}(s).
\]
Zero modes localized at $s=0$ have profile
\[
\psi(x,s)=\exp\left(-\int_0^sM(s')ds'\right)\psi_L(x)
\]
and chirality $\gamma^5\psi_L=-\psi_L$.  The number of boundary zero modes is controlled by the index of the four-dimensional Dirac operator.

For the toy model $M^4=S^2\times S^2$, let $\alpha,\beta$ be the fundamental two-classes and take a line bundle $L_{p,q}$ with
\[
c_1(L_{p,q})=p\alpha+q\beta.
\]
Since $\widehat A(S^2\times S^2)=1$,
\[
\Index(D_{L_{p,q}})=\frac12\int c_1(L_{p,q})^2=pq.
\]
The explicit choice $(p,q)=(1,3)$ gives
\[
\Index(D_{L_{1,3}})=3.
\]
Thus an index-three defect background supports three chiral boundary zero
modes.  The multiplicity is already encoded in the chosen Chern numbers;
this example does not explain why the dynamics should select the value
three.

\subsection*{Spectral-triple route}
A finite internal noncommutative geometry is specified by
\[
(\A_F,\Hh_F,D_F,J_F,\gamma_F).
\]
The full spectral triple is
\[
\A=C^\infty(M)\otimes\A_F,
\qquad
D_A=D_M\otimes1+\gamma_5\otimes D_F+A+JAJ^{-1}.
\]
Generation number is the index pairing
\[
Q_{\rm gen}=\langle\ch(D_A),[p]\rangle.
\]
In this route generations are $K$-theoretic invariants, not lattice copies.

\section{Response-superfluid dark branch}\label{app:superfluid}

\textbf{Status: phenomenological ansatz.}
Let $N_R=\sum_xn_x$ be a conserved response charge and $B_x$ a lowering operator,
\[
[N_R,B_x]=-B_x.
\]
If response charge condenses,
\[
\omega(B_x^*B_y)\to n_Re^{i(\varphi(y)-\varphi(x))},
\]
then the long-wavelength field is
\[
\Psi_R=\sqrt{n_R}e^{i\varphi}.
\]
The thermodynamic pressure is defined directly from the process functional:
\[
Z_{n,\Lambda}(\mu)=\Omega_{n,\Lambda}(e^{\beta\mu N_R}),
\qquad
P_R(\mu)=\lim_{\Lambda\nearrow\infty}\lim_{n\to\infty}\frac{1}{\beta|\Lambda|}\log Z_{n,\Lambda}(\mu).
\]
In the nonrelativistic, weak-field, slow-phase limit, the local chemical-potential combination is
\[
X=\dot\varphi-m\Phi-\frac{|\nabla\varphi|^2}{2m}.
\]
The hydrodynamic action is
\[
S_R^{\rm hyd}=\int dtd^3x\,P_R(X)-\alpha\int dtd^3x\,\varphi\rho_b+\cdots.
\]
If the infrared response fixed point has
\[
P_R(X)=\kappa X\sqrt{|X|},
\]
then the phonon equation in a static spherical galaxy is
\[
\nabla\cdot(P_R'(X)\nabla\varphi)=\alpha\rho_b.
\]
Outside the baryonic mass $M_b$, this gives
\[
r^2P_R'(X)\varphi'=\frac{\alpha M_b}{4\pi}.
\]
With $P_R'(X)\propto\sqrt{|X|}$ and $X\simeq-|\nabla\varphi|^2/(2m)$, one obtains
\[
\varphi'\propto \frac{\sqrt{M_b}}{r}.
\]
The baryonic force mediated by the phonon is therefore
\[
g_R\propto\frac{\sqrt{GM_ba_R}}{r},
\]
which yields
\[
v_c^4=GM_ba_R.
\]
Whether the same parameters yield a normal-fluid phase compatible with
cluster lensing, the CMB, and structure growth is an open cross-scale test.
The BTFR scaling alone does not establish the microscopic equation of state.

\section{Phenomenology and falsifiability}\label{app:phenom}

In the weak-field Newtonian gauge,
\[
ds^2=-(1+2\Phi)dt^2+a(t)^2(1-2\Psi)d\mathbf x^2.
\]
RQCP deviations can be parameterized by
\[
-k^2\Psi=4\pi Ga^2\mu_{\rm eff}(k,z)\rho_m\Delta_m,
\]
\[
\eta_{\rm slip}(k,z)=\frac{\Phi}{\Psi},
\]
and
\[
w_{\rm RQCP}(a)=-1-\frac13\frac{d\log\rho_{\rm resp}}{d\log a}.
\]
A minimal transfer-spectrum motivated expansion is
\[
\mu_{\rm eff}(k,z)=1+\sum_iA_i(1+z)^{-p_i}\frac{(k/k_i)^{\sigma_i}}{1+(k/k_i)^{\sigma_i}},
\]
\[
\eta_{\rm slip}(k,z)=1+\sum_jB_j(1+z)^{-q_j}\frac{(k/k_j)^{\tau_j}}{1+(k/k_j)^{\tau_j}}.
\]
The parameters are not arbitrary: $k_i^{-1}$ are response correlation lengths, $p_i,q_i$ are related to transfer-channel critical exponents, and $A_i,B_j$ measure the coupling of the corresponding response operators to baryonic and lensing observables.

\subsection*{Failure modes}
The framework becomes empirically untenable if any of the following occur:
\begin{enumerate}[label=(F\arabic*)]
\item no CSR-TPN model has a split-record RG fixed point with nontrivial center;
\item the response transfer channel lacks a spectral gap, so exact records fail to stabilize;
\item no response-charge condensate exists in viable 3+1D models;
\item the infrared response equation of state never approaches $P_R(X)\propto X\sqrt{|X|}$ or an observationally equivalent law;
\item the sequestering constraint fails to remove UV identity vacuum response from local curvature;
\item the predicted $\mu_{\rm eff}$ and $\eta_{\rm slip}$ cannot fit rotation curves, lensing, clusters, CMB, and structure growth simultaneously;
\item no anomaly-free chiral defect/spectral triple sector reproduces the observed gauge and fermion representations.
\end{enumerate}

\section{Choi--Effros infrared algebras and asymptotic Boolean records}\label{app:v62-choi-effros}

This appendix outlines the operator-algebraic derivation replacing the preliminary assumption of split-record laboratories.  Let $\A$ be a von Neumann algebra and let $\Gamma:\A\to\A$ be a normal UCP map preserving a faithful normal state $\omega$.

\subsection*{Mean ergodic projection}

Assume the Cesaro averages
\[
E_N(A)=\frac1N\sum_{k=0}^{N-1}\Gamma^k(A)
\]
converge ultraweakly to a normal UCP idempotent $E_\infty$.  Then
\[
E_\infty^2=E_\infty,
\qquad
\Gamma E_\infty=E_\infty\Gamma=E_\infty.
\]
The asymptotic range
\[
\A_{\rm IR}:=E_\infty(\A)
\]
is an operator system.  Choi--Effros theory
\cite{ChoiEffros1976,ChoiEffros1977} equips it with the product
\[
X\circ Y=E_\infty(XY),
\]
turning it into a $C^*$-algebra.

\begin{remark}
The product $\circ$ need not coincide with the ambient product.  Moreover, Choi--Effros alone does not imply commutativity.  Therefore it does not by itself prove emergence of classical facts.
\end{remark}

\subsection*{Asymptotic abelianness}

Define the GNS seminorm
\[
\|A\|_\omega=\omega(A^*A)^{1/2}.
\]
A coarse-graining flow is asymptotically abelian on its invariant range if
\[
\lim_{n\to\infty}\|[\Gamma^n(A),\Gamma^n(B)]\|_\omega=0
\]
for all \(A,B\in E_\infty(\A)\).  Since these elements are fixed by
\(\Gamma\), this is an explicit commutativity condition on what survives.
Decay established only for pairs of microscopic local observables is not
sufficient unless it extends continuously to the invariant range.

\begin{theorem}[Asymptotically abelian Choi--Effros range]
Let $E_\infty$ be the normal ultraweak limit of Cesaro averages of $\Gamma$,
and assume the invariant-range condition above.  Then the
Choi--Effros algebra \(\A_{\rm IR}\) is commutative.  For any faithful normal
representation \(\pi\),
\[
[\pi(X),\pi(Y)]=0,
\qquad X,Y\in\A_{\rm IR}.
\]
Consequently \(\pi(\A_{\rm IR})''\) is a commutative von Neumann algebra and
its projections form a complete Boolean algebra.
\end{theorem}

\begin{proof}
For \(X,Y\in\A_{\rm IR}\), mean ergodicity gives
\(\Gamma^n(X)=X\) and \(\Gamma^n(Y)=Y\).  The assumed limit therefore gives
\(\|[X,Y]\|_\omega=0\).  Faithfulness of \(\omega\) implies
\([X,Y]=0\).  Hence
\(E_\infty(XY)=E_\infty(YX)\), so the Choi--Effros product is commutative.
Faithful representation and von Neumann closure preserve commutativity, and
the resulting projection lattice is complete Boolean.
\end{proof}

\subsection*{Finite split-record approximation}

The split-record laboratory
\[
\R=\bigoplus_aB(\Hh_a)
\]
is a finite-stage approximation to the infrared fixed algebra before quotienting by asymptotically invisible intra-sector degrees of freedom.  Its center
\[
Z(\R)=\bigoplus_a\mathbb C P_a
\]
is the finite-dimensional shadow of the commutative Choi--Effros infrared algebra.

\section{Conditional central sectors and the response-horizon
ansatz}\label{app:v62-superselection-free-energy}

\textbf{Status: conditional.}
This appendix records the algebra of the earlier proposal while exposing the
inputs that are not generated by the current local channel.

\subsection*{Central response sector}

Assume that \(Z_R\) is a central element of the global IR representation:
\[
Z_R\in Z(\pi_\omega(\A_{\rm glob})'').
\]
For any local observable $A_{\rm loc}$,
\[
[A_{\rm loc},Z_R]=0.
\]
Then $Z_R$ labels superselection sectors.  A state decomposes as
\[
\omega=\int^\oplus \omega_z\,d\mu(z),
\]
where $z$ is the response-center label.  Operations belonging to the
block-diagonal observable algebra cannot detect cross-sector coherence.
The existence of this particular global center is an input here.

\subsection*{Vacuum identity response as a central gauge mode}

The bare identity vacuum response contributes
\[
T_{\mu\nu}^{\rm bare,vac}=-\rho_{\rm vac}^{\rm bare}g_{\mu\nu}.
\]
If the local response calculus is insensitive to the central label,
introduce a compensator $\lambda_R(Z_R)$ and define a sector-invariant
combination
\[
\rho_\Lambda^{\rm loc}=\rho_{\rm vac}^{\rm bare}+\lambda_R+\rho_{\rm rel}.
\]
The transformation
\[
\rho_{\rm vac}^{\rm bare}\mapsto \rho_{\rm vac}^{\rm bare}+\Delta\rho,
\qquad
\lambda_R\mapsto\lambda_R-\Delta\rho
\]
leaves response differences invariant by construction.  A gravitational
source defined on that quotient would depend only on the residual response
density.  Neither the compensator nor the quotient-to-curvature rule follows
from centrality alone.

\subsection*{Renormalized residual}

The earlier parametrization wrote
\[
\rho_\Lambda^{\rm ren}=\rho_{\rm IR}+\rho_{\rm bdry}+\rho_{\rm anomaly}.
\]
If the only macroscopic response scale is $L_R$, covariance and dimensional
analysis permit
\[
\rho_\Lambda^{\rm ren}=\xi M_{\rm Pl}^2L_R^{-2}+O(L_R^{-4})+O(e^{-L_R/\ell_R}).
\]
This scaling is not a calculation of the residual.  It becomes a prediction
only after deriving the central stiffness, the response length, and their
coupling to curvature in the same microscopic model.

\subsection*{Entanglement free-energy selection}

Define
\[
F_R(L)=E_{\rm mod}(L)-T_HS_{\rm gen}(L),
\qquad T_H=\frac{\hbar H}{2\pi}.
\]
The phenomenological late-time ansatz is
\[
E_{\rm mod}(L)=\alpha M_{\rm Pl}^2H^2L^3+O(L^2),
\]
and
\[
T_HS_{\rm gen}(L)=\beta M_{\rm Pl}^2HL^2+O(L).
\]
The leading free energy is
\[
F_R(L)=\alpha M_{\rm Pl}^2H^2L^3-\beta M_{\rm Pl}^2HL^2+\cdots.
\]
Thus
\[
\frac{\partial F_R}{\partial L}=M_{\rm Pl}^2HL(3\alpha HL-2\beta)+\cdots.
\]
The nonzero stationary point is
\[
L_R=\frac{2\beta}{3\alpha}H^{-1}+\cdots.
\]
Imposing first-law matching at \(L=H^{-1}\) sets
\(2\beta/3\alpha=1\), and therefore reproduces
\[
L_R=H^{-1}+\cdots,
\qquad
\rho_\Lambda^{\rm ren}\simeq \xi M_{\rm Pl}^2H^2.
\]

\begin{remark}
Because the matching point is the radius subsequently recovered, these
equations are a consistency condition, not a derivation of a horizon
attractor.  A closed calculation must determine \(\alpha\) and \(\beta\)
before extremizing, verify \(F_R''(L_R)>0\), and compute the relaxation rate.
\end{remark}

\section{Residual influence algebras and finite spectral
triples}\label{app:v62-ncg}

\textbf{Status: conditional classification interface.}
This appendix describes how a separately derived residual matter algebra
could be tested against finite noncommutative geometry.

\subsection*{Residual noncommutative influence algebra}

Let $\I_{\rm IR}$ be the infrared influence algebra after CP coarse-graining.  Its center gives classical facts.  The residual internal algebra is the noncentral finite part
\[
\I_{\rm rel}:=\I_{\rm IR}/Z(\I_{\rm IR})
\]
restricted to finite-energy localized sectors.  Assume it has a finite semisimple real form
\[
\A_F=\bigoplus_iM_{n_i}(\mathbb K_i),
\qquad
\mathbb K_i\in\{\mathbb R,\mathbb C,\mathbb H\}.
\]

\subsection*{Finite real spectral triple}

A finite real spectral triple consists of
\[
(\A_F,\Hh_F,D_F,J_F,\gamma_F),
\]
where $\A_F$ acts faithfully on $\Hh_F$, $D_F$ is self-adjoint, $J_F$ is a real structure, and $\gamma_F$ is a grading.  The product with spacetime gives
\[
\A=C^\infty(M)\otimes\A_F,
\qquad
\Hh=L^2(M,S)\otimes\Hh_F,
\]
\[
D=D_M\otimes1+\gamma_5\otimes D_F.
\]
Inner fluctuations
\[
D_A=D+A+JAJ^{-1}
\]
produce gauge fields and scalar/Higgs-like components.

\subsection*{Standard-Model-like algebra}

Under the particular finite-spectral-triple representation and consistency
conditions used in the Standard-Model construction---including a real
structure, the first-order condition, orientability, Poincare duality,
unimodularity, anomaly cancellation, and a minimal faithful nontrivial
charge representation---the familiar minimal candidate is
\cite{ConnesMarcolli,ChamseddineConnes}
\[
\A_F\simeq \mathbb C\oplus\mathbb H\oplus M_3(\mathbb C).
\]
Its unitary group modulo the unimodular condition gives
\[
G_F\simeq S(U(2)\times U(3))
\simeq \frac{SU(3)\times SU(2)\times U(1)}{\mathbb Z_6}.
\]
This is a classification of an assumed spectral triple.  It does not replace
the missing RQCP theorem that must first generate the residual algebra and
its representation.

\subsection*{Generation number}

The generation number is the index pairing
\[
N_{\rm gen}=\Index(D_A,p)=\langle \ch(D_A),[p]\rangle.
\]
A three-generation model is one in which the residual RQCP spectral triple contains a class $[p]\in K_0(\A_F)$ satisfying
\[
\langle \ch(D_A),[p]\rangle=3.
\]
This does not prove that RQCP dynamics selects three generations.  Once the
class \([p]\) is chosen, the multiplicity is topologically stable; the value
three remains encoded in the assumed index pairing.

\subsection*{Spectral action and coupling to gravity}

The spectral action has the form
\[
\operatorname{Tr} f(D_A/\Lambda_s)+\langle J\Psi,D_A\Psi\rangle.
\]
Its asymptotic expansion contains Einstein--Hilbert, Yang--Mills, scalar, and
Dirac/Yukawa terms under the standard spectral-action hypotheses.  In the
present program this is a candidate continuum functional, not a result
derived from the Boolean or causal-order channels.

\section{Microscopic mechanisms for asymptotic abelianness}\label{app:v62-asymptotic-origin}

This appendix spells out the physical content behind the
asymptotic-abelian condition.  A candidate coarse-graining channel is
expected to be quasi-local, to have a finite information-propagation speed,
and to discard phases inaccessible to long-wavelength probes.
Lieb--Robinson bounds motivate the quasi-local part
\cite{LiebRobinson1972}, while ETH and scrambling motivate intra-sector
mixing \cite{Deutsch1991,Srednicki1994}.  Neither body of results, without
additional hypotheses on the channel and observables, proves the
block-primitive estimates used here.

Mathematically, this motivates the block-primitive estimates used in Theorem~\ref{thm:block-primitive-abelian}.  If the invariant projections $P_a$ label macroscopic response sectors, the channel has the asymptotic form
\[
\Gamma^n(X)=\sum_a \omega_a(P_aXP_a)P_a+R_n(X),
\qquad
\|R_n(X)\|_\omega\le C_Xe^{-\gamma n}.
\]
The term $R_n$ contains both off-sector coherences and nonthermal intra-sector fluctuations.  Consequently,
\[
[\Gamma^n(A),\Gamma^n(B)]=[R_n(A),R_n(B)]+[R_n(A),X_B^{\rm cl}]+[X_A^{\rm cl},R_n(B)],
\]
where $X_A^{\rm cl}=\sum_a\omega_a(P_aAP_a)P_a$.  The right side is
$O(e^{-\gamma n})$ in the GNS norm.  Thus the block-primitive estimates give
a concrete sufficient route to Boolean IR records.  Deriving those estimates
from a generic local Hamiltonian remains open; the dephasing--exchange model
realizes them directly.

\begin{remark}[Nonabelian exceptions]
If a channel has a protected nonabelian noiseless subsystem, the IR algebra is not Boolean.  RQCP-QG then treats that subsystem as a residual quantum or topological sector, not as a classical record.  This is compatible with the matter-sector interpretation: classical spacetime facts are associated with the abelian center, while noncommutative residual algebras may encode gauge and internal degrees of freedom.
\end{remark}

\section{Causal-diamond coefficient matching}\label{app:v62-coefficients}

The equality $2\beta/3\alpha=1$ is a first-law matching condition between the modular-volume term and the boundary area-capacity term.  Let
\[
E_{\rm mod}(L)=\alpha M_{\rm Pl}^2H^2L^3+O(L^2\partial H,H^4L^5),
\]
\[
T_HS_{\rm gen}(L)=\beta M_{\rm Pl}^2HL^2+O(L\partial H,H^3L^4).
\]
If causal-diamond first-law matching is imposed at the de Sitter radius, one
sets
\[
\delta E_{\rm mod}=T_H\delta S_{\rm gen}
\]
for first-order radial variations at $L=H^{-1}$ in the late-time quasi-de Sitter regime.  Substituting the two scaling laws yields
\[
3\alpha M_{\rm Pl}^2H^2L^2\delta L
=
2\beta M_{\rm Pl}^2HL\delta L
\]
at $L=H^{-1}$, hence $3\alpha=2\beta$.  The coefficient relation is
therefore equivalent to imposing the Bekenstein--Hawking area normalization
and Jacobson-type entanglement equilibrium at the target radius.  It is not
an independent calculation of that radius.

In an explicit CSR-TPN model, $\alpha$ should be computed from the modular response density and $\beta$ from the boundary record capacity.  Agreement with $3\alpha=2\beta$ is a nontrivial consistency check.  Failure of the relation does not invalidate the algebraic core of RQCP-QG, but it invalidates the specific claim that the response horizon locks exactly to $H^{-1}$.

\section{Conditional status of the finite spectral-triple completion}\label{app:v62-ncg-status}

The NCG completion is a conditional bridge from residual influence algebras to Standard-Model-like matter.  The input is not merely a finite-dimensional algebra; it is a finite real spectral triple
\[
(\A_F,\Hh_F,D_F,J_F,\gamma_F)
\]
with reality, grading, first-order condition, orientability, Poincare duality, unimodularity, anomaly cancellation, and minimal faithful charge representation.  These constraints are strong.  RQCP-QG proposes that the residual noncommutative influence algebra is the physical source of $\A_F$, but it does not yet prove the dynamical selection of this algebra from arbitrary microscopic data.

The relevant result is an external classification input:
\begin{external}[Finite-spectral-triple classification
input]\label{thm:conditional-ncg-selection}
Within the specific representation, irreducibility, KO-dimension, reality,
first-order, orientability, Poincare-duality, unimodularity, and anomaly
conditions of the cited finite-geometry construction
\cite{ConnesMarcolli,ChamseddineConnes}, the familiar minimal
Standard-Model-like internal algebra is
\[
\A_F\simeq \mathbb C\oplus\mathbb H\oplus M_3(\mathbb C),
\]
with gauge group
\[
S(U(2)\times U(3))\simeq (SU(3)\times SU(2)\times U(1))/\mathbb Z_6.
\]
For a supplied \(K\)-theory class, generation multiplicity can be represented
by an index pairing.
\end{external}

This classification is not a proof that the Standard Model is inevitable
from RQCP.  The missing theorem must derive the residual sector category,
its representation, and the relevant index data from a local process before
the external classification can be applied.

\section{Logical dependencies of the conditional synthesis}
\label{app:synthesis}

This appendix records the origin and scope of the conclusions used in the
conditional synthesis and in the finite regulator results.  The list is not
an additional proof.  Its purpose is to prevent statements established in
different models from being combined without their interface hypotheses.
\begin{enumerate}[label=(D\arabic*)]
\item Boolean events follow abstractly from the center-projection theorem.
The explicit collision and dephasing--exchange models prove that a nontrivial
local dynamics can realize the required center and gap.
\item A transfer-channel gap and threshold margin freeze an already specified
finite-window edge set.  Acyclicity is obtained separately from defect
extinction in the block-cactus model.  In the affine bridge, plaquette Hodge
loss instead produces an exact response one-form and its reconstructed rank.
\item The relative-entropy first law is algebraic.  Stress-energy, boost
geometry, area response, several null directions, and tensor completion are
supplied in the affine regulator by decoupled wires, global exchange
stationarity, a product-reference matching condition, discrete focusing, the
finite response frame, conservation, and an exact lattice Bianchi identity.
The nonlinear regulator replaces the first-law and linear-focusing
truncations by exact finite information and Riccati identities and extends
the null-frame rank theorem to arbitrary invertible coframes.  It does not
derive the remaining finite Ward relation, interacting modular theorem, or
common holonomy/stress update.  The separate capacity--complement regulator
selects the reference used by this balance for factorized tokens, but does
not derive the record--area map or couple the selection to that same
nonlinear update.
\item The de Sitter sprinkling appendix is a background-relative benchmark;
the Ramsey rank-defect model is an inverse theorem in a specified
field--state family.  The affine rank-nine metric is a distinct flat
finite-frame theorem.  The generic curved order--volume limit is instead the
intrinsic volume-clock and strong-profile compactness theorem of
Sec.~\ref{sec:generic-order-volume}, conditional on its displayed operational
gates.  The unified-continuum theorem then proves, for one FLRW--TT
refinement family, that this order--volume branch and the independently
evaluated action-phase branch converge on the same reconstructed geometry;
dimension, action, capacity normalization, and \(G\) remain inputs.
\item The same-link branch uses a stable anonymous constructor to produce an
unlabeled member of a supplied \(P_L\times T_L^3\) language.  Directed order,
integer record volume, the scalar operator, the induced TT kernel, and the
large-\(N_s\) TT measure use that one support.  The positive coefficient
comes from the scalar determinant rather than a bare curvature phase.
Dimension, program/clock, volume unit, matter content, fixed spacing, and
the restricted metric sector remain inputs.
\item The relational-history branch replaces the external update parameter
by a stationary finite clock correlation and proves common causal-rank and
record-volume refinement.  The clock boundary, program, port language, and
volume unit remain inputs.  The quadratic completion branch uniquely extends
the same TT normalization to all ten tensor components and supplies its
finite Ward/ghost/BRST complex after the linear gauge law is declared.  The
spectral-balance branch removes the Newton cutoff within its regulator class
after two signed supertrace moments are imposed.  The positive-Hilbert
companion derives those moments from exterior parity and an affine internal
mass ladder geometry by geometry in every supplied finite graph ensemble,
and proves exact finite graph-move response plus an obstruction to exact
all-level pairing.  The selective-Ward companion then protects the moments
for arbitrary finite positive interacting physical operators and realizes a
positive geometry channel on a supplied finite language.  The
projective-growth companion further constructs compatible cylinder laws and
an infinite positive measure on a supplied naturally labelled causal-birth
language.  The covariant-growth companion then maps response-generated
	contact weights to an exactly discretely covariant, Bell-causal classical
	sequential-growth law and constructs an extendible strongly positive
	rank-one coherent history measure.  The fixed-strip companion separately
	constructs an extendible local-unitary noncommuting rank-four history
	measure.  The arbitrary-causet holonomy companion further combines that
	classical CSG density with intrinsic fresh-event noncommuting gates and
	proves a relabeling-covariant countably additive operator measure with a
	non-Abelian integrated range.  Their composition is a
	regulator-independent linear Gaussian response with finite interacting
	protection, an exact positive history law, and a covariant coherent
	subfamily, together with a distinct covariant non-Abelian
	operator-density branch.  It is not a derivation of the protected commutant, the
	maximal-birth language, contact motifs, response multiplet, coherent arity
	or phase, an operator-Markov or square-operator Bell law, quantum
	geometry backreaction, a tight Lorentzian
continuum, spin--statistics or anomaly freedom, nonlinear BV symmetry, or a
nonperturbative quantum metric measure.
\item Fixed-sector normalization cancels a source-independent pure-volume
factor if charge has already been identified with metric volume.  No
cosmological residual or response horizon follows from that identity.
\item Gauge registers, domain-wall data, finite spectral triples, and the
response-superfluid equation of state are conditional constructions or
phenomenological ansatzes, not outputs of the foundational channel.
\end{enumerate}

Accordingly, no equation in these appendices is used to claim that the
current program autonomously derives generic quantum Einstein dynamics, the
observed cosmological constant, the Standard Model, three generations, or
dark matter.  Einstein response is reached only in the declared
model-specific affine, same-update, unified-continuum, and induced-TT
closures under their visible inputs.  The appendices otherwise specify the mathematical
interfaces that future closed models must replace.

\newpage
\bibliographystyle{unsrturl}
\bibliography{references_v2}

\end{document}